\documentclass{report}
\usepackage[utf8]{inputenc}
\usepackage[margin=3cm]{geometry}
\usepackage{url}
\usepackage{graphicx}
\usepackage{comment}
\usepackage{hyperref}
\usepackage{bbm}
\usepackage{bigints}
\usepackage{amsthm}
\usepackage{amsfonts}
\usepackage[ruled,vlined,lined,boxed,commentsnumbered]{algorithm2e}
\usepackage{setspace}
\usepackage{mathtools}
\usepackage{titling}
\usepackage{stackrel}
\usepackage[nottoc,numbib]{tocbibind}
\usepackage[edges]{forest}
\usepackage[makeroom]{cancel}
\usepackage{xcolor,colortbl}
\usepackage{exercise}

\newcommand{\chapquote}[3]{\begin{quotation} \textit{#1} \end{quotation} \begin{flushright} - #2, \textit{#3}\end{flushright} }

\usepackage{listings}
\newtheorem{theorem}{Theorem}[section]
\newtheorem{corollary}[theorem]{Corollary}
\newtheorem{lemma}[theorem]{Lemma}
\newtheorem{definition}[theorem]{Definition}
\newtheorem{example}[theorem]{Example}
\newtheorem{remark}[theorem]{Remark}

\renewcommand{\arraystretch}{1.5}
\newcommand\tab[1][1cm]{\hspace*{#1}}
\newcommand{\gr}{\textcolor{gray}}

\title{\Huge CM0622 - Algorithms for Massive Data}
\author{\Large
Nicola Prezza\\
Ca' Foscari University of Venice
}
\date{\vspace{50pt} Version: \today}

\begin{document}

\begin{titlingpage}
\maketitle
\end{titlingpage}

\newpage

\tableofcontents

\newpage

\section*{Overview of the course}\label{sec:overview}

The goal of this course is to introduce algorithmic techniques for dealing with \emph{massive data}: data so large that it does not fit in the computer's memory. 
In this course, massive data will be modeled using two fundamental combinatorial objects: sequences (strings) and sets of integers. Solutions will be divided into two broad categories, covered in the two modules of the course: (lossless) \textbf{compressed data structures} and (lossy) \textbf{data sketches}. 

\vspace{10pt}

\subsection*{Module 1 - Compressed data structures}

A \emph{compressed data structure} supports fast queries on the data and uses a space proportional to the compressed data. This solution is typically \emph{lossless}: the representation allows to fully reconstruct the original data. Here we exploit the fact that, in some applications, data is extremely redundant and can be considerably reduced in size (even by orders of magnitude) without losing any information by exploiting this redundancy. Interestingly it turns out that, usually, computation on compressed data can be performed \emph{faster} than on the original (uncompressed) data: the field of compressed data structures exploits the natural duality between compression and computation to achieve this goal. 

The main results we will discuss in Chapter \ref{chapter:cds}, \textbf{compressed dictionaries} (solving various queries on sets of integers) and \textbf{compressed text indexes} (solving indexed pattern matching on strings), find important applications in information retrieval. They allow us to turn large collection of documents/integer sets into a compressed data structure that allows locating  substrings/finding integers very quickly, without decompressing the data. These techniques today stand at the core of modern search engines and sequence-mapping algorithms in computational biology. 

Importantly, lossless techniques cannot break the information-theoretic lower bound for representing the input. For example, since the number of subsets of cardinality $m$ of $\{1,\dots,n\}$ is $\binom{n}{m}$, the information-theoretic lower bound for storing such a subset is $\log_2\binom{n}{m} = m\log(n/m) + O(m)$ bits. No lossless data structure can use asymptotically less space than this bound for \emph{all} such subsets. 

\subsection*{Module 2 - Data sketches}

In the second module  we will resort to lossy compression and throw away some features of the data in order to dramatically reduce its size, usually breaking the information-theoretic lower bound. This feat will be achieved through randomization (queries may fail with small probability) and approximation (queries may return an approximate result).

In Chapter \ref{ch:filters} we show that any set of cardinality $m$ can be stored with a \textbf{filter data structure} in just $O(m)$ bits, provided that we accept a small probability that membership queries fail. Note that this breaks the information-theoretic lower bound discussed above, $m\log(n/m) + O(m)$ bits.

In Chapter \ref{ch:sketch-stream} we see how to shrink even more the size of our data while still being able to compute some useful information on it. The most important concept here is \textbf{sketching}: our input data is squeezed into \emph{sub-linear} representations called \emph{data sketch}. For example, a set of cardinality $m$ over a universe of cardinality $n$ can be stored in just $O(\mathrm{polylog}(m) + \mathrm{polylog}(n))$ bits (exponential compression!) using a data sketch supporting useful queries such as set similarity and cardinality estimation (approximately and with a bounded probability of error).
Sketches can be computed off-line to reduce the dataset's size and/or speed up the computation of distances, or on-line on \textbf{data streams}, where data is thrown away as soon as it arrives (only data sketches are kept in main memory).

\paragraph{Sources}

The proofs of these notes have been put together from various sources, reported in the bibliography. 
The notes collect also material from Leskovec et al.'s book \cite{leskovec2020mining} (\emph{Mining of massive data sets}), Amit Chakrabarti's notes on data stream algorithms \cite{ChakrabartiNotes}, Gonzalo Navarro's book \cite{navarro2016compact} (compressed data structures), and Demetrescu and Finocchi's chapter on algorithms for data streams from the book \cite{stojmenovic2007handbook}.

\chapter{Compressed Data Structures}\label{chapter:cds}

In this chapter we study data structures for representing strings and sets of integers. The goal is to design data structures using a space proportional to the compressed input. We will see two main notions of compressibility: worst-case entropy (used for sets of integers) and empirical entropy (used for string data structures). First, in Section \ref{sec: full-text indexing} we start with the pattern matching problem and give some classical uncompressed solutions.

\section{Full-text indexing - classical solutions}\label{sec: full-text indexing}

In this section we work with strings $\mathcal T \in \Sigma^n$ of length $n$ over alphabet $\Sigma$ of size $\sigma$. 

We consider the following problem:

\begin{definition}[Full-text indexing]
    Given a string $\mathcal T \in \Sigma^n$ (a text), build a data structure $\mathcal I(\mathcal T)$ (an index) solving the following queries:
    \begin{enumerate}
        \item $\mathtt{locate}(P)$: given a string $P\in \Sigma^m$, with $m\leq n$, return the exact occurrences of $P$ in $\mathcal T$, i.e. $\mathtt{locate}(P) = \{i\ :\ \mathcal T[i,i+m-1]=P\}$.
        \item $\mathtt{count}(P)$: given a string $P\in \Sigma^m$, with $m\leq n$, return the number $occ$ of exact occurrences of $P$ in $\mathcal T$, i.e. $\mathtt{count}(P) = occ = |\mathtt{locate(P)}|$.
        
    \end{enumerate}
\end{definition}

If we store the string $\mathcal T$ along with the index $\mathcal I(\mathcal T)$, then it's immediate that we can also access any character $\mathcal T[i]$ in constant time. In Sections \ref{sec:CSA} and \ref{sec:FMI}, however, we are going to present \emph{compressed indexes}: the bit-size of $\mathcal I(\mathcal T)$ is close to the size of the \emph{compressed} $\mathcal T$. It follows that we cannot afford to store $\mathcal T$ along with $\mathcal I(\mathcal T)$, because the former could be much larger than the latter! In such cases, it will be useful that $\mathcal I(\mathcal T)$ also provides \emph{random access} to the characters of $\mathcal T$. More formally, we want that $\mathcal I(\mathcal T)$ supports this extra query:

\begin{enumerate}
  \setcounter{enumi}{2}
    \item $\mathtt{extract}(i,\ell)$: return string $\mathcal T[i,\dots ,i + \ell -1]$
\end{enumerate}

An index supporting query \emph{extract} is called \emph{self-index}, to reflect that it is a replacement of the underlying string:

\begin{definition}
    If $\mathcal I(\mathcal T)$ supports queries \emph{count} and \emph{locate}, it is called a \emph{full-text index}. If $\mathcal I(\mathcal T)$ supports all three queries \emph{count}, \emph{locate}, and  \emph{extract}, it is called a \emph{full-text self-index}. 
\end{definition}

\subsection{Inverted indexes and suffix trie}

An \emph{inverted index} is a classical data structure solving a very restricted version of text indexing problem. If the text $\mathcal T$ is formed by words --- for example $\mathcal T$ = ``to be or not to be" --- then we may simply record the positions (counting from 1) where each word appears in $\mathcal T$, using a dictionary data structure $H$ (for example, a hash table): $H[to] = \{1, 14\}, H[be] = \{4,17\}, H[or] = 7$, etc. Inverted indexes do not work on texts without word breaks (for example, DNA:  $\mathcal T = ACCTGCATT\dots$), and can only answer quick locate queries on words (not on arbitrary substrings). 

A complete and fast (at query time) solution to full-text indexing is to build an inverted index over all the substrings $\mathcal T[i,j]$ of our string $\mathcal T$, and associate to each substring $\mathcal T[i,j]$ the result of $\mathtt{count}(\mathcal T[i,j])$ and $\mathtt{locate}(\mathcal T[i,j])$. 
This solution solves fast queries, but is very space-inefficient: there are $\Theta(n^2)$ such substrings, and each uses $O(n)$ space. The total space of this solution is, therefore $O(n^3)$ words of memory.

We first show how to reduce the space to $O(n^2)$ words. The main observation we will use is that every substring of $\mathcal T$ is a prefix of a suffix of $\mathcal T$.

\begin{example}
    Let $\mathcal T = BANANA\$$. Then, the substring $NAN$ is a prefix of the suffix $NANA\$$.
\end{example}

This observation allows us to index just the suffixes of a string (there are $n$ of them), rather than all substrings of the string (there are $\Theta(n^2)$ of them). We introduce the \emph{suffix trie} data structure; see Figure \ref{fig:strie} for an example.

\begin{definition}\label{def:STRIE}
The suffix trie of $\mathcal T$ is the trie of all suffixes of $\mathcal T$. On each leaf, corresponding to suffix $\mathcal T[i,n]$, we store position $i$. For each internal node $x$, corresponding to substring $\mathcal T[i,j]$ (i.e. the string read from the root of the trie to node $x$) we store (1) $\mathtt{count}(\mathcal T[i,j])$, and (2) a pointer to the leftmost and rightmost leaf in the subtree rooted in $x$. We moreover link all the trie's leaves, from the leftmost to rightmost, using a linked list. 
\end{definition}

We obtain: 

\begin{theorem}
    The suffix trie uses $O(n^2)$ words of space, supports \emph{count} queries in optimal time $O(m)$, and \emph{locate} queries in optimal time $O(m+occ)$.
\end{theorem}

\begin{figure}[h!]
    \centering
    \includegraphics[width=0.5\linewidth]{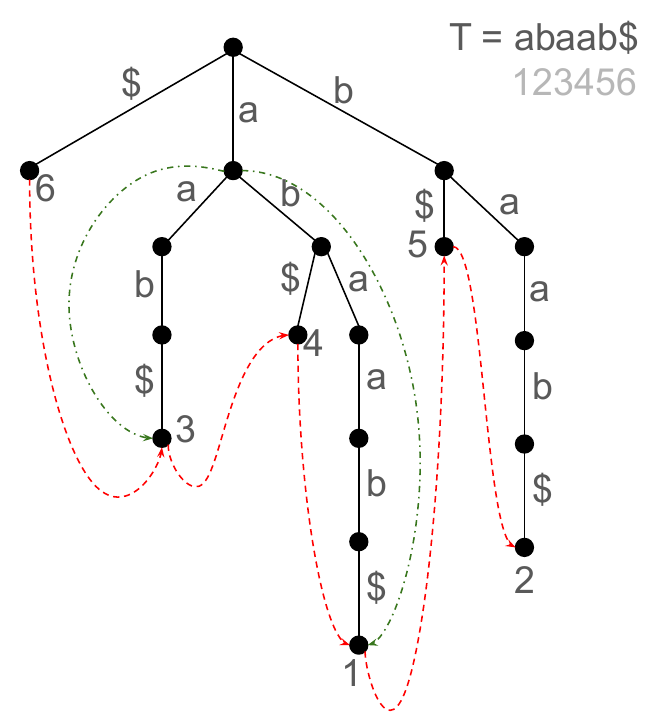}
    \caption{Suffix trie of the string $\mathcal T=abaab\$$. Leaves are labeled with the starting position in $\mathcal T$ of the corresponding suffix of $\mathcal T$.  In addition to the tree's edges, we also store additional information: (1) on each node $x$, we store two pointers to the leftmost and rigthtmost leaf in the subtree rooted in $x$ (for clarity, in the example we show these pointers --- in dashed green --- only on one node), (2) on each node $x$ such that the string read from the root to $x$ is $s$, we store $\mathtt{count}(s)$ (we do not show this information in the figure; for example, on the node reached by reading string $``a"$ from the root, this value would be $\mathtt{count}(``a")=3$), and (3) we link the leaves from left to right using a linked list (shown in dashed red in the figure): $6 \rightarrow 3 \rightarrow 4 \rightarrow 1 \rightarrow 5 \rightarrow 2$. \textbf{Pattern matching example}: to find all occurrences of the string $``a"$, descend from the root reading $``a"$, use the extra (dashed green) pointers to jump on the leftmost (3) and rightmost (1) leaves in the subtree of the node, and starting from the leftmost leaf (3) follow the linked list of leaves until reaching the rightmost leaf (1). Proceeding in this way, we navigate the sub-list $3\rightarrow 4 \rightarrow 1$, corresponding to all occurrences of $``a"$.}
    \label{fig:strie}
\end{figure}

\subsection{Suffix tree}

To reduce the space of the suffix trie, we observe that the strings we can read on the paths on the trie correspond to text's substrings. For example, in Figure \ref{fig:strie}, consider the node $x$ reached by reading string $b$ from the root, and the node $y$ reached by reading string $baab\$$. The path connecting $x$ to $y$ is labeled by string $aab\$ = \mathcal T[3,6]$. We can therefore replace this path by the two integers $[3,6]$, indicating that in order to reconstruct the original path, we need to read the text's substring $\mathcal T[3,6]$. 
We perform this operation for all maximal paths of nodes $x_1 \rightarrow x_2 \rightarrow \dots \rightarrow x_k$ such that $x_1, \dots, x_{k-1}$ have only one child each, and $x_k$ is either a leaf or has at least two children.
This process is called \emph{path compression}, and the resulting tree is called the \emph{suffix tree}: see 
Figure \ref{fig:stree}. 

\begin{figure}[h!]
    \centering
    \includegraphics[width=0.5\linewidth]{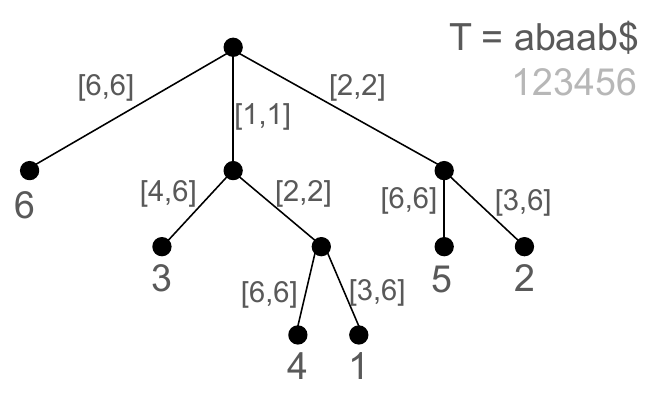}
    \caption{Suffix tree of the string $\mathcal T=abaab\$$. For each node, we store the same extra information shown in Figure \ref{fig:strie} (not shown here for simplicity).}
    \label{fig:stree}
\end{figure}

The suffix tree has exactly $n$ leaves, and its internal nodes are all branching (i.e. have at least 2 children). It follows that the total number of nodes is at most $2n-1$, and the number of edges at most $2n-2$. For each edge and each node, we store a constant number of pointers and integers (i.e. the same extra information as in Figure \ref{fig:strie}). Pattern matching works as with the suffix trie, except that (1) while descending the tree from the root, we also need to jump on the text to reconstruct the labels of the tree's edges, and (2) our search may end in the middle of an edge, for example: if we search $baa$ in Example \ref{fig:stree}, we end up in the middle of the edge entering in the leaf containing 2. 
Modification (1) does not increase the query running time, while issue (2) is not really an issue: if our search ends up in the middle of an edge, just jump at the end of the edge. For example: if we search $baa$ in Example \ref{fig:stree}, we end up in the middle of an edge; just jump at the end of the edge (corresponding to string $baab\$$): the result is correct, because all occurrences of $baa$ continue with characters $b\$$. We obtain:

\begin{theorem}
    The suffix tree uses $O(n)$ words of space, supports \emph{count} queries in optimal time $O(m)$, and \emph{locate} queries in optimal time $O(m+occ)$.
\end{theorem}

\subsection{Suffix array}\label{sec:suffix array}

Despite using linear $O(n)$ words of space, the constant hidden in the big-O notation of the suffix tree's space usage is very large. In fact, for every node and edge we need to record lots of information: pointers to children, count information, pointers to the leftmost/rightmost leaf in the node's subtree, etc. A good suffix tree implementation uses at least (approximately) $8n$ words, and each word takes 4 bytes (even 8 bytes if the text is very long). In total, at least $32n$ bytes. For example, if $\mathcal T$ is the DNA reference sequence of the Human genome, $n$ is approximately $3.2\cdot 10^9$. The text $\mathcal T$ can be stored in about 3 Gigabytes, while the suffix tree requires at least $32n \approx 95$ GiB.
In total, our data structure (text + suffix tree) uses approximately $98$ GiB: more than 32 times the text itself!

Observe that we ordered the leaves of the suffix tree and suffix trie according to the lexicographic order of their corresponding suffixes: 
$$
\$\ (6) < aab\$\ (3) < ab\$\ (4) < abaab\$\ (1) < b\$\ (5) < baab\$\ (2)
$$
This is not a coincidence: since suffixes are sorted, the idea is that we can use binary search to find efficiently all suffixes prefixed by any query pattern. This array of numbers on the leaves, 
$$
6,3,4,1,5,2 
$$
is called the \emph{suffix array} (SA). Suppose we want to find the suffixes prefixed by pattern $P=aa$. We do this by binary search. First, we jump in the middle of the suffix array, and read number 4. Then, we jump on the text and start reading $\mathcal T[4..] = ab..$. We compare this string with the pattern $aa$, discovering that it is larger: $ab... > aa$. It follows that we can ignore all entries of the suffix array 4,1,5,2, because they correspond to suffixes being larger than the pattern: $aa < ab\$\ (4) < abaab\$\ (1) < b\$\ (5) < baab$. We repeat recursively on the remaining suffix array interval (6,3), until finding the range of all suffixes starting with our pattern (exactly as when we perform binary search on an array of sorted integers to find all occurrences of some query integer). Note that we perform $O(\log n)$ binary search steps, and in each step we need to compare the pattern with a text substring of length at most $m$. Binary search thus runs in $O(m\log n)$ time and allows finding the suffix array range containing all suffixes prefixed by the pattern $P$. The length of this range is the answer to $\mathtt{count}(P)$, while the numbers contained inside the suffix array range are the answer to $\mathtt{locate}(P)$. We obtain:

\begin{theorem}
    The suffix array uses $n$ words of space on top of the text, supports \emph{count} queries in time $O(m\log n)$, and \emph{locate} queries in time $O(m\log n+occ)$.
\end{theorem}

On the Human genome ($n\approx 3.2\times 10^9$), each suffix array entry can be stored using words of 4 bytes. The text itself can be stored using 1 byte per character. In total, our data structure (text + SA) uses $5n \approx 15$ GiB: this is more than 6 times smaller than the suffix tree.

\section{Basics of information theory}

Consider again our example of the previous section where the text $\mathcal T$ is the human genome, of length $n= 3.2\times 10^9$ (we use equality here for simplicity). The human genome is a sequence over a 4-letters alphabet: $\Sigma = \{A,C,G,T\}$. But then, there is no need of using the whole ASCII alphabet (1 byte per character): 2 bits per character are sufficient! 
Let $\gamma:\Sigma\rightarrow \{0,1\}^+$ be an encoding of our alphabet. For example, we could encode $\gamma(A)=00$, $\gamma(C)=01$, $\gamma(G)=10$, $\gamma(T)=11$. Then, the whole human genome fits in $2n$ bits, which is just about $700$ MiB! This makes the suffix tree (95 GiB) and suffix array (12 GiB) appear even larger than they need to be: the former uses 135 times the space of the text, while the latter 17 times. 

Things could be much worse than this. Consider, for example, the following text:

$$
\mathcal T = ACGTTTTT \dots T
$$

Where we have one occurrence of A,C,G, and $3.2\times 10^9 - 3$ occurrences of $T$. What if we encode the characters as follows, $\gamma(A)=100$, $\gamma(C)=101$, $\gamma(G)=110$, $\gamma(T)=0$? The encoded text becomes (spaces are added just for clarity):

$$
\gamma(\mathcal T) = 100\ 101\ 110\ 0\ 0\ 0\ 0\ \dots\ 0
$$

Notice that we chose the encoding $\gamma()$ in such a way that the concatenation of the encoded characters, $\gamma(\mathcal T) = \gamma(A)\cdot \gamma(C) \cdot \gamma(G) \cdot \gamma(T) \dots \gamma(T)$ can be decoded unambiguously. A code with this property is called \emph{prefix code} and will be defined more precisely later.  
The total number of bits in $\gamma(\mathcal T)$ is $|\gamma(\mathcal T)| =  3\cdot 3 + 1\cdot (3.2\cdot 10^9-3) = 3.200.000.006$, so we are using an average of $|\gamma(\mathcal T)|/n \approx 1$ bit per symbol! Now, the encoded text uses just approximately $350$ MiB. We achieved \emph{compression} by observing that character's frequencies are very unbalanced. 

The comparison with the suffix tree and suffix array is even more dramatic than before: the former uses 270 times the space of the compressed text, while the latter 34 times. The goal of this chapter will be to design compressed indexes taking a space very close to the \emph{compressed text}, while at the same time including the text itself (i.e. also supporting access to the characters of $\mathcal T$). Using the above example, the compressed index and text will use \emph{in total} a space slightly larger than 350 MiB (let's say, around 400 MiB) and support very efficient \emph{count}, \emph{locate}, and \emph{access} ($\mathcal T[i]$) operations.

\subsection{Worst-case entropy}\label{sec:worst case entropy}

Consider our DNA alphabet $\Sigma = \{A,C,G,T\}$ and the first encoding we used: 
$\gamma(A)=00$, $\gamma(C)=01$, $\gamma(G)=10$, $\gamma(T)=11$. Observe that we used the same number of bits (2) for each character. Can we use 1 bit per character? It is easy to see that this will not work, since any encoding assigning 1 bit to \emph{all} characters, will assign the same code to at least two characters. For example: $\gamma(A)=0$, $\gamma(C)=1$, $\gamma(G)=0$, $\gamma(T)=1$. Such a code cannot be decoded because we cannot distinguish C from T and G from A. 

This reasoning leads to the conclusion that no uniquely-decodable code can assign less than $\lceil \log_2(|\Sigma|) \rceil$ bits to all characters in $\Sigma$.
We are now going to scale this result  to sets of combinatorial objects. Let $\mathcal U$ be a set (universe) of combinatorial objects, for example: all strings of length $n$ over alphabet $\Sigma$, or the number of subsets of $\{1,\dots, n\}$ of cardinality $m$.

\begin{definition}[Worst-case entropy]
    Let $\mathcal U$ be a universe of combinatorial objects. The \emph{worst-case entropy} of $\mathcal U$ is
    $$
    \mathcal H_{wc}(\mathcal U) = \lceil \log_2(|\mathcal U|) \rceil
    $$
\end{definition}

We conclude: 

\begin{corollary}
    Any uniquely-decodable encoding $\gamma : \mathcal U \rightarrow \{0,1\}^+$ of objects from $\mathcal U$ cannote use less than $\mathcal H_{wc}(\mathcal U)$ bits for every object of $\mathcal U$, i.e. there must exist  $x \in \mathcal U$ such that $|\gamma(x)| \geq \mathcal H_{wc}(\mathcal U)$.
\end{corollary}

The name \emph{worst-case} entropy comes from the fact that, for any encoding of  $\mathcal U$, the length (in bits) of the longest code for any object in $\mathcal U$ must be at least $\mathcal H_{wc}(\mathcal U)$.

\begin{example}
    Let $\mathcal B_n$ be the set of all sequences of length $n$ over binary alphabet $\{0,1\}$ (i.e. bitvectors). Then, $|\mathcal B_n| = 2^n$ and therefore $\mathcal H_{wc}(\mathcal B_n) = n$.
\end{example}

\begin{example}
    Let $\mathcal B_{n,m}$ be the set of all bitvectors of length $n$ containing exactly $m\leq n$ bits equal to 1. Then, $|\mathcal B_{n,m}| = \binom{n}{m}$ and therefore $\mathcal H_{wc}(\mathcal B_{n,m}) = \log_2 \binom{n}{m} = m\log (n/m) + O(m)$.
\end{example}

\begin{example}
    Let $\mathcal S_{n,m}$ be the set of all subsets of $\{1,\dots, n\}$ of cardinality $m$. Then, $|\mathcal S_{n,m}| = \binom{n}{m}$ and therefore $\mathcal H_{wc}(\mathcal S_{n,m}) = \log_2 \binom{n}{m} = m\log (n/m) + O(m)$.
\end{example}

Note that the worst-case entropy of the previous two examples is the same. In fact, it is easy to see that bitvectors are just a representation for sets (and vice-versa).

\begin{example}\label{example:ordinal}
    Let $\mathcal T_{n}$ be the set of all ordinal trees with $n$ nodes. Then, $|\mathcal T_{n}| = \frac{1}{n}\binom{2n-2}{n-1}$ and therefore $\mathcal H_{wc}(\mathcal T_{n}) = \log_2 \left(\frac{1}{n}\binom{2n-2}{n-1}\right) = 2n - \Theta(\log n)$.
\end{example}

The goal of the research field of \emph{succinct data structures} is to design data structures using $\mathcal H_{wc}(\mathcal U)$ bits to represent objects from $\mathcal U$ while supporting fast queries on them. For example, $\mathcal U$ could be the set of all texts of length $n$ over alphabet $\Sigma$, and the queries to support could be pattern matching queries (count, locate). With respect to this universe and these queries, the goal of a succinct data structure would be to use a space close to $\mathcal H_{wc}(\mathcal U) = n\log|\Sigma|$ bits while supporting fast \emph{count} and \emph{locate} queries.

\begin{Exercise}
    Consider example \ref{example:ordinal}. 
    Show how to encode any ordinal tree with $n$ nodes using exactly $2n$ bits. Note that such an encoding is (almost) worst-case optimal, because it (almost) matches the worst-case entropy (up to an additive term of $\Theta(\log n)$ bits, which can be ignored since it is exponentially smaller than the worst-case entropy).
\end{Exercise}

\subsection{Prefix-free codes}

Worst-case entropy tells us how much compressible is an object from some set $\mathcal U$, in the worst case.  
The goal of this chapter is to find encodings for our objects (in general, strings and sets) approaching this bound. 
If the objects we are trying to compress are strings, one way to approach worst-case entropy (as we will see) is to encode the string's individual characters using a good-enough \emph{prefix-free} code:

\begin{definition}
    A prefix-free code is an injective function $\gamma : \Sigma \rightarrow \{0,1\}^+$ such that $\gamma(x)$ is not a prefix of $\gamma(y)$, for any $x,y\in \Sigma$ such that $x\neq y$.
\end{definition}

\begin{example}
 The code $\gamma(A)=100$, $\gamma(C)=101$, $\gamma(G)=110$, $\gamma(T)=0$ is prefix-free because no character's encoding is a prefix of another. 
\end{example}

It is easy to see that, when using a prefix-free code,  we can just concatenate the encodings of all the string's characters, and the resulting bit sequence is invertible (i.e. we can reconstruct the string). More formally:

\begin{definition}
    If $\gamma$ is a prefix-free code and $S\in \Sigma^n$, we define the domain extension
    $$
    \gamma(S) = \gamma(S[1])\cdot \gamma(S[2]) \dots \gamma(S[n])
    $$
\end{definition}

\begin{remark}
    If $\gamma$ is prefix-free, we can reconstruct $S$ given $\gamma$ and $\gamma(S)$.
\end{remark}

Among many useful prefix-free codes we will use Huffman's (1952) which can be proved to be optimal among prefix codes. We show how to build a Huffman code via an example.

\begin{example}\label{ex:Huffman}
    Consider the string $S = abracadabra$. Count the frequency of every character in the string: $a$ appears 5 times, $b$ and $r$ appear 2 times each, and $c$ and $d$ appear 1 time each. The Figure below shows a possible Huffman code for these frequencies (the Huffman code is not necessarily unique, but all Huffman codes have the same cost). We are going to build a binary tree having in the leaves all the characters together with their frequencies. From each internal node, the left outgoing edge will be labeled with 0 and the right one with 1. The path from the root to each leaf will spell the code for that character. To build the tree, begin by creating the leaves. Then, pick two nodes $x,y$ without parent having the smallest frequencies (in the example, we pick $c$ and $d$, both with frequency 1), create a new node $xy$ whose frequency is the sum of the frequencies of $x$ and $y$ (in the example, we create node $cd$ of frequency 1+1=2), make it parent of $x$ and $y$, and repeat recursively until we have only 1 node left (the root) without any parent. In our example, we proceeded by grouping $b$ and $r$ (parented by $br$ of frequency 4), then we grouped $cd$ and $br$ (parented by $cdbr$ of frequency 6), and finally grouped $cdbr$ and $a$. At the end, we get the Huffman code $\gamma(c) = 000$, $\gamma(d) = 001$, $\gamma(b) = 010$, $\gamma(r) = 011$, $\gamma(a)=1$.
    Note that the most frequent character gets indeed the shortest code. The final bit-cost of the encoding is $|\gamma(abracadabra)| = 23$ bits. It can be proved that no prefix-free code for this string can use less than 23 bits: Huffman is optimal among all prefix-free codes.
    \tikzset{iv/.style={draw,fill=gray!10,circle,minimum size=20pt,inner
    sep=0pt,text=black},ev/.style={draw,fill=gray,rectangle,minimum
    size=20pt,inner sep=0pt,text=black}}
    \begin{figure}[h!]
        \centering
        \begin{forest}
        for tree={where n children={0}{ev}{iv},l+=8mm,
        if n=1{edge label={node [midway, left] {0} } }{edge label={node [midway, right] {1} } },}
        [cdbra(1)1
         [cdbr(6)
          [cd(2)
           [c(1)]
           [d(1)]
          ]
          [br(4)
           [b(2)]
           [r(2)]
          ]
         ]
         [a(5)]
        ]
        \end{forest}
    \end{figure}
\end{example}

\begin{remark}
    Note that after the first step in Example \ref{ex:Huffman} (after grouping $c$ and $d$) the nodes without parents are $cd(2)$, $b(2)$, $r(2)$, and $a(5)$. There are multiple valid ways to proceed: we could group $cd$ and $b$, or $cd$ and $r$, or $b$ and $r$ (in the example we chose the latter grouping). The resulting Huffman trees will have a different shape, but all the resulting codes will be optimal (i.e. use 23 bits, in that example). As an exercise, build all the valid Huffman trees for these frequencies. 
\end{remark}

\subsection{Zero-order empirical entropy}

We will see that Huffman encoding allows us to approach the worst-case entropy of the set of all strings with the same character frequencies of our input string. 
We show this by introducing a new compressibility measure for strings, which will turn out to be tightly related with worst-case entropy.

\begin{definition}
    Let $S\in \Sigma^n$ be a string. The \emph{zero-order empirical entropy} of $S$ is defined as:
    $$
    H_0(S) = \sum_{c\in \Sigma} \frac{n_c}{n}\log_2\left(\frac{n}{n_c}\right)
    $$
    Where $n_c = |\{i\ : \ S[i]=c\}|$ is the number of occurrences of character $c$ in $S$. When $S$ is clear from the context, we will just write $H_0$ instead of $H_0(S)$.
\end{definition}

Shannon's source coding theorem (1948) implies that (as $n$ grows to infinity) no prefix-free code can use less than $nH_0$ bits to encode a string $S\in \Sigma^n$. On the other hand, it can be shown that Huffman code uses no more than $n(H_0+1)$ bits: at most 1 bit per character over the empirical entropy. Other encodings (for example, arithmetic encoding) achieve $nH_0 + O(1)$ bits. In the next sections, we will see encodings achieving $nH_0 + o(n\log\sigma)$ bits which support also fast queries on the string (i.e. compressed data structures).

\begin{example}
    Observe that on a bitvector of length $n$ with $m\leq n$ bits equal to 1, $H_0 = \frac{m}{n}\log_2\left(\frac{n}{m}\right) + \frac{n-m}{n}\log_2\left(\frac{n}{n-m}\right)$. For example, take
    $$
    B = 001010001000
    $$
    Then, $n=12$ and $m=3$. The zero-order empirical entropy of $B$ is
    $$
    H_0(B) = \frac{3}{12}\log_2\left(\frac{12}{3}\right) + \frac{9}{12}\log_2\left(\frac{12}{9}\right) \approx 0.8\ \mathrm{bits}
    $$
    This means that $nH_0$ (zero-order empirical compression) is approximately $20\%$ smaller than $n$ bits (no compression). 
\end{example}

\begin{remark}
    In the example above, in order to achieve $nH_0$ we would need to use on average $0.8$ bits per bit.
    Observe, however, that Huffman is useless when encoding bitvectors, because it uses at least one bit per symbol!  In the next sections we will see that the trick to achieve this (apparently impossible) goal is to encode blocks of bits, rather than single bits. 
\end{remark}

It turns out that worst-case entropy and zero-order entropy are tightly connected. Later we will use the following lemma on bitvectors. Try to prove it by yourself as an exercise 
(full proof in \cite[Sec. 2.3.1]{navarro2016compact}), expanding the binomial coefficient and using Stirling's approximation of the logarithm of the factorial ($\log_2(n!) = n \log_2 n - n\log_2 e + O(\log n)$):

\begin{lemma}\label{lem:Hwc vs H0}
Let $\mathcal B_{n,m}$ be the set of all bitvectors of length $n$ with $m$ bits equal to $1$, and let $x \in \mathcal B_{n,m}$. Then:
$$
\mathcal H_{wc}(\mathcal B_{n,m}) = \left\lceil \log_2\binom{n}{m} \right\rceil = nH_0(x) - O(\log n)
$$
\end{lemma}

A similar relation as the one of Lemma \ref{lem:Hwc vs H0} holds between the worst-case entropy of the set of strings with fixed symbol probabilities, and the zero-order empirical entropy of any member of that family of strings (see \cite[Sec. 2.3.1]{navarro2016compact}).

\subsection{High-order empirical entropy}\label{sec:high order entropy}

Zero-order empirical entropy exploits knowledge about the character's frequencies to achieve compression. Can we do better? Yes: intuitively, we could apply the same reasoning of zero-order compression to the characters followed by the same context of $k$ characters, for some fixed (small) integer $k$. For example, consider a book about algorithms. If we pick $k=8$ and see string (context) $lgorithm$ (of length $k=8$), then we are almost sure that the preceding character is $a$. In other words, $a$ has frequency $100\%$ among the set of characters followed by string $lgorithm$, while all other characters have frequency $0\%$: the set of characters followed by $lgorithm$ is extremely compressible. By repeating the same reasoning to all possible contexts of length $k$, we get to high-order empirical entropy.

\begin{definition}
    Let $S\in \Sigma^n$, and choose $k\leq n$. Let $w\in \Sigma^k$. We define $S_w$ as the string formed by all characters preceding occurrences of $w$ in $S$, from the leftmost to rightmost occurrence, considering the string $S$ as being circular (i.e. $S[1]$ is preceded by $S[n]$).
\end{definition}

\begin{example}\label{ex:k-order}
    Let $S = aababbabab$ and $k=2$. Consider context $w = ab$. $S_{ab}$ is the string formed by all characters that precede occurrences of $ab$ in $S$: $S_{ab} = abbb$. Similarly, $S_{bb}=a$, $S_{aa} = b$, and $S_{ba}=abaa$ (remember that we consider the string to be circular).
\end{example}

\begin{definition}[$k$-th order empirical entropy]\label{def:high order entropy}
    Let $S\in \Sigma^n$, and choose $k\leq n$. The $k$-th order empirical entropy of $S$ is:
    $$
    H_k(S) = \sum_{w\in \Sigma^k} \frac{|S_w|}{n} H_0(S_w)
    $$
    When clear from context, we simply write $H_k$ instead of $H_k(S)$.
\end{definition}

\begin{example}
    Consider the string of example \ref{ex:k-order}, of length $n=10$. Then $H_0(S_{aa}) = H_0(b) = 0$, $H_0(S_{ab}) = H_0(abbb) = \frac{1}{4}\log_2(4) + \frac{3}{4}\log_2(4/3) \approx 0.8112$, $H_0(S_{ba}) = H_0(abaa) = \frac{3}{4}\log_2(4/3) + \frac{1}{4}\log_2(4) \approx 0.8112$, $H_0(S_{bb}) = H_0(a) = 0$, and we get:
    $$
    \begin{array}{ccl}
    H_2(S) & = & \frac{|S_{aa}|}{n} H_0(S_{aa}) + \frac{|S_{ab}|}{n} H_0(S_{ab}) + \frac{|S_{ba}|}{n} H_0(S_{ba}) + \frac{|S_{bb}|}{n} H_0(S_{bb})\\
    & \approx & \frac{1}{10}\cdot 0 + \frac{4}{10}\cdot 0.8112 + \frac{4}{10}\cdot 0.8112 + \frac{1}{10}\cdot 0\\
    & \approx & 0.65
    \end{array}
    $$
    This means that a compressor achieving second-order entropy would use just 0.65 bits per symbol to encode $S$. On the other hand, the zero-order entropy of $S$ is $H_0(S) = 1$: no compression at all (note that the alphabet's size is equal to 2, so a naive encoding using 1 bit per symbol achieves this result as well). 
\end{example}

\begin{remark}
 The larger $k$ is, the better compression we get. In fact, it holds: 
 $$
 \log_2\sigma \geq H_0 \geq H_1 \geq H_2 \geq \dots
 $$
 However, we cannot use an arbitrarily large $k$: an information-theoretic lower bound prevents from using $k$ larger than $\log_{\sigma}n$ (see \cite{GAGIE2006246} for more details). 
\end{remark}

As a concrete example (see \cite{pizzachiliText}), a zero-order compressed XML file is typically $35\%$ smaller than its uncompressed version. Using high-order compression helps a lot: $nH_5$ is about $90\%$ smaller than the uncompressed file. This means that XML files are very predictable, given short contexts.

\section{Zero-order compressed bitvectors}\label{sec:zero order bv}

The topics of this section are treated more in detail in \cite[Sec. 4.1.1]{navarro2016compact}.

Observe that Huffman encoding uses an integer number of bits to encode each character of our string: it cannot use, for example, 3.7 bits to encode a particular character! as a result, Huffman is useless on bitvectors: it will always use 1 bit to encode each bit of the bitvector. No compression. 

In this section, the goal is to encode a bitvector in $nH_0 + o(n)$ bits while also supporting particular queries on it in constant time:

\begin{definition}[zero-order bitvector data structure]
    Given a bitvector $B \in \{0,1\}^n$, a zero-order bitvector data structure is a structure using $nH_0 + o(n)$ bits and supporting the following queries in $O(1)$ time (we enumerate positions of $B$ from 1):
    \begin{enumerate}
        \item Random access: extract any bit $B[i]$
        \item Rank: $B.rank_b(i)$ is the number of bits equal to $b$ in the prefix $B[1,\dots, i]$.
        \item Select: $B.select_b(i)$ is the position in $B$ of the $i$-th bit equal to $b$.
    \end{enumerate}
\end{definition}

\begin{example}
    Let $B = 011100010100110011$. Then:
    \begin{itemize}
        \item $B[5] = 0$
        \item $B.rank_0(4) = 1$
        \item $B.rank_0(6) = 3$ 
        \item $B.rank_1(8) = 4$
        \item $B.select_1(4) = 8$
        \item $B.select_0(3) = 6$
    \end{itemize}
\end{example}

We will show a simplified version of the RRR structure \cite{rajeev2002succinct} (the name comes from the initials of the authors). A nice self-contained description is also given in this blog \cite{BoweRRR}. We will only show how to solve \emph{random access} and \emph{rank} queries (for \emph{select}, see the book \cite[Sec 4.1.1]{navarro2016compact}).

\subsection{Random access in $n$ bits}

First, we show how to pack a bitvector using exactly $n$ bits (plus a constant number of memory words, which we will ignore for simplicity since they are negligible with respect to the $n$ bits).

While a bit can be stored as an integer, it is not a good idea to allocate $n$ integers to represent $B$ in your favorite programming language: each integer typically uses 32 bits (4 bytes), so the whole array uses $32n$ bits!

The idea is that, precisely because an integer is formed by 32 bits, we can use it to pack 32 bits of $B$. To be more general, let us denote with $w$ the size of an integer in our computer. 
Parameter $w$ is known in the literature as \emph{memory word size}.
In standard computers, we have $w= 32$ (or $w=64$ if we use long int). Note that it is reasonable to assume $w \in \Omega(\log n)$ because it should be possible to store any position $i \in [n]$ pointing inside $B$ using one (or a constant number of) integers. A standard assumption in the research field is $w \in \Theta(\log n)$. We will follow this assumption.

The following example shows how we can pack a bitvector, so to use just $n$ bits.

\begin{example}\label{ex:bitvector n bits}
    Let $B = 011100010100110011$. Assume the word size is $w=6$. Then, we break $B$ in $n/w$ blocks of $w$ bits:
    $$
    B = 011100\ 010100\ 110011
    $$
    We store $B$ as an array $B'$ of integers of $w$ bits each (i.e. treat each block as a binary integer of $w$ bits):
    $$
    B' = 28, 20, 51
    $$
    Et voilà. $B'$ is a standard array of integers (you can create it in your favorite programming language), contains exactly the same information as $B$, and it uses $n$ bits.
\end{example}

\begin{remark}
    To be precise, the encoding shown in Example \ref{ex:bitvector n bits} uses $n + \Theta(w) = n + \Theta(\log n)$ bits because the computer also needs to store the address of $B'$ in memory and its length. Those extra integers however use $\Theta(w) = \Theta(\log n)$ bits so they are negligible with respect to the size $n$ of the bitvector. To simplify our discussion, in the rest of the notes we will ignore this small overhead. Be careful, however, that in general it is not always safe to ignore it! (it will be safe in our data structures).
\end{remark}

\begin{remark}
    In general, $w$ does not divide $n$ so the last block could contain less than $w$ bits. In this case, just pad $B$ with bits '0' until its length is a multiple of $w$. This adds just $O(w)$ bits of overhead and we will still be able to solve all our queries.
\end{remark}

We use notation $(\cdot)_2$ to indicate an integer written in base 2. Otherwise, integers are expressed in base 10. For example, $(011)_2 = 3$.

Let's see how to do random access in constant time, with an example.

\begin{example}\label{ex:bitvector n bits2}
    Consider again bitvector $B$ of the previous example, divided in blocks:
    $$
    B = 011100\ 010100\ 110011
    $$
    Stored as the following integer array:
    $$
    B' = 28, 20, 51
    $$
    Assume we want to extract bit $B[10] (=1)$. The idea is that we locate the block containing the 10-th bit: this is the second block, $B'[2]=20$. In general, $B[i]$ is inside the block $B'[\lfloor (i-1)/w \rfloor + 1]$. 
    
    Now, we know that our bit $B[10]$ is hidden inside the binary representation of $B'[2]=20$. In particular, $B[10]$ is the 4-th bit in the binary representation of $B'[2]=20 = (010100)_2$. In general, $B[i]$ is the bit in position $i' = (i-1)\mod w + 1$ of $B'[\lfloor (i-1)/w \rfloor + 1]$. 

    To extract the 4-th bit of $B'[2]=20 = (010100)_2$, we first right-shift this number by $w-4=2$ bits, i.e. move every bit to the right by 2 positions; the rightmost 2 bits will disappear, and 2 new bits equal to 0 will appear to the left. Our bit (shown underlined here) will end up at the end of the block of bits:
    $$
    20 \gg 2 = (010\underline{1}00)_2 \gg 2 = (00010\underline{1})_2 = 5
    $$
    This operation (right shift) is supported in $O(1)$ time in our computers. 
    In general, to move the bit $B[i]$ at the end of our block $B'[\lfloor (i-1)/w \rfloor + 1]$ we will perform shift $B'[\lfloor (i-1)/w \rfloor + 1] \gg w - ((i-1)\mod w) - 1$.
    
    We have only one operation left to do: set to 0 all bits except the rightmost one. This can be done applying a bitwise AND (symbol \emph{\&}) with mask $1 = (000001)_2$:
    $$
    5\ \&\ 1 = (00010\underline{1})_2\ \&\ (000001)_2 = (00000\underline{1})_2 = 1 = B[10]
    $$
    Also this operation is supported in $O(1)$ time on any computer.
\end{example}

\begin{remark}
    In general, with almost the same technique we can extract in $O(1)$ time any block of at most $w$ contiguous bits from $B$. The extracted block will fit in an integer of $w$ bits. Exercise: show how to do it.
\end{remark}

We obtained a \emph{packed bitvector} data structure: 

\begin{theorem}[Packed bitvector] \label{thm:packed bv}
    Any bitvector $B \in \{0,1\}$ can be stored with a data structure of $n$ bits supporting the extraction of any block $B[i,\dots, i+\ell-1]$ of $\ell \leq w$ contiguous bits in $O(1)$ time. 
\end{theorem}

We will often use the above theorem to represent efficiently \emph{packed arrays of integers}:

\begin{theorem}[packed integer array]\label{thm:packed array}
    Assume we want to store efficiently an array of $n$ integers of size at most $2^b$ each, with $b\le w$ (i.e. each integer uses exactly $b$ bits). Then, Theorem \ref{thm:packed bv} can be used to store the array in optimal $nb$ bits while supporting random access to any integer in $O(1)$ time. 
\end{theorem}

\subsection{Encoding a bitvector in $nH_0 + o(n)$ bits}

Let $B \in \{0,1\}$ be our bitvector. We choose a \emph{block size} $b = \lceil (\log_2 n)/2 \rceil = \Theta(\log n)$, and divide $B$ in blocks of $b$ bits each. We call the blocks $B_1, B_2, \dots, B_{n/b}$ (again, we assume for simplicity that $b$ divides $n$).

\begin{example}
    Let
    $$
    B = 010101001110000011111
    $$
    The bitvector's length is $n=21$ so $b = \lceil (\log_2 n)/2 \rceil = 3$. The division in blocks is the following:
    $$
    \begin{array}{ccccccc}
         B_1 & B_2 & B_3 & B_4 & B_5 & B_6 & B_7  \\
         010 & 101 & 001 & 110 & 000 & 011 & 111
    \end{array}
    $$
\end{example}

We will encode each block $B_i$ with two integers: $C_i$ (the \emph{class} of the block) and $O_i$ (the \emph{offset} of the block). 

The class $C_i$ of block $B_i$ is simply the number of bits equal to 1 in $B_i$:

\begin{example}
   Continuing the previous example, the classes are:
    $$
    \begin{array}{ccccccc}
         C_1 & C_2 & C_3 & C_4 & C_5 & C_6 & C_7  \\
         1 & 2 & 1 & 2 & 0 & 2 & 3
    \end{array}
    $$
\end{example}

Note that each integer $C_i$ is a number between 0 and $b$, so it requires $\lceil \log_2(b+1) \rceil = \Theta(\log\log n)$ bits to be represented. We store this array of integers using exactly $\lceil \log_2(b+1) \rceil$ bits each, using Theorem \ref{thm:packed array}, and obtain:

\begin{lemma}
    Integers $C_1, \dots, C_{n/b}$ can be stored in $O\left(\frac{n}{b}\log(b+1)\right) = O\left(\frac{n}{\log n}\log\log n\right) = o(n)$ bits, and each $C_i$ can be extracted in $O(1)$ time. 
\end{lemma}

We now show what are the offsets $O_i$. Consider the lists $B_{b,j}$, for all $0\le j \le b$, containing all bitvectors (sorted in increasing order) of length $b$ with $j$ bits equal to 1:

\begin{example}
    In our example, $b=3$ so the lists are:
    \begin{itemize}
        \item $B_{3,0} = (000)$
        \item $B_{3,1} = (001, 010, 100)$
        \item $B_{3,2} = (011, 101, 110)$
        \item $B_{3,3} = (111)$
    \end{itemize}
\end{example}

Observe that a pointer inside $B_{b,j}$ requires $\lceil \log_2\binom{b}{j}\rceil$ bits. Offset $O_i$ is precisely such pointer (in this case, we count positions  starting from 0): formally, $O_i \ge 0$ is the integer such that $B_i = B_{b,C_i}[O_i]$. In other words, to compute $O_i$:

\begin{enumerate}
    \item Count the number $C_i$ of 1's in $B_i$
    \item Search $B_i$ inside list $B_{b,C_i}$
    \item $O_i$ is the position (counting from 0) of $B_i$ inside list $B_{b,C_i}$.
\end{enumerate}

\begin{example}
   Continuing the previous example, the offsets are:
    $$
    \begin{array}{ccccccc}
         O_1 & O_2 & O_3 & O_4 & O_5 & O_6 & O_7  \\
         1 & 1 & 0 & 2 & 0 & 0 & 0
    \end{array}
    $$
\end{example}

Above we observed that $O_i$ requires $\lceil \log_2\binom{b}{C_i}\rceil$ bits to be represented; we are going to use exactly this number of bits. Notice that, differently from the classes $C_1, \dots, C_{n/b}$, we use a variable-length encoding for the offsets $O_1, \dots, O_{n/b}$.

\begin{example}\label{ex:offsets}
   Continuing the previous example, this is our variable-length encoding of the offsets:
    $$
    \begin{array}{ccccccc}
         O_1 & O_2 & O_3 & O_4 & O_5 & O_6 & O_7  \\
         01 & 01 & 00 & 10 & 0 & 00 & 0
    \end{array}
    $$
    Why? Because a pointer inside $B_{3,0}$ and $B_{3,3}$ (both of length 1) requires just 1 bit: this is the case of offsets $O_5$ and $O_7$. A pointer inside $B_{3,1}$ and $B_{3,2}$ (both of length 3) requires 2 bits. This is the case of all other offsets $O_1, O_2, O_3, O_4, O_6$.
\end{example}

We are going to simply concatenate the binary encodings of $O_1, \dots, O_{n/b}$, without any separator, using Theorem \ref{thm:packed bv}. Importantly, notice that this encoding is not prefix-free; however, given $C_1, \dots, C_{n/b}$ we can deduce the bit-length of the encoding of each $O_i$, so the joint encoding of  $C_1, \dots, C_{n/b}$ and  $O_1, \dots, O_{n/b}$ is invertible (i.e. it allows us to reconstruct any $C_i$ and any $O_i$).

How much space does our encoding of the offsets use? interestingly, zero-order compressed space:

\begin{lemma}
    Our variable-length encoding of $O_1, \dots, O_{n/b}$ uses $nH_0 + o(n)$ bits.
\end{lemma}
\begin{proof}
   $O_i$ is stored using $\lceil \log_2\binom{b}{C_i}\rceil$ bits. Summing over all offsets, we obtain: 
   $$
   \begin{array}{ccl}
        \sum_{i=1}^{n/b} \lceil \log_2\binom{b}{C_i}\rceil & \leq &  \sum_{i=1}^{n/b} \left(\log_2\binom{b}{C_i} + 1 \right)\\
        & = & \left( \sum_{i=1}^{n/b} \log_2\binom{b}{C_i} \right) + n/b\\
        & = & \left( \log_2\prod_{i=1}^{n/b} \binom{b}{C_i} \right) + n/b
   \end{array}
   $$
   We use the fact that $\binom{x}{y}\cdot \binom{z}{k} \le \binom{x+z}{y+k}$: the number $\binom{x+z}{y+k}$ of possible ways to choose $y+k$ elements from $x+z$ slots include all combinations of choosing $y$ elements from $x$ slots and $k$ elements from $z$ slots. Then, $\binom{b}{C_1} \cdot \binom{b}{C_2} \cdots \binom{b}{C_{n/b}} \le \binom{b+b+\dots + b}{C_1 + C_2 + \dots + C_{n/b}} = \binom{n}{m}$, where $m$ is the number of bits equal to '1' in $B$. Recalling from Lemma \ref{lem:Hwc vs H0} that $\log_2\binom{n}{m} \leq nH_0$, we finally obtain that the number of bits is
   $$
   nH_0 + n/b = nH_0 + O(n/\log n) = nH_0 + o(n)
   $$
\end{proof}

Notice that, given $C_i$ and $O_i$, we can reconstruct $B_i$ since we know how to build lists $B_{b,j}$ (we do not need to store them together with the encoding). To sum up (note that we need to store the number ``$n$'' as well to retrieve $b$ and the number of blocks):

\begin{lemma}
    Our encoding of $C_1, \dots, C_{n/b}$ and $O_1, \dots, O_{n/b}$, in addition to the number ``$n$'' (stored using an integer of $w$ bits), uses $nH_0 + o(n)$ bits of space and allows reconstructing the original bitvector $B$.
\end{lemma}

\subsection{Random access and rank in $nH_0+o(n)$ bits}

To extract any bit $B[i]$, we proceed in a similar way as in Example \ref{ex:bitvector n bits2}, except that now we replace $w$ with $b$ since our blocks are made of $b$ bits. To extract $B[i]$, we:

\begin{enumerate}
    \item Identify which block $B_{j}$ contains position $i$. Like in Example \ref{ex:bitvector n bits2}, the block number is $j = \lfloor (i-1)/b \rfloor + 1$,
    \item Extract $C_j$ and $O_j$,
    \item Decompress the block: from $(C_j,O_j)$, obtain the bitvector $B_j$ of length $b$ (packed in a memory word: note that $b\leq w$), and
    \item Like in Example \ref{ex:bitvector n bits2}, extract bit in position $i' = (i-1)\mod b + 1$ of $B_j$.
\end{enumerate}

Steps 1 and 4 can be performed in $O(1)$ time, as seen in Example  \ref{ex:bitvector n bits2}. In step 2, we already know how to extract $C_j$ in constant time: recall that each class $C_i$ is stored using $\lceil \log_2(b+1) \rceil$ bits, so we can use Theorem \ref{thm:packed array}. 
We cannot use the same technique for the offsets, because they are encoded using a variable-length encoding. 
Let's see how to extract also any $O_j$ in constant time. 

\subsubsection{Extracting offsets}

In Example \ref{ex:offsets}, our offsets are stored as the following bitvector: 011000100000 (no separators!). The idea is to store extra information telling us the position where each offset starts in this bitvector. Ideally, we would like to store the position of every offset, like this:

$$
\begin{array}{lllllll}
    O_1 & O_2 & O_3 & O_4 & O_5 & O_6 & O_7  \\
     01 & 10 & 00 & 10 & 0 & 00 & 0\\
     1  & 3  & 5  & 7  & 9 & 10  & 12
\end{array}
$$

However, this solution uses too much space: we are storing $n/b$ positions, and each uses $\log_2 n$ bits. In total, $(n/b)\log_2 n = \Theta(n)$ bits (we want to use only extra $o(n)$ bits). The solution is to use a two-levels scheme: 

\begin{enumerate}
    \item We group the blocks into macroblocks of $b$ blocks each,
    \item we remember the \emph{absolute} starting position of every macroblock in the encoded offsets, and
    \item we remember the \emph{relative} starting position of each block inside its macroblock.
\end{enumerate}

\begin{example}
    Continuing our example, we store the following extra information:
    $$
    \begin{array}{rllll|lll|l}
        &&O_1 & O_2 & O_3 & O_4 & O_5 & O_6 & O_7  \\
        \mathrm{encoded\ offsets} & & 01 & 10 & 00 & 10 & 0 & 00 & 0\\
        \mathrm{starting\ position\ of\ macroblock} & & 1  &    &    & 7  &   &    & 12\\
        \mathrm{relative\ starting\ position\ of\ offset\ inside\ macroblock} & & 1  &  3 &  5 & 1  & 3 & 4  & 1
    \end{array}
    $$
    This information is clearly sufficient to reconstruct the initial position of any offset. Take for example $O_6$: using the stored information, we know that $O_6$ starts in relative position 4 inside its macroblock (the second), which starts in absolute position 7. The absolute starting position of $O_6$ is then $7 + 4 - 1 = 10$.
\end{example}

The total number of bits in the encoded offsets is $nH_0 + o(n) = O(n)$, so absolute macroblock positions use $O(\log_2 n)$ bits each. We store them using Theorem \ref{thm:packed array}. Each of those positions uses many bits, but they are few: there are in total $n/b^2$ macroblocks, so these integers use in total $(n/b^2) \cdot O(\log_2 n) = O(n/\log n) = o(n)$ bits.

Relative block positions inside macroblocks, on the other hand, are more numerous: $n/b$. However, these relative positions are small numbers: recall that each offset is encoded using $\lceil \log_2 \binom{b}{C_i}\rceil \leq \log_22^b = b$ bits. This means that the $b$ offsets inside each macroblock are encoded using at most $b^2$ bits, so each relative position inside a macroblock will use at most $\lceil \log_2b^2 \rceil$ bits. We encode the $n/b$ relative positions using $\lceil \log_2b^2 \rceil$ bits each with Theorem \ref{thm:packed array}. In total, these integers use $(n/b) \lceil \log_2b^2 \rceil \in O((n/b)\log b) = O\left( \frac{n\log\log  n}{\log n} \right) = o(n)$ bits.

This gives us: 

\begin{lemma}
    Our encoding of $C_1, \dots, C_{n/b}$ and $O_1, \dots, O_{n/b}$, in addition to absolute macroblock positions and relative block positions, uses $nH_0 + o(n)$ bits of space and allows retrieving any class $C_j$ and any offset $O_j$ in $O(1)$ time. 
\end{lemma}

\subsubsection{Decoding the block given its class and offset}

The last problem to solve is how to locally decompress a block: given $C_j$ and $O_j$, retrieve $B_j$ packed in a memory word. Our solution is simply to pre-compute all combinations of class and offset. We store a two-dimensional table

$$
T[C_j][O_j] = B_j
$$

How much space does the table use? the total number of classes is $b+1 = O(\log n)$. The total number of offsets for each class is $\binom{b}{C_i} \leq 2^b = 2^{\lceil (\log_2 n)/2 \rceil} = O(\sqrt n)$. For each class-offset combination, we store a bitvector $B_j$ of $b = O(\log n)$ bits. In total: $O(\sqrt n \log^2 n) = o(n)$ bits.

\begin{example}\label{ex:tableT}
    Table $T$ is just a way to encode our lists $B_{b,j}$:
        \begin{itemize}
        \item $B_{3,0} = (000)$
        \item $B_{3,1} = (001, 010, 100)$
        \item $B_{3,2} = (011, 101, 110)$
        \item $B_{3,3} = (111)$
    \end{itemize}
    The table is shown in Figure \ref{fig:tableT}.
    \begin{figure}[h!]
        \centering
        \begin{tabular}{|c|c|c|c|}
            \hline
             & 0 & 1 & 2 \\\hline
             0 & 000 & - & - \\\hline
             1 & 001 & 010 & 100 \\\hline
             2 & 011 & 101 & 110 \\\hline
             3 & 111 & - & - \\\hline
        \end{tabular}\caption{Table $T$ of Example \ref{ex:tableT}. Classes correspond to rows and offsets to columns.}\label{fig:tableT}
    \end{figure}
\end{example}

Putting everything together:

\begin{lemma}
    Our encoding of $C_1, \dots, C_{n/b}$ and $O_1, \dots, O_{n/b}$, in addition to absolute macroblock positions, relative block positions, and table $T$, uses $nH_0 + o(n)$ bits of space and allows retrieving any bit $B[i]$ in $O(1)$ time. 
\end{lemma}

\subsubsection{Rank}

The idea for supporting $B.rank_1()$ is almost the same as for random access. Instead of absolute and relative positions in the encoded offsets, we store the rank (number of bits equal to 1) before each macroblock (excluding the first macroblock position) and before each block inside its macroblock (excluding the first block position). The solution for $B.rank_0()$ is the same.

In addition, we store a three-dimensional table $R$ pre-computing the answer to any rank query on any possible bitvector of length $b$, specified as a class and offset:

$$
R[C_j][O_j][i] = B_j.rank_1(i)
$$

\begin{example}
    Continuing our example, we store the following extra information to support $B.rank_1()$:
    $$
    \begin{array}{rllll|lll|l}
         && B_1 & B_2 & B_3 & B_4 & B_5 & B_6 & B_7  \\
        \mathrm{blocks} && 010 & 101 & 001 & 110 & 000 & 0\underline 11 & 111\\
        \mathrm{absolute\ rank\ before\ macroblock} & & 0  &    &    & 4  &   &    & 8\\
        \mathrm{relative\ rank\ inside\ macroblock} & & 0  &  1 &  3 & 0  & 2 & 2  & 0
    \end{array}
    $$
    Let's see how to solve $B.rank_1(17)$ (position 17 is underlined). The rank before the macroblock containing position 17 is 4. Position 17 is inside block $B_6$, and the relative rank before $B_6$ in its macroblock is 2. Finally, position 17 is the second inside $B_6 = 0\underline11$. 
    We extract $C_6 = 2$ and $O_6 = 0$, and use table $R$ to get the rank inside  $B_6$: $R[2][0][2] = 1$. Finally, we add up the three partial ranks: $B.rank_1(17) = 4 + 2 + 1 = 7$.
\end{example}

Table $R$ uses in total at most $(b+1) \cdot 2^b \cdot b \cdot \lceil \log_2(b+1)\rceil = O(\sqrt{n}\log^2 n \log\log n) = o(n)$ bits (number of classes $\times$ number of offsets $\times$ number of positions inside a bitvector $B_j$ of length $b$ $\times$ number of bits to express the result of a rank query on $B_j$).

To solve $B.rank_0()$, simply observe that $B.rank_0(i) = i-B.rank_1(i)$.

The solution for query \emph{select} uses similar ideas, except that we build blocks of variable length containing a fixed number of bits equal to '1'. See \cite{navarro2016compact} for the details. We finally obtained the main result of the section:

\begin{theorem}[RRR bitvector \cite{rajeev2002succinct}]\label{thm:bitvector}
    Any bitvector $B\in \{0,1\}^n$ can be stored in a data structure of $nH_0 + o(n) \leq n + o(n)$ bits supporting random access, rank, and select queries in $O(1)$ time. 
\end{theorem}

Finally, it is not hard to prove that the bit-size of the RRR bitvector data structure built on the concatenation of two bitvectors $B_1$ and $B_2$ is almost the same as the sum of the sizes of the RRR data structure built separately on $B_1$ and $B_2$. 
More in general (prove this as an exercise):

\begin{lemma}\label{lem:sum RRR}
    Let $\mathrm{RRR\_size}(B)$ denote the bit-size of the RRR bitvector data structure built on bitvector $B$.
    Let $B_i \in \{0,1\}^{n_i}$, for $i=1, \dots, t$ be $t$ bitvectors. Let $n = \sum_{i=1}^t n_i$. Then:
    $$
    \mathrm{RRR\_size}(B_1B_2\dots B_t) \le \left(\sum_{i=1}^t n_iH_0(B_i)\right) + o(n) + O(t\log n)
    $$
\end{lemma}
\begin{proof}
    Recall that the RRR data structure partitions the bitvector $B_1B_2\dots B_t$ in small blocks of size $b = O(\log n)$. If we ignore the $t-1$ blocks overlapping the borders of each pair of adjacent bitvectors $B_{i}B_{i+1}$, the analysis of this section shows that the remaining blocks use at most  $\sum_{i=1}^t \left(n_iH_0(B_i) + o(n_i)\right) = \left(\sum_{i=1}^t n_iH_0(B_i)\right) + o(n)$ bits. On the other hand, the $t-1$ blocks at the borders of adjacent bitvectors require at most $O(t\log n)$ bits of space (this is a safe upper-bound to our encoding into classes and offsets).
\end{proof}

This result will be useful later when extending access, rank, and select queries on arbitrary strings and finally using them to index text in high-order compressed space.

\section{Sets of integers - the Elias-Fano data structure}\label{sec:Elias-Fano}

The topics of this section are treated more in detail in \cite[Sec. 4.4]{navarro2016compact}.

Consider the problem of encoding a set of integers $S \subseteq \{0,\dots,n-1\}$ of cardinality $|S|=m$. As seen in Section \ref{sec:worst case entropy}, $m\log(n/m) + O(m)$ bits are necessary to encode $S$. In this section, we show that  $m\log(n/m) + \Theta(m)$ bits  are sufficient by showing a data structure of this size supporting several queries on $S$.

In Section \ref{sec:worst case entropy} we observed that sets of cardinality $m$ over universe $\{0,\dots,n-1\}$ and bitvectors of length $n$ with $m$ bis equal to '1' are essentially the same objects. In fact (note: in this section the smallest integer in our sets is $0$):

\begin{example}
    Consider the set $S = \{3,7,8\} \subseteq [9]$. The following bitvector of length $9$ with $3$ bits equal to '1' is an encoding of this set: $0001000110$. The other way round is also true: any bitvector of length $n$ with $m$ bis equal to '1' can be represented as a set of cardinality $m$ on universe $\{0,\dots,n-1\}$.
\end{example}

Then, this means that the data structure of the previous section, of size $nH_0 + o(n) = m\log(n/m) + O(m) + o(n)$ can be used to store also any set of cardinality $m$ on universe $\{0,\dots,n-1\}$. Notice that this space is almost optimal, except for the term $o(n)$. If $m\approx n$, then this term is indeed small and can be ignored. On the other hand, if $m\ll n$ then this term is too large! Recall that in our data structure this extra term is $O(n\log\log n / \log n)$. This can be much larger than $m$ if the set is very sparse.

\begin{example}
    Consider the problem of storing $m = 1000$ IPv4 addresses. Each such address is an integer of 32 bits, i.e. $n = 2^{32}$. 
    A naive solution using 32 bits per address would use $32\cdot 1000 = 32000$ bits $\approx 4$ KiB. 
    The information-theoretic lower bound, on the other hand, is approximately $m\log(n/m) \approx 3$ KiB.
    The solution of the previous section, on the other hand, uses $m\log(n/m) + O(n\log\log n / \log n) + O(m) \geq m\log(n/m) + n\log\log n / \log n \approx 6.7\cdot 10^8$ bits $\approx 80$ MiB! The term $O(n\log\log n / \log n) = o(n)$ is negligible only when $m \approx n$ (not in this case). 
\end{example}

We now introduce a data structure removing this extra $o(n)$ term. In the literature, the data structure is known with the name \emph{Elias-Fano} (EF), from the name of the researchers that first described these ideas. We present the EF data structure via an example. For simplicity we assume that $n$ and $m$ are powers of two (but with more care the structure can be described for arbitrary $n,m$).

\begin{example}
    We are going to build the Elias-Fano data structure on set $\{0,5,8,12,14,17,20,31\} \subseteq \{0,\dots, 31\}$. The universe's size is $n=32$, and the set's cardinality is $m=8$.
    We will first write the (sorted) integers in binary, using $\log n = 5$ bits each:
    $$
    \begin{array}{c|c}
        x\in S & (x)_2 \\\hline
        0 & 00000 \\
        5 & 00101 \\
        8 & 01000 \\
        12 & 01100 \\
        14 & 01110 \\
        17 & 10001 \\
        20 & 10100 \\
        31 & 11111 
    \end{array}
    $$
    The next step is to break each integer into a prefix of $\log m = 3$ bits, and a suffix of $\log n - \log m = \log(n/m) = 2$ bits. The suffixes form the first component $EF_1$ of the EF data structure:
    $$
    \begin{array}{c|c|c}
        x\in S & \mathrm{prefix\ of} (x)_2 & EF_1 = \mathrm{suffix\ of} (x)_2\\\hline
        0 & 000 & 00 \\
        5 & 001 & 01 \\
        8 & 010 & 00 \\
        12 & 011 & 00 \\
        14 & 011 & 10 \\
        17 & 100 & 01 \\
        20 & 101 & 00 \\
        31 & 111 & 11 
    \end{array}
    $$

    Note that the suffixes use in total $m\log(n/m)$ bits. We store $EF_1 = (00,01,00,00,10,01,00,11)$ with the packed array data structure of Theorem \ref{thm:packed array}. 
    As far as the prefixes are concerned, observe that they form a list of $m$ nondecreasing integers of $\log m$ bits each. We will store the differences between consecutive integers in the list, in unary encoding: difference $\delta_i = p_i-p_{i-1}$, where $p_i$ is the $i$-th prefix, is encoded as $0^{\delta_i}1$ (i.e. $\delta_i$ occurrences of bit '0' followed by one bit '1'). For example, the difference $\delta_8$ between the two adjacent prefixes $p_7 = 101$ and $p_8 = 111$ is $(111 - 101)_2  =2$; we encode this difference as $\delta_8 = 0^21 = 001$. 
    The first difference is defined as $\delta_1 = p_1$. The concatenated bits of this encoding of $\delta_1, \dots, \delta_m$ is the second part $EF_2$ of the EF data structure:
    
    $$
    \begin{array}{c|c|c|c}
        x\in S & \mathrm{prefix\ of} (x)_2 & EF_2 = \mathrm{difference}\ \delta_i\ \mathrm{in\ unary} & EF_1 = \mathrm{suffix\ of} (x)_2\\\hline
        0 & 000 & 1 & 00 \\
        5 & 001 & 01 & 01 \\
        8 & 010 & 01 & 00 \\
        12 & 011 & 01 & 00 \\
        14 & 011 & 1 & 10 \\
        17 & 100 & 01 & 01 \\
        20 & 101 & 01 & 00 \\
        31 & 111 & 001 & 11 
    \end{array}
    $$
    We build the data structure of Theorem \ref{thm:bitvector} on bitvector $EF_2 = 101010110101001$. The Elias-Fano data structure is the pair $EF = (EF_1,EF_2)$.
\end{example}

Observe that the bitvector $EF_2$, storing differences in unary, has $m$ bits equal to '1'. Also the number of '0' is at most $m$, because the total number of '0' is equal to the difference between the largest integer (in the example, $(111)_2$) and $(000)_2$; this difference is at most $m-1$ since the integers' prefixes are formed by $\log m$ bits each. As a result, $EF_2$ has at most $2m$ bits. We conclude that the Elias-Fano data structure $EF=(EF_1,EF_2)$ uses in total $m\log(n/m) + \Theta(m)$ bits.

We now show how to access any integer efficiently:

\begin{example}
    Consider the previous example. Suppose we want to extract the 5-th integer $14 = (01110)_2$. The suffix $(10)_2$ of this number is explicitly stored in the packed array $EF_1$, so we can extract it in constant time. To extract the prefix $011$ of the number, observe that $(011)_2 = 3$ is  equal to the number of '0's before the 5-th '1' in $EF_2$. In general, the prefix of the $i$-th integer is equal to $EF_2.rank_0(EF_2.select_1(i))$. Why? because the '0's in $EF_2$ store the differences between adjacent prefixes, so by adding up those differences we retrieve the prefixes.
\end{example}

Once obtained the prefix $p$ and suffix $s$ of an integer, we combine them in constant time to get the original integer: $(p \ll k) + s$, where $k=\log(n/m)$ is the length of the suffixes. We obtain:

\begin{theorem}\label{thm:Elias-Fano}
    The Elias-Fano data structure represents a sorted list of $m$ integers from universe $\{0,\dots,n-1\}$ using optimal $m\log(n/m)+\Theta(m)$ bits of space, and allows retrieving the $i$-th integer (random access) in $O(1)$ time. 
\end{theorem}

The Elias-Fano structure actually supports other interesting queries such as membership of integers in the set (i.e. testing $x\in S$), and finding predecessors  (i.e. finding the largest $x\in S$ such that $x < y$, where $y$ is the query) and successors, all in \emph{logarithmic} time. 
Notice that these query times are larger than the $O(1)$-query times of the structure of Section \ref{sec:zero order bv}. This is the price we have to pay to remove the $O(n\log\log n / \log n)$ term from the space of our data structure!
We do not describe these queries since we will not need them in these notes. Exercise: try to show how to support these queries using $EF = (EF_1,EF_2)$.

\section{The compressed suffix array (CSA)}\label{sec:CSA}


We are finally ready to show how to compress the suffix array (SA) of Section \ref{sec:suffix array} in $nH_0 + \Theta(n)$ bits of space.

We are going to introduce the Compressed Suffix Array (CSA) via a running example. Consider Figure \ref{fig:SA banana}, showing the suffix array of string $S = BANANA\$$.

\begin{figure}[h!]
    \centering
    \begin{tabular}{ccccccccc}
        \gr {$i$}     & \gr = & \gr 1 & \gr 2 & \gr 3 & \gr 4 & \gr 5 & \gr 6 & \gr 7 \\\hline 
        $SA[i]$ & = & 7 & 6 & 4 & 2 & 1 & 5 & 3 \\\hline 
        sorted suffixes & : & \$ & A  & A  & A  & B  & N  & N  \\  
                        &   &    & \$ & N  & N  & A  & A  & A  \\  
                        &   &    &    & A  & A  & N  & \$ & N  \\  
                        &   &    &    & \$ & N  & A  &    & A  \\  
                        &   &    &    &    & A  & N  &    & \$ \\  
                        &   &    &    &    & \$ & A  &    &    \\  
                        &   &    &    &    &    & \$ &    &    \\  
    \end{tabular} \caption{The suffix array and sorted suffixes of $S = BANANA\$$.} \label{fig:SA banana}
\end{figure}

Recall that we can find the suffix array range of all suffixes prefixed by a given pattern using binary search. For example, the range of pattern $AN$ is $SA[3,4]$. At each step of binary search, we find ourselves on a particular position $SA[i]$ of the suffix array and we jump on the string $S$ in position $SA[i]$ to extract the suffix $S[SA[i], SA[i]+1, \dots]$, to be compared with the pattern. The first observation is that, in order to run binary search, we really only need to be able to extract the $i$-th suffix in alphabetic order. If we can do that, we don't need the suffix array! Let's therefore thrash the suffix array. We also keep only the first character of each sorted suffix. See Figure \ref{fig:sorted suffixes}: for clarity we show the whole suffixes, but we actually store only the characters in black (those in gray are shown just for clarity). We are going to call $F$ the string storing the first character of each suffix (in alphabetic order).

\begin{figure}[h!]
    \centering
    \begin{tabular}{ccccccccc}
        \gr {$i$}     & \gr = & \gr 1 & \gr 2 & \gr 3 & \gr 4 & \gr 5 & \gr 6 & \gr 7 \\\hline 
        F & = & \$ & A  & A  & A  & B  & N  & N  \\  
                        &   &    & \gr \$ & \gr N  & \gr N  & \gr A  & \gr A  & \gr A  \\  
                        &   &    &    & \gr A  & \gr A  & \gr N  & \gr \$ & \gr N  \\  
                        &   &    &    & \gr \$ & \gr N  & \gr A  &    & \gr A  \\  
                        &   &    &    &    & \gr A  & \gr N  &    & \gr \$ \\  
                        &   &    &    &    & \gr \$ & \gr A  &    &    \\  
                        &   &    &    &    &    & \gr \$ &    &    \\  
    \end{tabular} \caption{The sorted suffixes of $S = BANANA\$$. Let's store just the first character (in black) of each suffix, in a string $F = \$AAABNN$.} \label{fig:sorted suffixes}
\end{figure}

We know that the first character of the $i$-th suffix is $F[i]$. How about the second, third, fourth character? Suppose $i = 3$, so the $i$-th suffix is $ANA\$$. The second character of $ANA\$$ is equal to the first character of $NA\$$, i.e. suffix $ANA\$$ without the first letter. The third character of $ANA\$$ is equal to the first character of $A\$$, i.e. suffix $NA\$$ without the first letter. This observation suggests to link suffixes: we create an array $\psi$ and define $\psi[i]$ to be the position $j$ such that, if we remove the first character from the $i$-th suffix, we obtain the $j$-th suffix. In the example above, $ANA\$$ is the 3rd suffix and $NA\$$ is the 6th suffix, so $\psi[3]=6$.  See Figure \ref{fig:psi}.

\begin{figure}[h!]
    \centering
    \begin{tabular}{ccccccccc}
        \gr {$i$}     & \gr = & \gr 1 & \gr 2 & \gr 3 & \gr 4 & \gr 5 & \gr 6 & \gr 7 \\\hline 
        \gr{$SA[i]$} & \gr = & \gr 7 & \gr 6 & \gr 4 & \gr 2 & \gr 1 & \gr 5 & \gr 3 \\\hline 
        $\psi[i]$ & = & - & 1  & 6  & 7  & 4  & 2  & 3  \\  \hline
        F & = & \$ & A  & A  & A  & B  & N  & N  \\  
                        &   &    & \gr \$ & \gr N  & \gr N  & \gr A  & \gr A  & \gr A  \\  
                        &   &    &    & \gr A  & \gr A  & \gr N  & \gr \$ & \gr N  \\  
                        &   &    &    & \gr \$ & \gr N  & \gr A  &    & \gr A  \\  
                        &   &    &    &    & \gr A  & \gr N  &    & \gr \$ \\  
                        &   &    &    &    & \gr \$ & \gr A  &    &    \\  
                        &   &    &    &    &    & \gr \$ &    &    \\  
    \end{tabular} \caption{The sorted suffixes of $S = BANANA\$$. Let's store just the first character (in black) of each suffix, in a string $F = \$AAABNN$.} \label{fig:psi}
\end{figure}

In Figure \ref{fig:psi} we showed again the suffix array SA to make it clear that array $\psi$ is just a transformation of SA. If we denote by $SA^{-1}$ the \emph{inverse suffix array}, i.e. the array such that $SA^{-1}[SA[i]] = i$, then we have that:

\begin{definition}\label{def: psi}
    For all $2\le i \le n$:
    $$
    \psi[i] = SA^{-1}[SA[i]+1]
    $$
\end{definition}

\begin{example}
    Let's see why the formula of Definition \ref{def: psi} is true. Given a position $i$ in the suffix array, the formula tells us to first extract the suffix array entry $SA[i]$, and add 1: $j = SA[i]+1$. This is equivalent to moving one position forward in the text or, equivalently, to remove the first character from suffix $S[SA[i], SA[i]+1, \dots]$, obtaining suffix $S[SA[i]+1, SA[i]+2, \dots]$. Finally, $SA^{-1}[j]$ is defined to be the position $z$ in SA such that $SA[z]=j$.

    For example, let's see how to compute $\psi[3]$ with the formula. We extract $SA[3]=4$, and add 1: $SA[3]+1 = 5$. Finally, $SA^{-1}[SA[3]+1] = SA^{-1}[5]$ is the position in SA containing number 5: this position is $SA[6]=5$, so we obtain $\psi[3] = 6$.
\end{example}

Arrays $\psi$ and $F$ are all we need to perform binary search. To see this, observe that they can be used to extract the $i$-th suffix in lexicographic order.

\begin{example}
    Let's see how to extract the 4-th suffix in lexicographic order (ANANA\$). By definition, $F[4] = A$ is the first character of this suffix. We then jump to position $\psi[4] = 7$. By definition of $\psi$, the 7-th suffix (NANA\$) is obtained from the 4-th suffix by removing the first letter. But then, $F[\psi[4]] = F[7] = N$ is the second letter of ANANA\$. This process allows us to extract the whole suffix ANANA\$:
    $$
    \begin{array}{cc}
                     & F[4] = A  \\
        \psi[4] = 7; & F[7] = N\\ 
        \psi[7] = 3; & F[3] = A\\ 
        \psi[3] = 6; & F[6] = N\\ 
        \psi[6] = 2; & F[2] = A\\ 
        \psi[2] = 1; & F[1] = \$\\ 
    \end{array}
    $$
\end{example}

In general, we can state:

\begin{lemma}
    The $k$-th letter ($k\ge 1$) of the $j$-th suffix (in lexicographic order) is 
    $$
    F[\psi^{k-1}[j]]
    $$
    where notation $\psi^{k}$ is defined recursively as follows:
    $$
    \begin{array}{ll}
        \psi^0[j] = j & \mathrm{and} \\
         \psi^k[j] = \psi[\psi^{k-1}[j]] & \mathrm{for}\ k>0 
    \end{array}
    $$
\end{lemma}

Remember, from Section \ref{sec:suffix array}, that all we needed to run binary search on the suffix array was the ability to extract any suffix in lexicographic order. It follows that, if we have constant-time random access on $F$ and $\psi$, then we can solve query \emph{count} in the same time as in section \ref{sec:suffix array}: $O(m\log n)$. We now show how to compress arrays $F$ and $\psi$ while guaranteeing constant-time random access to them. 

\subsection{Compressing $F$ and $\psi$ and counting patterns}

String $F$ contains all characters of $S$, in alphabetic order. This string is therefore easy to compress. The idea is to store a bitvector $FO$ (first occurrence) marking with a bit '1' the first occurrence of each character in $F$, and build our constant-time rank data structure of Section \ref{sec:zero order bv} on $FO$. Then, we store in array $\Sigma$ all alphabet characters in alphabetic order (without duplicates), stored using $w$ bits each (or less; for example, on ASCII alphabet, 7 bits per character suffice). 
Note: array $\Sigma$ contains all the alphabet's characters, so we use for it the same symbol $\Sigma$ used to denote the alphabet.
Finally, we observe that $F[i] = \Sigma[FO.rank_1(i)]$. See Figure \ref{fig:FO} for an example.

    \begin{figure}[h!]
    \centering
    \begin{tabular}{ccccccccc}
        \gr {$i$}     & \gr = & \gr 1 & \gr 2 & \gr 3 & \gr 4 & \gr 5 & \gr 6 & \gr 7 \\\hline 
        \gr F & \gr= & \gr{\$} & \gr A  & \gr A  & \gr A  & \gr B  & \gr N  & \gr N  \\ \hline
        FO & = & 1 & 1 & 0 & 0 & 1 & 1 & 0 \\
        $\Sigma$ & = & \$ & A & B & N & & & 
    \end{tabular} \caption{Suppose we want to extract $F[4]=A$. Then, we compute $FO.rank_1(4) = 2$ and access $FO[FO.rank_1(4)] = FO[2] = A$. Similarly, to access $F[5]=B$ we compute $FO.rank_1(5) = 3$ and access $FO[FO.rank_1(5)] = FO[3] = B$. } \label{fig:FO}
\end{figure}

We obtain: 

\begin{lemma}
 $FO$ and $\Sigma$ use at most $n + o(n) + \sigma\cdot w = O(n + \sigma\log n)$ bits of space and allow extracting any $F[i]$ in constant time.
\end{lemma}

\begin{remark}
    Note that ASCII and UNICODE alphabets are constant-sized (their size does not depend on the text length $n$), so in those cases $\sigma = O(1)$ and the space in the above lemma simplifies to $O(n)$.
\end{remark}

To compress $\psi$, we observe that it is \emph{piecewise increasing}. For $c\in \Sigma \setminus \{\$\}$, denote with $\psi_c$ the sub-array corresponding to positions $\psi[i]$ such that $F[i]=c$. 

\begin{example}
Continuing our example, we have:
$$
    \begin{array}{lll}
        \psi_A & = & 1,6,7 \\
        \psi_B & = & 4 \\
        \psi_N & = & 2,3 \\
    \end{array}
$$
Observe that each of these three arrays is increasing. 
\end{example}

To see why each $\psi_c$ is increasing, consider any two suffixes starting with the same letter, for example the 3-rd (ANA\$) and 4-th (ANANA\$). We want to show that $\psi[3]<\psi[4]$. Since $ANA\$ < ANANA\$$ and they start with the same letter, the order is preserved if we remove the first letter: $\xcancel{A}NA\$ < \xcancel{A}NANA\$$. By definition of $\psi$, suffixes $NA\$$ and $NANA\$$ are the $\psi[3]$-th and $\psi[4]$-th in lexicographic order. But then, $NA\$ < NANA\$$ is equivalent to $\psi[3]<\psi[4]$.

We are going to encode with an Elias-Fano data structure (Section \ref{sec:Elias-Fano}) each $\psi_c$. Amazingly, in this way we achieve zero-order compression:

\begin{lemma}
    The Elias-Fano data structures $EF(\psi_c)$ (for $c\in\Sigma$) use in total $nH_0 + \Theta(n)$ bits of space and allow retrieving any $\psi_c[i]$ in $O(1)$ time.
\end{lemma}
\begin{proof}
    Let $n_c$ be the number of characters equal to $c$ in $S$. In our running example, $n_A = 3$, $n_B=1$, and $n_N = 2$. $\psi_c$ contains $n_c$ integers of size at most $n$ each. Using Theorem \ref{thm:Elias-Fano}, we obtain:
    $$
    \begin{array}{lll}
        \sum_{c\in \Sigma} \mathrm{bitsize}(EF(\psi_c))& \leq  & \sum_{c\in \Sigma} \big( n_c \log(n/n_c) + \Theta(n_c) \big) \\
        & =  & \sum_{c\in \Sigma} \big( n_c \log(n/n_c)\big) + \Theta(n) \\
        & =  & n\sum_{c\in \Sigma} \big( (n_c/n) \log(n/n_c)\big) + \Theta(n) \\
        & =  & nH_0 + \Theta(n) \\
    \end{array}
    $$
\end{proof}

The last thing to do is to show how to access $\psi[i]$ given arrays $\psi_c$. We are going to use again bitvector $FO$. Similarly to what we did to simulate string $F$, it is not hard to see that:

\begin{lemma}
    Let $j = FO.select_1(FO.rank_1(i))$ and $c=F[i]$. Then:
    $$
    \psi[i] = \psi_{c}[i-j+1] 
    $$
\end{lemma}

\begin{remark}
    $j = FO.select_1(FO.rank_1(i))$ is the position of the bit '1' immediately preceding (or equal to) position $i$.
\end{remark}

\begin{example}
    Suppose we want to extract $\psi[4]$. Then, $F[4] = A$ (we know how to retrieve $F[4]$ using $FO$ and $\Sigma$), and $j = FO.select_1(FO.rank_1(4)) = FO.select_1(2) = 2$. 
    We obtain $\psi[4] = \psi_{A}[4 -2+1] = \psi_{A}[3] = 7$.
\end{example}

Putting everything together:

\begin{lemma}\label{lem:CSA count}
    The Compressed Suffix Array uses $nH_0 + \Theta(n) + \Theta(\sigma\log n)$ bits of space and supports \emph{count} queries in $O(m\log n)$ time. 
\end{lemma}

\begin{remark}
    Again, ASCII and UNICODE alphabets are constant-sized (their size does not depend on the text length $n$), so in those cases $\sigma = O(1)$ and the CSA uses $nH_0 + \Theta(n)$ bits of space.
\end{remark}

\subsection{Locating patterns and extracting substrings}\label{sec:locate,extract}

We are left to show how to locate all $occ$ occurrences of a pattern $P$ in the indexed string $S$, i.e. how to retrieve $SA[l], SA[l+1], \dots, SA[r]$ where $[l,r]$ is the suffix array range of all suffixes prefixed by $P$ and $occ = l-r+1$ (for instance, in our running example: if $P=``AN''$ then $l=2$ and $r=3$; if $P=``A''$ then $l=2$ and $r=4$). With Lemma \ref{lem:CSA count} we can find $l$ and $r$ by binary search in $O(m\log n)$ time and $nH_0 + \Theta(n)$ bits of space. We now show how to extract any $SA[i]$ in $O(\log n)$ time without changing asymptotically the space usage.

The idea is that function $\psi$ moves us \emph{forward} in the text by one position. For instance, in our running example: $SA[4]=2$, $SA[\psi[4]] = SA[7] = 3$, $SA[\psi[7]] = SA[3] = 4$, etc. Choose $\rho = \lceil \log n\rceil$; parameter $\rho$ is called the \emph{sample rate}. We explicitly store $SA[i]$ either if $SA[i]\mod \rho = 0$, or if $SA[i]=n$.
The sampled values of $SA$ are stored in an array SSA called the \emph{sampled suffix array}. We moreover use a bitvector $M$ to mark with a bit '1' the entries of SA that have been sampled. See Figure \ref{fig:SSA} for a running example.

\begin{figure}[h!]
    \centering
    \begin{tabular}{ccccccccc}
        \gr {$i$}     & \gr = & \gr 1 & \gr 2 & \gr 3 & \gr 4 & \gr 5 & \gr 6 & \gr 7 \\\hline 
        \gr{$SA[i]$} & \gr = & \gr 7 & \gr 6 & \gr 4 & \gr 2 & \gr 1 & \gr 5 & \gr 3 \\\hline 
        $M[i]$ & = & 1 & 1  & 0  & 0  & 0  & 0  & 1 \\\hline 
        $SSA[i]$ & = & 7 & 6 & 3  &  &   &   &  \\\hline 
        $\psi[i]$ & = & - & 1  & 6  & 7  & 4  & 2  & 3 
    \end{tabular} \caption{Sampled Suffix array SSA and bitvector M on our running example. We have $\rho = \lceil \log n\rceil = 3$, so we sample SA values 3,6 (divisible by $\rho = 3$) and $7=n$. Elements in gray are not stored.} \label{fig:SSA}
\end{figure}

The idea is that, if $M[i]=1$ then $SA[i]$ is sampled and we can retrieve it from SSA with our usual rank trick: 
$$
SA[i] = SSA[M.rank_1(i)]
$$

Otherwise ($M[i]=0$), we apply $\psi$ to $i$ for $k$ times until obtaining a value $j = \psi^k[i]$ such that $M[j]=1$. Since we sample every other $\rho = \Theta(\log n)$ entries of SA, we will apply $\psi$ at most $k \in O(\log n)$ times before finding a sampled value. 
The number $k\geq 0$ of times we apply $\psi$ is therefore the smallest integer such that $M[\psi^k[i]] = 1$. 
Applying $k$ times function $\psi$ means walking by $k$ positions forward on string $S$ so in general the formula to retrieve any $SA[i]$ is:

$$
SA[i] = SSA[M.rank_1(\psi^k[i])] - k,\ \mathrm{where\ }k\geq 0\mathrm{\ is\ the \ smallest\ integer\ s.t.\ }M[\psi^k[i]] = 1
$$

\begin{example}
    Suppose we want to extract $SA[5]=1$. We have $M[5]=0$, so we apply $\psi[5]=4$. Also $M[4]=0$, so
    we apply again $\psi[4]=7$. Now, $M[7]=1$ so we stop (we found a sampled SA value). Note that we applied $\psi$ in total $k=2$ times, so our formula gives:
    $$
    \begin{array}{lll}
        SA[5] & = & SSA[M.rank_1(\psi^2[5])] - 2 \\
              & = & SSA[M.rank_1(7)] - 2 \\
              & = & SSA[3] - 2 \\
              & = & 3 - 2 \\
              & = & 1
    \end{array}
    $$
\end{example}

Array $SSA$ contains $n/\rho + 1 = O(n/\log n)$ integers of $\log n$ bits each, so it uses $O(n)$ bits of space. Array $M$ occupies $n+o(n)$ bits. Each application of function $\psi$ costs $O(1)$ time, and to extract any $SA[i]$ we apply $\psi$ at most $O(\log n)$ times. It follows that we can extract any $SA[i]$ (i.e. solve \emph{locate queries}) in $O(\log n)$ time using in total $nH_0 + \Theta(n)$ bits of space. 

\begin{Exercise}\label{ex:extract}
    Show how to support query \emph{extract}$(i,\ell)$: given integers $i,\ell$, extract from the compressed index the substring $S[i, \dots, i+\ell-1]$. The query's running time must be $O(\log n + \ell)$ and the space must not increase:  $nH_0 + \Theta(n)$ bits in total.
    Hint: sample the array $SA^{-1}$, similarly to how we sampled $SA$ to support \emph{locate}.
\end{Exercise}

An index supporting query \emph{extract} of Exercise \ref{ex:extract} is called a \emph{self-index} because it does not need the string $S$ to be stored: the index is sufficient to solve pattern matching \emph{and} extracting any substring of $S$ efficiently.

We finally obtained our first \emph{full-text compressed self-index}:

\begin{theorem}\label{thm:CSA}
The Compressed Suffix Array (CSA) uses $nH_0 + \Theta(n) + \Theta(\sigma\log n)$ bits of space and supports: 
\begin{itemize}
    \item \emph{count} queries in $O(m\log n)$ time,
    \item \emph{locate} queries in $O((m+occ)\log n)$ time, and
    \item \emph{extract} queries (extract any substring of length $\ell$ of $S$) in $O(\log n+\ell)$ time. 
\end{itemize}
\end{theorem}

\begin{remark}
    Again, ASCII and UNICODE alphabets are constant-sized (their size does not depend on the text length $n$), so in those cases $\sigma = O(1)$ and the CSA uses $nH_0 + \Theta(n)$ bits of space.
\end{remark}

More advanced techniques \cite{CSA03} work for any alphabet size $\sigma$ and can reduce the space to $nH_k + o(n\log\sigma)$ bits, for an appropriate entropy order $k$.

\section{The FM-index}\label{sec:FMI}

In this section we present an index of size $nH_k + o(n\log\sigma)$ bits supporting faster queries than the CSA.
We need two ingredients:

\begin{itemize}
    \item Wavelet trees (Section \ref{sec:WT}): given a string $S$ on arbitrary alphabet, $WT(S)$ is an elegant data structure supporting \emph{rank, select, access}, and other queries on $S$.
    \item The Burrows-Wheeler transform (Section \ref{sec:BWT}). $BWT(S)$ is a magic permutation of $S$ with many interesting properties. In particular:
    \begin{enumerate}
        \item local zero-order entropy compression of $BWT(S)$ is equivalent to high-order entropy compression of $S$, and 
        \item \emph{count} queries on $S$ can be reduced to simple operations on $BWT(S)$ (mainly \emph{rank}).
    \end{enumerate}
\end{itemize}

The FM-index of $S$ is simply $WT(BWT(S))$. This data structure (explained in Section \ref{subsec:FMI}) achieves high-order entropy compression and supports \emph{count} queries on $S$. Adding the same additional structures of Section \ref{sec:locate,extract}, we will finally gain also \emph{locate} and \emph{extract} queries.

\subsection{Wavelet Trees: WT}\label{sec:WT}

In this section we generalize the result obtained in Section \ref{sec:zero order bv}  to strings on \emph{general alphabets}. More formally, we want to design a data structure satisfying these bounds:

\begin{definition}[zero-order string data structure]
    Given a string $S \in \Sigma^n$, a zero-order string data structure is a structure using $nH_0(S) + o(n\log\sigma)$ bits (where $\sigma = |\Sigma|$) and supporting efficiently the following queries:
    \begin{enumerate}
        \item Random access: extract any symbol $S[i]$
        \item Rank: $S.rank_c(i)$ is the number of characters equal to $c\in\Sigma$ in the prefix $S[1,\dots, i]$.
        \item Select: $S.select_c(i)$ is the position in $S$ of the $i$-th character equal to $c\in\Sigma$.
    \end{enumerate}
\end{definition}

\emph{Wavelet trees}, introduced by Grossi et al.~\cite{CSA03},  are an extremely elegant data structure that permits to reduce the above (and many other) queries to rank/select/access queries on bitvectors (a problem that we learned how to solve in Section \ref{sec:zero order bv}).  

The wavelet tree depends on a prefix-free encoding $\gamma:\Sigma \rightarrow \{0,1\}^+$ of the alphabet. 
In a wavelet tree, the string is represented as a binary tree with $\sigma$ leaves. 
Each internal node has at most two outgoing edges, labeled '0' and '1'.
The concatenated labels of the edges on the path from the root to the leaf corresponding to character $c$ yield $\gamma(c)$.
Each internal node of the tree stores a bitvector data structure.
We describe the wavelet tree data structure in Example \ref{ex: WT1} and Figure \ref{fig:wt}:

\begin{figure}[h!]
		\begin{center}
			\includegraphics[width=0.5\textwidth]{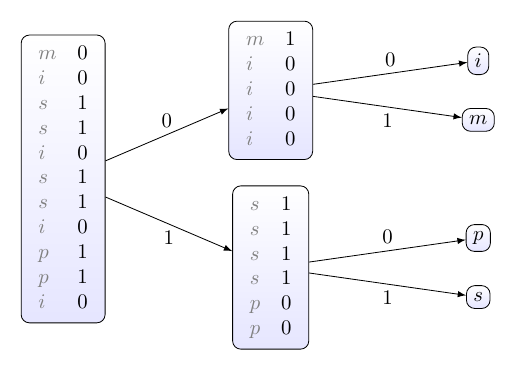}
		\end{center}\caption[Balanced wavelet tree]{Wavelet tree of the string $mississippi$. The $i$-th level stores the $i$-th bit in the encoding of the characters. 
        Since we use a balanced encoding ($\log\sigma$ bits per character), the wavelet tree is balanced and has height $\log\sigma$.
        There are $\sigma$ leaves, each labeled with a character. Characters in gray are shown just for clarity but are not explicitly stored in the data structure.}\label{fig:wt}
\end{figure}
 
\begin{example}\label{ex: WT1}
We show how to build the WT of string $S = mississippi$. We need any prefix-free binary encoding of the alphabet; here we choose $\gamma(i) = 00$, $\gamma(m) = 01$, $\gamma(p) = 10$, $\gamma(s) = 11$. Figure \ref{fig:wt} depicts $WT(S)$. The bitvector in the root is formed by the first bit of the binary encodings of the string's characters. The root has two outgoing edges, labeled '0' and '1'. We send to child '0' all string's characters whose encoding starts with a '0', and to child '1' all string's characters whose encoding starts with a '1'. We repeat recursively: in the second level, we extract the second bit of the encoding of every character and use this bit to decide if the character goes to child '0' or '1'. In general, in level $i$ ($i=1$ is the root) the bitvectors store the $i$-th bit in the encodings of the characters.

\end{example}

\begin{remark}
    We \emph{do not} store characters (shown in gray in the figure) inside internal nodes. In Example \ref{ex: WT1}, characters are shown inside internal nodes only for explanatory purposes.
\end{remark}

We now show how to solve \emph{access} and \emph{rank} queries on $WT(S)$. 
We will not need \emph{select} in these notes, so we do not show how to solve it (exercise: do it yourself).

We use the following notation: $root$ is the root node of the WT, $N.bitvector$ is the bitvector associated with the node $N$, and $N.child(b)$, with $b\in\{0,1\}$ is the node child of $N$ reached following the edge labeled with bit $b$ from node $N$. If $N$ is a leaf, $N.label$ is the character stored in the leaf.

\paragraph{Random access}

To extract $S[i]$ using $WT(S)$, start from the bitvector $root.bitvector$ at the root of the wavelet tree.  $root.bitvector[i]$ is the first bit of $\gamma(S[i])$. 
Note that the bitvector $root.child(0).bitvector$ 
in child '0' of the root
stores the second bit of the encoding of characters in $S$ whose encoding starts with bit '0' (similarly for bitvector $root.child(1).bitvector$). Let $b = root.bitvector[i]$. It follows that the second bit of $\gamma(S[i])$ is $root.child(b).bitvector[root.bitvector.rank_b(i)]$. By recursively repeating this strategy, we finally reach the leaf of $WT(S)$ labeled with character $S[i]$. Since we descend the tree from the root to a leaf and spend only constant time at each level (one access and one rank), running time is $O(|\gamma(S[i])|)$ (number of bits of the encoding of $S[i]$). If the tree is balanced, this running time is $O(\log\sigma)$.\\

\begin{algorithm}[H]\label{alg: WT access}
	\footnotesize
	\caption{access $S[i]$}
	\SetKwInOut{Input}{input}
	\SetKwInOut{Output}{output}
	\SetSideCommentLeft
	\LinesNumbered
	
	\BlankLine
	
	$N\leftarrow root$\;
	
	\While{$N\ is\ not\ leaf$}{
	
		$B \leftarrow N.bitvector$\;
		$b \leftarrow B[i]$\;
		$N \leftarrow N.child(b)$\;
		$i \leftarrow B.rank_b(i)$\;
		
	}
	
	\Return $N.label$\;
	
\end{algorithm}

\paragraph{Rank}

Rank does not differ much from access. Let us start from bitvector $root.bitvector$ at the root of the tree. In order to answer $S.rank_c(i)$ using $WT(S)$, first count how many characters in $S$ start with bit $b = \gamma(c)[1]$ (first bit of $\gamma(c)$) before position $i$: $i' = root.bitvector.rank_b(i)$. Then, move at position $i'$ of the bitvector $root.child(b).bitvector$ and repeat with the second bit of $\gamma(c)$. By recursively repeating this strategy for all bits of $\gamma(c)$, we reach the leaf of $WT(S)$ labeled with character $c$. The last rank operation performed at the parent of this leaf yields exactly the value $S.rank_c(i)$. The running time of this procedure is $O(|\gamma(c)|)$. If the tree is balanced, this running time is $O(\log\sigma)$.\\

\begin{algorithm}[H]\label{alg: WT rank}
	\footnotesize
	\caption{$rank_c(i)$}
	\SetKwInOut{Input}{input}
	\SetKwInOut{Output}{output}
	\SetSideCommentLeft
	\LinesNumbered
	
	\BlankLine
	
	$N\leftarrow root$\;
	$k\leftarrow 1$\;
	
	\While{$N\ is\ not\ leaf$}{
		
		$B \leftarrow N.bitvector$\;
		$b \leftarrow \gamma(c)[k]$;
		
		$i \leftarrow B.rank_b(i)$\;
		
		$N \leftarrow N.child(b)$\;
		$k\leftarrow k+1$\;
		
	}
	
	\Return $i$\;
	
\end{algorithm}

\paragraph{Space}

We now analyze the space of the WT data structure.

Let  $n'$ be the total number of bits in the bitvectors of $WT(S)$.
The space used by the wavelet tree is equal to the space used by the data structure chosen for storing those bitvectors, plus $O(\sigma\log n)$ bits for the tree topology and the characters stored in the leaves. 
If we use a balanced encoding (like in the previous example), then $n' = n\log\sigma$. 
We are going to build the data structure of Theorem \ref{thm:bitvector} (supporting rank/select/access) on each bitvector in $WT(S)$. Remember that this data structure uses at most $n' + o(n')$ bits. Since $n' = n\log\sigma$, we immediately obtain that the data structure $WT(S)$ uses \emph{at most} $n' + o(n') + \Theta(\sigma\log n) \leq n\log\sigma +o(n\log\sigma) + \Theta(\sigma\log n)$ bits. 

Remarkably, it is possible to improve this bound. Zero-order compression of the bitvectors \emph{propagates to zero-order compression of $S$}:

\begin{lemma}\label{lem:zero-order WT}
    Using the RRR bitvector of Theorem \ref{thm:bitvector} to store the bitvectors of $WT(S)$, the structure $WT(S)$ uses in total $n H_0(S) + o(n\log\sigma) + \Theta(\sigma\log n)$ bits of space and answers all queries in $O(\log\sigma)$ time.
\end{lemma}
\begin{proof}
The proof works inductively: take a node $x$ of the wavelet tree. The node has two children (left/right). Assuming inductively that the sub-Wavelet trees of the two children achieve zero-order compression of their respective sequences, it is possible to prove that also the sub-Wavelet tree rooted in $x$ achieves zero-order compression of its associated sequence. The inductive assumption holds in the lowest nodes, which contain a RRR bitvector (zero-order compressed). Propagating the induction up to the root, this will prove the main claim: the sub-Wavelet tree rooted in the root achieves zero-order compression.  

We show a simplified proof on a Wavelet tree of 2 levels of internal nodes (i.e. ignoring leaves): the root and two children, each containing a RRR bitvector. 
Extending the proof to an arbitrary number $\log\sigma$ of levels is then easy.
The whole Wavelet tree encodes a sequence $S \in \Sigma^n$ over alphabet $\Sigma = \{0,1,2,3\}$, encoded in binary as usual using 2 bits per character. Let $N_0$ and $N_1$ be the number of characters in $S$ whose binary encoding starts with $0$ and $1$, respectively. Observe that $n=N_0+N_1$, and that $N_i$ is the number of bits in the $i$-th children ($i\in\{0,1\}$). 
Let moreover $m_i$  be the number of 1's in children $i\in\{0,1\}$.

Since the $i$-th children, for $i\in\{0,1\}$, is zero-order compressed, it uses
$$
m_i \log\frac{N_i}{m_i} + (N_i-m_i)\log\frac{N_i}{N_i-m_i} + o(N_i)\ \ \ \mathrm{bits}
$$

Summed over the two children, we obtain that the RRR bitvectors in the two children use in total

\begin{equation}\label{eq:child wt zero ord}
\sum_{i\in\{0,1\}} \left( m_i \log\frac{N_i}{m_i} + (N_i-m_i)\log\frac{N_i}{N_i-m_i} \right) + o(n) \ \ \ \mathrm{bits}
\end{equation}

The bitvector in the root, on the other hand, has $N_0$ zeros and $N_1$ ones by definition of Wavelet tree. It uses therefore 

$$
N_0 \log\frac{n}{N_0} + N_1 \log\frac{n}{N_1} + o(n)\ \ \ \mathrm{bits}
$$

Note that we can rewrite $N_i = m_i + (N_i-m_i)$. Substituting this into the previous formula, we obtain that the bitvector in the root uses

\begin{equation}\label{eq:root wt zero ord}
 \sum_{i\in \{0,1\}} \left( m_i \log\frac{n}{N_i} + (N_i-m_i)\log\frac{n}{N_i} \right) + o(n)\ \ \ \mathrm{bits}   
\end{equation}

Summing up quantities \ref{eq:child wt zero ord} and \ref{eq:root wt zero ord} we obtain that the Wavelet tree uses in total

\begin{equation}\label{eq:wt zero ord}
 \sum_{i\in \{0,1\}} \left( m_i \log\frac{n}{m_i} + (N_i-m_i)\log\frac{n}{N_i-m_i} \right) + 2\cdot o(n)\ \ \ \mathrm{bits}   
\end{equation}

of course, the constant 2 multiplying $o(n)$  disappears in asymptotic notation. However, extending the proof to trees of arbitrary depth $\log \sigma$, it is easy to see that this term becomes $o(n\log\sigma)$.

Finally, observe that the four terms of the summation in Equation \ref{eq:wt zero ord} correspond precisely to the frequencies of the four characters in $\Sigma = \{0,1,2,3\}$: 

\begin{itemize}
    \item $N_0-m_0$, the number of zeros in the left children, is precisely the number of characters whose encoding is 00, i.e. the number $n_0$ of $0$'s in sequence $S$. 
    \item $m_0$, the number of ones in the left children, is precisely the number of characters whose encoding is 01, i.e. the number $n_1$ of $1$'s in sequence $S$.
    \item $N_1-m_1$, the number of zeros in the right children, is precisely the number of characters whose encoding is 10, i.e. the number $n_2$ of $2$'s in sequence $S$. 
    \item $m_1$, the number of ones in the right children, is precisely the number of characters whose encoding is 11, i.e. the number $n_3$ of $3$'s in sequence $S$.
\end{itemize}

It follows that Equation \ref{eq:wt zero ord} is equivalent to: 

$$
\sum_{c\in \Sigma} n_c\log\frac{n}{n_c} + 2\cdot o(n)\ \ \ \mathrm{bits}
$$

That is, the zero-order entropy of $S$. Generalizing the argument to Wavelet trees of height $\log \sigma$, we obtain the main claim. 
\end{proof}

To sum up, we obtained:

\begin{theorem}\label{thm:zero-WT}
    Given a string $S\in \Sigma^n$, the balanced wavelet tree $WT(S)$ described in this section uses $nH_0(S) + o(n\log\sigma) + \Theta(\sigma\log n)$ bits of space and supports queries \emph{access}, \emph{rank}, and \emph{select} in $O(\log\sigma)$ time.
\end{theorem}

Let $S_1, \dots, S_t \in \Sigma^*$ be $t$ strings. 
What happens when building the Wavelet tree on the concatenation $S_1\dots S_t$ of the $t$ strings? It is not hard to see that the bitvectors on corresponding nodes (i.e. nodes reached by reading on the tree's edges the same sequence of bits from the root) of the $t$ Wavelet trees of $S_1, \dots, S_t$ get concatenated. 
Since the Wavelet tree has $O(\sigma)$ nodes, Lemma \ref{lem:sum RRR}, directly gives us:

\begin{lemma}\label{lem:sum WT}
    Denote with $\mathrm{WT\_size}(S) = |S| H_0(S) + o(|S|\log\sigma) + O(\sigma\log|S|)$ the bit-size of the Wavelet tree (Theorem \ref{thm:zero-WT}) built on string $S\in \Sigma^*$.
    Let $S_i \in \Sigma^{n_i}$, for $i=1, \dots, t$ be $t$ strings of total length $n = \sum_{i=1}^t n_i$. Then:
    $$
    \mathrm{WT\_size}(S_1S_2\dots S_t) \le \left(\sum_{i=1}^t n_iH_0(S_i)\right) + o(n\log\sigma) + O(t\sigma \log n)
    $$
\end{lemma}

This lemma will be crucial to analyze the size of the FM-index in Section \ref{subsec:FMI}.

\subsubsection{Using another encoding}

In the previous paragraphs, we used a balanced encoding assigning  $\log\sigma$ bits to each character. 
The techniques we described,  however, work with any prefix-free encoding of the alphabet. 

Even if in the FM-index of Section \ref{subsec:FMI} we will use a balanced encoding (so you can safely ignore this paragraph), it is interesting to see what happens when other encodings are used.

The Huffman encoding represents an interesting case: with this encoding, the total number of bits stored in the bitvectors of the wavelet tree is at most $n(H_0+1)$. This number is compressed, even if we do not use zero-order compressed bitvectors inside the internal nodes of the tree! 

\begin{example}
Consider again string $S = mississippi$. Use the Huffman encoding $\gamma(i) = 01$, $\gamma(m) = 001$, $\gamma(p) = 000$, $\gamma(s) = 1$. Figure \ref{fig:hwt} depicts $WT(S)$ with the shape of this encoding.
	\begin{figure}[ht!]
		\begin{center}
			\includegraphics[width=0.65\textwidth]{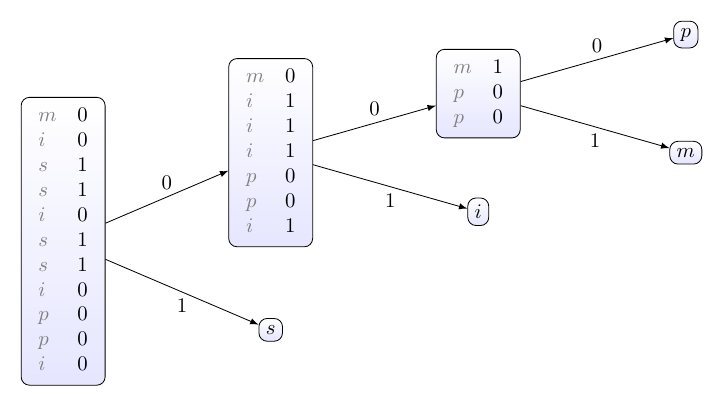}
		\end{center}\caption[Huffman-shaped wavelet tree]{Huffman-shaped wavelet tree of the string $mississippi$. The depth of each leaf is inversely proportional to the frequency of its label. This means that if we access $S[i]$ at a uniform position $i$, on average we traverse $H_0+1$ levels of the wavelet tree. On average, this is less than the $\log\sigma$ levels that we always traverse in balanced wavelet trees (but it can be larger than $\log\sigma$ in the worst case).}\label{fig:hwt}
	\end{figure}
\end{example}

\subsection{The Burrows-Wheeler Transform: BWT}\label{sec:BWT}

In 1994, Michael Burrows and David Wheeler~\cite{burrows1994block} introduced the following string permutation, today known as the \emph{Burrows-Wheeler transform}.
In this section we will assume that the string $S\in \Sigma^n$ ends with a special character $\$$ being alphabetically smaller than all the other characters and not appearing anywhere else in $S$.

\begin{enumerate}
    \item Sort all circular rotations of $S$. This yields a matrix: the BWT matrix of $S$. See Figure \ref{fig:BWT example} for an example.
    \item Concatenate in a new string $BWT(S)$ the last character of each rotation of $S$, in sorted order.
\end{enumerate}

    \begin{figure}[h!]
        \centering
        \setlength{\tabcolsep}{3pt}
        \renewcommand{\arraystretch}{1.15}
        \begin{tabular}{c|cccccccccc|ll}
            F &&&&&&&&&&& L & = $BWT(S)$\\\hline
             \$ & m & i & s & s & i & s & s & i & p & p & \cellcolor{green!25} i &\\
             i & \$ & m & i & s & s & i & s & s & i & p & \cellcolor{green!25} p &\\
             i & p & p & i & \$ & m & i & s & s & i & s & \cellcolor{green!25} s &\\
             i & s & s & i & p & p & i & \$ & m & i & s & \cellcolor{green!25} s &\\
             i & s & s & i & s & s & i & p & p & i & \$ & \cellcolor{green!25} m &\\
             m & i & s & s & i & s & s & i & p & p & i & \cellcolor{green!25} \$ &\\
             p & i & \$ & m & i & s & s & i & s & s & i & \cellcolor{green!25} p &\\
             p & p & i & \$ & m & i & s & s & i & s & s & \cellcolor{green!25} i &\\
             s & i & p & p & i & \$ & m & i & s & s & i & \cellcolor{green!25} s &\\
             s & i & s & s & i & p & p & i & \$ & m & i & \cellcolor{green!25} s &\\
             s & s & i & p & p & i & \$ & m & i & s & s & \cellcolor{green!25} i &\\
             s & s & i & s & s & i & p & p & i & \$ & m & \cellcolor{green!25} i &\\
        \end{tabular}
        \caption{BWT matrix of $S = mississippi\$$. 
        The rows of the BWT matrix contain all and only the circular rotations of $S$.
        The last column $L$ (in green) is $BWT(S) = ipssm\$pissii$. We have already encountered the first column $F = \$iiiimppssss$ in Section \ref{sec:CSA}.}
        \label{fig:BWT example}
    \end{figure}

As we will see, the BWT possesses several remarkable properties which boost compression and enable indexing. At the core of most properties, stands the \emph{LF property} (see Figure \ref{fig: LF}):

\begin{theorem}{\emph{LF property.}}\label{th: LF mapping} 
	The $i$-th occurrence of $c\in\Sigma$ on column L (i.e. $L[L.select_c(i)]$) corresponds to the $i$-th occurrence of $c$ on column F (i.e. $F[F.select_c(i)]$). In other words, these two characters correspond to the same position in the original string $S$.
\end{theorem}

It is not hard to prove the LF-property. The key idea is to think about the \emph{ordering} of equal characters on the two columns. In column F, equal characters (for example, the 'i's) are ordered according to the lexicographic order of the text suffixes that follow them. 
For example, in Figure \ref{fig:BWT example} the second 'i' in column F is followed by $ppi\$$, which is smaller than the suffix $ssissippi\$$ which follows the fourth 'i' in column F.
It is easy to see that the same property holds for equal characters on the L column.
For example, in Figure \ref{fig:BWT example} the second 'i' in column L is followed by $ppi\$$, which is smaller than the suffix $ssissippi\$$ which follows the fourth 'i'.
It follows that occurrences of any $c\in\Sigma$ appear in the same order in the F and L columns.

\begin{figure}[h!]
	\begin{center}
		\includegraphics[trim=0cm 17cm 7cm 3.5cm, clip=true, width=1\textwidth]{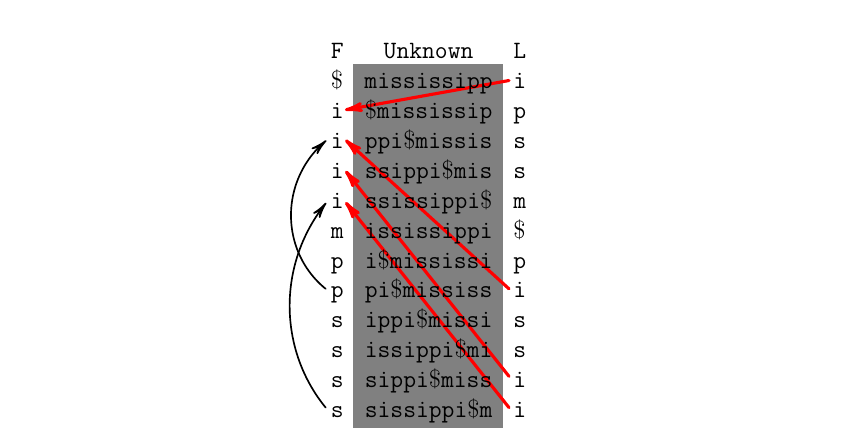}
	\end{center}\caption[LF mapping]{Red arrows: LF property (only on character `i' is shown, but the property holds for any character). Black arrows: function $LF$ induced by the LF property (only $LF(12)=5$ and $LF(8) = 3$ are shown in the figure). Note that $LF$ corresponds to backward navigation of $S$: $SA[LF(i)] = SA[i]-1$. In other words, $LF = \psi^{-1}$ (see Section \ref{sec:CSA}). The central part of the matrix is not needed for all the algorithms that we will present, so here we declare it as ``Unknown''.}\label{fig: LF}
\end{figure}

As shown in Figure \ref{fig: LF}, the LF property yields a function (in fact, a permutation) $LF : [n] \rightarrow [n]$. Note that $LF = \psi^{-1}$ ($\psi$ is the function we defined in Section \ref{sec:CSA} to compress the suffix array). In fact, the FM-index we are going to describe in Section \ref{sec:FMI} is just another way to compress the suffix array!

A first important consequence of the LF mapping is that $BWT(S)$ can be inverted. As displayed in Figure \ref{fig:invert BWT}, by repeatedly applying function $LF$ starting from position $1$ (note: $L[1]$ always contains the character preceding $\$$), we reconstruct $S$ \emph{backwards}. Note that $F$ can be obtained by sorting $BWT(S)$, so all we need to reconstruct $S$ is just $BWT(S)$.

\begin{figure}[h!]
$$
\begin{array}{rcc}
    L[1] & = & i \\
    L[LF[1]] & = & p \\
    L[LF^2[1]] & = & p \\
    L[LF^3[1]] & = & i \\
    L[LF^4[1]] & = & s \\
    L[LF^5[1]] & = & s \\
    L[LF^6[1]] & = & i \\
    L[LF^7[1]] & = & s \\
    L[LF^8[1]] & = & s \\
    L[LF^9[1]] & = & i \\
    L[LF^{10}[1]] & = & m
\end{array}
$$
\caption{Reconstructing $S = mississippi\$$ (backwards) by iteratively applying the LF mapping.}
\label{fig:invert BWT}
\end{figure}

In the next section we will use two additional properties of $BWT(S)$:
\begin{enumerate}
    \item In the BWT matrix, the string's suffixes are sorted lexicographically. This suggests that it should be possible to use $BWT(S)$ to solve pattern matching queries (for the same reason we are able to solve this task on the CSA of Section \ref{sec:CSA}).
    \item The first $k$ columns of the BWT matrix contain all substrings of length $k$ of $S$, in sorted order. But then, $BWT(S)$ is formed by the characters preceding those contexts: in other words,  $BWT(S)$ is the concatenation of strings $S_w$ (see Section \ref{sec:high order entropy}). This suggests that local zero-order entropy compression of $BWT(S)$ yields $k$-order entropy compression of $S$.
\end{enumerate}

\subsection{FM-index = WT(BWT)}\label{subsec:FMI}

The FM-index of string $S$ is $WT(BWT(S))$: a zero-order compressed wavelet tree over the Burrows-Wheeler transform of $S$. In the next subsection we show that this data structure supports efficient $count$ queries, then we use the same strategy of Section \ref{sec:locate,extract} to support also \emph{locate} and \emph{extract} queries. 
Finally, we show that actually $WT(BWT(S))$ automatically achieves high-order entropy-compressed space, \emph{without any further intervention} (just analysis).
This index has been described for the first time by Ferragina and Manzini in the year 2000 \cite{ferragina2000opportunistic}. The original paper \cite{ferragina2000opportunistic} used a more complicated compression strategy without Wavelet trees and only achieved $O(nH_k) + o(n\log \sigma)$ bits of space. The version of the FM-index that we present is more modern, easier to describe and analyze, and more efficient than the original one. 

\subsubsection{\emph{Count} queries: the backward search algorithm}\label{sec:FM count}

We show how to solve query $count(P)$: return the number of occurrences of string $P\in\Sigma^m$ in $S$ (more in general, the query returns the suffix array range $SA[l,r]$ of all suffixes prefixed by $P$).
The reason for this name is that we search the pattern from its last to first character (right-to-left instead of left-to-right as done in suffix trees, suffix arrays, and in the CSA). The algorithm is based on two observations:

\begin{enumerate}
	\item All occurrences of $P$ appear \emph{contiguously} in a range of rows in the BWT matrix, and 
	\item Let $SA[l,r]$ be the range of suffixes prefixed by some string $P'$. Then, the characters in $BWT[l,r]$ \emph{precede} the occurrences of $P'$ in $S$. In other words, String $BWT[i]\cdot P'$ is a substring of $S$ for all $l\leq i \leq r$ (unless $BWT[i]=\$$).
\end{enumerate}

Observations (1) and (2), together with the LF property, imply that we can find the range $[l,r]$ of all suffixes prefixed by $P$ iteratively, in $m$ steps. See Figure \ref{fig:bwsearch} for an example of one such step. The resulting algorithm is called \emph{backward search}.

\begin{figure}[h!]
\setlength{\tabcolsep}{2.5pt}
\renewcommand{\arraystretch}{0.7}
    \centering
    \begin{tabular}{r | c | cccccccccccccccc| l }
\textcolor{gray}{$i$} & F & & & & & & & & & & & & & & & & & $L=BWT(S)$ \\ \hline
{\footnotesize $C[\$] = $} \textcolor{gray}{0} & & & & & & & & & & & & & & & & & & \\
{\footnotesize $C[a] = $} \textcolor{gray}{1} & \$ & a & a & b & b & b & a & b & a & b & b & b & a & a & b & a & b & a \\ 
\textcolor{gray}{2} & a & \$ & a & a & b & b & b & a & b & a & b & b & b & a & a & b & a & b \\ 
\textcolor{gray}{3} & a & a & b & a & b & a & \$ & a & a & b & b & b & a & b & a & b & b & b \\ 
\textcolor{gray}{4} & a & a & b & b & b & a & b & a & b & b & b & a & a & b & a & b & a & \$ {\footnotesize $\leftarrow BWT.rank_b(5-1)+1 = 3$rd  'b' in L} \\\hline 
\textcolor{gray}{5} & \cellcolor{green!25}a & \cellcolor{green!25}b & a & \$ & a & a & b & b & b & a & b & a & b & b & b & a & a & \textcolor{red}{\textbf b}  \\ 
\textcolor{gray}{6} & \cellcolor{green!25}a & \cellcolor{green!25}b & a & b & a & \$ & a & a & b & b & b & a & b & a & b & b & b & a \\ 
\textcolor{gray}{7} & \cellcolor{green!25}a & \cellcolor{green!25}b & a & b & b & b & a & a & b & a & b & a & \$ & a & a & b & b & b \\ 
\textcolor{gray}{8} & \cellcolor{green!25}a & \cellcolor{green!25}b & b & b & a & a & b & a & b & a & \$ & a & a & b & b & b & a & \textcolor{red}{\textbf b} \\ 
{\footnotesize $C[b] = $} \textcolor{gray}{9} & \cellcolor{green!25}a & \cellcolor{green!25}b & b & b & a & b & a & b & b & b & a & a & b & a & b & a & \$ & a {\footnotesize $\leftarrow BWT.rank_b(9) = 5$th 'b' in L} \\\hline 
\textcolor{gray}{10} & b & a & \$ & a & a & b & b & b & a & b & a & b & b & b & a & a & b & a \\ 
\textcolor{gray}{11} & b & a & a & b & a & b & a & \$ & a & a & b & b & b & a & b & a & b & b \\ 
{\footnotesize $3$rd 'b' in F $\rightarrow$} \textcolor{gray}{12} &\cellcolor{blue!25}\textcolor{red}{\textbf b} & \cellcolor{blue!25}a & \cellcolor{blue!25}b & a & \$ & a & a & b & b & b & a & b & a & b & b & b & a & a \\ 
\textcolor{gray}{13} & \cellcolor{blue!25}b & \cellcolor{blue!25}a & \cellcolor{blue!25}b & a & b & b & b & a & a & b & a & b & a & \$ & a & a & b & b \\ 
{\footnotesize $5$th 'b' in F $\rightarrow$} \textcolor{gray}{14} &\cellcolor{blue!25}\textcolor{red}{\textbf b} & \cellcolor{blue!25}a & \cellcolor{blue!25}b & b & b & a & a & b & a & b & a & \$ & a & a & b & b & b & a \\ 
\textcolor{gray}{15} & b & b & a & a & b & a & b & a & \$ & a & a & b & b & b & a & b & a & b \\ 
\textcolor{gray}{16} & b & b & a & b & a & b & b & b & a & a & b & a & b & a & \$ & a & a & b \\ 
\textcolor{gray}{17} & b & b & b & a & a & b & a & b & a & \$ & a & a & b & b & b & a & b & a \\ 
\textcolor{gray}{18} & b & b & b & a & b & a & b & b & b & a & a & b & a & b & a & \$ & a & a \\ 
\end{tabular}
    \caption{BWT matrix of string $S = aabbbababbbaababa\$$. We show one step of \textbf{backward search}. Suppose we are looking for pattern $P=bab$, and suppose we already know the range $[5,9]$ of suffixes prefixed by string ``ab'' (shown in green). Then, the characters in $BWT[5,9] = babba$ are those preceding all occurrences of ``ab'' in $S$. Since we want to left-extend ``ab'' with character 'b' in order to obtain $P=bab$, we are only interested in 'b' in $BWT[5,9] = babba$. We identify the first and last 'b' in that range: $BWT[5,9] =$ \textcolor{red}{\textbf b}ab\textcolor{red}{\textbf b}a. These are $BWT[5] = BWT[8] =$ \textcolor{red}{\textbf b}, shown in red in the figure. To find them, two rank queries at the extremes of interval $[5,9]$ are sufficient: $BWT.rank_b(5-1)+1 = 3$ (third 'b' in BWT) and $BWT.rank_b(9) = 5$ (fifth 'b' in BWT). Finally, we apply the LF mapping: the third and fifth 'b' in column L=BWT correspond to the third and fifth 'b' in column F. To find them, we access $C[b] = 9$ (position preceding the first 'b' in F) and conclude that the third 'b' appears in position $C[b]+3 = 12$ in column F, and the fifth 'b' appears in position $C[b]+5 = 14$ in column F. Then, the range of  $P=bab$ is $[12,14]$ (shown in blue). }
    \label{fig:bwsearch}
\end{figure}

Backward search works as follows. We start from the range $[1,n]$ of all suffixes prefixed by the empty string and LF-map all characters equal to $P[m]$ inside $BWT[1,n]$ to column $F$. 
The range $[l_m,r_m]$ we obtain is the range of all suffixes prefixed by $P[m]$. We repeat this process with $P[m-1]$: we LF-map all characters equal to $P[m-1]$ inside $BWT[l_m,r_m]$ to the F column, obtaining the range $[l_{m-1},r_{m-1}]$ of all suffixes prefixed by $P[m-1,m]$. By repeating this process with $P[m-2], P[m-3], \dots$, we finally obtain the range $[l_1,r_1]$ of all suffixes prefixed by $P$. 

Observe that at step $i$ in the above process we do not really need to LF-map \emph{all} the characters equal to $P[i]$ inside $BWT[l_{i+1}, r_{i+1}]$ to column F: we only need to LF-map the \emph{first} and \emph{last} occurrence of $P[i]$ in $BWT[l_{i+1}, r_{i+1}]$. These occurrences can be quickly located with two \emph{rank} queries on the BWT. Let $c = P[i]$.  
The first occurrence of $c$ in $BWT[l_{i+1}, r_{i+1}]$ is the $L$-th occurrence of $c$ in $BWT(S)$, where
$$
L = BWT.rank_c(l_{i+1}-1) + 1
$$
Similarly, the last occurrence of $c$ in $BWT[l_{i+1}, r_{i+1}]$ is the $R$-th occurrence of $c$ in $BWT(S)$, where
$$
R = BWT.rank_c(r_{i+1})
$$

Note that $R<L$ if and only if there are no occurrences of $P[i]$ in $BWT[l_{i+1}, r_{i+1}]$.

Let $C[c] = F.select_c(1)-1$ be the array of length $\sigma$ storing, for each character $\sigma\in \Sigma$, the position preceding the first occurrence of $c$ in $F$ (i.e. $F[C[c]+1]$ is the first occurrence of $c$ in $F$). 
We pre-compute this array at index construction time and store it along with $WT(BWT(S))$ \footnote{Technically, $C[c]$ can be computed  in $O(\log\sigma)$ time from $WT(BWT(S))$ using an additional wavelet tree query that we have not introduced in these notes, so array $C$ is not strictly necessary (here we use it because it simplifies the description). Exercise: can you show how to compute $C[c]$ in $O(\log\sigma)$ time using $WT(BWT(S))$?}.
The range of suffixes prefixed by $P[i, \dots, m]$ is then $[C[c] + L, C[c] + R]$.

The backward search algorithm is formalized in Algorithm \ref{alg: BSearch}. 
We abbreviate $WT(BWT(S))$ as $BWT$. 
The algorithm returns the interval $[l,r]$ of suffixes prefixed by $P$. The number $occ$ of occurrences of $P$ in $S$ is $occ = r-l+1$ ($r=l-1$ if and only if $P$ does not occur in $S$). \\

\begin{algorithm}[H]\label{alg: BSearch}
	\footnotesize
	\caption{$count(P)$ (\emph{backward search} algorithm)}
	\SetKwInOut{Input}{input}
	\SetKwInOut{Output}{output}
	\SetSideCommentLeft
	\LinesNumbered
	
	\BlankLine

        $m \leftarrow |P|$\;

	$l \leftarrow 1$\;
	$r \leftarrow |S|$\;
		
	\BlankLine

	\For{$i=m\ \mathbf{downto}\ 1$}{
		
		$c\leftarrow P[i]$\;
		
		$l \leftarrow C[c] + BWT.rank_c(l-1) + 1$\;
		$r \leftarrow C[c] + BWT.rank_c(r)$\;
				
	}
	
	\BlankLine
	\Return $[l,r]$\;
\end{algorithm}

\ \\Note that each step of backward search requires just two \emph{rank} queries on $WT(BWT(S))$. We conclude:

\begin{lemma}
    $WT(BWT(S))$ supports $count(P)$ queries in time $O(m\log\sigma)$.
\end{lemma}

Notice that $O(m\log\sigma)$ is much faster than the time $O(m\log n)$ that we obtained with the compressed suffix array. In particular, if $\sigma$ is a constant with respect to $n$ (for example, ASCII or UNICODE), then this running time is $O(m)$: optimal!

\subsubsection{\emph{Locate} and \emph{extract} queries}\label{sec: FM locate}

The solution for \emph{locate} and \emph{extract} is essentially the same as the one described in Section \ref{sec:locate,extract}: we use the same arrays $SSA$ and $M$. The only two differences are:

\begin{enumerate}
    \item $LF = \psi^{-1}$, so applying $LF$ makes us move \emph{backwards} in the string by one position.
    We therefore need to make the following small adjustment. 
    Letting $\rho$ being the sample rate, we sample $SA[i]$ whenever $SA[i]\mod \rho = 0$ or $SA[i]=1$, and define
    $$
    SA[i] = SSA[M.rank_1(LF^k(i))] + k
    $$
    where $k\geq 0$ is the smallest integer s.t. $M[LF^k(i)] = 1$.
    \item We do not want to pay $O(n)$ bits for $SSA$ and array $M$: we can only afford an extra term $o(n\log\sigma)$. We use a different sample rate: $\rho = (\log n)^{1+\epsilon} / \log\sigma$, for any desired constant $\epsilon>0$ (for example, $\rho = (\log n)^{1.0001} / \log\sigma$). 
    We build the zero-order bitvector data structure of Theorem \ref{thm:bitvector} on $M$. 
    The bit-size of $SSA$ and $M$ is therefore $nH_0(M) + (n/\rho)\log n$ bits, where $\rho = (\log n)^{1+\epsilon}/\log\sigma$.
    Exercise: prove that $nH_0(M) + (n/\rho)\log n = o(n\log\sigma)$ bits.
\end{enumerate}

Since $M[LF^k(i)] = 1$ is true for $k \leq \rho$ (because we sample every other $\rho$ text positions) and computing the $LF$ function costs $O(\log\sigma)$ time, extracting any value $SA[i]$ costs $O(\rho\log\sigma) = O((\log n)^{1+\epsilon})$ time.
We obtain:

\begin{lemma}\label{lem:FMI1}
    Choose any constant $\epsilon>0$. The structure
    $WT(BWT(S))$, in addition to the locate and extract data structures described in this section, 
    uses $nH_0(S) + o(n\log\sigma) + \Theta(\sigma\log n)$ bits of space and
    supports:
    \begin{itemize}
        \item  $count(P)$ queries in time $O(m\log\sigma)$.
         \item  $locate(P)$ queries in time $O(m\log\sigma + occ\cdot (\log n)^{1+\epsilon})$.
         \item  $extract(i,\ell)$ queries in time $O((\log n)^{1+\epsilon} + \ell\log\sigma)$.
    \end{itemize}
\end{lemma}

\subsubsection{High-order entropy compressed space}\label{sec: FM space}

Amazingly, we now show that the index of Lemma \ref{lem:FMI1} already achieves $nH_k(S) + o(n\log\sigma) + \Theta(\sigma\log n)$ bits of space (for some $k\geq 0$ that we will derive below): we do not need to modify the index to achieve high-order entropy compressed space! See Figure \ref{fig:high order FMI} for an example with $k=2$.

\begin{figure}[h!]
\setlength{\tabcolsep}{2.5pt}
\renewcommand{\arraystretch}{0.7}
    \centering
    \begin{tabular}{r | cc | ccccccccccccccc| l }
\textcolor{gray}{$i$} & F & & & & & & & & & & & & & & & & & $L=BWT(S)$ \\ \hline
\textcolor{gray}{1} & \$ & a & a & b & b & b & a & b & a & b & b & b & a & a & b & a & b & a\ \ \ \ \textcolor{gray}{$S'_{\$a} = a$} \\ \hline 
\textcolor{gray}{2} & a & \$ & a & a & b & b & b & a & b & a & b & b & b & a & a & b & a & b\ \ \ \ \textcolor{gray}{$S'_{a\$} = b$} \\ \hline 
\textcolor{gray}{3} & a & a & b & a & b & a & \$ & a & a & b & b & b & a & b & a & b & b & b\ \ \ \ \textcolor{gray}{$S'_{aa} = b\$$} \\ 
\textcolor{gray}{4} & a & a & b & b & b & a & b & a & b & b & b & a & a & b & a & b & a & \$\\\hline 
\textcolor{gray}{5} & a & b & a & \$ & a & a & b & b & b & a & b & a & b & b & b & a & a & b\ \ \ \ \textcolor{gray}{$S'_{ab} = babba$}\\ 
\textcolor{gray}{6} & a & b & a & b & a & \$ & a & a & b & b & b & a & b & a & b & b & b & a \\ 
\textcolor{gray}{7} & a & b & a & b & b & b & a & a & b & a & b & a & \$ & a & a & b & b & b \\ 
\textcolor{gray}{8} & a & b & b & b & a & a & b & a & b & a & \$ & a & a & b & b & b & a & b \\ 
\textcolor{gray}{9} & a & b & b & b & a & b & a & b & b & b & a & a & b & a & b & a & \$ & a\\\hline 
\textcolor{gray}{10} & b & a & \$ & a & a & b & b & b & a & b & a & b & b & b & a & a & b & a\ \ \ \ \textcolor{gray}{$S'_{ba} = ababa$}\\ 
\textcolor{gray}{11} & b & a & a & b & a & b & a & \$ & a & a & b & b & b & a & b & a & b & b \\ 
\textcolor{gray}{12} &b & a & b & a & \$ & a & a & b & b & b & a & b & a & b & b & b & a & a \\ 
\textcolor{gray}{13} & b & a & b & a & b & b & b & a & a & b & a & b & a & \$ & a & a & b & b \\ 
\textcolor{gray}{14} &b & a & b & b & b & a & a & b & a & b & a & \$ & a & a & b & b & b & a \\ \hline 
\textcolor{gray}{15} & b & b & a & a & b & a & b & a & \$ & a & a & b & b & b & a & b & a & b\ \ \ \ \textcolor{gray}{$S'_{bb} = bbaa$} \\ 
\textcolor{gray}{16} & b & b & a & b & a & b & b & b & a & a & b & a & b & a & \$ & a & a & b \\ 
\textcolor{gray}{17} & b & b & b & a & a & b & a & b & a & \$ & a & a & b & b & b & a & b & a \\ 
\textcolor{gray}{18} & b & b & b & a & b & a & b & b & b & a & a & b & a & b & a & \$ & a & a \\ 
\end{tabular}
    \caption{Partitioning the BWT matrix by contexts of length $k=2$ (first two columns). This BWT gets partitioned into 6 strings: $BWT(S) = S'_{\$a}\cdot S'_{a\$}\cdot S'_{aa}\cdot S'_{ab}\cdot S'_{ba}\cdot S'_{bb}$. These strings are permutations of the strings $S_w$ encountered in Section \ref{sec:high order entropy}, and have therefore their same zero-order entropy: $H_0(S_w) = H_0(S'_w)$.}
    \label{fig:high order FMI}
\end{figure}

Figure \ref{fig:high order FMI} shows that, for any $k\geq 0$, we can partition the BWT as follows: 

$$
BWT(S) = S'_{w_1}\cdots S'_{w_t}
$$

where $w_1, \dots, w_t \in \Sigma^k$ are all the $t \leq \sigma^k$ distinct strings of length $k$ appearing in $S$ (considered as a circular string). Here, each $S'_{w_i}$ is a permutation of the string $S_{w_i}$ encountered in Section \ref{sec:high order entropy}. Since $S'_{w_i}$ and $S_{w_1}$ are permutations of each other, and zero-order empirical entropy only depends on the symbols' frequencies, we conclude that 

$$
|S'_{w_i}|\cdot H_0(S'_{w_i}) = |S_{w_i}|\cdot H_0(S_{w_i})
$$

But then, this means (see Definition \ref{def:high order entropy}) that:

\begin{lemma}\label{lem:high order S}
    Zero-order compression of each $S'_{w_i}$ gives us $k$-th order compression of $S$:
    $$
    \sum_{i=1}^t |S'_{w_i}|\cdot H_0(S'_{w_i}) = nH_k
    $$
\end{lemma}

Recall that Lemma \ref{lem:sum WT} tells us that a wavelet tree over the concatenation of $t$ strings takes the same space of the $t$ Wavelet trees for the strings, plus $O(t\sigma\cdot \log n)$ bits, where $n$ is the total length of the $t$ strings. 
In our case, the $t$ strings are $w_1, \dots, w_t$, with $t\le \sigma^k$.
We conclude that Lemmas \ref{lem:high order S} and \ref{lem:sum WT} immediately imply:

\begin{lemma}
    $WT(BWT(S))$ uses $nH_k + o(n\log\sigma) + O(\sigma^{k+1}\log n)$ bits of space, for any $k\geq 0$.
\end{lemma}

We don't want the overhead $O(\sigma^{k+1}\log n)$  to be larger than the index itself. This translates to the constraint 
 $\sigma^{k+1}\log n \le \Theta(\sigma\log n) + o(n)$. This inequality is true for $k = \lfloor \max\{0, \alpha\log_\sigma n - 1\}\rfloor$, where $\alpha < 1$ is any constant fixed at construction time. We obtain: 

\begin{theorem}\label{thm:FMI}
    Choose any constants $\epsilon>0$ (as small as possible) and $\alpha < 1$ (as large as possible). 
    Let $k = \lfloor \max\{0, \alpha\log_\sigma n - 1\}\rfloor$.
    The FM-index uses $nH_k + o(n\log\sigma) + \Theta(\sigma\log n)$ bits of space and supports:
    \begin{itemize}
        \item  $count(P)$ queries in time $O(m\log\sigma)$.
         \item  $locate(P)$ queries in time $O(m\log\sigma + occ\cdot (\log n)^{1+\epsilon})$.
         \item  $extract(i,\ell)$ queries in time $O((\log n)^{1+\epsilon} + \ell\log\sigma)$.
    \end{itemize}
\end{theorem}

\begin{remark}
    Again, ASCII and UNICODE alphabets are constant-sized (their size does not depend on the text length $n$): $\sigma = O(1)$. In those cases, the FM-index uses $nH_k(S) + o(n\log\sigma)$ bits of space, 
    $count(P)$ queries are supported in optimal time $O(m)$, $locate(P)$ queries in time $O(m + occ\cdot (\log n)^{1+\epsilon})$, and $extract(i,\ell)$ queries in time $O((\log n)^{1+\epsilon} + \ell)$.
\end{remark}

To convince yourself that the FM-index is really small, consider the following example:

\begin{example}
Consider a string of length $n = 10^9$ on  English alphabet with $\sigma = 26$. 
By storing this string in plain ASCII format, the string will take about 1 GiB of space. 
What if we build an FM-index?
Choose, for example, $\alpha = 0.999$. Then, the FM index achieves $k$-th order empirical entropy for $k = \lfloor \max\{0, \alpha\log_\sigma n - 1\}\rfloor = 5$. When compressed to $H_5$, English text requires just $1.9$ bits per character \cite{pizzachiliText}. This means that $nH_5 \approx 240$ MiB. The FM-index adds on top of this space a small $o(n\log\sigma)$ term, which in practice is about $10\%$ of $n\log\sigma \approx 600$ MiB. Therefore, the FM-index will use approximately $300$ MiB of space. In this tiny space, we have both the text and the index! 
\end{example}

\chapter{Probability Theory and Hashing}\label{ch:basics}

Streaming and sketching techniques (next chapters) will require some basics of probability theory, which we cover in this chapter.

\section{Random variables}

A random variable (R.V.) $X$ is a variable that takes values from some sample space $\Omega$ according to the outcomes of a random phenomenon. $\Omega$ is also called the \emph{support} of $X$. Said otherwise, $X$ takes values in $\Omega$ according to some probability distribution. A random variable can be discrete if $|\Omega|$ is countable (examples: coin tosses or integer numbers), or continuous (for example, if it takes any real value in some interval). 
When considering multiple R.V.s with supports $\Omega_1, \dots, \Omega_n$, the sample space is the Cartesian product of the individual sample spaces: $\Omega = \Omega_1 \times \dots \times \Omega_n$.
In these notes the support of a R.V. will either be a set of integers or an interval of real numbers. 

\subsection{Distribution function}

We indicate with $F(x) = P(X \leq x)$ the \emph{cumulative distribution function of X}: the probability that $X$ takes a value in $\Omega$ smaller than or equal to $x$. 
$P(X = x) = f(x)$ is the \emph{probability mass function} (for discrete R.V.s) or the \emph{probability density function} (for continuous R.V.s). For discrete R.V.s, this is the probability that $X$ takes value $x$. For continuous R.V.s, it's the function satisfying $F(x) = \int_{-\infty}^{x} f(x) \,dx $.

\begin{example}
Take the example of fair coin tosses. Then, $X \in \{0,1\}$ (0=tail, 1=head) is a discrete random variable with probability mass function $P(X = 0) = P(X = 1) = 0.5$.
\end{example}

\subsection{Events}

An \emph{event} is a subset of the sample space, i.e. a set of assignments for all the R.V.s under consideration. Each event has a probability to happen. For example, $A = \{0 \leq X \leq 1\}$ is the event indicating that $X$ takes a value between $0$ and $1$. 
$P(A \cup B)$ is the probability that either $A$ or $B$ happens. 
$P(A \cap B)$ is the probability that both $A$ and $B$ happen. Sometimes we will also use the symbols $\vee$ and $\wedge$ in place of $\cup$ and $\cap$ (with the same meaning).
$P(A|B)$ indicates the probability that $A$ happens, provided that $B$ has already happened. In general, we have:
$$
P(A\cap B) = P(A) \cdot P(B|A)
$$

We say that two events $A$ and $B$ are \emph{independent} if $P(A\cap B) = P(A) \cdot P(B)$ or, equivalently, that $P(A|B) = P(A)$ and $P(B|A) = P(B)$: the probability that both happen simultaneously is the product of the probabilities that they happen individually. Said otherwise, the fact that one of the two events has happened, does not influence the happening of the other event. 

\begin{example}
Consider throwing two fair coins, and indicate $A = \{first\ coin =  head\}$ and $B = \{second\ coin = head\}$. The two events are clearly independent, so 
$$
P(A\cap B ) = P(A) \cdot P(B) = 0.5 \cdot 0.5 = 0.25
$$ 
On the other hand, consider throwing a coin in front of a mirror, and the two events $A = \{coin = head\}$ and $B = \{coin\ in\ the\  mirror = head\}$. 
We still have $P(A) = P(B) = 0.5$ (the events, considered separately, have both probability $0.5$ to happen), but 
the two events are clearly dependent! In fact, $P(B|A) = 1 \neq P(B) = 0.5$. So: $P(A\cap B) = P(A) \cdot P(B|A) = P(A)\cdot 1 = 0.5$.
\end{example}

We can generalize pairwise-independence to a sequence of R.V.s:

\begin{definition}[$k$-wise independence]\label{def:k-wise independence}
    Let $W = \{X_1, \dots, X_n\}$ be a set of $n$ random variables. We say that this set is $k$-wise independent, for $k\leq n$, iff $P(\bigwedge_{j=1}^k X_{i_j} = x_{i_j}) = \prod_{j=1}^k P(X_{i_j} = x_{i_j})$ for any subset $\{X_{i_1}, \dots, X_{i_k}\} \subseteq W$ of $k$ random variables. For $k=n$, we also say that the random variables are \emph{fully} independent.
\end{definition}

We will often deal with dependent random variables. A useful bound that we will use is the following: 

\begin{lemma}[Union bound]
For any set of (possibly dependent) events $\{A_1, A_2, \dots, A_n\}$ we have that: 
$$
P(\cup_{i=1}^n A_i) \leq \sum_{i=1}^n P(A_i)
$$
\end{lemma}

The union bound can sometimes give quite uninformative results since the right hand-side sum can exceed 1. The bound becomes extremely useful, however, when dealing with \emph{rare} events: in this case, the probability on the right hand-side could be much smaller than 1. This will be indeed the case in some of our applications.  

We finally mention the law of total probability:

\begin{lemma}[Law of total probability]\label{lem:total prob}
    If $B_i$ for $i= 1, \dots, k$ is a partition of the sample space, then for any event $A$: 
    $$
    P(A) = \sum_{i=1}^k P(A\cap B_i) = \sum_{i=1}^k P(B_i) P(A|B_i)
    $$
\end{lemma}

\subsection{Expected value and variance}

Intuitively, the \emph{expected value} (or mean) $E[X]$ of a numeric random variable $X$ is the arithmetic mean of a large number of independent realizations of $X$. Formally, it is defined as $E[X] = \sum_{x\in \Omega} x\cdot f(x)$ for discrete R.V.s and $E[X] = \int_{-\infty}^{+\infty} x\cdot f(x) \,dx$ for continuous R.V.s.

Some useful properties of the expected value that we will use: 

\begin{lemma}[Linearity of expectation]\
Let $a_i$ be constants and $X_i$ be (any) random variables, for $i=1, \dots, n$. Then
   $E[\sum_{i=1}^n a_i X_i] = \sum_{i=1}^n a_i E[X_i]$
\end{lemma}
\begin{proof}
For simplicity we consider the cases of $E[X+Y]$ and $E[aX]$. The claim follows easily. 
$E[X+Y]$ is computed using the law of total probability: 
$$
\begin{array}{rcl}
    E[X+Y]  & = & \sum_{i,j}(x_i+y_j)P(X=x_i \wedge Y=y_j) \\
            & = &  \sum_{i,j}x_i \cdot P(X=x_i \wedge Y=y_j) + \sum_{i,j}y_j\cdot P(X=x_i \wedge Y=y_j)\\
            & = &  \sum_{i}x_i \sum_j P(X=x_i \wedge Y=y_j) + \sum_{j}y_j \sum_i P(X=x_i \wedge Y=y_j)\\
            & = &  \sum_{i}x_i \cdot P(X=x_i) +  \sum_{j}y_j \cdot P(Y=y_j)\\
            & = & E[X] + E[Y]
\end{array}
$$

and $E[aX] = \sum_i a\cdot x_i \cdot P(X=x_i) = a\sum_i x_i \cdot P(X=x_i) = a\cdot E[X]$.
\end{proof}

Also, the expected value of a constant $a$ is the constant itself: $E[a]=a$ (a constant $a$ can be regarded as a random variable that takes value $a$ with probability 1).

In general $E[X\cdot Y] \neq E[X] \cdot E[Y]$. Equality holds if $X$ and $Y$ are independent, though:

\begin{lemma}\label{lem:expectation product}
    if $X$ and $Y$ are independent, then $E[XY] = E[X]E[Y]$.
\end{lemma}
\begin{proof}
$$
\begin{array}{rcl}
E[XY] & = & \sum_{i,j} x_iy_j P(X= x_i \wedge Y= y_j)\\ 
 & = &  \sum_{i,j} x_iy_j P(X = x_i) P(Y = y_j) \\
 & = & \Big( \sum_i x_iP(X=x_i) \Big) \cdot \left(\sum_j y_j P(Y=y_j)\right)\\
 & = & E[X]E[Y]
\end{array}
$$
\end{proof}

More in general (prove it as an exercise):

\begin{lemma}
    If $X_1, \dots, X_n$ are fully independent, then $E[\prod_{i=1}^n X_i] = \prod_{i=1}^nE[X_i]$.
\end{lemma}

Note that in the above lemma pairwise-independence is not sufficient: we need full independence. 
The expected value does not behave well with all operations, however. For example, in general $E[1/X] \neq 1/E[X]$.

The expected value of a non-negative R.V. can also be expressed as a function of the cumulative distribution function. We prove the following equality in the continuous case (the discrete case is analogous), which will turn out useful later in these notes. 

\begin{lemma}\label{lem:expected CDF}
For a non-negative continuous random variable X, it holds:
$$
E[X] = \int_{0}^{\infty} P(X\geq x) \,dx
$$
\end{lemma}
\begin{proof}
First, express  $P(X\geq x) = \int_{x}^{\infty} f(t) \,dt$:
$$
\int_{0}^{\infty} P(X\geq x) \,dx = \int_{0}^{\infty} \int_{x}^{\infty} f(t) \,dt \,dx
$$
In the latter integral, for a particular value of $t$ the value $f(t)$ is included in the summation for \emph{every} value of $x \leq t$. This observation allows us to invert the order of the two integrals as follows:
$$
\int_{0}^{\infty} \int_{x}^{\infty} f(t) \,dt \,dx = \int_{0}^{\infty} \int_{0}^{t} f(t) \,dx \,dt
$$
To conclude, observe that $\int_{0}^{t} f(t) \,dx = f(t) \cdot \int_{0}^{t} 1 \,dx = f(t) \cdot t$, so the latter becomes:
$$
\int_{0}^{\infty} \int_{0}^{t} f(t) \,dx \,dt = \int_{0}^{\infty} t\cdot f(t) \,dt = E[X]
$$

\end{proof}

The \emph{Variance} of a R.V. $X$ tells us how much the R.V. deviates from its mean: $Var[X] = E[(X-E[X])^2]$. The following equality will turn out useful:

\begin{lemma}
$Var[X] = E[X^2] - E[X]^2$
\end{lemma}
\begin{proof}
From linearity of expectation:
$Var[X] = E[(X-E[X])^2] = E[X^2 - 2XE[X] + E[X]^2] = E[X^2] - 2E[X]\cdot E[E[X]] + E[E[X]^2]$. If $Y$ is a R.V., note that $E[Y]$ is a constant (or, a random variable taking one value with probability 1). The expected value of a constant is the constant itself, thus the above is equal to $E[X^2] - E[X]^2$.
\end{proof}

If $X$ and $Y$ are independent, then one can verify that $Var[X+Y] = Var[X] + Var[Y]$. More in general, 

\begin{lemma}\label{lem:variance pairwise ind}
If $X_1, \dots, X_n$ are pairwise-independent, then $Var[\sum_{i=1}^n X_i] = \sum_{i=1}^n Var[X_i]$.
\end{lemma}
\begin{proof}
We have
$Var[\sum_{i=1}^n X_i]  = E[(\sum_{i=1}^n X_i)^2] - E[\sum_{i=1}^n X_i]^2 $. The first term evaluates to 
$$
E[(\sum_{i=1}^n X_i)^2] = \sum_{i=1}^n E[X_i^2] + 2\sum_{i\neq j} E[X_iX_j]
$$
Recalling that $E[X_i]E[X_j] = E[X_iX_j]$ if $X_i$ and $X_j$ are independent, the second term evaluates to:
$$
E[\sum_{i=1}^n X_i]^2 = \left(\sum_{i=1}^n E[X_i]\right)^2 =  \sum_{i=1}^n E[X_i]^2 + 2\sum_{i\neq j} E[X_i]E[X_j] = \sum_{i=1}^n E[X_i]^2 + 2\sum_{i\neq j} E[X_iX_j]
$$
Thus, the difference between the two terms is equal to
$$
\sum_{i=1}^n E[X_i^2] - \sum_{i=1}^n E[X_i]^2 = \sum_{i=1}^n \left( E[X_i^2] - E[X_i]^2  \right) = \sum_{i=1}^n Var[X_i]
$$
\end{proof}

Crucially, note that the above proof does not require full independence (just pairwise-independence). This will be important later.

Let $X$ be a R.V. and $A$ be an event. The \emph{conditional expectation} of $X$ conditioned on $A$ is defined as $E[X|A] = \sum_x x\cdot P(X=x | A)$.
The law of total probability (Lemma \ref{lem:total prob}) implies (prove it as an exercise):

\begin{lemma}[Law of total expectation]\label{lem:total exp}
	If $B_i$ for $i= 1, \dots, k$ are a partition of the sample space, then for any random variable $X$: 
	$$
	E[X] = \sum_{i=1}^k P(B_i) E[X|B_i]
	$$
\end{lemma}

\subsection{Bernoullian R.V.s}

\emph{Bernoullian} R.V.s model the event of flipping a (possibly biased) coin:

\begin{definition}
A Bernoullian R.V. $X$ takes the value 1 with some probability $p$ (parameter of the Bernoullian), and value 0 with probability $1-p$. The notation $X \sim Be(p)$ means that $X$ is Bernoullian with parameter $p$.
\end{definition}

\begin{lemma}\label{lem:Ber Exp Var}
If $X \sim Be(p)$, then $X$ has expected value $E[X] = p$ and variance $Var[X] = p(1-p)$. 
\end{lemma}
\begin{proof}
    Note that $X^2=X$ since $X\in\{0,1\}$.
    $E[X] = 0\cdot (1-p) + 1\cdot p = p$. $Var[X] = E[X^2]-E[X]^2 = p - p^2 = p(1-p)$.
\end{proof}

Note, as a corollary of the previous lemma, that $Var[X] \leq E[X]$ and $Var[X] \leq 1-E[X]$ for Bernoullian R.V.s.

\section{Concentration inequalities}

Concentration inequalities provide bounds on how likely it is that a random variable deviates from some value (typically, its expected value). These will be useful in the next sections to calculate the probability of obtaining a good enough approximation with our randomized algorithms.

\subsection{Markov's inequality}

Suppose we know the mean $E[X]$ of a nonnegative R.V. $X$.
Markov's inequality can be used to bound the probability that a random variable takes a value larger than some positive constant. It goes as follows:

\begin{lemma}[Markov's inequality]
For any nonnegative R.V. $X$ and any $a>0$ we have: 
$$
P(X \geq a) \leq E[X]/a
$$
\end{lemma}
\begin{proof}
We prove the inequality for discrete R.V.s (the continuous case is similar). $E[X] = \sum_{x=0}^\infty x\cdot P(X = x) \geq \sum_{x=a}^\infty x\cdot P(X = x) \geq a\cdot \sum_{x=a}^\infty P(X = x) = a\cdot P(X\geq a)$.
\end{proof}

\subsection{Chebyshev's inequality}

Chebyshev's inequality gives us a stronger bound than Markov's, provided that we know the random variable's variance. This inequality bounds the probability that the R.V. deviates from its mean by some fixed value. Note that we do not require $X$ to be nonnegative.

\begin{lemma}[Chebyshev's inequality]\label{lem:chebyshev}
For any $k>0$:
$$
P(|X-E[X]| \geq k) \leq Var[X]/k^2
$$
\end{lemma}
\begin{proof}
We simply apply Markov to to the (nonnegative) R.V. $(X-E[X])^2$:
$$
P(|X-E[X]| \geq k) = P((X-E[X])^2 \geq k^2) \leq E[(X-E[X])^2]/k^2 = Var[X]/k^2
$$
\end{proof}

\subsubsection{Boosting by averaging}

A trick to get a better bound is to draw $s$ pairwise-independent realizations of  the R.V. $X$ and average them out.

\begin{lemma}[Boosted Chebyshev's inequality]\label{lem:boosted Cheb}
Let $X$ be a R.V.
Fix $k>0$ and integer $s\geq 1$.
Let $X_1, \dots, X_s$ be pairwise independent and distributed as $X$, and define  $\hat X = \sum_{i=1}^s X_i / s$. Then:
$$
P(|\hat X-E[X]| \geq k) \leq \frac{Var[X]}{s\cdot k^2}
$$
\end{lemma}
\begin{proof}
Since $X_i\sim X$, it is easy to see by linearity of expectation that $E[\sum_{i=1}^s X_i / s] = E[X]$. Moreover, by pairwise independence we have $Var[\sum_{i=1}^s X_i] = s\cdot Var[X]$.
Then:
$$
\begin{array}{ccl}
P(|\hat X - E[X]| \geq k) & = & P\left(| \frac{1}{s}\sum_{i=1}^s X_i - E\left[\frac{1}{s}\sum_{i=1}^s X_i\right] | \geq k\right)\\
& = & P\left(\frac{1}{s} |\sum_{i=1}^s X_i - E\left[\sum_{i=1}^s X_i\right] | \geq k\right)\\
& = & P\left(|\sum_{i=1}^s X_i - E\left[\sum_{i=1}^s X_i\right] | \geq  sk\right)\\
& \le & Var[\sum_{i=1}^s X_i]/(s^2k^2)\\
& = & s\cdot Var[X] /(s^2k^2) \\
& = & Var[X] / (sk^2)
\end{array}
$$
\end{proof}

\subsection{Chernoff-Hoeffding's inequalities}

Chernoff-Hoeffding's inequalities are used to bound the probability that the sum $Y = \sum_{i=1}^n Y_i$ of $n$ independent identically distributed (iid)  R.V.s $Y_i$ exceeds by a given value its expectation.
The inequalities come in two flavors: with additive error and with relative (multiplicative) error. We will prove both for completeness, but only use the former in these notes. 
The inequalities give a much stronger bound w.r.t. Markov precisely because we know a particular property of the R.V.  $Y$ (i.e. it is a sum of iid. R.V.'s).
Here we study the simplified case of Bernoullian R.V.s. 

\begin{lemma}[Chernoff-Hoeffding bound, additive form]\label{lem:add Hoeffding}
Let $Y_1, \dots, Y_n$ be fully independent $Be(p)$ random variables. Denote $Y = \sum_{i=1}^n Y_i$. Then, for all $t\geq 0$:

\begin{itemize}
    \item $P(Y\geq E[Y] + t) \leq e^{\frac{-t^2}{2n}}$  [one sided, right]
    \item $P(Y\leq E[Y] - t) \leq e^{\frac{-t^2}{2n}}$  [one sided, left]
    \item $P(|Y-E[Y]|\geq t) \leq 2e^{\frac{-t^2}{2n}}$ [double sided] 
\end{itemize}
\end{lemma}
\begin{proof}
Let $X_i = Y_i - E[Y_i]$. Each $X_i$ is distributed in the interval $[-1,1]$, has mean $E[X_i] = E[Y_i-E[Y_i]] = E[Y_i]-E[Y_i] =  0$, and takes value $1-E[Y_i] = 1-p$ with probability $p$ and value $0-E[Y_i] = -p$ with probability $1-p$. Let $X = \sum_{i=1}^n X_i$. In particular, $X = Y - E[Y]$.
We  prove $P(X\geq t) \leq e^{\frac{-t^2}{2n}}$. The same argument will hold for $P(-X\geq t) \leq e^{\frac{-t^2}{2n}}$, so by union bound we will get $P(|Y-E[Y]|\geq t)  = P(|X|\geq t) \leq 2e^{\frac{-t^2}{2n}}$.

Let $s>0$ be some free parameter that we will later fix to optimize our bound. The event $X \geq t$ is equivalent to the event $e^{sX} \geq e^{st}$, so:
$$
P(X \geq t) = P(e^{s\sum_{i=1}^nX_i} \geq e^{st})
$$
We apply Markov to the (non-negative\footnote{Note: $X$ might be negative, but $e^{sX}$ is always positive, so we can indeed apply Markov's inequality}) R.V. $e^{s\sum_{i=1}^nX_i}$, obtaining
$$
P(X \geq t) \leq E[e^{s\sum_{i=1}^nX_i}] / e^{st} = E\left[\prod_{i=1}^n e^{sX_i}\right] /  e^{st}
$$
which, since the $X_i$'s are fully independent and identically distributed (in particular, they have the same expected value), yields the inequality 
\begin{equation}\label{eq:Hbound1}
P(X \geq t) \leq \left(E[e^{sX_1}]\right)^n  / e^{st}
\end{equation}

The goal is now to bound the expected value $E[e^{sX_1}]$ appearing in the above quantity. Define $A = \frac{1+X_1}{2}$ and $B = \frac{1-X_1}{2}$. Note that:

\begin{itemize}
	\item $A \geq 0$ and $B\geq 0$
	\item $A+B = \frac{(1+X_1)+(1-X_1)}{2} = 1$
	\item $A - B = \frac{(1+X_1)-(1-X_1)}{2} = X_1$
\end{itemize}

We recall Jensen's inequality: if $f(x)$ is convex, then for any $0\leq a \leq 1$ we have $f(ax + (1-a)y) \leq af(x) + (1-a)f(y)$. Note that $e^x$ is convex so, by the above three observations:

$$
\begin{array}{rcl}
e^{sX_1} & = & e^{s(A-B)}\\
& = & e^{sA - sB}\\
& \leq & Ae^{s} + Be^{-s}\\
& = & \frac{1+X_1}{2}e^s + \frac{1-X_1}{2}e^{-s}\\
& = & \frac{e^s+e^{-s}}{2} + X_1 \cdot \frac{e^s-e^{-s}}{2}
\end{array}
$$

Since $E[X_1] = 0$, we obtain:

$$
\begin{array}{rcl}
E[e^{sX_1}] & \leq & E\left[\frac{e^s+e^{-s}}{2} + X_1 \cdot \frac{e^s-e^{-s}}{2}\right]\\
& = & \frac{e^s+e^{-s}}{2} +  \frac{e^s-e^{-s}}{2} \cdot E\left[X_1\right]\\
& = & (e^s+e^{-s})/2\\
\end{array}
$$

The Taylor expansion of $e^s$ is $e^s = 1+ s + \frac{s^2}{2!} + \frac{s^3}{3!} + \dots$, while that of $e^{-s}$ is $ e^{-s} = 1 - s + \frac{s^2}{2!} - \frac{s^3}{3!} + \dots$ (i.e. odd terms appear with negative sign). 
Let $even$ and $odd$ denote the sum of even and odd terms, respectively.
Replacing the two Taylor series in the quantity $(e^s+e^{-s})/2$, we obtain 
$$
\begin{array}{ccl}
    (e^s+e^{-s})/2  & = & (even+odd)/2 + (even-odd)/2  \\
     & = & even\\
     & = & \sum_{i = 0,2,4,\dots} \frac{s^{i}}{i!}\\
     & = & \sum_{i = 0}^\infty  \frac{s^{2i}}{(2i)!}
\end{array}
$$

Now, note that $(2i)! = 1\cdot 2 \cdot 3 \cdots i \cdot (i+1) \cdots 2i \geq i! \cdot 2^i$, so 
$$
E[e^{sX_1}] \leq \sum_{i=0}^\infty \frac{s^{2i}}{(2i)!} \leq \sum_{i=0}^\infty \frac{s^{2i}}{i!\cdot 2^i} = \sum_{i=0}^\infty \frac{(s^2/2)^i}{i!}
$$
The term $\sum_{i=0}^\infty \frac{(s^2/2)^i}{i!}$ is precisely the Taylor expansion of $e^{s^2/2}$. We conclude
$$
E[e^{sX_1}] \leq e^{s^2/2}
$$
and Inequality \ref{eq:Hbound1} becomes
\begin{equation}\label{eq:Hbound2}
P(X \geq t) \leq \left(e^{s^2/2}\right)^n  / e^{st} = e^{(ns^2-2st)/2}
\end{equation}

Recall that $s$ is a free parameter. In order to obtain the strongest bound, we have to minimize $e^{(ns^2-2st)/2}$ as a function of $s$. This is equivalent to minimizing $ns^2-2st$. The coefficient of the second-order term is $n>0$, so the polynomial indeed has a minimum. In order to find it, we find the root of its derivative: $2ns-2t = 0$, which tells us that the minimum occurs at $s = t/n$. Replacing $s = t/n$ in Inequality \ref{eq:Hbound2}, we finally obtain 
$P(X \geq t) \leq e^{-t^2/(2n)}$.
\end{proof}

We will sometimes apply Chernoff-Hoeffding to the average of Bernoullian random variables. 
By simply replacing $t/n = \epsilon$ in 
Lemma \ref{lem:add Hoeffding}, we obtain:

\begin{corollary}\label{cor:avg Chernoff-Hoeffding}
Let $Y_1, \dots, Y_n$ be fully independent $Be(p)$ random variables. Let $\hat Y = \frac{1}{n}\sum_{i=1}^n Y_i$ be an estimator for $p$. Then, for all $0 < \epsilon < 1$:
$$
P(|\hat Y - p |\geq \epsilon) \leq 2e^{-\epsilon^2n/2}
$$
\end{corollary}

If $E[Y]$ is small, a bound on the relative error is often more useful:

\begin{lemma}[Chernoff-Hoeffding bound, multiplicative form]\label{lem:mult chernoff}
Let $Y_1, \dots, Y_n$ be fully independent $Be(p)$ random variables. Denote $Y = \sum_{i=1}^n Y_i$ and $\mu = E[Y] = np$. Then, for all $0 < \epsilon < 1$:

\begin{itemize}
    \item $P(Y \geq (1+\epsilon)  \mu) \leq e^{-\mu\epsilon^2/3}$  [one sided, right]
    \item $P(Y \leq (1-\epsilon)  \mu) \leq e^{-\mu\epsilon^2/2}$ [one sided, left]
    \item $P(|Y-\mu|\geq \epsilon  \mu) \leq 2e^{-\epsilon^2 \mu/3}$ [double sided] 
\end{itemize}

\end{lemma}
\begin{proof}
We first study $P(Y \geq (1+\epsilon)  \mu)$.
Note that $\mu=np$, since our R.V.s are distributed as $Be(p)$. 

The first step is to upper-bound the quantity $P(Y \geq t)$ (later we will fix $t = (1+\epsilon)  \mu$).
Let $s>0$ be some parameter that we will later fix to optimize our bound. The event $Y \geq t$ is equivalent to the event $e^{sY} \geq e^{st}$, so:
$$
P(Y \geq t) = P(e^{s\sum_{i=1}^nY_i} \geq e^{st})
$$
We apply Markov to the R.V. $e^{s\sum_{i=1}^nY_i}$, obtaining
$$
P(Y \geq t) \leq E[e^{s\sum_{i=1}^nY_i}] / e^{st} = E\left[\prod_{i=1}^n e^{sY_i}\right] /  e^{st}
$$
which, since the $Y_i$'s are fully independent and identically distributed (in particular, they have the same expected value), yields the inequality 
\begin{equation}\label{eq:chernoff1}
P(Y \geq t) \leq \left(E[e^{sY_1}]\right)^n  / e^{st}
\end{equation}

Replacing $t = (1+\epsilon)  \mu$, we obtain
\begin{equation}\label{eq:chernoff2}
P(Y \geq (1+\epsilon)  \mu) \leq \left(E[e^{sY_1}]\right)^n  / e^{s(1+\epsilon)  \mu} = \left(E[e^{sY_1}]  e^{-sp(1+\epsilon)}\right)^n
\end{equation}

The expected value $E[e^{sY_1}]$ can be bounded as follows:
$$
\begin{array}{rcl}
   E[e^{sY_1}] & = & p\cdot e^{s\cdot 1} + (1-p)\cdot e^{s\cdot 0} \\
     & = & p\cdot e^{s} + 1-p\\
     & = & 1 + p(e^{s}-1)\\
     & \leq & e^{p(e^{s}-1)}
\end{array}
$$
where in the last step we used the inequality $1+x \leq e^{x}$ with $x = p(e^{s}-1)$. Combining this with Inequality \ref{eq:chernoff2} we obtain: 
\begin{equation}\label{eq:chernoff3}
P(Y \geq (1+\epsilon)  \mu) \leq \left(e^{p(e^{s}-1)} e^{-sp(1+\epsilon)}\right)^n = \left(e^{e^{s}-1} e^{-s(1+\epsilon)}\right)^\mu
\end{equation}
By taking $s = \log(1+\epsilon)$ (it can be shown that this choice optimizes the bound), we have $e^{e^{s}-1} e^{-s(1+\epsilon)} = e^{\epsilon-\log(1+\epsilon)^{(1+\epsilon)}}$, thus Inequality \ref{eq:chernoff3} becomes:
\begin{equation}\label{eq:chernoff4}
P(Y \geq (1+\epsilon)  \mu) \leq \left( \frac{e^\epsilon}{(1+\epsilon)^{(1+\epsilon)}} \right)^\mu = \rho
\end{equation}
To conclude, we bound $\log \rho = \mu(\epsilon - (1+\epsilon)\log(1+\epsilon))$. We use the inequality $\log(1+\epsilon) \geq \frac{\epsilon}{1+\epsilon/2}$, which holds for all $\epsilon \geq 0$, 
and obtain: 
\begin{equation}\label{eq:chernoff5}
\log \rho \leq \mu \left( \epsilon - \frac{\epsilon(1+\epsilon)}{1+\epsilon/2}  \right) = \frac{-\mu\epsilon^2}{2+\epsilon} \leq \frac{-\mu\epsilon^2}{3}
\end{equation}
Where the latter inequality holds since we assume $\epsilon < 1$.
Finally, Bounds \ref{eq:chernoff4} and \ref{eq:chernoff5} yield:
\begin{equation}\label{eq:chernoff6}
P(Y \geq (1+\epsilon)  \mu) \leq e^{-\mu\epsilon^2/3}
\end{equation}
We are left to find a bound for the symmetric tail $P(Y \leq (1-\epsilon)  \mu) = P(-Y \geq -(1-\epsilon)  \mu)$. Following the same procedure used to obtain Inequality \ref{eq:chernoff2} we have
$$
P(-Y \geq -(1-\epsilon)  \mu) \leq \left( E[e^{-sY_1}]  e^{s(1-\epsilon)p} \right)^n
$$
We can bound the expectation as follows:
$E[e^{-sY_1}] = p\cdot e^{-s}+(1-p) = 1+p(e^{-s}-1) \leq e^{p(e^{-s}-1)}$ and obtain:
\begin{equation}\label{eq:chernoff7}
P(-Y \geq -(1-\epsilon)  \mu) \leq \left( e^{e^{-s}-1} e^{s(1-\epsilon)} \right)^\mu
\end{equation}
It can be shown that the bound is minimized for $s = -\log(1-\epsilon)$. This yields:
\begin{equation}\label{eq:chernoff8}
P(Y \leq (1-\epsilon)  \mu) \leq \left(\frac{ e^{-\epsilon}}{(1-\epsilon)^{(1-\epsilon)}}  \right)^\mu = \rho
\end{equation}
Then, $\log \rho = \mu (-\epsilon - (1-\epsilon)\log (1-\epsilon))$. We plug the bound $\log(1-\epsilon) \geq \frac{\epsilon^2/2-\epsilon}{1-\epsilon}$, which holds for all $0  \leq \epsilon < 1$.  Then, $\log\rho \leq \mu(-\epsilon - (\epsilon^2/2-\epsilon)) = -\mu\epsilon^2/2$.
This yields 
\begin{equation}\label{eq:chernoff9}
P(Y \leq (1-\epsilon)  \mu) \leq e^{-\mu\epsilon^2/2} \leq e^{-\mu\epsilon^2/3} 
\end{equation}
and by union bound we obtain our double-sided bound. 
\end{proof}

Equivalently, we can bound the probability that the arithmetic mean of $n$ independent R.V.s deviates from its expected value. This yields a useful estimator for Bernoullian R.V.s (i.e. the arithmetic mean of $n$ independent observations of a Bernoullian R.V.). Note that the bound improves \emph{exponentially} with the number $n$ of samples.

\begin{corollary}\label{cor:mult mean}
Let $Y_1, \dots, Y_n$ be fully independent $Be(p)$ random variables. Consider the estimator $\hat Y = \frac{1}{n}\sum_{i=1}^n Y_i$ for the value $p$ ($=E[\hat Y]$). Then, for all $0 < \epsilon < 1$:
$$
P(|\hat Y - p|\geq \epsilon  p) \leq 2e^{-\epsilon^2 n p/3}
$$
\end{corollary}

\subsection{Combining bounds: the \emph{mean} and \emph{median} tricks}   \label{sec:mean+median}

Suppose that we want to estimate $E[X]$ with relative error at most $\epsilon$ for a random variable $X$ of which we know the variance.  
We want strong guarantees that our estimate is within this error rate; let us fix a desired failure probability $\delta$ (i.e. we want to exceed relative error $\epsilon$ with probability at most $\delta$).
In these notes, $X$ will be the output of some randomized algorithm whose goal is to estimate $E[X]$ as accurately as possible. 
Chebyshev (Lemma \ref{lem:chebyshev}) tells us that the probability that $X$ exceeds relative error $\epsilon$ is at most 
$P(|X-E[X]|\geq \epsilon E[X]) \leq \frac{Var[X]}{E[X]^2\epsilon^2}$. Unfortunately, this probability depends on $\epsilon$. Even worse, for small values of $\epsilon$ this probability is larger than $1$, so this bound is useless. The first step is to make this probability \emph{constant} (we choose constant 1/3 but any smaller constant will work). The following is an easy corollary of the Boosted Chebyshev bound (Lemma \ref{lem:boosted Cheb}):

\begin{corollary}[``mean'' trick]\label{cor:mean trick}
    Let $X$ be a random variable of which we know $E[X]$ and $Var[X]$ and suppose that our goal is to compute an estimate of $E[X]$ exceeding relative error $\epsilon$ with small probability. 

    Draw $s \ge \frac{3Var[X]}{E[X]^2\epsilon^2}$ i.i.d. realizations $X_1, \dots, X_s$ of $X$ and define $Y = \frac{1}{s}\sum_{i=1}^s X_i$. Then, $E[Y]=E[X]$  and the following holds: 
    $$
    P(|Y-E[Y]| \geq \epsilon E[Y]) \leq 1/3.
    $$
\end{corollary}

To further reduce the failure probability to any desired $\delta>0$, we apply the so-called \emph{median trick}: we will draw several independent realizations of $Y$ of Corollary \ref{cor:mean trick} and take their \emph{median}. As it turns out, this will allow us to apply Chernoff-Hoeffding and reduce exponentially the failure probability. We summarize the median trick in the following theorem, which will be used several times in the notes:

\begin{theorem}[``median'' trick]\label{thm:median trick}
    Fix any desired failure probability $\delta > 0$. 
    Let $Y$ be a random variable such that 
    $$
    P(|Y-E[Y]| \geq \epsilon E[Y]) \leq 1/3
    $$
    for some relative error $\epsilon>0$.
    Draw $t = 72 \ln(1/\delta)$ i.i.d. realizations $Y_1,\dots, Y_t$ of $Y$ and define
    $$
    \hat Y = \mathtt{median}(Y_1, Y_2, \dots, Y_t).
    $$
    Then, the following holds: 
    $$
    P(|\hat Y-E[Y]| \geq \epsilon E[Y]) \leq \delta.
    $$
\end{theorem}
\begin{proof}
Consider the following indicator (Bernoullian) random variables $\mathbbm{1}_1, \dots, \mathbbm{1}_t$:

$$
\mathbbm{1}_j = \left\{
\begin{array}{cl}
   1  &  \mathrm{if\ }|Y_j - E[Y]| \geq \epsilon E[Y]\\
    0  & \mathrm{otherwise}
\end{array}\right.
$$

That is, $\mathbbm{1}_j$ is equal to 1 if and only if $Y_j$ ``fails'', i.e. if its relative error with respect to $E[Y]$ exceeds $\epsilon$. 
By assumption, note that $\mathbbm{1}_j$ takes value 1 with probability at most $1/3$.

What is the probability that $\hat Y$ (the median of all $Y_j$) ``fails'', i.e. that $|\hat Y - E[Y]| \geq \epsilon E[Y]$? If the median fails, then it is either too small (below $(1-\epsilon)E[Y]$) or too large (above $(1+\epsilon)E[Y]$). In either case, by definition of median, at least $t/2$ 
variables $Y_j$
return a result which is either too small or too large, and thus fail. In other words,

$$
P(|\hat Y - E[Y]| \geq \epsilon E[Y])  \leq P\left(\sum_{j=1}^t \mathbbm{1}_j \geq t/2\right)
$$

As seen above, each $\mathbbm{1}_j$ is a Bernoullian R.V. taking value 1 with probability at most $1/3$. We thus have $\mu = E[\sum_{j=1}^t \mathbbm{1}_j] \leq t/3$. 

Recall the one-sided right variant of the Chernoff-Hoeffding additive bound (Lemma \ref{lem:add Hoeffding}):
$$
P\left(\sum_{j=1}^t \mathbbm{1}_j\geq \mu + k\right) \leq P\left(\sum_{j=1}^t \mathbbm{1}_j\geq \frac{t}{3} + k\right) \leq e^{\frac{-k^2}{2t}}
$$

Solving $t/3 + k = t/2$, we obtain $k = t/6$. Replacing this value into the previous inequality, we obtain:

$$
P\left(\sum_{j=1}^t \mathbbm{1}_j\geq t/2\right) \leq e^{-t/72}
$$

We want the probability on the right-hand side to be equal to our desired failure probability $\delta$. Solving $\delta = e^{-t/72}$, we obtain $t = 72 \ln(1/\delta)$. 
\end{proof}

The mean and median tricks are usually  combined together as follows:

\begin{corollary}[combined ``mean+median'' trick]\label{cor:mean+median trick}
    Let $X$ be a random variable such that $Var[X] = O(E[X]^2)$ and suppose that our goal is to compute an estimate of $E[X]$ exceeding relative error $\epsilon$ with arbitrarily small probability $\delta$, for any pre-defined parameters $\epsilon>0$ and $\delta>0$. 
    Let $t= 72\ln(1/\delta)$ and define a new random variable $\hat Y$ as follows: $\hat Y = \mathtt{median}(Y_1, Y_2, \dots, Y_t)$, where each $Y_i$ is built by drawing $s = \frac{3Var[X]}{E[X]^2\epsilon^2} = O(\epsilon^{-2})$ i.i.d. realizations $X_{i,1}, \dots, X_{i,s}$ of $X$ and defining $Y_i = \frac{1}{s}\sum_{j=1}^s X_{i,j}$. Then, the following properties hold: 
    \begin{enumerate}
        \item $\hat Y$ is constructed from $O(st) = O(\epsilon^{-2}\log(1/\delta))$ independent realizations of $X$, and  
        \item $P(|Y-E[X]| \geq \epsilon E[X]) \leq \delta$.
    \end{enumerate}
\end{corollary}

In other words, the above corollary says that if a random variable $X$ satisfies $Var[X] = O(E[X]^2)$, then we can boost the accuracy of our prediction of $E[X]$ by sampling $O(\epsilon^{-2}\log(1/\delta))$ realizations of $X$, combining them as described in the corollary (i.e. taking a median of means) and obtaining a random variable $\hat Y$ whose realizations exceed $E[X]$ by more than a relative fraction $\epsilon$ with arbitrarily small probability $\delta$. We will use this powerful corollary several times in these notes.

\section{Hashing}\label{sec:hash functions}

A \emph{hash function} is a function $h : U \rightarrow [0,M)$ from some universe $U$ (usually, an interval of integers) to an interval of numbers (usually the integers, but we will also work with the reals). Informally speaking, $h$ is used to \emph{randomize our data} and should have the following basic features:

\begin{enumerate}
\item $h(x)$ should be ``as random'' as possible. Ideally, $h$ should map the elements of $U$ completely uniformly (but we will see that this has a big cost).
\item $h(x)$ should be quick to compute algorithmically. Ideally, we would like to compute $h(x)$ in time proportional to the time needed to read $x$ ($O(1)$ if $x$ is an integer, or $O(n)$ if $x$ is a string of length $n$).
\item $h$ occupies space in memory, since it is implemented with some kind of data structure. This space should be as small as possible (ideally, $O(1)$ words of space, or logarithmic space).
\end{enumerate}

$h(x)$ will also be called \emph{the fingerprint} of $x$.
 
Note that, while $h$ accepts as input any value from $U$, typically the algorithms using $h$ will apply it to much smaller subsets of $U$ (for example, $U$ might be the set of all $2^{32}$ possible IPv4 addresses, but the algorithm will work on just a small subset of them). 

We now formalize the notion of hashing. 

\begin{definition}
    The set of all functions $h : U \rightarrow [0,M)$ from domain $U$ to codomain $[0,M)$ is denoted as $[0,M)^U$.
\end{definition}

\begin{definition}
    A \emph{family of hash functions} $\mathcal H$ is a subset of $[0,M)^{U}$.
\end{definition}

\begin{remark}\label{remark:cardinality of M^U}
Note that each $h\in [0,M)^{U}$ assigns a unique value from $[0,M)$ to each of the $|U|$ universe elements. Then, a function     $h\in \mathcal H$ is fully characterized by the vector $( h(x_1), h(x_2), \dots, h(x_{|U|}) )$ of length $|U|$, where $U = \{x_1, x_2, \dots, x_{|U|}\}$. Conversely, each such vector defines exactly one function of $[0,M)^{U}$. We conclude that the cardinality of $[0,M)^{U}$ is $M^{|U|}$.
\end{remark}

Given a family $\mathcal H$ of hash functions, our randomized algorithm will first extract a uniform\footnote{One (strong) assumption is always needed for this to work: we can draw uniform integers. This is actually impossible, since computers are deterministic. However, there is a vast literature on pseudo-random number generators (PRNG) which behave reasonably well in practice. We will thus ignore this problem for simplicity.} $h\in \mathcal H$. Then, the algorithm will \emph{randomize} the input data by applying $h$ to it. The expected-case analysis of the algorithm will take into account the structure of $\mathcal H$ and the fact that $h$ has been  chosen uniformly from it. 

\begin{remark}
From Section \ref{sec:worst case entropy}, in the worst case we need at least $\log_2|\mathcal H|$ bits in order to represent (and store in memory) any $h \in \mathcal H$. This shows a trade-off: the larger   $\mathcal H$ is, the more \emph{random} (i.e. higher independence) our hash function $h$ will be. On the other hand, a large $\mathcal H$ means that we have to waste more bits to represent $h$.
\end{remark}

Ideally, we would like our hash function to be completely uniform: 

\begin{definition}[Uniform hash function]
	Assume that $h\in \mathcal H \subseteq [0,M)^{U}$ is chosen uniformly, and let $U = \{x_1, \dots, x_{|U|}\}$.
	We say that $\mathcal H$ is \emph{uniform} if for any $y_1, \dots, y_{|U|} \in [0,M)$,  we have  $P(( h(x_1), \dots, h(x_{|U|})) = ( y_1, \dots, y_{|U|}) ) = \frac{1}{M^{|U|}}$.
\end{definition}

As it turns out, the requirements (1-3) introduced at the beginning of this section are in conflict. In fact, it is impossible to obtain all three simultaneously:

\begin{remark}
    Suppose that our goal is to obtain a uniform hash family and let $U = \{x_1, \dots, x_{|U|}\}$.
    In other words, our family $\mathcal H$ should be such that, for any choice of $y_1, \dots, y_{|U|} \in [0,M)$ and uniform $h\in \mathcal H$: 
    $$
    P(( h(x_1), \dots, h(x_{|U|})) = ( y_1, \dots, y_{|U|}) ) = \frac{1}{M^{|U|}}.
    $$
    However, this is possible only if $\mathcal H = [0,M)^U$: in any other case, there would exist at least one choice of $y_1, \dots, y_{|U|} \in [0,M)$ such that  $P(( h(x_1), \dots, h(x_{|U|})) = ( y_1, \dots, y_{|U|}) ) < \frac{1}{M^{|U|}}$. Therefore, a uniform hash function must take $\log_2|\mathcal H| = \log_2 M^{|U|} = |U| \log_2 M$ bits of memory, which is typically too 
    much (for example: $U$ could be the set of all IPv4 addresses, so $|U| = 2^{32}$).
    It is easy to construct such a hash function: fill a vector $V[1,|U|]$ with uniform integers from $[0,M)$, and define $h(x) = V[x]$.
Note that such a hash function satisfies requirements (1) and (2), but not (3).
\end{remark}

In the next subsection we study good compromises that will work for many algorithms: $k$-wise independent and universal hashing. Then, we briefly discuss functions mapping $U$ to real numbers.

\subsection{k-wise independent hashing}  \label{sec: kwise indep hashing}

In this section we work with hash functions on the integers: $h : [1,n] \rightarrow [0,M)$ (i.e. our universe is $U=[1,n]$).

$k$-wise independent (or $k$-uniform / $k$-independent) hashing is a weaker version of uniform hashing:

\begin{definition}\label{def:k-wise hashing}
We say that the family $\mathcal H $ is $k$-wise independent if and only if, for a uniform choice of $h\in \mathcal H$, we have that:
\begin{enumerate}
    \item For any $x\in [1,n]$, $h(x)$ is a uniform random variable in $[0,M)$.
    \item The random variables $h(1), h(2), \dots, h(n)$ are $k$-wise independent.
\end{enumerate}
\end{definition}

\begin{remark}
Definition \ref{def:k-wise hashing} is equivalent to  
$$
P\left(  \bigwedge_{i=1}^k h(x_i) = y_i \right) = M^{-k}
$$
for any choice of distinct $x_1, \dots, x_k \in [1,n]$ and (not necessarily distinct) $y_1, \dots, y_k \in [0,M)$. Another equivalent characterization is: the
$k$-tuple $(h(x_1) , \dots, h(x_k))$ is a uniform random variable over $[0,M)^k$ when $x_1, \dots, x_k$ are distinct.
\end{remark}

We say that $\mathcal H$ is \emph{fully-uniform} (or simply \emph{uniform}) when $k = n$. 

In the next sections we will see that $k=2$ is already sufficient in many interesting cases: this case is also called \emph{pairwise-independent hashing}.

\subsection{Universal hashing}  \label{sec: universal hashing}

Another important concept is \emph{universal hashing}, which measures the probability of \emph{collisions}: cases where $x_1\neq x_2$ and $h(x_1)=h(x_2)$.

\begin{definition}
We say that $\mathcal H $ is universal if and only if, for a uniform choice of $h\in \mathcal H$, we have that 
$$
P\left(  h(x_1) = h(x_2) \right) \leq 1/M
$$
for any choice of distinct $x_1 \neq x_2 \in [1,n]$.
\end{definition}

Note that this is at most the probability of collision we would expect if the hash function assigned truly random outputs to every key.
It is easy to see that pairwise-independence implies universality (the converse is not true). Consider the partition of the sample space $\{h(x_2)=y\}_{y\in[0,M)}$. By the law of total probability (Lemma \ref{lem:total prob}) and by pairwise-independence: 

$$
\begin{array}{rcl}
  P\left(  h(x_1) = h(x_2) \right)  & = & \sum_{y\in [0,M)} P(h(x_1) = h(x_2) \wedge h(x_2) = y) \\ 
    & = & \sum_{y\in [0,M)} P(h(x_1) = y\ \wedge\ h(x_2) = y)\\ 
  & = & \sum_{y\in [0,M)} M^{-2} \\
  & = & 1/M
\end{array}
$$

Next, we show a construction (not the only possible one) yielding a pairwise-independent hash function. 
Let $M\geq n$ be a prime number, and define
$$
h_{a,b}(x) = (a\cdot x + b) \mod M
$$
We define our family $\hat{\mathcal H}$ as follows:
$$
 \hat{\mathcal H} = \{ h_{a,b}\ :\ a,b\in [0,M) \}
$$
In other words, a uniform $h_{a,b} \in \hat{\mathcal H}$ is a uniformly-random polynomial of degree 1 over $\mathbb Z_M$. Note that this function is fully specified by $a, b, M$ and thus it can be stored in $O(\log M)$ bits. Moreover, $h_{a,b}$ can clearly be evaluated in $O(1)$ time. 
In our applications, the primality requirement for $M$ is not restrictive since for any $x$, a prime number always exists between $x$ and $2x$ (and, on average, between $x$ and $x + \ln(x)$). This will be enough since we will only require asymptotic guarantees for $M$.

We now prove pairwise-independence.

\begin{lemma}\label{lem: h_ab 2 indep}
    $\hat{\mathcal H}$ is a pairwise-independent family.
\end{lemma}
\begin{proof}
Pick any distinct $x_1, x_2 \in [1,n]$ and (not necessarily distinct) $y_1, y_2 \in [0,M-1)$. Crucially, note that since $x_1\neq x_2$ and $M\geq n$, then $x_1 \not\equiv_M x_2$. Let $a,b$ be two uniform integers in $[0,M)$. Then:

$$
\begin{array}{rcll}
   P\left(  h_{a,b}(x_1) = y_1 \wedge h_{a,b}(x_2) = y_2 \right)  & = &  P(ax_1 + b  \equiv_M y_1 \wedge ax_2 + b  \equiv_M y_2)&\\
   & = & P\left( b \equiv_M y_1 - x_1\cdot \frac{y_2-y_1}{x_2-x_1} \wedge a \equiv_M \frac{y_2-y_1}{x_2-x_1} \right) & (1) \\
   & = & P\left( b \equiv_M y_1 - x_1\cdot \frac{y_2-y_1}{x_2-x_1}\right) \cdot P\left(a \equiv_M \frac{y_2-y_1}{x_2-x_1} \right) & (2) \\
   & = & M^{-2} &
\end{array}
$$

Notes:

\begin{itemize}
    \item[(1)] Simply solve the system in the variables $a$ and $b$. Note that $(x_2-x_1)^{-1}$ exists because $x_2\not\equiv_M x_1$ and $\mathbb Z_M$ is a field, thus every element (except 0) has a multiplicative inverse.
    \item[(2)] $a$ and $b$ are independent random variables. 
\end{itemize}
\end{proof}

In general, it can be proved that the family $\hat{\mathcal H} = \left\{ \sum_{i=0}^{k-1} a_i x^i \mod M\ :\ a_0,\dots, a_{k-1} \in [0,M) \right\}$ is $k$-wise independent whenever $M\geq n$ is a power of a prime number. Note that members of this family require $O(k\log M)$ bits of space to be stored and can be evaluated in $O(k)$ time.

\begin{remark}
An important note for later: a function $h\in \hat{\mathcal H}$ maps integers from $[1,n]$ to $[0,M)$, with $M\geq n$. The co-domain size $M\geq n$ might be too large in some applications. An
example is represented by hash tables, see 
Section \ref{sec:hash tables}: if the domain is the space of all IPv4 addresses, $n=2^{32}$ and the table's size must be $M \geq 2^{32}$. This is too much, considering that typically we will 
insert $m \ll n$ objects into the table. In Section \ref{sec:hash tables} we will describe a technique for reducing the co-domain size of $h$ while still guaranteeing good statistical properties.     
\end{remark}

\subsection{Perfect hashing}\label{sec:perfect hashing}

In some applications we will need a \emph{perfect} hash function: 

\begin{definition}
    A hash function $h:[1,n] \rightarrow [0,M)$ is perfect on a set $A \subseteq [1,n]$ if and only if for any $x_1,x_2\in A$ with $x_1\neq x_2$ we have $h(x_1) \neq h(x_2)$ (i.e. $h$ is injective on $A$).
\end{definition}

In general we will be happy with a function that satisfies this property \emph{with high probability}: 

\begin{definition}[with high probability (w.h.p.)]
    We say that an event holds with high probability with respect to some quantity $n$, if its probability is at least  $1-n^{-c}$ for an arbitrarily large constant $c$. Equivalently, we say that the event succeeds with inverse-polynomial probability. 
\end{definition}

Typically, in the above definition $n$ is the size of the input or the universe size (i.e. the number of hashed elements or the hash function's domain size).
In practice, it makes sense to simply ignore the incredibly small failure probability of events holding with high probability. Note that this is reasonable in practice: for example, on a small universe with $n = 10^6$ (typical universes are much larger than that), the small constant $c = 3$ already gives a failure probability of $10^{-18}$. It is far more likely that your program fails due to a cosmic ray flipping a bit in RAM \footnote{\url{www.stackoverflow.com/questions/2580933/cosmic-rays-what-is-the-probability-they-will-affect-a-program}}.

We prove: 

\begin{lemma}
    If a family $\mathcal H$ of functions     $h : [1,n] \rightarrow [0,M)$ is universal and $M \geq n^{c+2}$ for an arbitrarily large constant $c$, then a uniformly-chosen $h\in \mathcal H$ is perfect on any set $A \subseteq [1,n]$ with high probability, i.e. with probability at least $1-n^{-c}$.
\end{lemma}
\begin{proof}
Universality means that $P(h(x_1) = h(x_2)) \leq M^{-1}$ for any $x_1\neq x_2$. Since there are at most $|A|^2 \leq n^2$ pairs of distinct elements in $|A|$, by union bound the probability of having at least one collision is at most $n^2/M$. By choosing $M \geq n^{c+2}$ we have at least one collision with probability at most $n^{-c}$, i.e.
our function is perfect with probability at least $1-n^{-c}$. 
\end{proof}

Note that, by choosing $M \in \Theta(n^{c+2})$, one hash value (as well as the hash function itself) can be stored in $\log_2 M \in O(\log n)$ bits and can be evaluated in constant time using the hash family $\hat{\mathcal H}$ introduced in the previous section. 

Observe that function $h_{a,b}(x) = ax+b\mod M$ is always perfect with probability 1 on $[1,n]$ whenever $M \geq n$ and $a\neq 0$ (exercise: prove it). 

\subsection{Hashing integers to the reals}\label{sec:hash reals}

Let $\mathcal H$ be a family of functions $h : [1,n] \rightarrow \{x\in \mathcal R\ |\ 0\leq x \leq 1\}$  mapping the integers $[1,n]$ to the \textbf{real} interval $\{x\in \mathcal R\ |\ 0\leq x \leq 1\}$. To simplify notation, in these notes we will denote the codomain $\{x\in \mathcal R\ |\ 0\leq x \leq 1\}$ with $[0,1]$ (not to be confused with the interval of \textbf{integers} $[1,n]$ of the domain).  
The integer/real nature of the set $[a,b]$ will
always be clear from the context. We say that $\mathcal H$ is $k$-wise independent iff, for a uniformly-chosen $h\in\mathcal H$, $(h(x_1),\dots, h(x_k))$ is uniform in $[0,1]^k$ for any choice of distinct $x_1, \dots,  x_k$.

It is impossible to algorithmically draw (and store) a uniform function $h : [1,n] \rightarrow [0,1]$, since the interval $[0,1]$ contains infinitely-many numbers. 
However, we can aim at an approximation with any desired degree of precision (i.e. decimal digits of $h(x)$). Here we show how to simulate such a pairwise-independent hash function (enough for the purposes of these notes). 

We start with a pairwise-independent discrete hash function $h' : [1,n] \rightarrow [0,M]$ (just $O(\log M)$ bits of space, see the previous subsection) that maps integers from $[1,n]$ to integers from $[0,M]$ and define $h(x) = h'(x)/M \in [0,1]$. 
Since $h'$ is pairwise-independent, also $h$ is pairwise-independent (over our approximation of $[0,1]$).

In addition to being pairwise-independent, $h'$ should be perfect on the subset of $[1,n]$, of size $d \leq n$, on which the algorithm will work. This is required because, on a truly uniform $h : [1,n] \rightarrow [0,1]$, we have $P(h(x) = h(y)) = 0$ whenever $x\neq y$.
We will be happy with a guarantee that holds with high probability. 
Recalling that pairwise-independence implies universality, by the discussion of the previous section it is sufficient to choose $M \geq n^{c+2}$ to obtain a perfect hash function. 

\paragraph{Simplifications} For simplicity, in the rest of the notes we will simply say ``$h$ is $k$-wise independent/uniform hash function'' instead of ``$h$ is a function uniformly chosen from a $k$-wise independent/uniform family $\mathcal H$''.

\subsection{Hash tables}\label{sec:hash tables}
 
Our goal in this section is to use hashing to build a dynamic set data structure. We want to store sets of cardinality $m$ from universe $[1,n]$ supporting fast membership/insert/delete queries. More formally:

\begin{definition}[Dynamic set data structure]
    A \emph{dynamic set} (also called \emph{dictionary}) over universe $[1,n]$ is a data structure $H$ implementing a set and supporting the following operations for any $x\in [1,n]$ (ideally, in constant time): 
    \begin{itemize}
        \item Insert $x$ in the set
        \item Check if $x$ belongs to the set
        \item Remove $x$ from the set
    \end{itemize}
    After inserting $m$ (distinct) elements, the space of $H$ should be bounded by $O(m)$ words = $O(m \log n)$ bits. 
\end{definition}

\begin{remark}
Note: in these notes we allow a dynamic set structure to use $O(m \log n)$ bits. This is more than the information-theoretic lower bound of $m \log(n/m) + O(m)$ bits required to store a set of cardinality $m$ over universe $[1,n]$. It is actually possible to implement dynamic sets taking optimal space ($m \log(n/m) + O(m)$ bits) and supporting constant-time operations \cite{10.1007/3-540-45061-0_30}. In these notes we will only use the simpler (sub-optimal) $O(m \log n)$-bits variant described in the following paragraphs (hashing by chaining). 
\end{remark}

In the previous sections we introduced hash functions $h: [1,n] \rightarrow [0,M)$ having a low collision probability. 
This suggests that we could use $h(x)$ as an index inside an array $H[0,M-1]$; the low collision probability of $h$ should ensure that few distinct elements are associated with the same array entry $H[h(x)]$ (we will call these events \emph{collisions}).
The simplest hashing scheme following this idea is called \emph{hashing by chaining}. 
Assuming we know the number $m$ of elements that will be inserted in our structure, we choose a hash function $h : [1,n] \rightarrow [0,M)$ such that $M\in \Theta(m)$ and initialize an empty vector $H[0,M-1]$ (also called \emph{hash table}) of size $M$ (indexed from index 0). The cell $H[i]$ contains a linked list---we call it the \emph{$i$-th chain}---, initially empty. Then, operation \emph{insert} will be implemented by appending $x$ at the end of the list $H[h(x)]$.  
Note that this implementation allows us to associate with each $x\in H$ also some satellite date (e.g. a pointer). 

Of course, we want the collision probability of $h$ to be as low as possible. Here we will require universality: for any $x\neq y$, we want $P(h(x) = h(y)) \leq 1/M$. 

The function $h_{ab}:[1,n]\rightarrow [0,M)$ of the previous section requires $M > n$, so the space of $H$ would be prohibitively large (think about the IPv4 example: what if we want to store $m=1000$ values from the set of all possible $n=2^{32}$ IPv4 addresses?). 
We now describe a hash function with codomain size equal to $m$.

\begin{definition}\label{def:bar h}
    Choose a prime number $M > n$. Then, choose two uniform numbers $a\in (0,M)$ and $b\in [0,M)$. We define the following hash function:
    $$
    \bar h(x) = ((a\cdot x + b) \mod M) \mod m
    $$ 
\end{definition}

We prove that function $\bar h$ is universal:

\begin{lemma}
    Function $\bar h(x)$ is universal, i.e. for any $x\neq y$, it holds $P(\bar h(x) = \bar h(y)) \leq 1/m$.
\end{lemma}
\begin{proof}
    Let $x_1\neq x_2$. Our goal is to upper-bound the probability $P(\bar h(x_1) = \bar h(x_2))$. Let $\hat X_1 = ax_1+b \mod M$ and $\hat X_2 = ax_2 + b \mod M$, so $\bar h(x_1) = \hat X_1 \mod m$ and $\bar h(x_2) = \hat X_2 \mod m$. Then we can rewrite the above probability as:
    $$
    P(\bar h(x_1) = \bar h(x_2)) = P(\hat X_1 \equiv_m \hat X_2)
    $$
    
    First, we show that it must be $\hat X_1 \neq \hat X_2$. If, for a contradiction, it were $\hat X_1 = \hat X_2$, then $ax_1+b \equiv_M  ax_2+b$. Since $a\neq 0$, this is equivalent to $x_1 \equiv_M x_2$. However, $M \geq n \geq x_1$ and $M \geq n \geq x_2$, so $x_1 \equiv_M x_2$ implies $x_1=x_2$, a contradiction. 

    Let's go back to the goal of upper-bounding $P(\hat X_1 \equiv_m \hat X_2)$. Applying the law of total probability (on $\hat X_2$):

    $$
    P(\hat X_1 \equiv_m \hat X_2) = \sum_{i=0}^{M-1} P(\hat X_1 \equiv_m \hat X_2 \wedge \hat X_2 = i) = \sum_{i=0}^{M-1} P(\hat X_1 \equiv_m i \wedge \hat X_2 = i)
    $$

    Applying again the law of total probability (this time on $\hat X_1$), the above is equal to 

    $$
    \sum_{i=0}^{M-1}\ \  \sum_{j=0,\dots, M-1, j\neq i} P(\hat X_1 \equiv_m i \wedge \hat X_2 = i \wedge \hat X_1 = j)
    $$

    Let $W = \{j \in [0,M)\ :\ j\neq i \wedge j\equiv_m i\}$ be the set of all values different than $i$ but equivalent to $i$ modulo $m$. The probability $P(\hat X_1 \equiv_m i \wedge \hat X_2 = i \wedge \hat X_1 = j)$ is equal to zero for $j \notin W$, so the above summation is equal to 

    $$
    \sum_{i=0}^{M-1} \sum_{j\in W} P(\hat X_2 = i \wedge \hat X_1 = j)
    $$ 

    We bound the size of $W$. First of all, notice that $|W| \leq \lceil M/m\rceil -1$: in the range $[0,M)$ there are at most $\lceil M/m\rceil$ integers equivalent to $i$ modulo $m$ (from those integers, we need to exclude $i$ so the bound follows).
    We now prove that $\lceil M/m\rceil -1 \leq (M-1)/m$. If $m$ divides $M$, then $\lceil M/m\rceil -1 = M/m - 1 \leq (M-1)/m$. Otherwise, let $k = \lfloor M/m \rfloor$ and write $M = km + k'$, for $k' = M - km$. Note that $1\leq k' < m$ since $m$ does not divide $M$; in particular, $1/m \leq k'/m$. Then: 
    $$
    \frac{M-1}{m} = \frac{M}{m} - \frac{1}{m} \geq \left( \frac{km+k'}{m} \right) - \frac{k'}{m} = k = \lfloor M/m \rfloor \geq \lceil M/m \rceil -1
    $$

    We conclude that $|W| \leq \frac{M-1}{m}$.

    Using the same argument of Lemma \ref{lem: h_ab 2 indep} and noting that $a\neq 0$, one can easily obtain that $P(\hat X_2 = i \wedge \hat X_1 = j) \leq \frac{1}{M(M-1)}$. Finally, putting everything together:

    $$
    P(\bar h(x_1) = \bar h(x_2)) \leq \sum_{i=0}^{M-1} \sum_{j\in W} P(\hat X_2 = i \wedge \hat X_1 = j) \leq M \cdot \frac{M-1}{m} \cdot \frac{1}{M(M-1)} = 1/m
    $$
    which proves universality of $\bar h$.  
\end{proof}

We now show why we require universality of $h$: this guarantees constant expected running time for all queries.

\begin{lemma}
    If $h$ is universal, then the expected length of any fixed chain $H[i]$ is $O(1)$.
\end{lemma}
\begin{proof}
    Let $x_1,\dots, x_m$ be the $m$ distinct elements in the hash table.
    Let $\mathbbm{1}_{i,j}$ be the indicator (Bernoullian) random variable taking value 1 if and only if $h(x_i)=h(x_j)$.
    Since $h$ is universal, $E[\mathbbm{1}_{i,j}] \le 1/m$.
    
    Universality of $h$ implies $E[|H[h(x_i)]|] = E[\sum_{j\neq i} \mathbbm{1}_{i,j}] = \sum_{j\neq i} E[\mathbbm{1}_{i,j}] \leq m\cdot (1/m) = 1$. This is the expected length of any fixed chain, so each operation on the hash table takes expected $O(1)$ time.
\end{proof}

We can lift the assumption that we know $m$ in advance with a classic doubling technique. Initially, we allocate $m=1$ cells for $H$. After having inserted the $m$-th element, we allocate a new table of size $2m$, re-hash all elements in this new table using a new hash function modulo $2m$, and delete the old table. It is easy to see that the total space is linear and operations still take $O(1)$ amortized time. 

\paragraph{Expected longest chain} We have established that a universal hash function generates chains of expected length $O(1)$. This means that $n$ insertions in the hash table will take expected $O(n)$ time. 
Another relevant question is: what is the variance of the chain length, and what is the expected length of the \emph{longest} chain? This is interesting because this quantity is precisely the expected worst-case time we should expect for \emph{one} operation (the slowest one) when inserting $n$ elements in a hash of size $n$. 

We introduce some notation. 
Suppose $x_1, \dots, x_m$ are the elements we want to insert in the hash table. 
Let 

$$
\mathbbm{1}_{i,j} = \left\{
\begin{array}{cl}
   1  &  \mathrm{if}\  h(x_i)=j\\
    0  & \mathrm{otherwise}
\end{array}\right.
$$

be the indicator R.V. taking value $1$ if and only if $x_i$ hashes to the $j$-th hash bucket.
The length of the $j$-th chain is then $L_j = \sum_{i=1}^m \mathbbm{1}_{i,j}$. The quantity $L'_j  = |L_j - E[L_j]|$ indicates how much $L_j$ differs from its expected value; since for universal hash functions we have $E[L_j] = O(1)$ (assuming that the hash' codomain has size $m$), $L'_j = \Theta(L_j)$ so the two R.V.s are asymptotically equivalent (we will study $L'_j$). Let $L'_{max} = \max_{j}L'_j$. Our goal is to study $E[L'_{max}]$.

It turns out that, if $h$ is completely uniform, then $E[L'_{max}] \in O(\log m/\log \log m)$; this is the classic \emph{balls into bins problem}\footnote{\url{en.wikipedia.org/wiki/Balls_into_bins_problem}}. Surprisingly, a simple policy (the so-called \emph{power of two choices}) improves this bound exponentially: let's use two completely uniform hash functions $h_1$ and $h_2$. We insert each element $x$ either in $H[h_1(x)]$ or in $H[h_2(x)]$, choosing the bucket that contains the least number of elements. This simple policy yields $E[L'_{max}] \in O(\log \log m)$.

In practice, however, we almost never use completely uniform hash functions since they require too much space to be stored. What happens if $h$ is simply pairwise-independent? the following theorem holds for any 2-independent hash function (including $\bar h$ of Definition \ref{def:bar h}, even if that function is not completely 2-independent):

\begin{theorem}
    If $h$ is pairwise-independent, then $E[L'_{max}] \in O(\sqrt{m})$.
\end{theorem}
\begin{proof}
If $h$ is pairwise-independent, then it is also 1-independent so $\mathbbm{1}_{i,j} \sim Be(1/m)$. Then, $E[L_j] = m\cdot (1/m) = 1$. 
From pairwise-independence, we also get $Var[L_j] = Var[\sum_i \mathbbm{1}_{i,j}] = m\cdot Var[\mathbbm{1}_{1,j}] = m\cdot (1/m)\cdot(1-1/m) = 1-1/m \leq 1$. Then, applying Chebyshev:
$$
\begin{array}{rcl}
    P(L'_j \geq k) & = & P(|L_j-E[L_j]| \geq k)\\
    & \leq & Var[L_j]/k^2\\
    & \leq & 1/k^2
\end{array}
$$
By union bound: 
$$
\begin{array}{rcl}
    P(L'_{max} \geq k) & = & P(\bigvee_j L'_j \geq k)\\
    & \leq & m/k^2
\end{array}
$$
Let us rewrite $k = \sqrt t \cdot \sqrt m$:
$$
P(L'_{max} \geq \sqrt t \cdot \sqrt m) \leq  1/t
$$
We apply the law of total expectation on the partition of the event space $[0,\sqrt m),  [\sqrt{2^i}\sqrt m, \sqrt{2^{i+1}}\sqrt m)$, for all integers $i \geq 0$ (assume for simplicity that $L'_{max} \in [0,\infty)$: this does not affect our upper bound). 
The probability that $L'_{max}$ falls in the interval $[\sqrt{2^i}\sqrt m, \sqrt{2^{i+1}}\sqrt m)$ is at most $2^{-i}$; moreover, inside this interval the expectation of $L'_{max}$ is (by definition of the interval) at most $\sqrt{2^{i+1}}\sqrt m$, i.e.:
$$
E[L'_{max}|L'_{max}\in [\sqrt{2^i}\sqrt m, \sqrt{2^{i+1}}\sqrt m)] \leq \sqrt{2^{i+1}}\sqrt m$$
Applying the law of total expectation: 
$$
\begin{array}{rcl}
E[L'_{max}] & \leq & \sum_{i=0}^{\infty} 2^{-i}\sqrt{2^{i+1}}\sqrt m\\
& = & \sqrt{2m} \cdot \sum_{i=0}^{\infty} 2^{-i/2}
\end{array}
$$
It is easy to show that $\sum_{i=0}^{\infty} 2^{-i/2} = 2+\sqrt 2$ (prove it as an exercise), which proves our main claim. 
\end{proof}

Alon et al. \cite{alon1999linear} proved that there exist pairwise-independent hash functions with $E[L'_{max}] \in \Omega(\sqrt{m})$, so the above bound is tight in general for pairwise-independent hash functions. 
Nothing however prevents a \emph{particular} pairwise-independent hash function to beat the bound. In fact, 
Knudsen in \cite{knudsen2016linear} proved that the simple function $\bar h$ of Definition \ref{def:bar h} satisfies $E[L'_{max}] \in O(\sqrt[3]{m\log m})$.

\chapter{Randomized   filters}\label{ch:filters}

A \emph{filter} is a randomized data structure encoding a set of $m$ elements from a universe $U$ of cardinality $n$ and supporting typical set operations such as insertion of new elements, membership queries, union/intersection of two sets, frequency estimation (in the case of multi-sets). The data structure is \emph{probabilistic} in the sense that queries such as membership and frequency estimation may return a wrong result with a small (user-defined) probability. Typically, the smaller this probability is, the larger the space of the data structure will be. 

The name \emph{filter} comes from the typical usage case of these data structures: usually, they are used as an interface to a much larger and slower (but exact) set data structure; the role of the filter is to quickly discard negative queries in order to minimize the number of queries performed on the slower data structure. 
Another usage case is to filter streams: filters guaranteeing no false negatives (e.g. Bloom filters, Section \ref{sec:bloom}) can be used to quickly discard most stream elements that do not meet some criterion. 
A typical real-case example comes from databases: when implementing a database management system, a good idea could be to keep in RAM a fast (and small) filter guaranteeing no false negatives (e.g. a Bloom filter). A membership query first goes through the filter; the disk is queried if and only if the filter returns a positive answer. In situations where the user expects many negative queries, such a strategy speeds up queries by orders of magnitude. Another example is malicious URL detection: for example, the Google Chrome browser uses a local Bloom filter to detect malicious URLs. Only the URLs that pass the filter, are checked on Google's remote servers.

Note: also the \emph{sketches} discussed in Chapter \ref{ch:sketch-stream} (e.g. MinHash) are a randomized (approximate) representation of sets. The characterizing difference between those sketches and the filters described in this section, is that the former often require \emph{sub-linear} space (i.e. $o(m)$ bits, where
$m$ is the number of elements in the set), while the latter still require linear ($O(m)$ bits) space. The common feature of the both solutions is that they break the information-theoretic lower bound of $\log_2 \binom{n}{m} = m\log(n/m) + O(m)$ bits which are required in the worst case to represent a set of cardinality $m$ over a universe of cardinality $n$. In general, this is achieved at the price of returning wrong answers with  small probability.

\section{Bloom filters}\label{sec:bloom}

A Bloom filter (Burton H. Bloom, 1970 \cite{bloom1970space}) is a data structure representing a set $S$ under these operations: 

\begin{itemize}
    \item Insert: given an element $x$ (which may be already in the set $S$) update the set as $S \leftarrow S \cup \{x\}$
    \item Membership: given an element $x$, return \texttt{YES} if $x  \in S$ and \texttt{NO} otherwise. 
\end{itemize}

Bloom filters do not support \emph{delete} operations (\emph{counting Bloom filters} do: see Section \ref{sec:counting bloom}). 

Bloom filters (as well as the other filters introduced in this chapter) may return a wrong answer on membership queries, with small probability:

\begin{definition}
    Let $S$ be the set represented by the filter. Suppose we test membership of an element $x$ in $S$, using the filter. 
    \begin{itemize}
        \item If $x\notin S$ but the filter reports that $x$ belongs to $S$, we say that the filter has generated a \emph{False Positive} (FP). We denote with PFP the probability of having a false positive.
         \item If $x \in S$ but the filter reports that $x$ does not belong to $S$, we say that the filter has generated a \emph{False Negative} (FN). We denote with PFN the probability of having a false negative.
    \end{itemize}
\end{definition}

Bloom filters guarantee a bounded one-sided error probability on membership queries, as long as the maximum capacity $m$ of the filter is not exceeded: $PFP > 0$ and $PFN = 0$. We will denote with symbol $\delta = PFP$ the false positive probability.
As we show next, the Bloom filter uses $\Theta(m \log (1/\delta))$ bits of space to store at most $m$ elements ($m$ is the maximum filter capacity) from a universe of cardinality $n$. Notice that this space is independent from $n$ and breaks the lower bound of $m\log(n/m) + O(m)$ bits when $\delta$ is not too small (for example: if $\delta$ does not depend on $n$ and $m$, i.e. if $\delta$ is a constant) and $n$ is much larger than $m$ (which is typically the case: for example, if the universe is the set of all IPv4 addresses, then $n=2^{32}$ but typically we are going to store $m \ll n$ addresses in the set).

\subsection{The data structure}

There exist several variants of Bloom filters; here we cover \emph{partitioned Bloom filters} \cite{mullin1983second}, which are easier to analyze. 

Let $M>0$ and $k>0$ be two integer parameters that we will determine later (as a function of $m$ and of the desired false positive probability $\delta$). Without loss of generality, we assume that $M$ is a multiple of $k$.
Let $h_1, \dots, h_k$ be $k$ hash functions whose domains and codomains are $h_i : U \rightarrow [(i-1)M/k, iM/k )$, for $i=1,\dots, k$. In other words, we partition the interval $[0,M)$ into $k$ non-overlapping sub-intervals, each of size $M/k$. 
Function $h_i$ sends elements from $U$ to the $i$-th such sub-interval, where $U$ is the universe of size $n$ from which the set elements are chosen (for example: integers, strings, etc). 

In our analysis, we will assume that the functions $h_1, \dots, h_k$ are \emph{independent and completely uniform}; in other words, we require that $(h_1(x_1), \dots, h_k(x_k))$ is a $n^k$-wise independent random variable whenever $x_1, \dots, x_k$ are pairwise distinct \footnote{Actually, the original filter proposed in \cite{bloom1970space} is slightly different than the one we are describing: in that case, it is required that $h_1(x), \dots, h_k(x)$ are all distinct (i.e. hashing without replacement). The variant we describe is more common nowadays. See \cite{grandi2018analysis} for an analysis of both variants: they have similar performance.}.
Recall from the previous sections that this is not a realistic assumption, since full independence requires too much space (in this particular case, we would need $O(nk\log M)$ bits of space!).  
While in these notes we make this assumption to simplify the analysis, this is not a strict requirement: in \cite[Sec. 3.3]{mitzenmacher2008simple} the authors show that using $O(1)$-independent hash functions, we can obtain the same theoretical guarantees (false positive rate) of full-independent hashing, at the price of multiplying the filter's space usage by a constant.

The Bloom filter is simply a bit-vector $B[0,M-1]$ of length $M$, initialized with all entries equal to 0. Queries are implemented as follows: 

\begin{itemize}
    \item Insert: to insert $x$ in the set, we set $B[h_i(x)] \leftarrow 1$ for all $i=1, \dots, k$.
    \item Membership: to check if $x$ belongs to the set, we return $\bigwedge_{i=1}^k B[h_i(x)]$.
\end{itemize}

In other words, the filter returns \texttt{YES} if and only if all bits $B[h_1(x)], \dots, B[h_k(x)]$ are equal to 1.
It is easy to see that no false negatives can occur: if the Bloom filter returns \texttt{NO} on a membership query, then the element is not in the set. Equivalently, if an element is in the set then the filter returns \texttt{YES}. However, false positives may occur due to hash collisions. In the next section we analyze their probability. 

See \url{https://florian.github.io/bloom-filters/} for a nice online demo of classic Bloom filters (the ones we study in this section --- \emph{partitioned Bloom filters} --- are slightly different).

\subsection{Analysis}

Suppose we have inserted $S=\{x_1, \dots, x_m\}$ in the Bloom filter. We want to compute the probability $PFP$ that, for any $x\notin S$, the filter erroneously reports that $x$ belongs to $S$.

The insertion of one element in the filter causes the modification of at most $k$ bits in the bitvector $B$, one for each sub-interval of size $M/k$ of the range $[0,M)$ of indexes of $B$. 
Moreover, by definition of full-independence and uniformity of the hash functions, 
the random variables $h_t(x_j)$ are independent for all $1\le t \le k$ and $1\le j \le m$.

We now compute the probability that any bit $B[i]$ is equal to 1. Note that position $i$ belongs to the $t$-th sub-interval of size $M/k$ of $[0,M)$, where $t=\lfloor ik/M \rfloor + 1$.

$$
\begin{array}{ccl}
    P(B[i]=1) & = & P\left(\bigvee_{j\in[m]} h_t(x_j)=i\right) \\
    & = & 1 - P\left(\bigwedge_{j\in[m]} h_t(x_j)\neq i\right) \\
    & = & 1 - \prod_{j\in[m]} P\left(h_t(x_j)\neq i\right) \\
    & = & 1- \prod_{j\in[m]} (1-k/M) \\
    & = & 1  -  \left(  1-\frac{k}{M} \right)^{m} \\ 
    & = & 1  -  \left(\left(  1-\frac{k}{M} \right)^{M/k}\right)^{mk/M} \\
    & \stackrel{M/k\rightarrow \infty}{=} & 1-e^{-mk/M}
\end{array}
$$

Where the last equality holds for $M/k$ (the size of each sub-interval) tending to infinity; this hypothesis is justified in our case since we are performing an asymptotic analysis ($m\rightarrow \infty$) and, as we show next, $M/k >  m$  will hold by construction.

We obtain a false positive if and only if $B[h_j(x)]=1$ for all $1\le j \le k$. Then:
$$
PFP = P(B[h_1(x)]=1 \wedge B[h_2(x)]=1 \wedge \dots \wedge B[h_k(x)]=1)
$$
where $P(B[h_j(x)]=1) = 1-e^{-mk/M}$ for all $1\le j \le k$, as proved above. Since the functions $h_1, \dots, h_k$ are independent and have disjoint domains, the events $B[h_1(x)]=1, B[h_2(x)]=1, \dots, B[h_k(x)]=1$ are \emph{independent}. We obtain:

$$
PFP = P(B[h_1(x)]=1 \wedge B[h_2(x)]=1 \wedge \dots B[h_k(x)]=1) = \prod_{j=1}^k P(B[h_j(x)]=1) = (1-e^{-mk/M})^k
$$
Where the last equality holds for $M/k\rightarrow \infty$ (again, true asymptotically).

It can be shown that the quantity $(1-e^{-mk/M})^k$ is minimized for $k = (M/m)\ln 2$; replacing this value into the above probability, we get that the false positive probability is
$$
PFP = (1/2)^{(M/m)\ln 2}
$$
Solving $(1/2)^{(M/m)\ln 2} = \delta$ as a function of $M$, we finally get $M = m\log_2 e \cdot \log_2(1/\delta) \approx 1.44 \cdot m\log_2(1/\delta)$ and $k = (M/m)\ln 2 = \log_2(1/\delta)$. Note that, indeed, $M/k = m/\ln 2 > m$ (we assumed this earlier). 

\begin{remark}
    An interesting observation: using $k = (M/m)\ln 2$, after exactly $m$ insertions the probability that any bit $B[i]$ is 0 is (for $M\rightarrow \infty$) equal to $e^{-mk/M} = 1/2$. In other words, after $m$ insertions $B$ is a uniform bitvector. This makes sense because it means that the entropy of $B$ is maximized, i.e., we have packed as much information as possible inside it. 
\end{remark}

We obtained:

\begin{theorem}
    Let $0 < \delta < 1$ be a user-defined parameter (false positive rate), and let $m$ be a maximum capacity.
    By using $k = \log_2(1/\delta)$ fully-independent hash functions and $M = m\log_2 e \cdot \log_2(1/\delta)$ bits of space, the Bloom filter supports membership and insert queries,  uses $O\left(m \log(1/\delta)\right)$ bits of space (in addition to the space required to store the hash functions), and guarantees false positive probability at most $\delta$, provided that no more than $m$ elements are inserted into the filter.  
    The Bloom filter does not generate false negatives. 
    Assuming that the hash functions can be evaluated in constant time, all queries take $O(\log(1/\delta))$ time. 
\end{theorem}

In practice, however, $k$ and $M$ as computed above are often not integer values, and $k$ does not divide $M$ (we assumed this earlier). One solution is to choose $k$ as the closest integer to $\log_2(1/\delta)$ and then choose the smallest integer $M$ being a multiple of $k$ and such that $\left(1-e^{-mk/M}\right)^k \leq \delta$. These adjustments do not affect asymptotically the bounds of the theorem above.

\begin{example}\label{ex:bloom}
    Suppose we want to build a Bloom filter to store at most $m = 10^7$ malicious URLs, with false positive probability $\delta = 0.1$. 
    The average URL length is around 77 bytes (see e.g. \url{www.supermind.org/blog/740/average-length-of-a-url-part-2}), so just storing these URLs would require around $734$ MiB.
    Choosing $k = 3$ and $M = 48.100.000$, our Bloom filter uses just $5.73$ MiB of space (about 5 bits per URL) and returns false positives at most $10\%$ of the times. The filter uses 128 times less space than the plain URLs and speeds up negative queries by one order of magnitude (assuming that the filter resides locally in RAM and the URLs are on a separate server or on a local slow disk - so accessing the original data is much slower than using the Bloom filter).
\end{example}

\section{Counting Bloom filters}\label{sec:counting bloom}

What if we wanted to support deletions from the Bloom filter? The idea is to replace the bits of the bitvector $B$ with \emph{counters} of $t$ bits (i.e. able to store integers in the range $[0,2^t)$), for some parameter $t$ to be decided later: 
$B \in [0,2^t)^M$.
The resulting structure is called \emph{counting Bloom filter} and works as follows:

\begin{itemize}
    \item Insert: to insert $x$ in the set, we update $B[h_i(x)] \leftarrow \min\{B[h_i(x)]+1, 2^t-1\}$ for all $i=1, \dots, k$.
    \item Delete: to delete $x$ from the set, we update $B[h_i(x)] \leftarrow \max\{ B[h_i(x)]-1, 0 \}$ for all $i=1, \dots, k$. 
    \item Membership: we return \texttt{YES} if and only if $B[h_i(x)] \geq 1$ for all  $i=1, \dots, k$.
\end{itemize}

To simplify our analysis, we assume that we never insert an element that is already in the set. Similarly, we assume that we only delete elements which are in the set. 
In general, one can check these pre-conditions using the filter and, only if the filter returns a positive answer, checking the (slow) memory storing the set exactly, so we will assume they hold. 

\begin{definition}
    We say that an \emph{overflow} happens when we are trying to increment a counter $B[i]=2^t$ during an \emph{insert} query. Similarly, an \emph{underflow} happens when we are trying to decrement a counter $B[i]=0$ during a \emph{delete} query. 
\end{definition}

The first observation is that, as long as no overflows occur, the filter behaves exactly as a standard Bloom filter: no false negatives occur, and false positives occur with probability at most $\delta$. We therefore choose $M = m\log_2 e \cdot \log_2(1/\delta)$ and $k = (M/m)\ln 2 = \log_2(1/\delta)$ as in the previous section.

Let $T=2^t$. Our goal is to bound the probability that $B[i]$ overflows, for any fixed position $i$. 
Position $i$ belongs to bucket $j = \lfloor i/k \rfloor + 1$ (remember that $B$ is split into $k$ non-overlapping sub-vectors of size $M/k$ each) and therefore is associated with hash function $h_j$.
Let $A=\{x_1, \dots, x_m\}$ be the elements that have been inserted in the filter 
and let $i_1, \dots, i_{m}$ be the (not necessarily distinct) $m$ locations of $B$
that are incremented while inserting $x_1, \dots, x_m$ into the filter through hash function $h_j$. In other words:
$$
\begin{array}{l}
    i_1 = h_j(x_1), i_2 = h_j(x_2),\ \dots,\ i_m = h_j(x_m)
\end{array}
$$

Since we assume $h_j$ to be fully-independent, the indices $i_1, \dots, i_{m}$ are fully-independent random variables uniformly distributed in $[(j-1)M/k,jM/k)$. Then, for any $1\leq p \leq m$:

$$
P(i_p = i) = k/M
$$

If $B[i]$ overflows, then 
$T$ random variables among $i_1, \dots, i_{m}$ take value $i$. In other words, there exist $T$ distinct indices $j_1 < \dots < j_T$ such that $i_{j_1} = \dots = i_{j_T} = i$. As a consequence:
$$
P(B[i]\ \mathrm{overflows}) \le P(\exists\ j_1 < \dots < j_T\ :\ i_{j_1} = \dots = i_{j_T} = i)
$$

By independence of the variables $i_1, \dots, i_{m}$, for any \emph{fixed} choice of $T$ distinct indices $j_1 < \dots < j_T$ we have:

$$
P(i_{j_1} = i_{j_2} = \dots = i_{j_T} = i) = (k/M)^T
$$

There are $\binom{m}{T}$ ways of choosing the $T$ distinct indices $j_1, \dots, j_T$, so by union bound: 

$$
P(\exists\ j_1 < \dots < j_T\ :\ i_{j_1} = i_{j_2} = \dots = i_{j_T} = i) \leq \binom{m}{T} \cdot \left(\frac{k}{M}\right)^T
$$

We can upper-bound this quantity as follows:

$$
\begin{array}{ccl}
     \binom{m}{T} \cdot \left(\frac{k}{M}\right)^T & \leq &  \left( \frac{em}{T} \right)^T \cdot \left(\frac{k}{M}\right)^T \\ 
     & = & \left( \frac{emk}{TM} \right)^T \\
     & = &  \left( \frac{e\ln 2}{T} \right)^T \\
     &\leq & (1/2)^T\ \ \mathrm{for\ T\geq 4\ (i.e.\ for\ t\ge 2)}
\end{array}
$$

Here, the first inequality comes from the inequality $\binom{a}{b} \leq \left(\frac{e\cdot a}{b}\right)^b$, where $e$ is the base of the natural logarithm. 
Between the second and third line, we used the fact that $M=mk \log_2e$.
The last inequality holds for $2^t = T \geq 4$, i.e. $t \geq 2$.
To summarize, we obtained:

$$
P\left(B[i]\  \mathrm{overflows}\right) \leq (1/2)^{2^t}\ \ \mathrm{for\ t\geq 2}
$$

After $m$ insertions and any number of deletions, the filter could return a false negative on a particular query $x$ if at least one of the $k$ counters $B[h_i(x)]$ ($i\in [k]$)  associated with the query overflows at some point. 
These $k$ counters are independent, but for simplicity we use union bound and obtain that a particular query returns a false negative with probability upper-bounded by

$$
PFN \leq  k(1/2)^{2^t}
$$

As the next example shows, in practice the value $t = 4$ (4 bits per counter) is already sufficient to guarantee a negligible probability of false negatives for realistic values of $k$. 

\begin{example}
    Suppose we want to build a counting Bloom filter to store at most $m = 10^7$ malicious URLs, with false positive probability $\delta = 0.1$.
    From Example \ref{ex:bloom}, just storing these URLs would require around $734$ MiB.
    Choosing $k = 3$, $M = 48.100.000$, and $t = 4$, our Bloom filter uses just $22.92$ MiB of space (32 times less than the plain URLs) and returns false positives at most $10\%$ of the times. The probability that a query returns a false negative is at most $0.0046\%$. 
\end{example}

From the theoretical point of view, we may want to start from a user-defined false negative probability $\gamma$ and see how much space (as a function of $m$, $\delta$, and $\gamma$) the filter will take:

$$
PFN \leq  k(1/2)^{2^t} = \gamma
$$

Recalling that we choose $k=\log_2(1/\delta)$ and solving as a function of $t$, we obtain (logarithms are in base 2): 

$$
t = \log\log\left( \frac{\log(1/\delta) }{\gamma} \right)
$$

This is the number of bits per counter. There are $M=m\log_2 e \cdot \log_2(1/\delta)$ counters in total, so the final space usage of the counting Bloom filter is $O\left(m \log(1/\delta)  \log\log\left( \frac{\log(1/\delta) }{\gamma} \right)\right)$ bits. We can summarize this result in a theorem:

\begin{theorem}
    Let $0 < \delta < 1$ and $0<\gamma < 1$ be two user-defined parameter (false positive and false negative rate, respectively), and let $m$ be a maximum capacity.
    By using $k = \log_2(1/\delta)$ fully-independent hash functions and $M = m\log_2 e \cdot \log_2(1/\delta)$ counters of $\log\log\left( \frac{\log(1/\delta) }{\gamma} \right)$ bits each, the Counting Bloom filter supports membership, insert, and delete queries,  uses $O\left(m \log(1/\delta)  \log\log\left( \frac{\log(1/\delta) }{\gamma} \right)\right)$ bits of space (in addition to the space required to store the hash functions), and guarantees false positive probability at most $\delta$ and false negative probability at most $\gamma$, provided that no more than $m$ elements are inserted into the filter.
    Assuming that the hash functions can be evaluated in constant time, all queries take $O(\log(1/\delta))$ time. 
\end{theorem}

\section{Quotient filters}\label{sec:quotient}

Quotient filters (QF) were introduced in 2011 by Bender et al. in \cite{bender2011don}. This filter uses a space slightly larger than classic Bloom filters, with a similar false positive rate. In addition, the QF supports deletes without incurring into false negatives and has a much better cache locality (thus being faster than the Bloom filter in practice).

\subsection{The data structure}

Essentially, a QF is just a clever (space-efficient) implementation of hashing with chaining and quotienting, see Figure \ref{fig:QF1}. We first describe how the filter works by using a standard hash table $T[0,M-1]$ where each $T[i]$ stores a chain. Then, in the next subsection we show how to encode $T$ using just one array $H$ of small integers. We use a uniform hash function $h$ mapping our universe to $[0,2^p)$, for a value $p$ that will be chosen later \footnote{Again, the uniformity assumption is not realistic in practice, but the authors show that, by using ``good in practice'' hash functions, the practical performance follow those predicted by theory}. We break hash values $h(x)$ (of $p$ bits) into two parts: a suffix (\emph{remainder}) $R(x)$ of $r$ bits (i.e. the $r$ least significant bits of $h(x)$) and a prefix (\emph{quotient}) $Q(x)$ of $q = p-r$ bits (i.e. the $q$ most significant bits of $h(x)$). 
The table T has $M = 2^q$ cells. 
The value $q$ is chosen such that $M = 2^q \geq m$ ($m$ is the maximum number of elements that will be inserted in the set) and such that the \emph{load factor} $\alpha = m/M$ of the table, i.e. the fraction of occupied slots, is a small enough constant (a practical evaluation for different values of $\alpha$ is provided in the paper).

\begin{figure}[h!]
    \centering
    \includegraphics[width=0.6\textwidth]{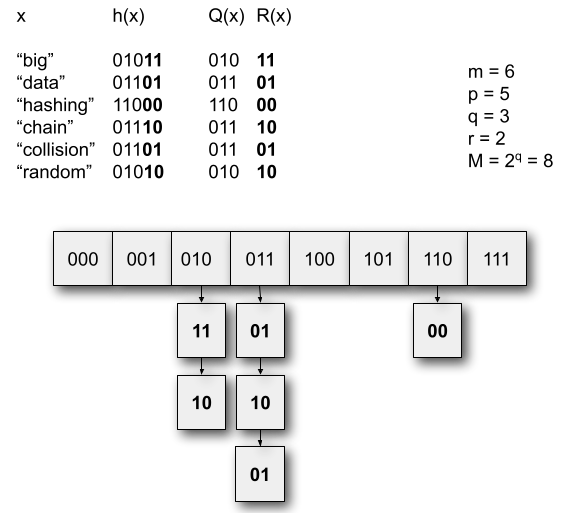}
    \caption{Hashing with chaining and quotienting. A quotient filter is a space-efficient implementation (avoiding pointers) of this hashing scheme, see Figure \ref{fig:QF2}}
    \label{fig:QF1}
\end{figure}

The operations on this simplified implementation of the filter work as follows: 

\begin{itemize}
    \item To insert $x$ in the set, we append $R(x)$ to the chain stored in $T[Q(x)]$. Importantly, we allow repetitions of remainders inside the same chain. 
    \item To remove $x$ from the set, we remove one occurrence of $R(x)$ from the chain stored in $T[Q(x)]$.
    \item To check if $x$ belongs to the set, we check if  $R(x)$ appears inside the chain stored in $T[Q(x)]$.
\end{itemize}

Notice that this scheme allows retrieving $h(x)$ from the table: if remainder $R$ is stored in the $Q$-th chain, then the corresponding fingerprint is $Q \cdot 2^r + R$. In other words, the trick is to exploit the location ($Q$) inside the hash table to store information implicitly, in order to reduce the information ($R$) that is explicitly inserted inside the table. This trick was introduced by Knuth in his 1973 book ``The Art of Computer Programming: Sorting and Searching'', and already allows to save some space with respect to a classic chained hash that stores the full fingerprints $h(x)$ inside its chains. 

Importantly, note that this implementation generates a false positive when we query an element $x$ which is not in the set, and the set contains another element $y\neq x$ with $h(y)=h(y)$. Later we will analyze the false positive probability, which can be reduced by increasing $p$. Note also that, thanks to the fact that we store all occurrences of repeated fingerprints in the table, the data structure does not generate false negatives.

\subsection{Reducing the space}

The QF encodes the table $T$ of the previous subsection using a circular\footnote{\emph{circular} means that the cell virtually following $H[m-1]$ is $H[0]$} array $H[0,M-1]$ of $M = 2^q$ slots, each containing an integer of $r+3$ bits: $r$ bits storing a remainder, in addition to the following $3$ metadata bits.

\begin{enumerate}
    \item \emph{is-occupied[i]}: this bit records whether there exists an element $x$ in the set such that $i = Q(x)$, i.e. if chain number $i$ contains any remainder.
    \item \emph{is-shifted[i]}: this bit is equal to 0 if and only if the remainder $R(x)$ stored in $H[i]$ corresponds to an element $x$ such that $Q(x)=i$, i.e. if $R(x)$ belongs to the $i$-th chain. In other words, \emph{is-shifted[i]}=1 indicates that the remainder $R(x)$ stored in $H[i]$ has been shifted to the right w.r.t. its ``natural'' position $H[Q(x)]$. 
    \item \emph{is-continuation[i]}: this bit is equal to 1 if and only if the remainder $R(x)$ stored in $H[i]$ belongs to the same chain of the remainder $R(y)$ stored in $H[i-1]$, i.e. if the two corresponding set elements $x,y$ are such that $Q(x)=Q(y)$.
\end{enumerate}

Figure \ref{fig:QF2} shows the QF implementation of the hash table of Figure \ref{fig:QF1}. In this example, the QF uses in total $m\cdot(r+3) = 40$ bits (i.e. the bitvector to the right of ``Metadata + remainders = QF'').

\begin{figure}[h!]
    \centering
    \includegraphics[width=0.7\textwidth]{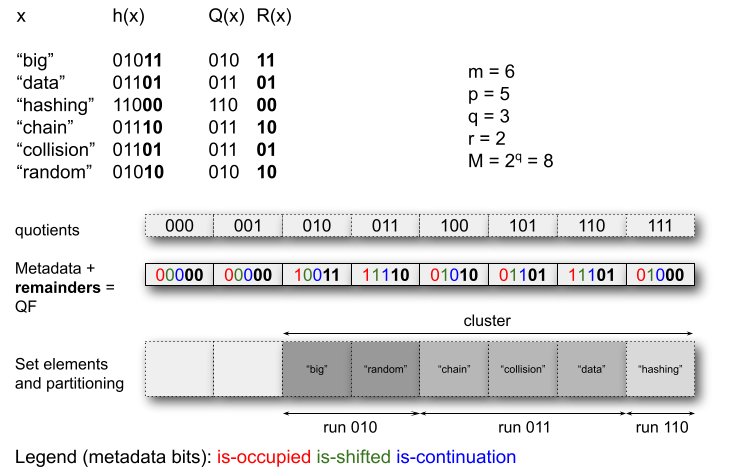}
    \caption{Quotient filter encoding of the hash table in Figure \ref{fig:QF1}.}
    \label{fig:QF2}
\end{figure}

While inserting elements, the following invariant is maintained: if $Q(x) < Q(y)$, then $R(x)$ comes before $R(y)$ in the table. We call \emph{runs} contiguous subsequences corresponding to the same quotient. See Figure \ref{fig:QF2}: there are three runs, sorted by their corresponding quotients. We say that a \emph{cluster} is a maximal contiguous portion $H[i,\dots,j]$ of runs; in particular, $H[i-1]$ and $H[j+1]$ are empty (i.e. do not store any remainder $R(x)$). In Figure \ref{fig:QF2}, there is just one cluster (the array $H$ is circular, so that after the cluster there is indeed an empty slot). 

It is not hard to see that this implementation allows to simulate chaining. Observe that:

\begin{enumerate}
    \item Empty cells are those such that $is\_occupied[i] = 0$ and $is\_shifted[i] = 0$. 
    \item Runs $H[i,\dots, i+k]$ of the same quotient can be identified because $is\_continuation[i]=0$ and $is\_continuation[i+j]=1$ for all $j=1,\dots, k$.   
    \item Points (1) and (2) allow us identifying clusters and runs inside a cluster. Looking at all '1'-bits $is\_occupied[i] = 1$ inside a cluster, we can moreover reconstruct which quotients $Q(x)$ are stored inside the cluster. Note also that $R(x)$ is always stored inside the cluster containing cell $H[Q(x)]$. 
    \item Since quotients in a cluster are sorted and known, and we know their corresponding runs, it is possible to insert/delete/query an element $x$ by scanning the cluster containing position $H[Q(x)]$ (see the original paper \cite{bender2011don} for the detailed algorithms). 
\end{enumerate}

Point (4) above implies that the average/worst-case query times are asymptotically equal to the average/largest cluster length, respectively.

\subsection{Analysis}

A false positive occurs when we query the QF on an element $x$ not in the set, and $h(x) = h(y)$ for some $y$ in the set ($x\neq y$). Since we assume $h$ to be completely uniform, the probability that $h(x) = h(y)$ is $1/2^p$. Then, the probability that $h(x) \neq h(y)$ is $1-1/2^p$, thus the probability that $h(x) \neq h(y)$ for all the $m$ elements $y$ in the set is (again by uniformity of $h$) $\left(1-1/2^p\right)^m$. We conclude that the false positive probability is bounded by

$$
1 - \left(1-\frac{1}{2^p}\right)^m =  1 - \left(1-\frac{1}{2^p}\right)^{2^p\cdot \frac{m}{2^p}} \approx 1-e^{-m/2^p} \leq \frac{m}{2^p} \leq \frac{2^q}{2^p} = 2^{-r}
$$

where the first inequality ($\leq$) follows from the inequality $x\leq \ln\left(\frac{1}{1-x}\right)$ for $x < 1$. By setting $\delta = 2^{-r}$ (where $0<\delta<1$ is the chosen false positive rate), we obtain that the space used by the QF is $M\cdot (r+3) = M\cdot \log_2(1/\delta) + 3M$ bits. Recalling that $M = m/\alpha$, where $0 < \alpha < 1$ is the table's load factor, we finally obtain that the space is $ (m/\alpha)\cdot \log_2(1/\delta) + 3m/\alpha = O(m \log(1/\delta))$ bits. Notice that this is smaller than the space of the Counting Bloom filter (whose space depends also on the false negative probability $\gamma$; recall that Quotient filters have no false negatives). 

The choice of the constant $\alpha$ affects the queries' running times. 
In any case, as the following theorem shows, the length of the longest cluster does not exceed $\Theta(\log m)$ with high probability:

\begin{theorem}\label{th:QF longest cluster}
    For constant load factor $0.5< \alpha < 1$ and any constant $c > 0$, 
    the probability that the longest cluster exceeds length
    $$
    \frac{3(1+c)\ln M}{1/\alpha + \alpha -2} = \Theta(\log m)
    $$
    is at most $m^{-c}$.
\end{theorem}
\begin{proof}
    If $H[i,\dots, i+k-1]$ is contained in a cluster, then $k$ elements $x_1,\dots, x_k$ in the set are such that $Q(x_j) \in [i,i+k-1]$ for all $j=1,\dots, k$. For a fixed $x$, the probability that $Q(x) \in [i,i+k-1]$ is $k/M$. Since the hash is fully uniform, the number $C_i$ of elements that hash inside $H[i,\dots, i+k-1]$ is the sum of $m$ independent Bernoulli variables $Be(k/M)$. 
    Note that $E[C_i] = m\cdot (k/M) = k\alpha$.
    Applying Lemma \ref{lem:mult chernoff} (multiplicative Chernoff), we obtain (note that for $0.5 < \alpha<1$ it holds $0 < 1/\alpha - 1 < 1$ so we can apply the lemma):
    $$
    P(C_i = k) = P(C_i = (1+(1/\alpha-1))E[C]) \leq e^{-k\alpha\cdot (1/\alpha-1)^2/3} = e^{-k\cdot (1/\alpha+\alpha-2)/3}
    $$
    So we have $P(H[i,\dots, i+k-1]\mathrm{\ is\ contained\ in\ cluster}) \leq e^{-k\cdot (1/\alpha+\alpha-2)/3}$.
    The longest cluster's length is greater than or equal to $k$ iff there exists an integer $0\leq i < M$ such that $H[i,\dots, i+k-1]$ is contained in a cluster, so by union bound:
    $$
    P(\mathrm{longest\ cluster\ length} \geq k) \leq M\cdot e^{-k\cdot (1/\alpha+\alpha-2)/3}
    $$
    The above probability is equal to $M^{-c} \leq m^{-c}$ for 
    $$
    k = \frac{3(1+c)\ln M}{1/\alpha + \alpha -2} = \Theta(\log m)
    $$
\end{proof}

Moreover, the expected cluster length is a constant:

\begin{theorem}\label{th:QF expected cluster}
    For constant load factor $0.5<\alpha < 1$, the expected cluster length is $O(1)$.
\end{theorem}
\begin{proof}
    If a cluster of length $k$ starts in position $i$, then $C_i = k$ so the probability that a particular cluster has length $k$ is at most $e^{-k\cdot (1/\alpha+\alpha-2)/3}$ (see the proof of Theorem \ref{th:QF longest cluster}). 
    As a result, the expected cluster length is at most $E[C_i]$.
    Using the fact that $\sum_{k=1}^\infty k\cdot e^{-kc} = \frac{e^c}{(e^c-1)^2}$, we obtain that, for constant $0.5<\alpha < 1$:
    $$
    E[C_i] \leq \sum_{k=1}^\infty k\cdot e^{-k \cdot (1/\alpha+\alpha-2)/3} \leq \frac{e^{(1/\alpha+\alpha-2)/3}}{(e^{(1/\alpha+\alpha-2)/3}-1)^2} \in O(1)
    $$
    See the original paper \cite{bender2011don} for a tighter bound as function of $\alpha$.
\end{proof}

In practice, choosing $\alpha \in [0.5,0.9]$ guarantees a good space-time trade-off. By choosing $\alpha = 0.5$, for example, the space of the filter is $2mr + 6m = 2m\cdot \log_2(1/\delta) + 6m$ bits and $99\%$ of the clusters have less than 24 elements (see \cite{bender2011don}). This space is slightly larger than that of the Bloom filter, but query times of the QF are much faster: each query requires scanning only one cluster which (due to the average cluster length) will probably fit into a single cache line, thus causing at most one cache miss. Bloom filters, on the other hand, generate one cache miss per hash function used: this makes them several times slower than Quotient filters. 

We summarize everything in the following result:

\begin{theorem}
    Let $0 < \delta < 1$ be a user-defined parameter (false positive rate), and let $m$ be a maximum capacity.
    The Quotient filter supports membership, insert, and delete queries,   uses $O\left(m \log(1/\delta)\right)$ bits of space (in addition to the space required to store one fully-independent hash function), and guarantees false positive probability at most $\delta$, provided that no more than $m$ elements are inserted into the filter. 
    The Quotient filter does not generate false negatives. 
    Assuming that the hash function can be evaluated in constant time, all queries take expected $O(1)$ time / worst-case $O(\log m)$ time w.h.p. 
\end{theorem}

\chapter{Sketching and Streaming algorithms}\label{ch:sketch-stream}

\chapquote{
\hfill
\begin{tabular}{l}
     ``Not all the answers are the same \\ Yet we still play the game''
\end{tabular}
}{Greta Van Fleet}{Broken Bells}

\section{Introduction to sketching}

Let $x$ be some data: a set, a string, an integer, etc. 
A \emph{data sketch} is the output of a function $f$ (usually, the combination of a certain number of hash functions) mapping $x$ to a sequence of bits $f(x)$ with properties 1-3 below, plus (depending on the application) also property 4:

\begin{enumerate}
	\item The bit-size of $f(x)$ is much smaller than the bit-size of $x$ (usually, sub-linear or even poly-logarithmic).
     \item $f(x)$ can be used to estimate some properties of $x$. For example, if $x$ is a multi-set then $f(x)$ could be used to compute an approximation of the number of distinct elements contained in $x$, or the most frequent element in $x$. 
     Sometimes, we are able to compute the exact quantity we are interested in. More often, however, it is not possible to compute the exact answer in sublinear space (see point 1 above) and we must accept an approximation (for example, a multiplicative $(1\pm\epsilon)$-approximation). Even worse, we may even fail in computing an approximation; in these cases, we will require that the failure probability $\delta>0$ is arbitrarily small.
    Typically, the success probability will be boosted by repetition (concentration bounds): different instances of the randomized sketch will return different answers concentrated around the value we want to estimate, and a mean+median trick (Section \ref{sec:mean+median}) will do the job.
	\item $f(x)$ can be updated (efficiently) if $x$ gets updated. Importantly, it should be possible to update $f(x)$ \emph{without knowing $x$}. For example:
    \begin{itemize}
        \item if we add an element $y$ to a set $x$, it should be possible to compute $f(x\cup \{y\})$ knowing just $f(x)$ and $y$ (not $x$).
        \item More in general, given two sketches $f(x_1)$ and $f(x_2)$, it should be possible to compute the sketch of the composition  of $x_1$ and $x_2$ (under some operator). For example, 
        if $x_1$ and $x_2$ are sets we could be interested in obtaining the sketch of $f(x_1 \cup x_2)$ given $f(x_1)$ and $f(x_2)$, \emph{without knowing $x_1$ and $x_2$}.
    \end{itemize}
    \item If $x$ and $y$ are similar according to some measure of similarity (e.g. Jaccard similarity between sets), then $f(x)$ and $f(y)$ are likely to be similar (according to some measure of similarity, not necessarily the same as the one between $x$ and $y$). 
\end{enumerate}

Note that, 
when the function $f$ is randomized (e.g. it uses hashing),
 $f(x)$ and $f(y)$ are (in general, dependent) random variables.

We first focus on sketches possessing property 4 and then move to data streams, focusing more on property 2.

\section{Identity - Rabin's hash function (polynomial hashing)}\label{sec:KR}

The most straightforward measure of similarity is \emph{identity}: is $x$ equal to $y$? Without loss of generality, let $x$ be a string of length $n$ over alphabet $\Sigma=  [0,\sigma-1]$, where for simplicity we assume $1<\sigma = |\Sigma|\le n$. Note that we can treat strings as integers of $n$ digits in base $\sigma$.
Note that this setting can also be used to represent subsets of $[1,n]$, letting $\Sigma = \{0,1\}$. Observe that, for any function $f$, if $bitsize(f(x)) < bitsize(x)$ (where $bitsize(\cdot)$ is the function returning the number of bits that an object takes in memory) then collisions must occur: there must exist pairs $x \neq y$ such that $f(x) = f(y)$. 

The first idea to solve the problem could be to use function $\bar h(x) = ((a\cdot x+b)\mod M)\mod m$ of Definition \ref{def:bar h}: we simply view the string $x$ as a number with $n$ digits in base $|\Sigma|$. Unfortunately, this is not a good idea: recalling that we require $M>x$ for any input $x$ of our function, we would need to perform modular arithmetic on integers with $n$ digits in order to update the sketch!

\emph{Rabin's hashing} (also known as \emph{polynomial hashing}) is a string hashing scheme that solves the above problem (but it cannot achieve universality --- even if it guarantees a very low collision probability, see below):

\begin{definition}[Rabin's hash function / polynomial hashing \cite{rabin1981fingerprinting}]
	Fix a prime number $q$, and pick a uniform $z \in [0,q)$. 
	Let $x[1,n] \in \Sigma^n$ be a string of length $n$.
	Rabin's hash function $\kappa_{q,z}(x)$ is defined as: 
	$$
	\kappa_{q,z}(x) = \left(\sum_{i=1}^{n} x[i]\cdot z^{n-i}\right)\mod q 
	$$
\end{definition}

In other words: $\kappa_{q,z}(x)$ is a polynomial modulo $q$ evaluated in $z$ (a random point in $[0,q)$) and having as coefficients the characters of $x$.  \footnote{Another variant of Rabin's hashing draws a uniform prime $q$ instead, and fixes $z = |\Sigma|$} 

Let $|x|$ denote the length of string $x$.
 We define the \emph{Rabin sketch} $f(x)$ of the string $x$ to be the pair
$$
f(x) = (\kappa_{q,z}(x),\ z^{|x|}\mod q)
$$

Note that $f(x)$ uses just $O(\log q)$ bits. Later we show that we can choose $q \leq n^{O(1)}$, so the sketch will use just $O(\log n)$ bits.

First, we show that this sketch is easy to compute and update. Suppose we wish to append a character $c\in\Sigma$ at the end of $x$, thereby obtaining the string $x\cdot c$ (where $x\cdot c$ means \emph{$x$ concatenated with $c$}). The hash value of $x$ can be updated as follows (Horner's method for evaluating polynomials):

\begin{lemma}
	$
	\kappa_{q,z}(x\cdot c) = (\kappa_{q,z}(x) \cdot z + c ) \mod q
	$
\end{lemma}

The length of $x\cdot c$ is $|x|+1$ and $z^{|x|+1} \mod q = (z^{|x|}\mod q)\cdot z \mod q$.

The above lemma gives us also an efficient algorithm for computing $\kappa_{q,z}(x)$: start from $\kappa_{q,z}(\epsilon) = 0$ (where $\epsilon$ is the empty string) and append the characters of $x$ one by one. 

Using a similar idea, we can concatenate the sketches of two strings in constant time, as follows:

\begin{lemma}
	$
	\kappa_{q,z}(x\cdot y) = (\kappa_{q,z}(x) \cdot z^{|y|} + \kappa_{q,z}(y))\mod q
	$
\end{lemma}

The length of the string $x\cdot y$ is $|x|+|y|$ and the value $z^{|x|+|y|}\mod q$ can be computed efficiently as  $z^{|x|+|y|}\mod q = ((z^{|x|}\mod q) \cdot (z^{|y|}\mod q)) \mod q$. 

We prove another important property of Rabin's hashing: if $x\neq y$, then $\kappa_{q,z}(x) \neq \kappa_{q,z}(y)$ with high probability. This is implied by the following lemma:

\begin{lemma}
	Let $x \neq y$, with $\max(|x|,|y|) = n$. Then:
	$$
	P(\kappa_{q,z}(x) = \kappa_{q,z}(y)) \leq n/q
	$$
\end{lemma}
\begin{proof}
	Note that $P(\kappa_{q,z}(x) = \kappa_{q,z}(y)) = P(\kappa_{q,z}(x) - \kappa_{q,z}(y) \equiv_q 0)$. Now, the quantity $\kappa_{q,z}(x) - \kappa_{q,z}(y)$ is, itself, a polynomial. Let $x-y$ be the string such that $(x-y)[i] = x[i]-y[i] \mod q$, where we left-pad with zeros the shortest of the two strings (so that both have $n$ characters). Then, it is easy to see that:
	$$
	\kappa_{q,z}(x) - \kappa_{q,z}(y) \mod q = \kappa_{q,z}(x-y)
	$$
	It follows that the above probability is equal to  $P(\kappa_{q,z}(x-y) \equiv_q 0)$. Since $x\neq y$, $\kappa_{q,z}(x-y)$ is a polynomial of degree at most $n$ over $\mathbb Z_q$ (evaluated in $z$) and it is not the zero polynomial. 
	Recall that any non-zero univariate polynomial of degree $n$ over a field has at most $n$ roots. Since $q$ is prime, $\mathbb Z_q$ is a field and thus there are at most $n$ values of $z$ such that $\kappa_{q,z}(x-y) \equiv_q 0$. Since we pick $z$ uniformly from $[0,q)$, the probability of picking a root is at most $n/q$.
\end{proof}

\begin{corollary}
	Choose a prime $n^{c+1} \leq q \leq 2\cdot n^{c+1}$ for an arbitrarily large constant $c$. Then, $bitsize(\kappa_{q,z}(x)) \in O(\log n)$ bits and, for any $x\neq y$:
	$$
	P(\kappa_{q,z}(x) = \kappa_{q,z}(y)) \leq n^{-c}
	$$
	that is, $x$ and $y$ collide with low (inverse polynomial) probability. 
\end{corollary}

Later in these notes, Rabin's hashing will be used to solve pattern matching in the streaming model. 
As noted above, Rabin hashing can be used also to sketch sets of integers under the following operations (prove it as an exercise):

\begin{enumerate}
    \item Inserting an element in the set, provided that the element does not belong to the set before the insertion takes place. This operation can be implemented in $O(\log q) = O(\log n)$ time (hint: here we need fast exponentiation to evaluate arbitrary powers $z^i \mod q$).
    \item Deleting an element from the set, provided that the element belongs to the set before the deletion takes place. This operation can be implemented in $O(\log q) = O(\log n)$ time.
    \item Computing the sketch of the union of two disjoint sets in constant time. 
    \item Checking the identity of two sets in constant time (succeeding w.h.p.). 
\end{enumerate}

Observe that, as opposed to the filters of Chapter \ref{ch:filters}, Rabin hashing allows us to squeeze an arbitrary subset of $[1,n]$ in just $O(\log n)$ bits! The price to pay is that we are limited just to the operations (1-4) above. Since operation 4 fails with low probability, we are not able to reconstruct the underlying set and therefore we do not break any information-theoretic lower bound. 

\section{Metrics}

Our next goal is to devise similarity-preserving sketches for \emph{distance metrics}.
A distance metric over a set $A$ is a function $d:A\times A \rightarrow \mathbb R$ with the following properties: 

\begin{itemize}
    \item Non-negativity: $d(x,y) \geq 0$
    \item Identity: $d(x,y) = 0$ iff $x=y$
    \item Simmetry: $d(x,y) = d(y,x)$
    \item Triangle inequality: $d(x,z) \leq d(x,y) + d(y,z)$
\end{itemize}

For example, the \emph{Jaccard distance} $d_J(x,y) = 1 - J(x,y) = 1 - |x \cap y|/|x\cup y|$ defined over sets (see also Section \ref{sec:minhash}) is indeed a distance metric.
Some examples of distances among vectors $x,y \in \mathbb R^n$ are:

\begin{itemize}
    \item $L_p$ norm (or Minkowski distance): $L_p(x,y) = \left( \sum_{i=1}^n |x_i-y_i|^p  \right)^{1/p}$
    \item $L_2$ norm (or Euclidean distance): $L_2(x,y) = \sqrt{\sum_{i=1}^n (x_i-y_i)^2}$
     \item $L_1$ norm (or Manhattan distance): $L_1(x,y) = \sum_{i=1}^n |x_i-y_i|$
     \item $L_\infty$ norm: $L_\infty(x,y) = \max\{ |x_1-y_1|, \dots, |x_n-y_n|\}$
     \item Cosine distance: $d_{cos}(x,y) = 1-cos(x,y) = 1 - \frac{x\cdot y}{\lVert x \rVert \cdot \lVert y \rVert} = 1- \frac{\sum_{i=1}^n x_iy_i}{\sqrt{\sum_{i=1}^n x_i^2} \cdot \sqrt{\sum_{i=1}^n y_i^2}}$
\end{itemize}

Between strings, we have:

\begin{itemize}
    \item Hamming distance between two equal-length strings: $H(s_1,s_2)$ is the number of positions $s_1[i]\neq s_2[i]$ in which the two strings differ. On alphabet $\{0,1\}$ it is equal to $L_1(s_1,s_2)$.
    \item Edit distance between any two strings: $Ed(s_1,s_2)$ is the minimum number of edits (substitutions, single-character inserts/deletes) that have to be applied to $s_1$ in order to convert it into $s_2$. 
\end{itemize}

All these metrics allow efficient sketching schemes. In the next two sections we study two particular cases: Hamming and Jaccard. 
In Section \ref{sec:P&P - Hamming} we will provide another solution for the Hamming distance. 
In Section \ref{sec: dimensionality reduction} we will provide a data sketch preserving the Euclidean distance.

\section{Hamming distance}\label{sec:Hamming sketch}

We devise a simple sketching mechanism for the \emph{normalized Hamming distance} between strings $x,y\in \Sigma^n$. We assume $|\Sigma| \le n$. 
Given two strings $x,y\in \Sigma^n$ of the same length $n$, the \emph{normalized Hamming distance} $d_H(x,y)$ of $x$ and $y$ is the number of positions where $x$ and $y$ differ, normalized by $n$:

$$
d_H(x,y) = \frac{1}{n} \sum_{i=1}^n (x[i]\neq y[i])
$$ 

where $(x[i]\neq y[i]) = 1$ if $x[i]\neq y[i]$, and 0 otherwise. 
Note that two strings are equal if and only if $d'_H(x,y)=0$ ($d_H(x,y)$ is indeed a metric). 

The normalized Hamming distance admits a very simple similarity-preserving sketching mechanism. Choose a uniform $i \in [1,n]$. Our first estimator is simply $(i,x[i])$, i.e. a character of $x$ at a uniform position $i$ (in addition to the position itself).

It is easy to see that $P(x[i] \neq y[i]) =  d_H(x,y)$: since $i$ is uniform, $x[i] \neq y[i]$ is true if and only if we choose a position where $x$ and $y$ differ. But, since there are $n\cdot d_H(x,y)$ such positions, the probability of choosing one of them is $n\cdot d_H(x,y)/n = d_H(x,y)$.

Let $(x[i] \neq y[i])$ be the indicator (Bernoullian) random variable (where $i$ is uniform in $[1,n]$) taking value 1 if and only if $x[i] \neq y[i]$. By the above observation, $E[(x[i] \neq y[i])] = d_H(x,y)$.

Unfortunately, $(x[i] \neq y[i])$ is not a good estimator since it has a large variance: in the worst case ($d_H(x,y)=0.5$), we have $Var[(x[i] \neq y[i])] = 0.25$ and thus the expected error (standard deviation) of $(x[i] \neq y[i])$ is $\sqrt{Var[(x[i] \neq y[i])]} = 0.5$. This means that on expectation  we are off by $50\%$ from the true value of $d_H(x,y)=0.5$. 

To reduce the variance, we average out $k$ i.i.d. indicators $(x[i] \neq y[i])$. In other words, we choose $k$ uniform indices $i_1, \dots, i_k \in [1,n]$ and define our sketch $f(x)$ for $x$ to be
$$
f(x) = ((i_1,x[i_1]), \dots, (i_k,x[i_k]))
$$

Note that, since we are storing also the positions $i_j$ in the sketch $f(x)$, we can easily update the sketch when the string $x$ gets updated by a substitution of the form $x[i] \leftarrow c$, for any $c\in \Sigma$. The sketch uses $O(k\log n)$ bits of space (remember that we assume $|\Sigma| \le n$).

To compare strings $x$ and $y$ using just their sketches $f(x)$ and $f(y)$, we compute the estimator:

$$
H(x,y) = \frac{1}{k}\sum_{j=1}^k (x[i_j] \neq y[i_j])
$$

Applying Corollary \ref{cor:avg Chernoff-Hoeffding} (Chernoff-Hoeffding on the average of $k$ Bernoullian RVs): 

\begin{theorem}
	Fix any desired absolute error $0 \leq \epsilon \leq 1$ and failure probability $0 < \delta \leq 1$.
    Then, $f(x)$ and $f(y)$ use $O(\epsilon^{-2} \log(1/\delta) \log n)$ bits of space and the estimator $H(x,y)$ exceeds absolute error $\epsilon$ with probability at most $\delta$, i.e. 
	$$
	P(|H(x,y) - d_H(x,y)| \geq \epsilon) \leq \delta
	$$
\end{theorem}
\begin{proof}
    By linearity of expectation: 
    $$
    E[\sum_{j=1}^k (x[i_j] \neq y[i_j])] = \sum_{j=1}^k E[(x[i_j] \neq y[i_j])] = k \cdot d_H(x,y)
    $$
    Corollary \ref{cor:avg Chernoff-Hoeffding} (Chernoff-Hoeffding) gives us:
    $$
    \begin{array}{ccl}
         P(|H(x,y) - d_H(x,y)| \geq \epsilon) & = & P(|\frac{1}{k}\sum_{j=1}^k (x[i_j] \neq y[i_j]) - d_H(x,y)| \geq \epsilon) \\
         & = & P(|\sum_{j=1}^k (x[i_j] \neq y[i_j]) - k\cdot d_H(x,y)| \geq k\epsilon) \\
         & \le & 2e^{-(k\epsilon)^2/(2k)} \\
         & = & 2e^{-\epsilon^2k/2}
    \end{array}
    $$
    Solving $2e^{-\epsilon^2k/2} = \delta$ as a function of $k$, we obtain $k = 2\epsilon^{-2}\ln(2/\delta)\in O\left( \frac{\log(1/\delta)}{\epsilon^2} \right)$. Our claim follows.
\end{proof}

\section{Jaccard similarity - MinHash}\label{sec:minhash}

MinHash is a sketching technique used to estimate the \emph{Jaccard similarity} of sets. It was invented by Andrei Broder in 1997 and initially used in the AltaVista search engine to detect duplicate web pages and eliminate them from search results.

Here we report just a definition and analysis of MinHash. For more details and applications see Leskovec et al.'s book \cite{leskovec2020mining}, Sections 3.1 - 3.3. 

The Jaccard similarity $J(A,B)$ of two sets $A$ and $B$ is defined as follows:

\begin{definition}[Jaccard similarity]
    $J(A,B) = \frac{|A\cap B|}{|A\cup B|}$
\end{definition}

The symmetric \emph{Jaccard distance} $d_J$ is defined as $d_J(A,B) = 1- J(A,B)$. It can be shown that $d_J$ is a metric. While in this section we show a sketching technique  estimating $J(A,B)$, the same technique works also for $d_J(A,B)$.

Without loss of generality, we may assume that we work with sets of integers from the universe $[1,n]$. This is not too restrictive, since we can always convert any object into an integer using a collision-free (w.h.p.) hash function.

\begin{definition}[MinHash hash function]\label{def:MinHash function}
Let $h$ be a hash function. The MinHash hash function of a set $A$ is defined as $\hat h(A) = \min\{h(x)\ :\ x\in A\}$, i.e. it is the minimum of $h$ over all elements of $A$.    
\end{definition}

\begin{definition}[MinHash estimator]\label{def:minhash estimator}
Let $\hat J_h(A,B)$ be the indicator R.V. defined as follows:
\begin{equation}
\hat J_h(A,B) = 
\begin{cases*}
1 & if $\hat h(A) = \hat h(B)$ \\
0        & otherwise
\end{cases*}
\end{equation}
\end{definition}

Note that $\hat J_h(A,B)$ is a Bernoullian R.V. We prove the following remarkable property:

\begin{lemma}\label{lem:jaccard minhas}
    If $h:[1,n]\rightarrow [1,n]$ is a uniform permutation, then
    $E[\hat J_h(A,B)] = J(A,B)$
\end{lemma}
\begin{proof}
    Let $|A\cup B| = N$.
  For  $i \in A\cup B$, consider the event $smallest(i)$ being true if and only if $h(i) = \hat h(A\cup B)$, i.e. if $i$ is the element of $A\cup B$ mapped to the smallest hash $h(i)$ (among all elements of $A\cup B$). Since $h$ is a permutation, exactly one element from $A\cup B$ will be mapped to the smallest hash (i.e. $smallest(i)$ is true for exactly one $i\in A\cup B$), so $\{smallest(i)\}_{i \in A\cup B}$ is a partition of cardinality $N = |A\cup B|$ of the event space.
  Moreover, the fact that $h$ is a uniform permutation implies that $P(smallest(i)) = P(smallest(j))$ for all $i,j \in A\cup B$: every element of $A\cup B$ has the same chance to be mapped to the smallest hash (among elements of  $A\cup B$). This implies that $P(smallest(i)) = 1/|A\cup B| = 1/N$ for every $i \in A\cup B$.

  Note that, if we know that $smallest(i)$ is true and $i \in A\cap B$, then $\hat J_h(A,B) = 1$ (because $i$ belongs to both $A$ and $B$ and $h$ reaches its minimum $\mathtt{min}$ on $i$, thus $\hat h(A) = \hat h(B) = \mathtt{min}$). On the other hand, if we know that $smallest(i)$ is true and $i \in (A\cup B) - (A\cap B)$, then $\hat J_h(A,B) = 0$ (because $i$ belongs to either $A$ or $B$ --- not both --- and $h$ reaches its minimum $\mathtt{min}$ on $i$, thus either $\hat h(A) \neq \hat h(B) = \mathtt{min}$ or $\mathtt{min} = \hat h(A) \neq \hat h(B)$ holds).
  
  Using this observation and applying the law of total expectation (Lemma \ref{lem:total exp}) to the partition $\{smallest(i)\}_{i \in A\cup B}$ of the event space we obtain:

  $$
  \begin{array}{ccl}
  	E[\hat J_h(A,B)] & = & \sum_{i\in A\cup B} P(smallest(i)) \cdot E[\hat J_h(A,B)\ |\ smallest(i)]\\
  	& = & \sum_{i\in A\cup B} \frac{1}{N} \cdot E[\hat J_h(A,B)\ |\ smallest(i)] \\
  	& = & \sum_{i\in A\cap B} \frac{1}{N} \cdot E[\hat J_h(A,B)\ |\ smallest(i)] + \sum_{i\in (A\cup B)-(A\cap B)} \frac{1}{N} \cdot E[\hat J_h(A,B)\ |\ smallest(i)]\\
  	& = & \sum_{i\in A\cap B} \frac{1}{N} \cdot 1 + \sum_{i\in (A\cup B)-(A\cap B)} \frac{1}{N} \cdot 0\\
  	& = & \frac{1}{N} \cdot \sum_{i\in A\cap B} 1\\
  	& = & \frac{1}{N} |A\cap B| \\
  	& = & \frac{|A\cap B|}{|A\cup B|}\\
  	& = & J(A,B)
  \end{array}
  $$  
\end{proof}

The above lemma states that $\hat J_h(A,B)$ is an unbiased estimator for the Jaccard similarity. Note that evaluating the estimator only requires knowledge of $\hat h(A)$ and $\hat h(B)$: an entire set is squeezed down to just one integer!

\subsection{Min-wise independent permutations}

The main drawback of the previous approach is that $h$ is a random permutation. There are $n!$ random permutations of $[1,n]$, so $h$ requires $\log_2 (n!) \in \Theta(n\log n)$ bits to be stored. What property of $h$ makes the proof of Lemma \ref{lem:jaccard minhas} work? It turns out that we need the following: 

\begin{definition}[Min-wise independent hashing]
    Let $h:[1,n] \rightarrow [0,M)$ be a function from some family $\mathcal H$.
    For any subset $A\subseteq [1,n]$ and $i\in A$, let $smallest_h(A,i) = (\forall j\in A - \{i\})(h(i) < h(j))$.
    
    The family $\mathcal H$
    is said to be min-wise independent if, for a uniform $h\in \mathcal H$, $P(smallest_h(A,i)) = 1/|A|$ for any $A\subseteq [1,n]$ and $i\in A$.
\end{definition}

In other words, $\mathcal H$ is min-wise independent if, for any subset of the domain, any element is equally likely to be the minimum (through a uniform $h\in \mathcal H$). The definition could be made more general by further relaxing the uniformity requirement on $h$.

Unfortunately, Broder et al. \cite{broder2000min} proved that any family of min-wise independent permutations must include at least $e^{n-o(n)}$ permutations, so a min-wise independent function requires at least $n \log_2 e \approx 1.44 n$ bits to be stored. This lower bound is easy to prove. First, observe that any $h\in\mathcal H$ identifies exactly one minimum in $A$. Since every $i\in A$ should have the same probability to be mapped to the minimum through a uniform $h\in\mathcal H$, it follows that $|A|$ must necessarily divide $|\mathcal H|$. This should hold for every $A\subseteq [1,n]$, so each $k=1, 2, \dots, n$ should divide $|\mathcal H|$ and therefore $|\mathcal H|$ cannot be smaller than the least common multiple of all numbers $1, 2, \dots, n$. The claim follows from the fact that $lcm(1,2,\dots, n) = e^{n-o(n)}$.  \footnote{\footnotesize 
See \url{https://en.wikipedia.org/wiki/Chebyshev_function}.}

There are two solutions to this problem:
\begin{enumerate}
\item (\emph{k-min-wise independent hashing}) We require $P(smallest_h(A,i)) = 1/|A|$ only for sets of cardinality $|A|\leq k$.
\item (\emph{Approximate min-wise hashing}): we require $P(smallest_h(A,i)) = (1\pm \epsilon)/|A|$ for a small error $\epsilon > 0$.
\end{enumerate}

Also combinations of (1) and (2) are possible. 
A hash with property (1) can be stored in $O(k)$ bits of space and is a good compromise: in practice, $k$ is the cardinality of the union of the two largest sets in our dataset (much smaller than the universe's size $n$).
As far as solution (2) is concerned, there exist hash functions of size $\Theta(\log(1/\epsilon) \cdot \log n)$ bits with this property.
Such functions can be used to estimate the Jaccard similarity with absolute error $\epsilon$.
For more details, see \cite{indyk2001small,patracscu2012power}. 

\subsection{Reducing the variance}

To reduce the variance, we average out $k$ independent estimators, for sufficiently large $k$.
Let $h_i:[1,n] \rightarrow [1,n]$, with $i=1,\dots, k$, be $k$ independent uniform permutations. 
We define the MinHash sketch of a set $A$ to be the $k$-tuple:

\begin{definition}[MinHash sketch]
	$h_{min}(A) = (\hat h_1(A), \hat h_2(A), \dots, \hat h_k(A))$
\end{definition}

In other words: the $i$-th element of $h_{min}(A)$ is the smallest hash $h_i(x)$, for $x\in A$. 
Note that the MinHash sketch of a set $A$ can be easily computed in $O(k|A|)$ time, provided that $h$ can be evaluated in constant time. 
Then, we estimate $J(A,B)$ using the following estimator: 

\begin{definition}[Improved MinHash estimator]
$$
J^+(A,B) = \frac{1}{k}\sum_{i=1}^k \hat J_{h_i}(A,B)
$$
\end{definition}

In other words, we compute the average of $\hat J_{h_i}(A,B)$ for $i=1, \dots, k$. Note that the improved MinHash estimator can be computed in $O(k)$ time given the MinHash sketches of two sets. 

We can immediately apply the Chernoff-Hoeffding bound for the average of Bernoullian R.V. (Corollary \ref{cor:avg Chernoff-Hoeffding}) and obtain that $P(|J^+(A,B) - J(A,B)| \geq \epsilon) \leq 2e^{-\epsilon^2k/2}$ for any desired absolute error $0 < \epsilon \leq 1$. Fix now any desired failure probability $0 < \delta \leq 1$. By solving $2e^{-\epsilon^2k/2} = \delta$ we obtain $k = 2\ln(2/\delta)/\epsilon^2$. We can finally state: 

\begin{theorem}
	 Fix any desired absolute error $0 < \epsilon \leq 1$ and failure probability $0 < \delta \leq 1$. 
    The MinHash sketch uses $O(\epsilon^{-2}\log(1/\delta)\log n)$ bits of space and allows computing in $O(\epsilon^{-2}\log(1/\delta))$ time an estimator $J^+(A,B)$ which exceeds absolute error $\epsilon$ with probability at most $\delta$, i.e. 
	 $$
	 P(|J^+(A,B) - J(A,B)| \geq \epsilon) \leq \delta
	 $$
\end{theorem}

Note that it is easy to combine the MinHash sketches of two sets $A$ and $B$ so to obtain the MinHash sketch of $A\cup B$ (similarly, to compute the MinHash sketch of $A\cup \{x\}$ given the MinHash sketch of $A$): $h_{min}(A\cup B) = (\min\{\hat h_1(A), \hat h_1(B)\}, \dots, \min\{\hat h_k(A), \hat h_k(B)\})$.

\section{Locality-sensitive hashing (LSH)}\label{sec:LSH}

Suppose our task is to find all similar pairs of elements (small distance $d(x,y)$, for some metric $d$) in a data set $A \subseteq U$ ($U$ is some universe). While a distance-preserving sketch (e.g. for Jaccard distance) speeds up the computation of $d(x,y)$, we still need to compute $O(|A|^2)$ distances in order to find all similar pairs! On big data sets this is clearly not feasible. 

Locality-sensitive hash functions are used to accelerate the search of similar elements in a data set, where similarity is usually measured in terms of a distance metric. 
The main intuition behind LSH is that we want similar items to be hashed to the same value. In other words, we want to \emph{maximize} hash collisions between similar items.


\subsection{The theory of LSH}

A locality-sensitive hash function for some distance metric $d : U\times U \rightarrow \mathbb R$ is a function $h : U \rightarrow [0,M)$ such that similar elements (i.e. $d(x,y)$ is small) are likely to collide: $h(x) = h(y)$. This is useful to drastically reduce the search space  with the following algorithm:

\begin{enumerate}
    \item Scan the data set $A$ and put each element $x\in A$ in bucket $H[h(x)]$ of a hash table $H$.
    \item Compute distances only between pairs inside each bucket $H[i]$.
\end{enumerate}

Classic hash data structures use $O(m)$ space for representing a set of $m$ elements and support insertions and lookups in $O(1)$ expected time (see Section \ref{sec:hash tables}). More advanced data structures\footnote{\footnotesize Dietzfelbinger, Martin, and Friedhelm Meyer auf der Heide. ``A new universal class of hash functions and dynamic hashing in real time.'' International Colloquium on Automata, Languages, and Programming. Springer, Berlin, Heidelberg, 1990.} support queries in $O(1)$ worst-case time with high probability. In the following, we will therefore assume constant-time operations for our hash data structures. 

LSH works by first defining a distance threshold $t$. Ideally, we would like the collision probability to be equal to 0 for pairs such that $d(x,y) > t$ and equal to 1 for pairs such that $d(x,y) \leq t$. For example, using a distance $d:U\times U \rightarrow [0,1]$ (e.g. Jaccard distance) the ideal LSH function should be the one depicted in Figure \ref{fig:LSH1}. 

\begin{figure}[h!]
\centering
	\includegraphics[scale=0.45]{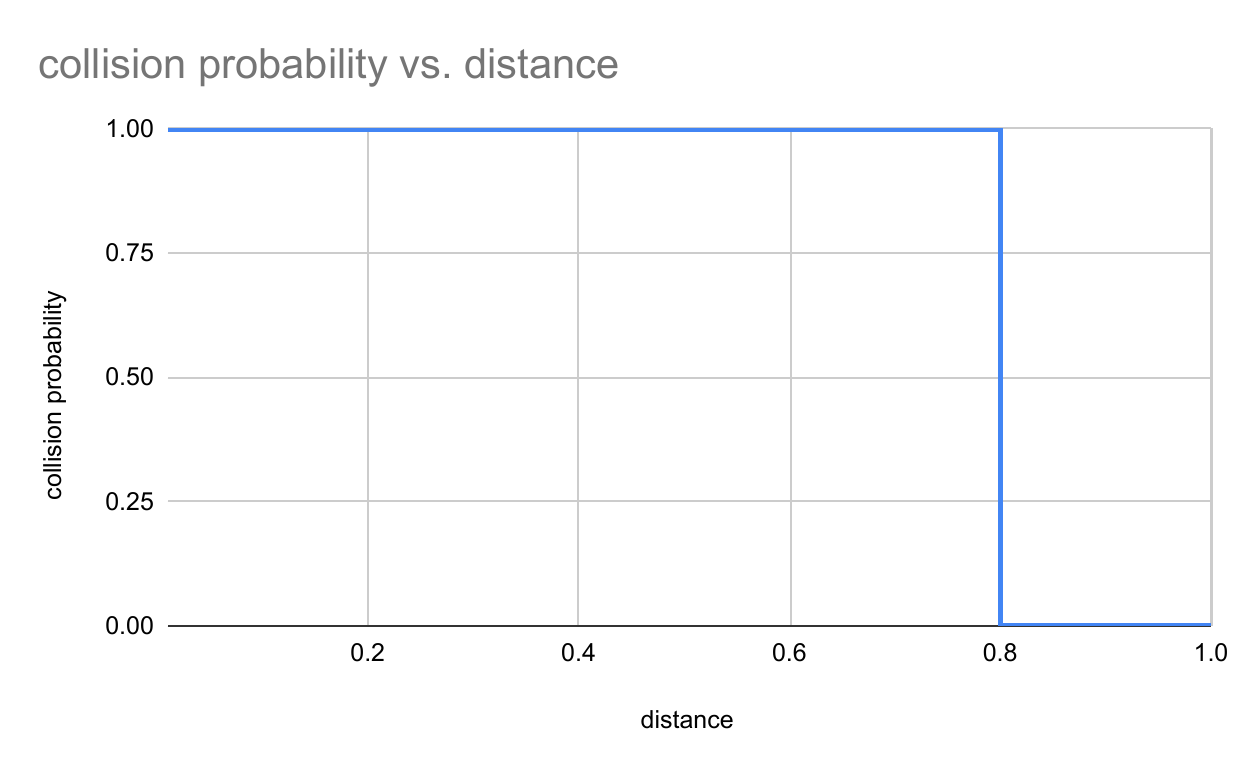}
	\caption{The ideal locality-sensitive hash function: elements whose distance is below the threshold $t=0.8$ collide with probability 1; elements whose distance is above the threshold do not collide.}\label{fig:LSH1}
\end{figure}

In practice, we are happy with a good approximation:

\begin{definition}
    A $(d_1,d_2, p_1, p_2)$-sensitive family $\mathcal H$ of hash functions is such that, for a uniformly-chosen $g \in \mathcal H$, we have:
    \begin{itemize}
        \item If $d(x,y) \leq d_1$, then $P(g(x) = g(y)) \geq p_1$.
        \item If $d(x,y) \geq d_2$, then $P(g(x) = g(y)) \leq p_2$.
    \end{itemize}
\end{definition}

Intuitively, we want $d_1$ and $d_2$ to be as close as possible ($d_1\leq d_2$), $p_1$ as large as possible, and $p_2$ as small as possible. 
To abbreviate, in the following we will say that \emph{$h$ is a $(d_1,d_2, p_1, p_2)$-sensitive hash function} when it is uniformly drawn from a $(d_1,d_2, p_1, p_2)$-sensitive family. 
For example, Figure \ref{fig:LSH2} shows the behaviour of a $(0.4,0.7, 0.999, 0.007)$-sensitive hash function for Jaccard distance (see next subsection for more details). 

We now show how locality-sensitive hash functions can be \emph{amplified} in order to obtain different (better) parameters. 

\subsubsection{AND construction}

Suppose $\mathcal H$ is a $(d_1,d_2, p_1, p_2)$-sensitive family. Pick uniformly $r$ independent hash functions $h_1, \dots, h_r \in \mathcal H$, and define:

\begin{definition}[AND construction]
$h^{AND}(x) = (h_1(x), \dots, h_r(x))$    
\end{definition}

Then, if two elements $x,y\in U$ collide with probability $p$ using any of the $h_i$, now they collide with probability $p^r$ using $h^{AND}$ (because the $h_i$ are independent).
In other words, the curve becomes $P(collision) = p^r$ and we conclude: 

\begin{lemma}
    $h^{AND}$ is a $(d_1,d_2, p_1^r, p_2^r)$-sensitive hash function.
\end{lemma}

Observe that, if the output of $h$ is one integer, then $h^{AND}$ outputs $r$ integers. However, we may use one additional collision-free hash function $h'$ to reduce this size to one integer: $x$ is mapped to $y = h'(h^{AND}(x))$. This is important, since later we will need  to insert $y$ in a hash table (this trick reduces the space by a factor of $r$).

\subsubsection{OR construction}

Suppose $\mathcal H$ is a $(d_1,d_2, p_1, p_2)$-sensitive family. Pick uniformly $b$ independent hash functions $h_1, \dots, h_b \in \mathcal H$, and define:

\begin{definition}[OR construction]
We say that $x$ and $y$ collide iff $h_i(x)=h_i(y)$ for at least one $1\leq i \leq b$.
\end{definition}

Note: the OR construction can be simulated by simply keeping $b$ hash tables $H_1,\dots, H_b$, and inserting $x$ in bucket $H_i[h_i(x)]$ for each $1\leq i \leq b$. Then, two elements collide iff they end up in the same bucket in at least one hash table.

Suppose two elements $x,y\in U$ collide with probability $p$ using any hash function $h_i$. Then:

\begin{itemize}
    \item For a fixed $i$, we have that $P(h_i(x)\neq h_i(y)) = 1-p$
    \item The probability that all hashes do not collide is $P(\wedge_{i=1}^b h_i(x)\neq h_i(y)) = (1-p)^b$
    \item The probability that at least one hash collides is $$P(\vee_{i=1}^b h_i(x) = h_i(y)) = 1- P(\wedge_{i=1}^b h_i(x)\neq h_i(y)) = 1-(1-p)^b$$
\end{itemize}

We conclude that the OR construction yields a curve of the form $P(collision) = 1 - (1-p)^b$ so:

\begin{lemma}
    The OR construction yields a $(d_1,d_2, 1-(1-p_1)^b, 1-(1-p_2)^b)$-sensitive hash function.
\end{lemma}

\subsubsection{Combining AND+OR}

By combining the two constructions, each $x$ is hashed through $rb$ hash functions: we keep $b$ hash tables and insert each $x\in U$ in buckets $H_i[h^{AND}_i(x)]$ for each $1\leq i \leq b$, where $h^{AND}_i$ is the combination of $r$ independent hash values. We obtain: 

\begin{lemma}
    If  $\mathcal H$ is a $(d_1,d_2, p_1, p_2)$-sensitive family, then the AND+OR constructions with parameters $r$ and $b$ yields a $(d_1,d_2, 1-(1-p_1^r)^b, 1-(1-p_2^r)^b)$-sensitive family.
\end{lemma}

It turns out (see next subsections) that by playing with parameters $r$ and $b$ we can obtain a function as close as we wish to the ideal LSH of Figure \ref{fig:LSH1}.

\subsection{LSH for Jaccard distance}

Let $\hat h$ be the MinHash function of Definition \ref{def:MinHash function}. 
In Section \ref{sec:minhash} we have established that $P(\hat h(A) = \hat h(B)) = J(A,B)$, i.e. the probability that two elements collide through $\hat h$ is exactly their Jaccard similarity. Recall that we have defined the \emph{Jaccard distance} (a metric) to be $d_J(A,B) = 1 - J(A,B)$. But then, $P(\hat h(A) = \hat h(B)) = 1 - d_J(A,B)$ and we obtain that $\hat h$ is a $(d_1,d_2, 1-d_1, 1-d_2)$-sensitive hash function for any $0\leq d_1 \leq d_2 \leq 1$, see Figure \ref{fig:LSH3}.

\begin{figure}[h!]
\centering
	\includegraphics[scale=0.45]{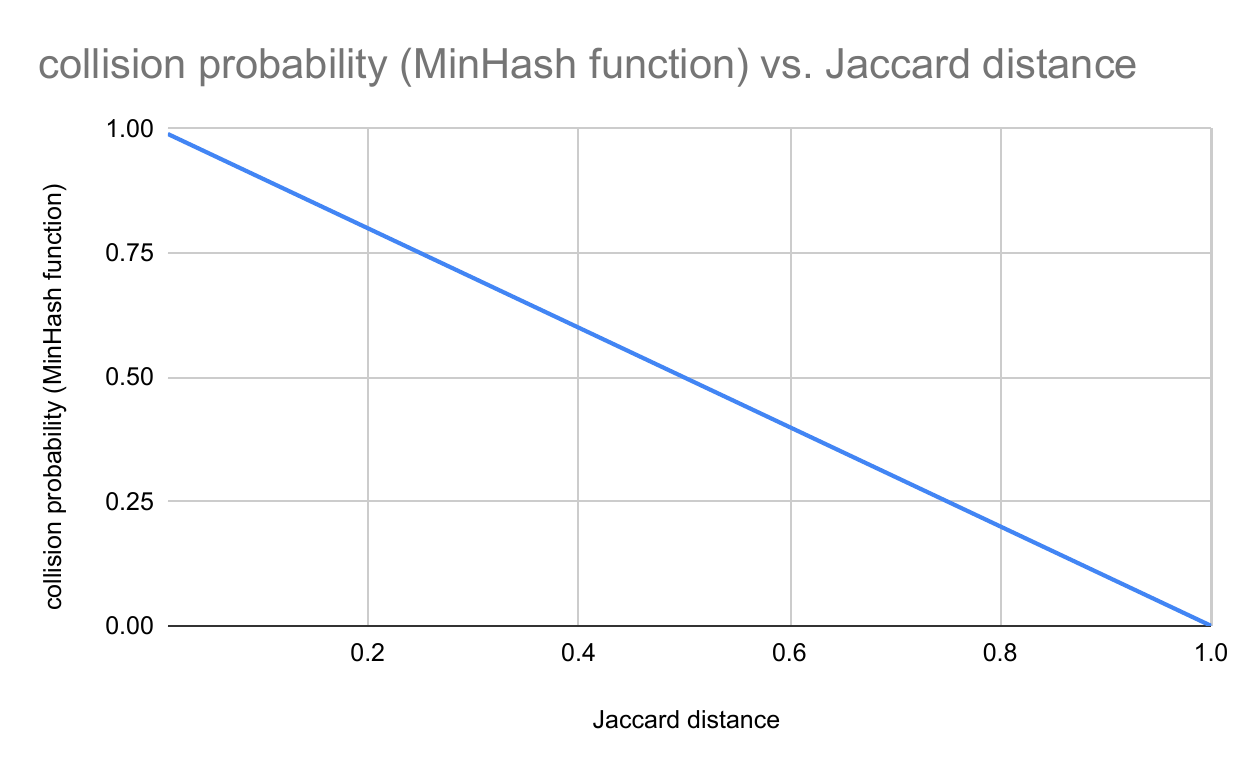}
	\caption{The MinHash function $\hat h$ of Definition \ref{def:MinHash function} is a $(d_1,d_2, 1-d_1, 1-d_2)$-sensitive function for any $0\leq d_1 \leq d_2 \leq 1$.}\label{fig:LSH3}
\end{figure}

Using the AND+OR construction, we can amplify $\hat h$ and obtain a  $(d_1,d_2, 1-(1-(1-d_1)^r)^b, 1-(1-(1-d_2)^r)^b)$-sensitive function for any $0\leq d_1 \leq d_2 \leq 1$. For example, with $r = 10$ and $b = 1200$ we obtain a function whose behaviour is depicted in Figure \ref{fig:LSH2}.

\begin{figure}[h!]
\centering
	\includegraphics[scale=0.45]{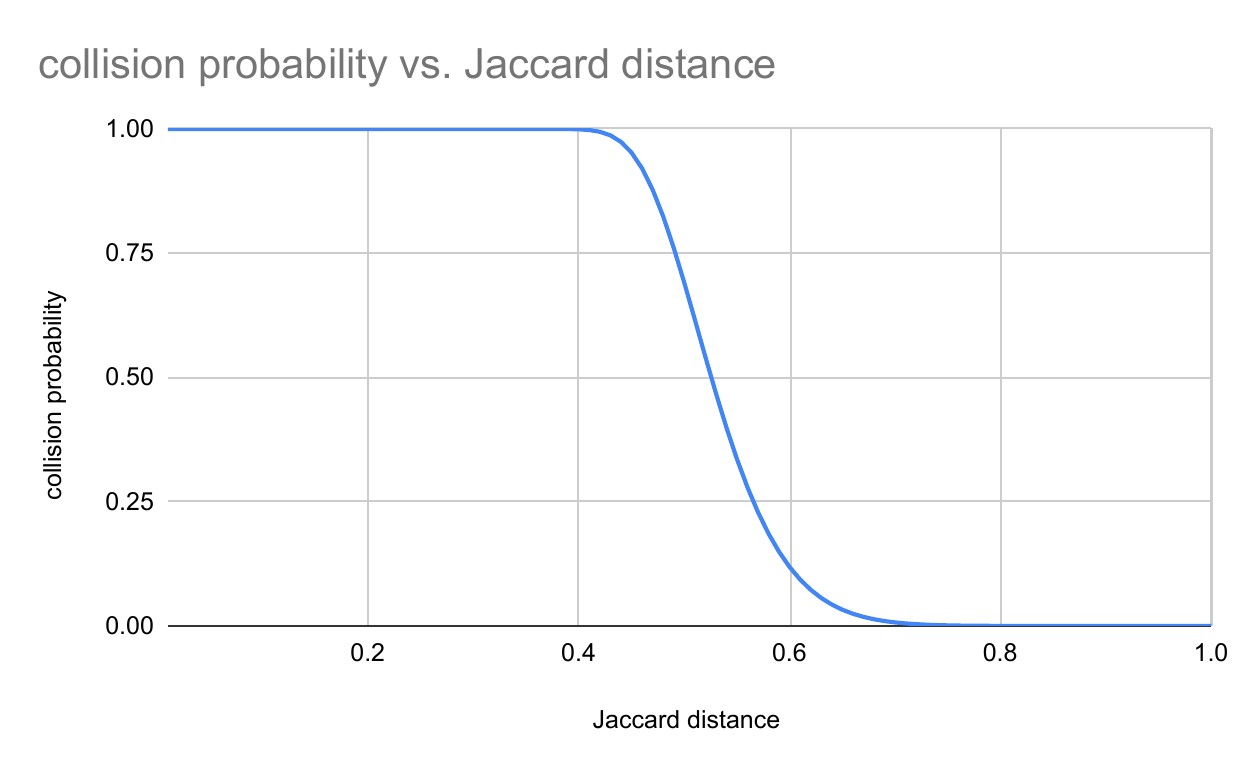}
	\caption{A $(0.4,0.7, 0.999, 0.007)$-sensitive hash function for Jaccard distance built with AND+OR construction with parameters $r = 10$ and $b = 1200$ starting from a  $(0.4,0.7,0.6,0.3)$-sensitive LSH function. 
	Equivalently, we can take two closer points $d_1$ and $d_2$ on the curve: for example, this function is also $(0.5,0.6, 0.69, 0.12)$-sensitive.
	}\label{fig:LSH2}
\end{figure}

The shape of the s-curve is dictated by the parameters $b$ and $r$. As it turns out, $b$ controls the steepness of the slope, that is, the distance between the two points where the probability becomes close to 0 and close to 1. 
The larger $b$, the steeper the s-curve is. 
In other words, $b$ controls the distance between $d_1$ and $d_2$ in our LSH: we want $b$ to be large. Parameter $r$, on the other hand, controls the position of the slope (the point where the curve begins to decrease). 

Let $p$ be the collision probability and $d_J$ be the Jaccard distance. The s-curve follows the equation $p = 1-(1-(1-d_J)^r)^b$
By observing that the center of the slope is approximately around $p=1/2$, one can determine the parameters $b$ and $r$ as a function of the slope position $d_J$. Let's solve the following equation as a function of $r$:

$$
1-(1-(1-d_J)^r)^b = 1/2
$$

We obtain (note that $r$ should be an integer so we must approximate somehow):

$$
r = \left\lfloor\frac{\ln\left( 1-2^{-1/b} \right)}{\ln(1-d_J)}\right\rfloor
$$

The fact that we have to approximate $r$ to an integer means that the slope of the resulting curve will not be centered exactly at $d_J$. By playing with parameter $b$, one can further adjust the curve. 

\begin{example}\label{ex:LSH Jaccard}
    Suppose we want to build a LSH to identify sets with Jaccard distance at most $0.9$. We choose a large $b = 100000$. Then, the above equation gives us $r = \left\lfloor\frac{\ln\left( 1-2^{-1/100000} \right)}{\ln(1-0.9)}\right\rfloor = 5$. Using these parameters, we obtain the LSH shown in Figure \ref{fig:LSH4}. For example, one can extract two data points from this curve and see that this is a $(0.85,0.95,0.99949,0.03076)$-sensitive function. 
\end{example}

\begin{figure}[ht!]
\centering
	\includegraphics[scale=0.44]{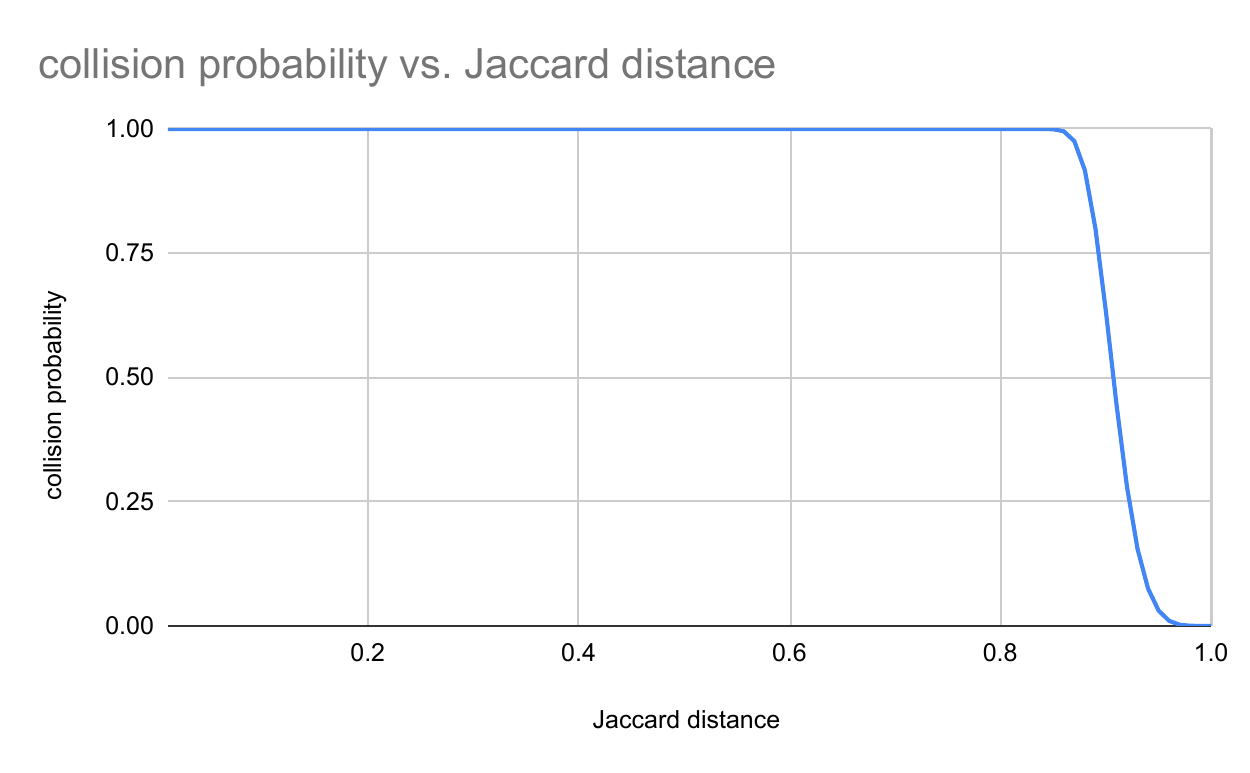}
	\caption{LSH built for Example \ref{ex:LSH Jaccard}.}\label{fig:LSH4}
\end{figure}

Clearly, a large $b$ has a cost: in Example \ref{ex:LSH Jaccard}, we have to compute $r\cdot b = 5\cdot 10^5$ MinHash functions for each set, which means that we have to apply $5\cdot 10^5$ basic hash functions $h$ (see Definition \ref{def:MinHash function}) to each element of each set. 
Letting $t=r\cdot b$, this translates to $O(|A|\cdot t)$ running time for a set $A$. Dahlgaard et al. \cite{dahlgaard2017fast} improved this running time to $O(|A| + t\log t)$. Another solution is to observe that the $t$ MinHashes are completely independent, thus their computation can be parallelized optimally (for example, with a MapReduce job running over a large cluster). 

Observe also that a large value of $b$ requires a large family of hash functions. While this is not a problem with the Jaccard distance (where the supply of $n!$ permutations is essentially unlimited), it could be a problem with the sketch for Hamming distance presented in Section \ref{sec:Hamming sketch}. There, we could choose only among $n$ hash functions, $n$ being the strings' length. It follows that the resulting LSH scheme is not good for small strings (small $n$).

\subsection{Nearest neighbour search}

One application of LSH is nearest neighbour search: 

\begin{definition}[Nearest neighbour search (NNS)]
	For a given distance threshold $D$, preprocess a data set $A$ of size $|A|=m$ in a data structure such that later, given any data point $x$, we can quickly find a point $y\in A$ such that $d(x,y) \leq D$.
\end{definition}

To solve the NNS problem, let $\mathcal H$ be a $(D',D, p_1, p_2)$-sensitive family, with $D'$ as close as possible to (and smaller than) $D$. Suppose moreover that $h(x)$ can be evaluated in time $t_h$ (this time is proportional to the size/cardinality of $x$) and $d(x,y)$ can be computed in time $t_d$. Note that $t_d$ can be reduced considerably by employing sketches --- see Section \ref{sec:minhash}. We amplify $\mathcal H$ with an AND+OR construction with parameters $r$ (AND) and $b$ (OR). Our data structure is formed by $b$ hash tables $H_1, \dots, H_b$. For each of the $m$ data points $x\in A$, we compute the $b$ functions $h^{AND}_i(x)$ in total time $O(m\cdot b \cdot r \cdot t_h)$ and insert in $H_i[h^{AND}_i(x)]$ a pointer to the original data point $x$ (or to its sketch). Assuming that a hash table storing $m$ pointers occupies $O(m)$ words of space and can be constructed in (expected) $O(m)$ time, we obtain:

\begin{lemma}
	Our NNS data structure can be constructed in $O(m\cdot b\cdot r \cdot t_h)$ time and occupies $O(m\cdot b)$ space (in addition to the original data points --- or their sketches).
\end{lemma}

To answer a query $x$, note that we are interested in finding just \emph{one} point $y$ such that $d(x,y) \leq D$: we can stop our search as soon as we find one. In $O(t_h\cdot b \cdot r)$ time we compute the hashes $h^{AND}_i(x)$ for all $1\leq i \leq b$. In the worst case, all the $m$ data points $y$ are such that $d(x,y) > D$. The probability that one such point ends up in bucket $H_i[h^{AND}_i(x)]$ is at most $p_2^r$. As a result, the expected number of false positives in each bucket $H_i[h^{AND}_i(x)]$ is at most $m\cdot p_2^r$; in total, this yields $m\cdot b \cdot p_2^r$ false positives that need to be checked against $x$. For each of these false positives, we need to compute a distance in time $t_d$. We obtain: 

\begin{lemma}\label{lem:NNS}
	Let: 
	\begin{itemize}
		\item $FP = m\cdot b \cdot p_2^r$ be the expected number of false positives in the worst case.
		\item $T = b\cdot r$ be the total number of independent hash functions used by our structure.
	\end{itemize}
	Our NNS data structure answers a query in expected time $O(t_h\cdot T + FP \cdot t_d)$. If there exists a point within distance at most $D'$ from our query, then we return an answer with probability at least $1-(1-p_1^r)^b$.
\end{lemma}

\begin{example}
	Consider the $(0.4,0.7, 0.999, 0.007)$-sensitive family of Figure \ref{fig:LSH2}. This function has been built with AND+OR construction with parameters $r = 10$ and $b = 1200$ taking as starting point the $(0.4,0.7,0.6,0.3)$-sensitive hash function of Figure \ref{fig:LSH2} (in fact, $1-(1-0.6^r)^b \approx 0.999$ and $1-(1-0.3^r)^b \approx 0.007$). 
	We can therefore use this hash to solve the NNS problem with threshold $D = 0.7$.
	Lemma \ref{lem:NNS} states that \emph{at most} $FP = m \cdot b \cdot p_2^r \approx 0.007\cdot m$ false positives need to be explicitly checked against our query (compare this with a naive strategy that compares $100\%$ of the $m$ points with the query). Moreover, if at least one point within distance $D' = 0.4$  from our query exists, we will return a point within distance $0.4$ with probability at least $1-(1-0.6^r)^b \approx 0.999$. The data structure uses space proportional to $b = 1200$ words (a few kilobytes) for each data point; note that, in big data scenarios, each data point (for example, a document) is likely to use much more space than that so this extra space is negligible. 
\end{example}

\section{Introduction to streaming algorithms}

We will start now using  sketches for solving problems on data streams. 
A data stream is a sequence $x = x_1, x_2, \dots, x_m$ of elements (without loss of generality, integers from the range $[1,n]$). We receive these elements one at a time, from $x_1$ to $x_m$. Typically, $m$ is too large and we cannot keep all the stream in memory. The goal of streaming algorithms is to compute useful statistics on the stream while using as little memory as possible (usually, poly-logarithmic in $m$ and $n$). Example of quantities that we may be interested in computing: average, sum, most frequent element, estimated frequency of every element, frequency moments (e.g. sum of the squares of frequencies). Importantly, the algorithm must be able to output the desired statistic at any point in time: immediately after receiving the $i$-th stream element $x_i$, the algorithm must be able to output the statistic for $x_1, x_2, \dots, x_i$.

A streaming algorithm is evaluated on these parameters:

\begin{enumerate}
	\item \textbf{Working space}  as a function of, e.g., stream length $m$ and universe size $n$.
	\item \textbf{Delay} per element: the worst-case time taken by the algorithm to process each stream element. 
	\item \textbf{Probability} of obtaining a correct solution or a good approximation of the correct result. 
	\item \textbf{Approximation ratio} (e.g. the value returned by the algorithm is a $(1\pm \epsilon)$ approximation of the correct answer, for a small $\epsilon \geq 0$).
\end{enumerate}

A nice introduction to data sketching and streaming is given in \cite{cormode2017data}.
Trivial problems that can be solved efficiently on  data streams are:

\begin{itemize}
    \item Compute the sum/average of all elements. 
    \item Find the minimum/maximum element in the stream. 
\end{itemize}

It is trivial to solve the above problems using  $O(\log n + \log m)$ bits of working space and $O(1)$ delay. The answer is always exact and the success probability is 1.

Our goal in the next sections will be to find  efficient algorithms for less trivial problems: pattern matching, frequency estimation, counting events, and computing frequency moments. 

\section{Pattern matching on streams}

The first example of stream statistic we consider is \emph{pattern matching}. Say the elements $x_i$ belong to some alphabet $\Sigma$: the stream is a string of length $m$ over $\Sigma$. Suppose we are given a pattern $y = y_1y_2\dots y_n \in \Sigma^n$. The pattern's length $n$ is smaller than $m$, but also $n$ could be very large (so that $y$ too does not fit in memory or cache). The question we tackle in this section is: how many times does $y$ appear in $x$ as a substring $y = x_{i}x_{i+1}\dots x_{i+n-1}$?

\begin{example}[Intrusion Detection and Prevention Systems (IDPSs)] 
IDPSs are software tools that scan network traffic in search of known patterns such as virus fragments or malicious code. The searched patterns are usually very numerous, so the memory usage and delay of the used pattern matching algorithm is critical. Ideally, the algorithm should work entirely in cache in order to achieve the best performance. See also the paper \cite{gupta2014pattern}.
\end{example}

\subsection{Karp-Rabin's algorithm}

Rabin's hashing is the main tool we will use to solve the problem. First, we note that the technique itself yields a straightforward solution, even though in $O(n)$ space. In the next section we refine this solution to use $O(\log n)$ space. 

Suppose we have processed the stream up to $x_1,\dots, x_i$ ($i\geq n$) and that we know the hash values $\kappa_{q,z}(x_{i-n+1}x_{i-n+2}\dots x_i)$ and $\kappa_{q,z}(y)$. By simply comparing these two hash values (in constant time) we can discover whether or not the patter occurs in the last $n$ stream's characters. The crucial step is to update the hash of the stream when a new element $x_{i+1}$ arrives. This is not too hard: we have to subtract character $x_{i-n+1}$ from the stream's hash and add the new character $x_{i+1}$. This can be achieved as follows:
$$
\kappa_{q,z}(x_{i-n+2}x_{i-n+2}\dots x_{i+1}) = (\kappa_{q,z}(x_{i-n+1}x_{i-n+2}\dots x_i) - x_{i-n+1}\cdot z^{n-1})\cdot z + x_{i+1} \mod q
$$
The value $z^{n-1} \mod q$ can be pre-computed, so the above operation takes constant time. Note that, since we need to access character $x_{i-n+1}$, at any time the algorithm must keep the last $n$ characters seen in the stream, thereby using $O(n)$ space. 

\subsubsection{Analysis} From Section \ref{sec:KR}, the probability that the pattern collides with a stream's substring is at most $n/q$. By union bound, the probability that the pattern collides with at least one substring is $mn/q \leq m^2/q$. We want this to happen with small (inverse polynomial probability): this can be achieved by choosing a prime $q$ in the range $[m^{c+2}, 2\cdot m^{c+2}]$, for any constant $c$. 
Such a prime (and therefore the output of Rabin's hash function) can be stored in $O(\log m)$ bits = $O(1)$ words.
We obtain: 

\begin{theorem}
	The Karp-Rabin algorithm solves the pattern matching problem in the streaming model using $O(n)$ words of memory and $O(1)$ delay. The correct solution is returned with high (inverse-polynomial) probability $1-m^{-c}$, for any constant $c\geq 1$ chosen at initialization time. 
\end{theorem}

There exist also deterministic algorithms with $O(1)$ delay and $O(n)$ space. However, as we show in the next section, Karp-Rabin's  randomization enables an exponentially more space-efficient solution. 

\subsection{Porat\&Porat's algorithm}

The big disadvantage of Karp-Rabin's algorithm is that it uses too much memory: $O(n)$ words per pattern. In this section we study an algorithm described by Benny Porat and Ely Porat in \cite{porat2009exact} that uses just $O(\log n)$ words of space and has $O(\log n)$ delay per stream's character \footnote{Note that, no matter how large $n$ is, $O(\log n)$ words will fit in cache. $O(\log n)$ delay in cache is by far more desirable than $O(1)$ delay in RAM: the former is hundreds of times faster than the latter.}. Other algorithms are able to reduce the delay to the optimal $O(1)$ (see \cite{breslauer2014real}).                          
For simplicity, assume that $n$ is a power of two: $n = 2^e$ for some $e\geq 0$. The algorithm can be generalized to any $n$ in a straightforward way. The overall idea is to: 

\begin{itemize}
    \item Keep the hashes of all $1+e = 1+\log_2 n$ prefixes of $y$ whose length is a power of two.
    \item Keep the occurrences of those prefixes of $y$ on the stream, working in $e$ levels: level $0 \leq i < e$ stores the set $W_i$ of all occurrences of the prefix $y[1,2^i]$ in the window containing the last $2^{i+1}$ stream's characters. Using a clever argument based on string periodicity, show that this set can be ``compressed'' in just $O(1)$ words of space per level ($O(\log n)$ space in total).
    \item When a new stream character $x_j$ arrives: 
    \begin{itemize}
        \item If it is an occurrence of $y_1$, insert $j$ in $W_0$.
        \item For every level $i$, position $p = j-2^{i+1}$ (the leftmost position in the window of level $i$) exits from the window: remove it, if it belongs to $W_i$.
        \item If position $p' = j-2^{i+1}+1 \in W_i$ check if $p'$ is an occurrence of $y[1,2^{i+1}]$ (do this check using fingerprints). If this is the case, then insert $p'$ in  $W_{i+1}$ (i.e. promote it to the next level). If $i+1 = e$ then we have found an occurrence of $y$.
    \end{itemize}
\end{itemize}

Crucially, in the last step we need to compute the fingerprint of a whole window. We will show that this can be done using just logarithmic space. 

Figure \ref{fig:PP} depicts two steps of the algorithm: before and after the arrival of a new stream character. Algorithm \ref{alg:PoratPorat} implements one step of the above procedure (hiding details such as compression of the occurrences and update of the hashes, which are discussed below). The window at level $i$ is indicated as $W_i$ and it is a set of positions (integers).

\begin{algorithm}[th!]
	\setstretch{1.35}
	\caption{new\_stream\_character($x_j$)}\label{alg:PoratPorat}
	
	\SetSideCommentLeft
	\LinesNumbered
	
	\BlankLine

 	\If{$x_j == y_1$}{

	   	$W_0 \leftarrow W_0 \cup \{j\}$\;	

    	}
 
	\ForEach{\emph{level} $i = 0,\dots, e-1$}{

      $W_i \leftarrow W_i - \{j - 2^{i+1}\}$\tcp*[r]{delete occurrences exiting the window}
	
     \If{$j - 2^{i+1} + 1 \in  W_i$}{

			\If{$\kappa_{q,z}(y[1,2^{i+1}]) == \kappa_{q,z}(x[j - 2^{i+1}+1,j])$}{

				\eIf{$i == e-1$}{
	
					\textbf{report} $j-2^e+1$\tcp*[r]{pattern occurrence found!}
					
				}{
	
					$W_{i+1} \leftarrow W_{i+1} \cup \{j - 2^{i+1}+1\}$		\tcp*[r]{promote position to next level}	
			
				}
		
			}

		}
		
	}

\end{algorithm}

\begin{figure}[h!]
	\includegraphics[scale=0.65]{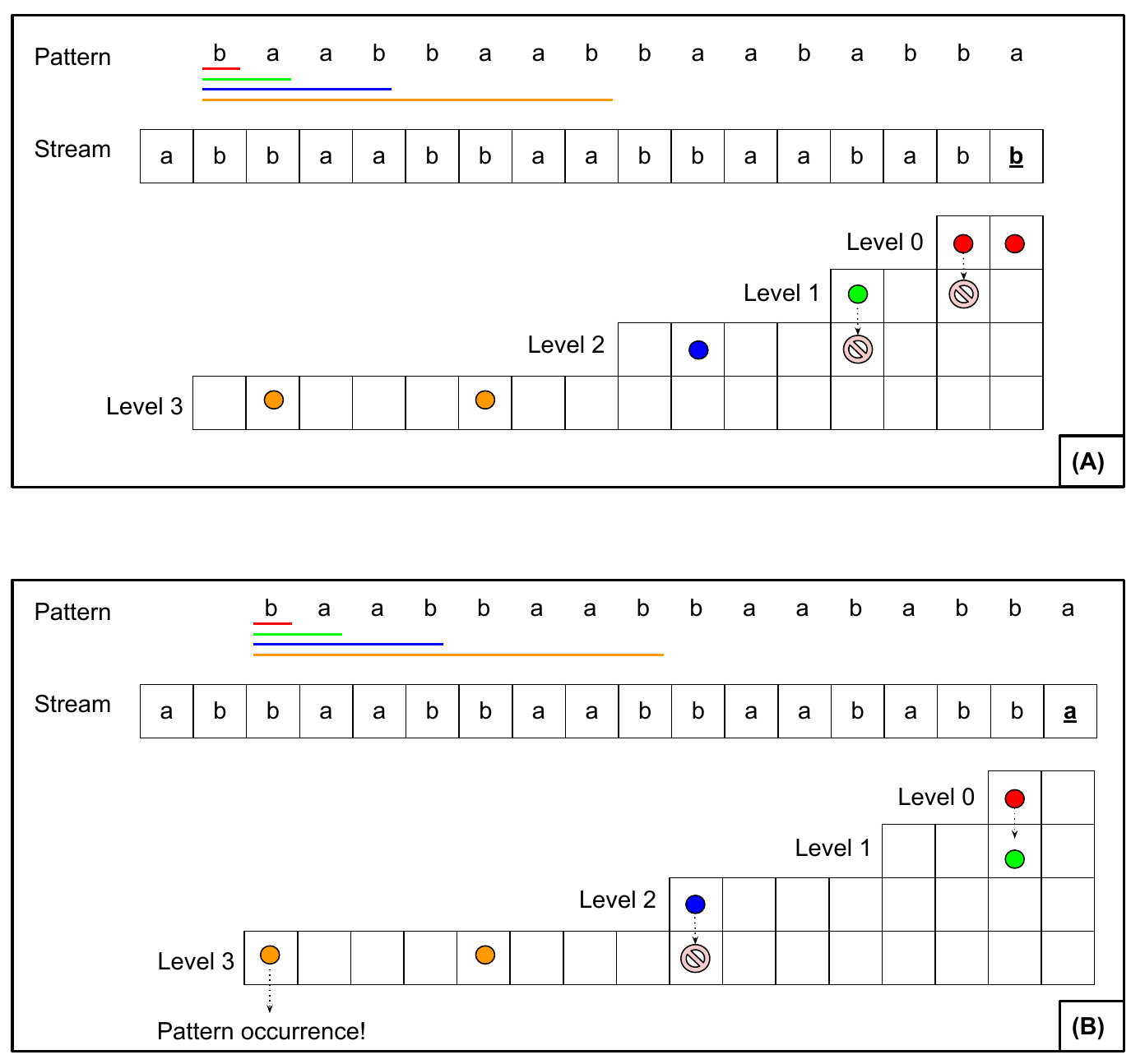}
	\caption{Each colored dot at level $i$ represents an occurrence of a prefix of length $2^i$ of the pattern (underlined with corresponding color). \textbf{(A)} 
    Suppose that the last stream's character ('b', underlined) has just arrived.
    Since the new character 'b' matches the pattern's prefix of length $2^0 = 1$, at level 0 we store the corresponding occurrence.
    In each level, some of the stored occurrences are candidates that could be promoted to the next level: these are the leftmost occurrences at levels $i=0$ and $i=1$. None of these occurrences can be promoted since they are not occurrences of a prefix of length $2^{i+1}$ of the pattern: the leftmost occurrence at level 0 is not an occurrence of ``ba'' and the occurrence at level 1 is not an occurrence of ``baab''.  
    \textbf{(B)} A new stream character ('a', underlined) has arrived. 
    Some occurrences of figure (A) exit their window and get deleted: the leftmost red (level 0) and green (level 1) occurrences of figure (A). 
    Now, in figure (B) three occurrences are candidates that could be promoted to the next level. The occurrence at level 0 (red) is indeed an occurrence of ``ba'', therefore it is promoted to level 1. The occurrence at level 2 (blue) is not an occurrence of ``baabbaab'', therefore it is not promoted. Finally, the leftmost occurrence at level 3 (orange) is promoted. Since this is the last level, we found a pattern occurrence.
	}\label{fig:PP}
\end{figure}

\subsubsection{Compressing the occurrences}

We have $\log  n$ levels, however this is not sufficient to claim that the algorithm uses $O(\log n)$ space: in each level $i$, there could be up to $2^i$ occurrences of the pattern's prefix $y[1,2^i]$. In this paragraph we show that all the occurrences in a window can be compressed in just $O(1)$ words of space. 

The key observation is that, in each level, we store occurrences of the pattern's prefix of length $K = 2^i$ in a window of size $2K = 2^{i+1}$. Now, if there are at least three such occurrences, then at least two of them must overlap. But these are occurrences of the same string $y[1,2^i]$, so if they overlap then the string must be periodic. Finally, if the string is periodic then all its occurrences in the window must be equally-spaced: we have an occurrence every $p$ positions, for some integer $p$ (a period of the string). Then, all $t$ occurrences $W_i = \{r_1, r_2, \dots, r_t\}$ in the window can be encoded in just $O(1)$ space by storing the first occurrence $r_1$, the number $t$ of occurrences, and the period $p$. 
This representation is also easy to update (in constant time) upon insertion of new occurrences to the right (which must follow the same rule) and removal of an occurrence to the left. 
We now formalize this reasoning.

\begin{definition}[Period of a string]
	Let $S$ be a string of length $K$. We say that $S$ has period $p$ if and only if $S[i] = S[i+p]$ for all $1 \leq i \leq K-p$.
\end{definition}

\begin{example}
	The string $S = abcabcabcabca$, of length $K=13$, has periods 3, 6, 9, 12.
\end{example}

\begin{theorem}[Wilf's theorem]
	Any string having periods $p$, $q$ and length at least $p+q-gcd(p,q)$ also has $gcd(p,q)$ as a period.
\end{theorem}

\begin{example}
	Consider the string above: $S = abcabcabcabca$. The string has periods $6$, $9$ (with $gcd(6,9)=3$) and has length $13 > 6+9-3 = 12$. Wilf's theorem can be used to deduce that the string must also have period $gcd(6,9)=3$.
\end{example}

Wilf's theorem can be used to prove the following: 

\begin{lemma}\label{lem:periodicity}
	Let $P$ be a string of length $K$, and $S$ be a string of length $2K$. If $P$ occurs in $S$ at positions $r_1 < r_2 < \dots < r_{t}$, with $t\geq 3$, then $r_{j+1} = r_j + p$, where $p = r_2-r_1$.
\end{lemma}

The lemma provides a compressed representation for all the occurrences $W_i = \{r_1, r_2, \dots, r_t\}$ in the window at each level $i$: just record $(r_1, t, p)$. This representation is easy to update in constant time when $r_1$ is removed from $W_i$ and when a new occurrence $r_{t+1}$ (greater than $r_t$) is inserted into $W_i$ (these are the only two updates we perform on the sets $W_i$).

\subsubsection{Updating the fingerprints}

The last thing to show is how to efficiently compute $w_i = \kappa_{q,z}(x[j - 2^{i+1},j-1])$ at level $i$ (needed at Line 1 of the algorithm), that is, the fingerprint of the whole window when the first occurrence stands at the beginning of the window: $r_1 = j - 2^{i+1}$. Consider the window $W_i$ at level $i$, and the two smallest positions $r_1,r_2 \in W_i$. 
Let $x = x[1,j-1]$ be the current stream.
We keep in memory three fingerprints (see Figure \ref{fig:PP2}):

\begin{enumerate}
    \item[(A)] $\kappa_{q,z}(x)$: the fingerprint of the whole stream.
    \item[(B)] $\kappa_{q,z}(x[r_1,r_2-1])$: the fingerprint of the stream's substring standing between $r_1$ (included) and $r_2$ (excluded), whenever $W_i$ contains at least two positions. 
    \item[(C)] $\kappa_{q,z}(x[1,r_1-1])$: the fingerprint of the stream's prefix ending at $r_1-1$, whenever $W_i$ contains at least one position. 
\end{enumerate}

\begin{figure}[h!]
\centering
	\includegraphics[scale=0.7]{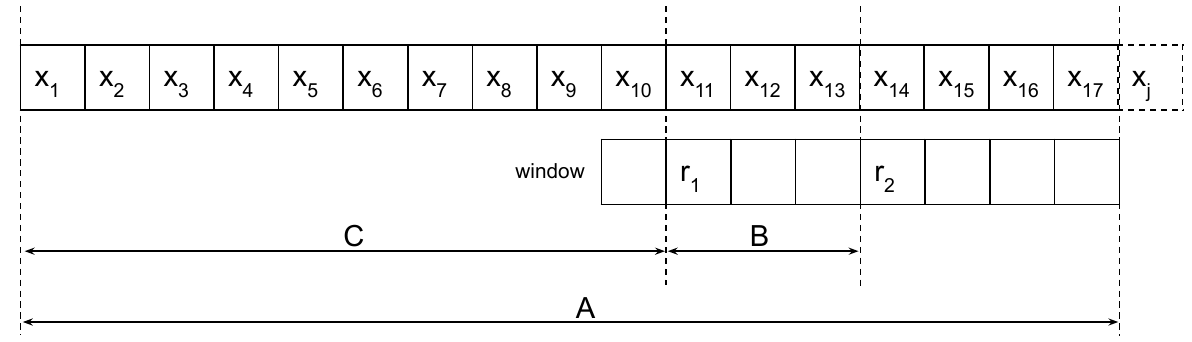}
	\caption{For each level (window), we keep three fingerprints: A (full stream, unique for all windows), B (string between first two pattern's occurrences), and C (from beginning of the stream to the first pattern's occurrence).}\label{fig:PP2}
\end{figure}

Knowing A,B, and C we can easily compute the fingerprint $w_i$ of the whole window when $r_1 = j-2^{i+1}$:

$$
w_i =  \left(A - C \cdot z^{2^{i+1}}\right) \mod q
$$

Note that  $z^{2^{i+1}} \mod q$ can easily be pre-computed for any $i\leq \log n$ at the beginning of the algorithm using the recurrence $z^{2^{i+1}} = (z^{2^{i}})^2$. We now show how to update the three fingerprints A, B, C.

\paragraph{Updating A}

Fingerprint A - the full stream - can be updated very easily in constant time each time a new stream character arrives (see Section \ref{sec:KR}). 

\paragraph{Updating B - case 1}

B needs to be updated in two cases. The first case happens when $r_2$ enters in the window (before that, only $r_1$ was in the window): see Figure \ref{fig:PP-B1}. Then, notice that $x[r_2,j-1] = y[1,2^i]$,  so we have the fingerprint $D = \kappa_{q,z}(y[1,2^i]) = \kappa_{q,z}(x[r_2,j-1])$.

\begin{figure}[h!]
\centering
	\includegraphics[scale=0.7]{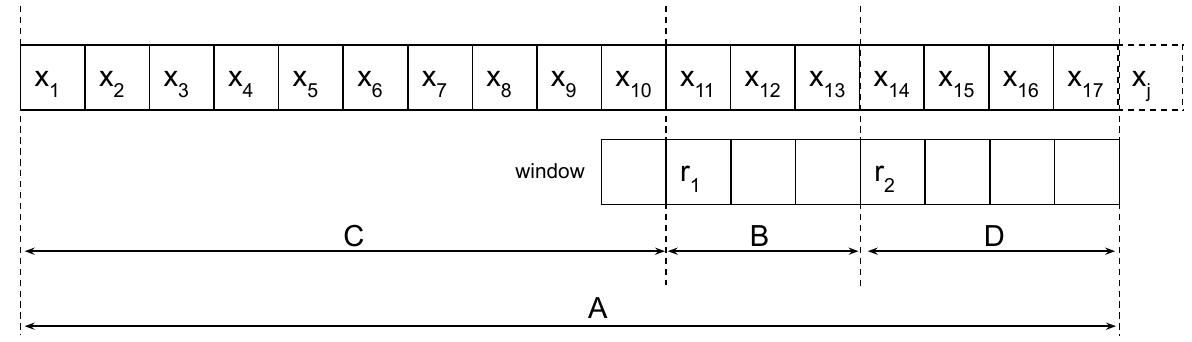}
	\caption{Updating B - case 1: $r_2$ enters in the window.}\label{fig:PP-B1}
\end{figure}

It follows that B can be computed as:

$$
B =  \left(\left(A - D - C\cdot z^{|D|+|B|}\right) \cdot z^{-|D|}\right) \mod q
$$
In the above equation, note that $z^{|D|+|B|} \equiv_q z^{2^i + (r_2-r_1)} \equiv_q z^{2^i} \cdot z^{(r_2-r_1)}$. The constants $z^{2^i}\mod q$ (just $\log n$ constants) can be pre-computed at the beginning, when processing the pattern. Then, $z^{2^i} \cdot z^{(r_2-r_1)}\mod q$ can be obtained by multiplying $z^{2^i}\mod q$ by $z$ a number of times equal to $r_2-r_1$ (i.e. $z^{2^i} \cdot z \cdot z \cdot \dots \cdot z$) while the stream characters from $r_1$ to $r_2$ are arriving (constant time per stream character). 
Also the $\log n$ values $z^{-|D|} = z^{-2^i} \mod q$ can be pre-computed before the stream arrives in $O(\log m)$ time as follows. $z^{2^{i+1}} \equiv_q (z^{2^i})^2$, and $z^{-2^i} \mod q$ can be computed in $O(\log q) = O(\log m)$ time using the equality $a^{-1} \equiv_q a^{q-2}$ and fast exponentiation: $z^{-2^i} \equiv_q z^{2^i\cdot (q-2)}$.

\paragraph{Updating B - case 2}

The second case where we need to update B is when $r_1$ exits the window and $r_3$ is in the window: $B$ should become the fingerprint of the string between $r_2$ and $r_3$. See Figure  \ref{fig:PP-B2}. 

\begin{figure}[h!]
\centering
	\includegraphics[scale=0.7]{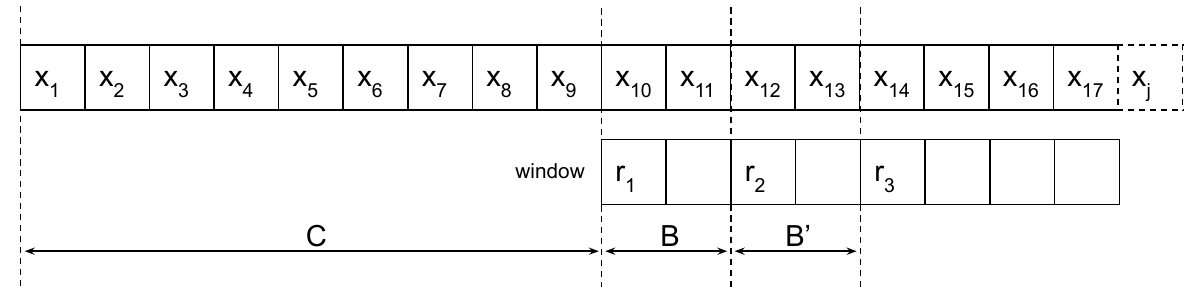}
	\caption{Updating B - case 2: $r_1$ exits the window and $r_3$ is in the window.}\label{fig:PP-B2}
\end{figure}

It turns out that in this case nothing needs to be done: The new fingerprint is $B' = B$. To see this, note that (1) $r_3-r_2 = r_2 - r_1$ by Lemma \ref{lem:periodicity}, and (2) $r_1$ and $r_2$ are both occurrences of the same string of length $2^i$. Since $r_2-r_1 \leq 2^i$, then $x[r_1,r_2-1] = x[r_2,r_3-1]$.

\paragraph{Updating C - case 1}

C needs to be updated in two cases. The first case happens when $r_1$ enters in the window (before that, the window was empty: $W_i = \emptyset$). See Figure \ref{fig:PP-C1}. 
As in case B1, notice that we have the fingerprint $D = \kappa_{q,z}(y[1,2^i]) = \kappa_{q,z}(x[r_1,j-1])$.

\begin{figure}[h!]
\centering
	\includegraphics[scale=0.7]{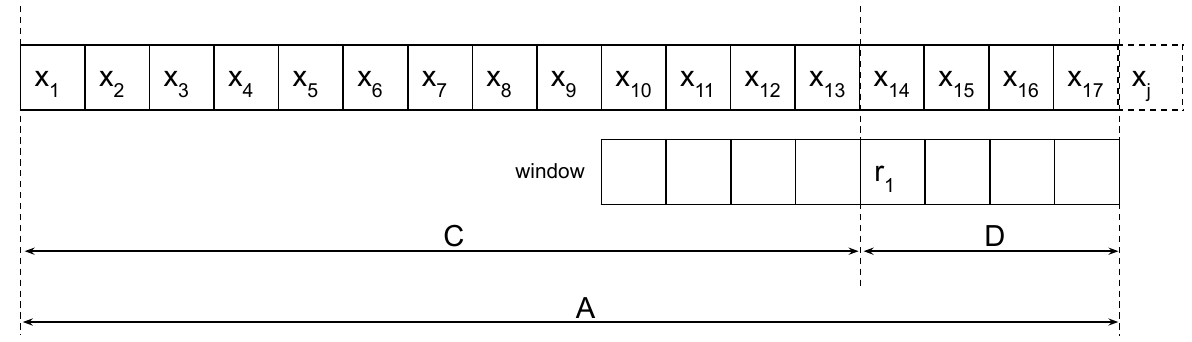}
	\caption{Updating C - case 1: $r_1$ enters in the window.}\label{fig:PP-C1}
\end{figure}

Then:

$$
C =  \left((A - D) \cdot z^{-|D|}\right) \mod q
$$ 

Where $ z^{-|D|} \equiv_q z^{-2^i}$ can be pre-computed as described above. 

\paragraph{Updating C - case 2}

The last case to consider is when $r_1$ exits the window and $r_2$ is in the window. See Figure \ref{fig:PP-C2}.

\begin{figure}[h!]
\centering
	\includegraphics[scale=0.7]{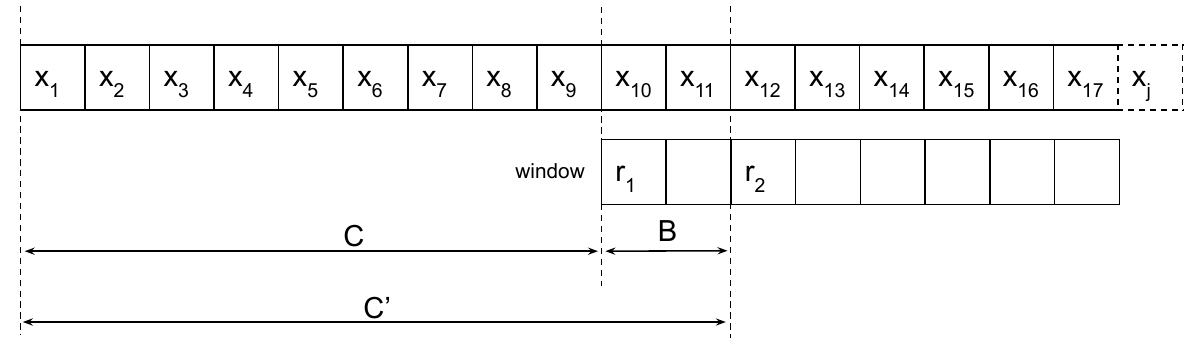}
	\caption{Updating C - case 2: $r_1$ exits the window and $r_2$ is in the window.}\label{fig:PP-C2}
\end{figure}

This is achieved as follows:

$$
C' = \left( C \cdot z^{|B|} + B \right) \mod q 
$$

where $z^{|B|} \equiv_q z^{r_2-r_1}$ is computed as described above while stream characters number $r_1$ and $r_2$ are arriving. 

\subsubsection{Final result}

Observe that each fingerprint update can be performed in constant time (per level, thus $O(\log n)$ time per stream's character). 
We obtain:

\begin{theorem}
	Let $m$ be the stream's length and $n\leq m$ be the pattern's length.
	Porat\&Porat's algorithm solves the pattern matching problem in the streaming model using $O(\log n)$ words of memory and $O(\log n)$ delay. The correct solution is returned with high (inverse-polynomial) probability $1-m^{-c}$, for any constant $c\geq 1$ chosen at initialization time. 
\end{theorem}

Breslauer and Galil in \cite{breslauer2014real} reduced the delay to $O(1)$ while still using $O(\log n)$ words of space.

\subsection{Extension to approximate pattern matching}\label{sec:P&P - Hamming}

We describe a modification of Porat\&Porat's algorithm that allows finding all stream occurrences $x_{i,n} = x_i\dots x_{i+n-1}$ of a pattern $y = y_1\dots y_n$ such that $D_H(x_{i,n},y) \leq k$ for any parameter $k$, where $D_H$ is the Hamming distance between strings:

$$
D_H(A,B) = \sum_{i=1}^n (A[i]\neq B[i])
$$ 

In the next paragraphs we introduce a sketch for Hamming distance between strings of the same length $n$. 
Then, we will apply this sketch to the approximate pattern matching problem on streams.
Differently from the sketch of Section \ref{sec:Hamming sketch} (which can estimate $d_H(A,B) = D_H(A,B)/n$ with bounded error and bounded failure probability), the technique that we will describe here will allow us to compute (with high probability of success) the \emph{exact} Hamming distance $D_H(A,B)$ between two strings, only if this distance is below some threshold (maximum distance) $k$ (otherwise, we simply report that $D_H(A,B) \ge k$). 
The sketch' space usage will depend on $k$ and $n$. 

For simplicity, we first describe the sketch and pattern matching algorithm for $k=1$, i.e. zero or one mismatch between the pattern and the stream. Then, we will extend the idea to arbitrary $k$.

\paragraph{A sketch for Hamming distance between strings of the same length}

Let $x,y$ be two strings of length $n$.
For some parameters $1 \le i,d \le n$, let $y_{i:d}$ be defined as 
$$
y_{i:d} = y_{i}\ y_{i+d}\ y_{i+2d} \dots
$$
In other words, $y_{i:d}$ is the sub-string of $y$ built by extracting every other $d$-th character from $y$, starting from character $y_i$. We call $y_{i:d}$ a \emph{shift} of $y$.

Clearly, $x=y$ if and only if $x_{i:d} = y_{i:d}$ for all $i=1,\dots, d$. Assume now that $D_H(x,y)=1$. Then, note that the error is captured by exactly one of the $d$ shifts: there exists one $i' \in [1,d]$ such that $x_{i':d} \neq y_{i':d}$, and $x_{i:d} = y_{i:d}$ for all $i\neq i'$.

\begin{example}
    Let $x = $\emph{ab\underline racadabra} and $y = $\emph{ab\underline bacadabra}, with $D_H(x,y)=1$ (the mismatch is underlined).
    Pick $d=2$ and consider the two shifts (per string) $x_{1:2} = $ \emph{arcdba}, $x_{2:2} = $\emph{baaar}, $y_{1:2} = $ \emph{abcdba}, $y_{2:2} = $\emph{baaar}. Then:
    \begin{itemize}
        \item $x_{1:2} \neq y_{1:2}$
        \item $x_{2:2} = y_{2:2}$
    \end{itemize} 
\end{example}

What if $D_H(x,y) = k > 1$? Then, the number of shifts $i'$ such that $x_{i':d} \neq y_{i':d}$ could be \emph{smaller} than $k$ (but never larger). Notice that this happens precisely when the distance $|j'-j|$ between \emph{two} mismatches $x_{j} \neq y_{j}$ and $x_{j'} \neq y_{j'}$ is a multiple of $d$.

\begin{example}
    Let $x = $\emph{ab\underline raca\underline dabra} and $y = $\emph{ab\underline baca\underline aabra}, with $D_H(x,y)=2$ (the two mismatches are underlined).
    Pick $d=2$ and consider the two shifts (per string) $x_{1:2} = $ \emph{arcdba}, $x_{2:2} = $\emph{baaar}, $y_{1:2} = $ \emph{abcaba}, $y_{2:2} = $\emph{baaar}. Then:
    \begin{itemize}
        \item $x_{1:2} \neq y_{1:2}$
        \item $x_{2:2} = y_{2:2}$
    \end{itemize} 
    In particular, the Hamming distance is 2 but only one of the two shifts generates a mismatch. This happens because the two mismatches are distanced 4 positions, which is a multiple of $d=2$.
\end{example}

It is easy to see that the above issue does not happen if $d$ does not divide the distance between the two mismatches. 

\begin{example}
    Let $x = $\emph{ab\underline raca\underline dabra} and $y = $\emph{ab\underline baca\underline aabra}, with $D_H(x,y)=2$ (the two mismatches are underlined).
    Pick $d=3$ and consider the three shifts (per string) $x_{1:3} = $\emph{aadr}, $x_{2:3} = $\emph{bcaa}, $x_{3:3} = $\emph{rab}, $y_{1:3} = $\emph{aaar}, $y_{2:3} = $\emph{bcaa}, $y_{3:3} = $\emph{bab}. Then:
    \begin{itemize}
        \item $x_{1:3} \neq y_{1:3}$
        \item $x_{2:3} = y_{2:3}$
        \item $x_{3:3} \neq y_{3:3}$
    \end{itemize} 
    Now, two shifts generates a mismatch.
\end{example}

This property can be summarized in a corollary:

\begin{corollary}
    Let $x,y$ be two strings of length $n$, and consider their $d$ shifts $x_{i:d}, y_{i:d}$ for $i=1, \dots, d$. Then
    $x_{i:d} \neq y_{i:d}$ for at least two values $i \in [1,d]$ if and only if (i) $D_H(x,y) > 1$ and (ii) there exist two mismatches whose distance $|j-j'|$ is not a multiple of $d$.
\end{corollary}

Consider the distance $|j-j'|$ between (the positions of) any two mismatches between $x$ and $y$. Consider moreover the smallest $\lceil \log_2 n\rceil$ prime numbers $P = \{p_1, p_2, \dots, p_{\lceil \log_2 n\rceil}\}$. Clearly, $|j-j'|$ cannot be a multiple of all numbers in $P$: this would imply that $|j-j'| \geq \prod_{p\in P} p > n$. This immediately implies the following corollary, which directly gives us a small sketch for deciding $D_H(x,y) \le 1$:

\begin{corollary}
    Let $x,y$ be two strings of length $n$. Let $P = \{p_1, p_2, \dots, p_{\lceil \log_2 n\rceil}\}$ be the set containing the $\lceil \log_2 n\rceil$ smallest prime numbers. Then, $D_H(x,y) > 1$ if and only if there exists $d\in P$ such that 
    $x_{i:d} \neq y_{i:d}$ for at least two values $i \in [1,d]$.
\end{corollary}

The corresponding sketch is:

$$
sketch(x) = \langle  \kappa(x_{i:d}) \rangle_{d\in P, i\in [d]}
$$

Where $\kappa$ is Rabin's hash function (for an opportune prime modulo guaranteeing low collision probability, see Section \ref{sec:KR}).
Notice that $\sum_{d\in P} d = O(\log^2n / \log\log n) = O(\log^2n)$ (we are going to use the latter bound for clarity, even if it is less precise), so the sketch uses $O(\log^2 n)$ words of space. Given $sketch(x)$ and $sketch(y)$, we can decide $D_H(x,y) \le 1$ with the following algorithm:

\vspace{5pt}

\begin{algorithm}[H]\label{alg:Hamming sketch 1}
	\footnotesize
	\caption{Decide if $D_H(x,y) \le 1$ for two strings $x,y\in \Sigma^n$}
	\SetKwInOut{Input}{input}
	\SetKwInOut{Output}{output}
	\SetSideCommentLeft
	\LinesNumbered
	
	\BlankLine
	
	\ForEach{$d \in P$}{
        $c \leftarrow 0$\;
        \ForEach{$i\in [d]$}{
            \If{$\kappa(x_{i:d}) \neq \kappa(y_{i:d})$}{
                $c \leftarrow c+1$\;
            }
            \If{$c > 1$}{
                \Return ``$D_H(x,y) > 1$''\;
            }
        }
    }
    \Return ``$D_H(x,y) \le 1$''\;
	
\end{algorithm}

\vspace{5pt}

As an exercise, try to modify Algorithm \ref{alg:Hamming sketch 1} so that it returns $D_H(x,y)$ if $D_H(x,y)\le 1$.

\paragraph{Approximate pattern matching on streams}

It is not hard to use the above sketch in the streaming scenario, where $|x|=m>n$: create $|P|$ copies of the stream $x$ 
\footnote{Note: this means that, when the stream character $x_j$ arrives, we just duplicate it $|P|$ times. Each copy is appended to the corresponding copy of the stream. This means that the original stream is read only once (i.e. we are not reading the stream $|P|$ times, which we cannot do in our setting!).}
and, for each $d\in P$, divide the corresponding copy of the stream into $d$ sub-streams $x_i\ x_{i+d}\ x_{i+2d}\dots$, for all $i \in [d]$.
Using the Porat\&Porat's algorithm, we then decide if $y_{j:d}$ matches a suffix of $x_{i:d}$, for all possible values of $i,j\in [d]$. 
By the discussion in the previous section, this is enough to decide if $y$ and the last $n$ characters of the stream $x$ match up to Hamming distance 1, with high probability.
Note that we run $O(\log^3 n)$ parallel instances of Porat\&Porat's algorithm. Since each instance uses $O(\log n)$ words of memory, we obtain: 

\begin{theorem}
	Let $m$ be the stream's length and $n\leq m$ be the pattern's length.
	The above modification of Porat\&Porat's algorithm finds
	all occurrences of the pattern at Hamming distance at most 1 in the stream
    using $O(\log^4 n)$ words of memory and $O(\log^4 n)$ delay. The correct solution is returned with high probability. 
\end{theorem}

\paragraph{Extension to any threshold $k$ of mismatches}

We can  extend the above idea to $k\geq 1$ mismatches. 
Take $x,y\in \Sigma^n$, and assume that $D_H(x,y) > k$. Consider any group of $k+1$ mismatches between $x$ and $y$, at positions $i_1 < \dots < i_{k+1}$. We want to find a prime number $d\geq k+1$  such that $d$ does not divide $i_{j}-i_{j'}$, for all $1\leq j' < j \leq k+1$. Then, we are guaranteed that
$x_{i:d}\neq y_{i:d}$ for at least $k+1$ shifts $i$, since no pair of mismatches $i_{j},i_{j'}$ can fall in the same shift (which would imply that $d$ divides $|i_{j}-i_{j'}|$). Note: we require that $d$ is greater than or equal to $k+1$ because the number of shifts is $d$ and we need at least $k+1$ distinct mismatching shifts of $x$ and $y$ in order to detect Hamming distance $>k$.

The integer $d$ does not divide $i_{j}-i_{j'}$ for all $1\leq j' < j \leq k+1$
if and only if $d$ does not divide their product $\prod_{1\leq j' < j \leq k+1} (i_{j}-i_{j'}) \leq n^{(k+1)^2}$. We will surely find such an integer $d$ in the set $P_k$ of the smallest $\log_2 n^{(k+1)^2} = O(k^2\log n)$ prime numbers larger than or equal to $k+1$. 

The new sketch is:

$$
sketch(x) = \langle  \kappa(x_{i:d}) \rangle_{d\in P_k, i\in [d]}
$$

From here, to simplify notation we use the $\tilde O$ notation: 

\begin{definition}[tilde-O notation]
    $\tilde O_n(t)$ denotes any complexity in $O(t\cdot \mathrm{polylog\ }n)$, where $\mathrm{polylog\ n} = (\log n)^{O(1)}$ is any  polynomial of $\log n$ of constant degree. When $n$ is clear from the context, we just write $\tilde O(t)$.
\end{definition}

In other words, the tilde-O notation hides polylogarithmic factors. In streaming scenarios this notation makes sense, since $\mathrm{polylog\ n}$ is negligible with respect to the input size $n$ and it is more useful to focus on the main complexity term, which in this particular scenario depends on $k$.

Since $|P_k| \in \tilde O(k^2)$ and the sketch contains $d$ hashes for every $d\in P_k$, the size of the sketch is of $\tilde O(k^4)$ words.

The algorithm for deciding $D_H(x,y) \le k$ works exactly as in the case $k=1$, except that we now look for a prime $d\in P_k$ such that $x_{i:d} \neq y_{i:d}$ for at least $k+1$ values of $i \in [1,d]$:

\vspace{5pt}

\begin{algorithm}[H]\label{alg:Hamming sketch 2}
	\footnotesize
	\caption{Decide if $D_H(x,y) \le k$ for two strings $x,y\in \Sigma^n$}
	\SetKwInOut{Input}{input}
	\SetKwInOut{Output}{output}
	\SetSideCommentLeft
	\LinesNumbered
	
	\BlankLine
	
	\ForEach{$d \in P_k$}{
        $c \leftarrow 0$\;
        \ForEach{$i\in [d]$}{
            \If{$\kappa(x_{i:d}) \neq \kappa(y_{i:d})$}{
                $c \leftarrow c+1$\;
            }
            \If{$c > k$}{
                \Return ``$D_H(x,y) > k$''\;
            }
        }
    }
    \Return ``$D_H(x,y) \le k$''\;
	
\end{algorithm}

\vspace{5pt}

As an exercise, try to modify the above algorithm so that it returns the exact value of $D_H(x,y)$ if $D_H(x,y)\le k$.

In the streaming scenario, for every $d\in P_k$ and every $i,j\in [1,d]$ we need to run a parallel instance of the Porat\&Porat's algorithm to check if $y_{j:d}$ matches a suffix of $x_{i:d}$: this is a cubic (in $|P_k|$) number of instances of Porat\&Porat's algorithm. We obtain:

\begin{theorem}
	The above modification of Porat\&Porat's algorithm finds
	all occurrences of the pattern at Hamming distance at most $k$ in the stream
    using $\tilde O(k^6)$ words of memory and $\tilde O(k^6)$ delay. The correct solution is returned with high probability. 
\end{theorem}

The above bounds can be easily improved with slightly more sofisticated techniques. 
In their original article \cite{porat2009exact}, Porat and Porat describe a more efficient solution using $\tilde O(k^3)$ space and $\tilde O(k^2)$ delay. Clifford et al. in \cite{clifford2016k} improved this to $\tilde O(k^2)$ space and $\tilde O(\sqrt k)$ delay.
These bounds were further improved in \cite{clifford2019streaming} to $\tilde O(k)$ space and $\tilde O(\sqrt k)$ delay. The authors of \cite{clifford2019streaming} prove that the latter bounds are optimal (up to poly-logarithmic factors).

\section{Counting ones in a window: Datar-Gionis-Indyk-Motwani's algorithm}

The DGIM algorithm \cite{datar2002maintaining} addresses the following basic problem. Consider an input stream of $m$ bits. What is the sum of the last $\bar m \leq m$ elements of the stream? 

This problem models several practical situations in which storing the entire stream is not practical, but we may be interested in counting the number of \emph{interesting} events among the last $\bar m\leq m$ events. 

\begin{example}
    Consider a stream of bank transactions for a given person; we mark a transaction with a 1 if it exceeds a given threshold (say, 50 euros) and with a 0 otherwise. Then, knowledge about the number of 1s in the last $\bar m$ transactions can be used to detect if the credit card's owner has changed behaviour (for example, has started spending much more than usual) and detect potential frauds (e.g. credit card has been cloned). 
\end{example}

It is easy to see that an exact solution requires $m$ bits of space (i.e. the entire stream). For any $0 < \epsilon \leq 1$, the DGIM algorithm uses $O(\epsilon^{-1}\log^2 m)$ bits of space and returns a multiplicative $(1+\epsilon)$-approximation (with certainty: DGIM is a deterministic algorithm). 

DGIM works as follows. Let $B = \lceil 1/\epsilon \rceil$. We group the stream's bits in groups $G_1, G_2, \dots, G_t$ that must satisfy the following rules:\ \\\ \\

\begin{enumerate}
    \item Each $G_i$ begins and ends with a 1-bit.
    \item Between two adjacent groups $G_i$, $G_{i+1}$ there are are only 0-bits, i.e. the stream is of the form $0^{m_0}\cdot G_1 \cdot 0^{m_1} \cdot G_2 \cdot \dots \cdot G_t \cdot 0^{m_{t}}$ for some $m_0,\dots, m_t \geq 0$.
    \item Each $G_i$ contains $2^k$ 1-bits, for some $k\geq 0$.
    \item For any $1\leq i < t$, if $G_i$ contains $2^k$ 1-bits, then $G_{i+1}$ contains either $2^{k}$ or $2^{k-1}$ 1-bits.
    \item For each $k$ except the largest one, the number $Z_k$ of groups containing $2^k$ 1-bits satisfies $B \leq Z_k \leq B+1$ (note that these groups must be adjacent). For the largest $k$, we only require $Z_k \leq B+1$.
\end{enumerate}

See Figure \ref{fig:DGIM1} for an example.

\begin{figure}[h!]
\centering
	\includegraphics[scale=0.7]{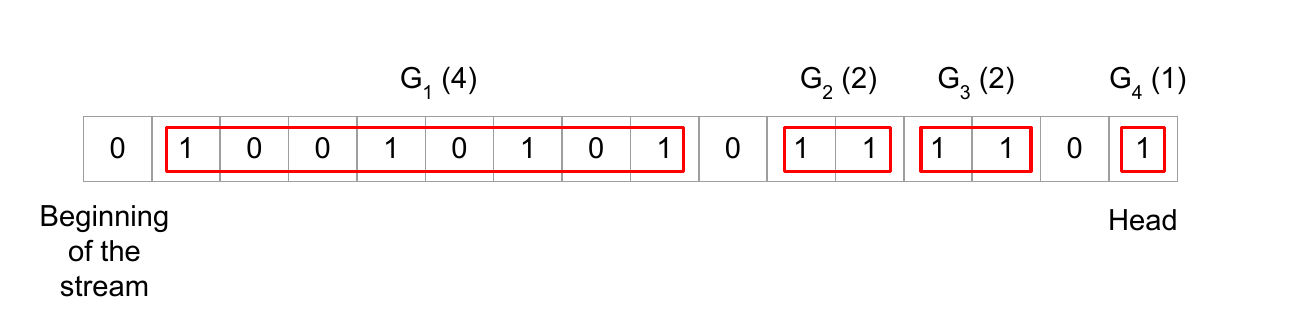}
	\caption{DGIM with parameter $B=1$. The head of the stream (the most recent element) is the rightmost bit. For each $k$, there are at least $B=1$ groups containing $2^k$ 1-bits, and at most $B+1 = 2$ groups containing $2^k$ 1-bits.}\label{fig:DGIM1}
\end{figure}

\subsection{Updates}

It is easy to see how to maintain the rules when a new bit arrives. If the bit is equal to 0, then nothing has to be done. If the bit is equal to 1, then:

\begin{enumerate}
    \item Create a new group with the new bit. 
    \item If there are $B+2$ groups containing $2^0 = 1$ 1-bits, merge the two leftmost such groups so now there are $B$ groups containing one $1$-bit.  This creates a new group containing $2^1 = 2$ 1-bits. \item Repeat with the groups containing $2^i$ 1-bits, for $i=1,2,\dots$.
\end{enumerate}

It is easy to see that
one update step takes $O(\log m)$ worst-case time using doubly linked lists (this time is the delay of the algorithm). Define a global list $L = \ell_q \leftrightarrow \ell_{q-1} \leftrightarrow \dots \leftrightarrow \ell_1$. Element $\ell_i$ contains all the groups with $2^i$ 1-bits and is itself a doubly-linked list: $\ell_i = G_{j_1} \leftrightarrow \dots \leftrightarrow G_{j_s}$, where $G_{j_1},\dots, G_{j_s}$ are all the groups (listed from left to right in the stream) containing $2^i$ 1-bits. 
Each group $G_j$ is simply a pair of integers $G_j = (left, right)$: the leftmost and rightmost positions of the group in the stream. 
For each linked list $\ell_i$, we store its head, tail, number of stream's bits contained in its groups, and number of '1'-bits contained in its groups. 
Then, finding the leftmost two groups in a given $\ell_i$, merging them, and moving the merged group to the end of $\ell_{i+1}$
takes $O(1)$ time. Overall, an update takes therefore $O(q) = O(\log m)$ time. 

Even better, updates take $O(1)$ \emph{amortized} time. To see this, suppose that a particular update increases $Z_k$ by one unit (recall that $Z_k$ is the number of groups containing $2^k$ 1-bits). But then, this means that before that update $Z_{k'} = B+1$ for all $k'<k$. In turn, this configuration required $2^k-1$ previous updates, which added to the new update yields $2^k$ updates in total. This shows that only one over $2^k$ updates costs $k$: the amortized cost is therefore at most $\sum_{k=1}^\infty k/2^k = O(1)$. 

\subsection{Space and queries}

The algorithm uses in total $O(\epsilon^{-1}\log^2m)$ bits of memory: each group uses $O(\log m)$ bits, and there are at most $B+1 = O(\epsilon^{-1})$ groups containing $2^k$ 1-bits, for each $k=0, \dots, \log m$. The lists' pointers also use asymptotically this space.

A query is specified by an integer $\bar m\leq m$ (the window size); our goal is to return the number of 1-bits contained in the most recent $\bar m$ bits of the stream. 
To solve a query, we simply find all the groups intersecting  with (i.e. containing at least one of) the last $\bar m$ stream's bits, and return the total number of $1$-bits they contain. 
This can be implemented by navigating the main list $\ell_1 \rightarrow \ell_2 \rightarrow \dots \rightarrow \ell_i$ until the number 
of bits contained in cells $\ell_1, \ell_2, \dots, \ell_i$ exceeds (or becomes equal to) $\bar m$. Then, we enter the cell $\ell_i$ and start navigating the linked list contained in $\ell_i$ until finding the leftmost group overlapping the last $\bar m$ stream's 
characters. Finally, we add up all '1' bits contained in the cells $\ell_1, \ell_2, \dots, \ell_{i-1}$ and in the groups scanned in the last cell $\ell_i$. The total running time is $O(\epsilon^{-1} + \log \bar m)$.
Finally, if $\bar m$ is fixed then it is easy to see that  queries take $O(1)$ time: at any time, we keep in memory only the groups overlapping with the last $\bar m$ stream's bits (together with the total number of 1-bits that they contain). This also reduces the total space usage to $O(\epsilon^{-1}\log^2 \bar m)$ bits.

\subsection{Approximation ratio}

Next, we analyze the approximation ratio of the algorithm. Consider Figure \ref{fig:DGIM2}, corresponding to the worst-case approximation ratio.

\begin{figure}[h!]
\centering
	\includegraphics[scale=0.7]{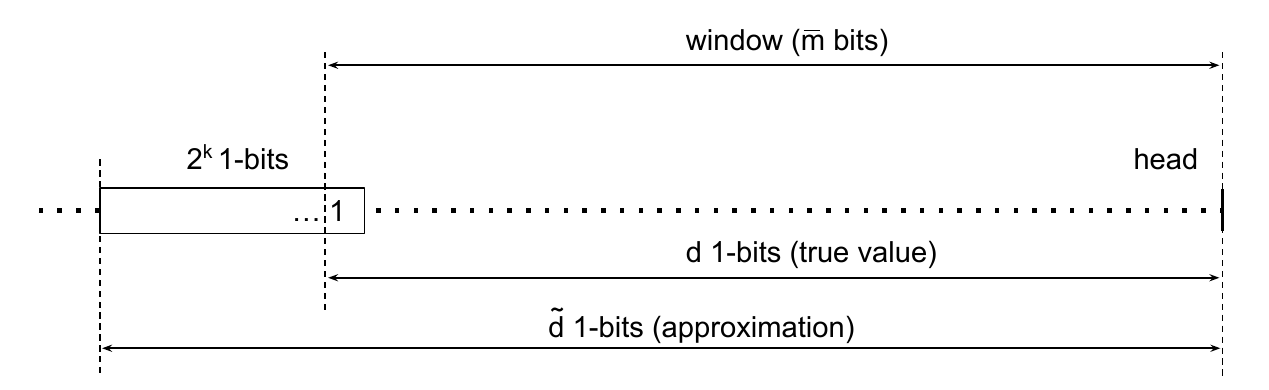}
	\caption{Worst case: our query (window of $\bar m$ bits) spans only the last bit of a block. 
	The true number of 1-bits in the window is $d$.
	The answer we return ($\tilde d$) includes the whole block depicted on the left and containing $2^k$ '1'-bits.}\label{fig:DGIM2}
\end{figure}

Let $k$ be the integer such that the leftmost (oldest) group intersecting the window has $2^k$ 1-bits. 
Let $d$ be the true number of 1-bits in the window, and $\tilde d$ be the sum of 1-bits in the groups intersecting the window (i.e. our approximate answer).
If $k=0$, then it is easy to see that $\tilde d=d$ because  every group intersecting the window contains 1 bit. 
We can therefore assume $k>0$. 
Clearly, $\tilde d \geq d$ since we count every block that overlaps with the window. 
We first compute a lower bound to $d$. Since the window spans a group containing $2^k$ 1-bits, then (by our invariants) the window surely contains at least $B$ groups containing $2^{j}$ 1-bits,  for all $0 \leq j < k$, i.e. 

$$
\begin{array}{rcl}
  d & \geq & B \cdot 2^{k-1} + B \cdot 2^{k-2} + \dots  + B\cdot 2^0\\
  & = &   B\cdot(2^k-1)
\end{array}
$$

On the other hand, $\tilde d \leq d + 2^{k} - 1$ (in the worst case we have equality: this is the case in the example in Figure \ref{fig:DGIM2}). We obtain:

$$
\begin{array}{rcl}
  \frac{\tilde d}{d} & \leq & \frac{d + 2^{k} - 1}{d}\\
  & = & 1 + \frac{2^k-1}{d}\\
  & \leq & 1 + \frac{2^k-1}{B\cdot(2^k-1)}\\
  & = & 1 + 1/B\\
  & \leq & 1 + \epsilon
\end{array}
$$

We conclude that $d \leq  \tilde d \leq d\cdot (1+\epsilon)$.

The web page \footnote{\url{https://observablehq.com/@andreaskdk/datar-gionis-indyk-motwani-algorithm}} implements a very nice simulator of the DGIM algorithm (note that the stream's head is on the left in this simulation).

\subsection{Generalization: sum of integers}

The algorithm can be used as a basis for many generalizations. Consider for example a stream formed by integers of $q$ bits each. We are interested in computing the sum of the last $\bar m$ integers in the stream. 

The solution is to break the stream into $q$ parallel streams, one per bit in the integers: see Figure \ref{fig:DGIM3}.

\begin{figure}[h!]
\centering
	\includegraphics[scale=0.7]{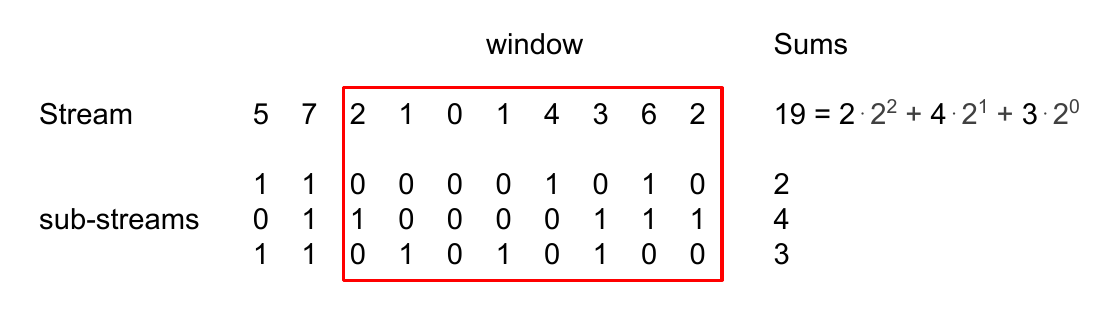}
	\caption{The sum of the last $\bar m$ $q$-bits integers can be reduced to the sum of the last $\bar m$ bits in the $q=3$ streams corresponding to the binary representations of the integers. }\label{fig:DGIM3}
\end{figure}

In other words: the $i$-th bit stream contains the binary weight of power $2^i$ in each integer of the original stream.
Let $s_i$ be the sum of the $i$-th bit stream in the window. The correct answer is $d = \sum_{i=0}^{q-1} s_i2^i$. From the analysis of DGIM, we conclude that the answer $\tilde d$ we return is  $d \leq \tilde d \leq \sum_{i=0}^{q-1} (1+\epsilon)s_i2^i = (1+\epsilon)\cdot d$.

\section{Estimating frequencies on a stream}\label{sec:frequencies}

 Suppose we observe a stream of $m$ integers $x_1, x_2, x_3, \dots, x_m$, each with domain $x_i \in [1,n]$. For any $y \in [1,n]$ define the frequency of $y$ as the number of times $y$ appears in the stream:

$$
f_y = |\{i\ :\ x_i = y\}|
$$

In this section and in and the following  we tackle a class of problems related with estimating particular functions of the frequency vector $(f_1, \dots, f_n)$. In this section, we discuss how to estimate each individual frequency $f_y$. Then, in the next sections we move to \emph{frequency moments}. The $k$-th frequency moment, for $k\ge 0$, is defined as follows:

$$
F_k = \sum_{y\in [1,n]} f_y^k
$$

If we take $0^0 = 0$, then $F_0 = |\{x_1,\dots, x_m\}|$ is simply the number of distinct elements in the stream. Counting distinct elements in a stream is probably the most emblematic problem in the field, and it is tackled in Section \ref{sec:distinct elements}.

Quantity $F_1$, instead, corresponds simply to the stream's length: $F_1 = \sum_{y\in[1,n]} f_y = m$. While it is trivial to compute $F_1$ exactly on a stream using $O(\log m)$ bits of space, using asymptotically less space is not trivial at all. In Section \ref{sec:counting} we will show an algorithm estimating $F_1$ in $O(\log\log m)$ bits of space. 

The second-order frequency moment, $F_2 = \sum_{y\in[1,n]} f_y^2$, is also called the \emph{surprise number} of the stream since it reflects how unbalanced the frequencies in the stream are: the more uneven the frequencies are, the larger $F_2$ is. Estimating this quantity finds unexpected deep applications in other contexts, such as estimating the size of table joins in relational databases, and sketches for dimensionality reduction under the Euclidean distance. These topics are covered in Section \ref{sec:F2}.

\subsection{Sampling}

The first solution that we discuss for estimating individual frequencies $f_y$ is just sampling. 
Assume that we know the stream length $m$. Choose $q$ i.i.d. uniform values $j_1, \dots, j_q \in [1,m]$ (i.e. $q$ values \emph{with replacement}),
where $q$ is a value that will be determined later. When the stream arrives, memorize the elements $x'_1 = x_{j_1}, \dots, x'_q = x_{j_q}$.  This sample is our sketch, of size $q$.

\begin{remark}
    It's  possible to sample $q$ elements with replacement from the stream on-the-fly, without knowing a priori the stream length $m$. The algorithm solving this task is called \emph{reservoir sampling with replacement} and is discussed in \cite{park2004reservoir}.
\end{remark}

Define the boolean indicator $\mathbbm 1_{x=y}$ to be equal to 1 if and only if $x=y$. Observe that, for any given $y$ and any $x'_j$: 

$$
E[\mathbbm 1_{x'_j=y}] = f_y/m
$$

In fact, $\mathbbm 1_{x'_j=y}$ is a Bernoullian variable taking value 1 if and only if $x'_j=y$; since $x'_j$ is a uniform stream element, the probability of choosing a value equal to $y$ is $f_y/m$. 
In particular: 

$$
E\left [\sum_{j=1}^q \mathbbm 1_{x'_j=y} \right] = (q/m) f_y
$$

Our estimator for $f_y$ is:

$$
\tilde f_y = \frac{m}{q}\sum_{j=1}^q \mathbbm 1_{x'_j=y}
$$

In other words, $\tilde f_y$ is the number of occurrences of $y$ in our sample of stream elements, scaled by the factor $m/q$.
It is easy to see that $E[\tilde f_y] = f_y$. Then:

$$
P(|\tilde f_y - f_y| \ge \epsilon\cdot m) = P( (q/m)|\tilde f_y - f_y| \ge \epsilon \cdot q ) = P\left( \left|  \sum_{j=1}^q \mathbbm 1_{x'_j=y} - E\left[\sum_{j=1}^q \mathbbm 1_{x'_j=y}\right] \right| \ge \epsilon\cdot q \right)
$$

Since our sampling is \emph{with replacement} (i.e. the sampled stream positions are i.i.d. uniform random variables in $[1,m]$), the random variables $\mathbbm 1_{x'_1=y}, \dots, \mathbbm 1_{x'_q=y}$ are independent so we can apply Chernoff-Hoeffding (additive form, double sided: Lemma \ref{lem:add Hoeffding}) to the random variable $\sum_{j=1}^q \mathbbm 1_{x'_j=y}$ and obtain that the above probability is upper-bounded by 
$$
2e^{-(\epsilon\cdot q)^2/(2q)} = 2e^{-\epsilon^2q/2} 
$$

We want this probability to be at most $2e^{-\epsilon^2q/2} = \delta$. Solving this equation as a function of $q$, this gives us the final result: 

\begin{theorem}
    By sampling $q = \lceil 2\epsilon^{-2}\ln(2/\delta)\rceil \in \Theta(\epsilon^{-2}\log(1/\delta))$ uniform stream values with replacement, we can approximate any frequency $f_y$ up to additive error $\epsilon\cdot m$ and failure probability $\delta$, i.e. for any $y\in [1,n]$, the sample yields an estimator $\tilde f_y$ such that: 
    $$
    P(|\tilde f_y - f_y| \ge \epsilon\cdot m) \le \delta.
    $$
\end{theorem}

Note that our sampling-based sketch has a double-sided error: the estimate $\tilde f_y$ can be smaller, equal to, or larger than the true value $f_y$.

\subsection{Count-Min sketch}

Count-Min sketch uses less asymptotic space than the sampling-based sketch described in the previous section, and additionally offers a \emph{one-sided} error guarantee. As we now show, the estimator returned by Count-Min sketch always satisfies $\tilde f_y \geq f_y$. In addition, with probability at least $1-\delta$ the bound $\tilde f_y \le f_y + \epsilon\cdot m$ will hold.

\paragraph{Definition of the sketch}
The Count-Min sketch is a matrix of (non-negative) integers $CM \in \mathbb N^{t\times s}$, initialized with all entries equal to zero. The size $s\times t$ of the matrix determines the error rate and the success probability, and will be determined later. Each row is associated with a \emph{universal} hash function (see Section \ref{sec: universal hashing}) $h_j : [1,n] \rightarrow [1,s]$, for $j=1, \dots, t$.

\paragraph{Insertion of an element in the sketch}
To insert a stream element $x$ in the sketch, we perform the following operation: 
$$
CM[j,h_j(x)] \gets CM[j,h_j(x)]+1, \mathrm{for\ all\ }j = 1, \dots, t.
$$ 
In other words, in the $j$-th row we increment by one unit the entry at column $h_j(x)$.

\paragraph{Estimating a frequency}

Observe that, by the way we defined the sketch, the counter $CM[j,h_j(y)]$ is incremented each time an occurrence of $y$ is seen in our stream (for each row $j=1, \dots t$). In fact, if the stream contains only one distinct integer $y$, then $CM[j,h_j(y)] = f_y$ is exactly the frequency of $y$, for any row $j=1, \dots, t$. Problems start occurring when the stream contains more than one distinct integer, due to hash collisions. Assume there is an element $y'\neq y$ such that $h_j(y)=h_j(y')$, for some $j\in [1,t]$. Then, $CM[j,h_j(y)] \ge f_y + f_{y'}$ (the inequality is strict if even more integers collide with $y$ and $y'$).

Due to hash collisions, we can therefore only assume $CM[j,h_j(y)] \ge f_y$ for every $j = 1, \dots, t$. Then, our best estimate for $f_y$ is the \emph{minimum} of those estimates:

$$
\tilde f_y = \min\{ CM[j,h_j(y)]\ :\ j=1, \dots, t  \}
$$

As observed above, $\tilde f_y \ge f_y$ is always true. The next goal is to show an upper bound as well. Fix a row $j\in [1,t]$. We study the probability that $CM[j, h_j(y)] > f_y + \epsilon\cdot m$, for some error $\epsilon > 0$. 

Define the boolean indicator $\mathbbm 1_{x,y}^j$ to be equal to 1 if and only if $h_j(x)=h_j(y)$, i.e. if $x$ and $y$ collide through hash function $h_j$.
Let $S = \{x_1, \dots, x_m\}$ denote the set of distinct elements appearing in the stream. 
The expected value of $CM[j, h_j(y)]$ can be written as:

$$
\begin{array}{rcl}
E[ CM[j, h_j(y)] ]  & = & E\left[\sum_{x \in S} f_x \cdot \mathbbm 1_{x,y}^j\right]\\
& = & f_y + E\left[\sum_{x \in S, x\neq y} f_x \cdot \mathbbm 1_{x,y}^j\right]\\
& = & f_y + \sum_{x \in S, x\neq y} f_x \cdot E\left[\mathbbm 1_{x,y}^j\right]\\
& \le & f_y + \sum_{x \in S, x\neq y} f_x \cdot (1/s)\\
& = & f_y + (1/s)\sum_{x \in S, x\neq y} f_x\\
& \le & f_y + m/s
\end{array}
$$

where in the first inequality we used the fact that the universality of $h_j$ implies $P(h_j(x)=h_j(y)) \le 1/s$ for any $x\neq y$.
We conclude that: 

$$
E[ CM[j, h_j(y)] - f_y] \le m/s
$$

that is, on expectation the difference between our estimate (in any row $j$) and $f_y$ is at most $m/s$. Applying Markov's inequality: 

$$
P(CM[j, h_j(y)] - f_y \ge 2m/s ) \le 1/2
$$

Choose any desired error $\epsilon>0$. By setting $2m/s =\epsilon \cdot m$, we obtain: 

\begin{lemma}\label{lem:CM1}
    By choosing the number of columns in the CM Sketch to be $s = \lceil 2\epsilon^{-1} \rceil$, for any row $j \in [1,t]$ it holds that: 
    $$
    P(CM[j, h_j(y)] \ge f_y + \epsilon \cdot  m ) \le 1/2
    $$
\end{lemma}

What is the probability that $\tilde f_y \ge f_y + \epsilon \cdot  m$? By definition, $\tilde f_y$ is the minimum of $CM[j, h_j(y)]$, for all rows $j=1, \dots, t$. Then, if the minimum $\tilde f_y$ exceeds  $f_y + \epsilon \cdot  m$, it means that $CM[j, h_j(y)] \ge f_y + \epsilon \cdot  m$ for \emph{all} $j=1, \dots, t$. Since the functions $h_1, \dots, h_t$ are independent, we apply Lemma  \ref{lem:CM1} and obtain that 
$$
P(\tilde f_y \ge f_y + \epsilon \cdot  m) \le (1/2)^t
$$

We want this probability to be $(1/2)^t = \delta$, for any desired failure probability $\delta > 0$. Solving in $t$, we finally obtain that 
by choosing $t = \lceil \log_2(1/\delta)\rceil$
our estimate $\tilde f_y$ satisfies $f_y \le \tilde f_y \le f_y + \epsilon \cdot  m$ with probability at least $1-\delta$. 
Assume that each $h_j$ can be evaluated in constant time. 
Putting everything together we obtain: 

\begin{theorem}
    Choose any desired error rate $\epsilon>0$ and failure probability $\delta > 0$.
    The CM Sketch uses $O(\epsilon^{-1}\log(1/\delta))$ words of space and, for any $y\in [1,n]$, returns an estimate $\tilde f_y$ such that:
    
    \begin{itemize}
        \item $\tilde f_y \ge f_y$
        \item $P(\tilde f_y \le f_y + \epsilon \cdot  m) \ge 1-\delta$
    \end{itemize}

    Frequency estimation queries are supported in $O(\log(1/\delta))$ time.
    When a new stream element arrives, the CM sketch can be updated in $O(\log(1/\delta))$ time. 
\end{theorem}

Note that this error is one-sided from above: the sketch never under-estimates $f_y$, but it could over-estimate it. Recall that in the previous subsection we obtained a double-sided error with larger space. It follows that Count-Min sketch is strictly better than sampling. 

\subsection{Misra-Gries sketch}

Interestingly, the problem can be solved also without randomization. Here we describe the Misra-Gries sketch, from \cite{misra1982finding}. The sketch provides a one-side error from below: it never over-estimates $f_y$, but it could under-estimate it (the opposite of the CM Sketch). 

\paragraph{Definition of the sketch}

The sketch is  a dictionary $H : [1,n] \rightarrow \mathbb N$ (a set of key-value pairs) that associates frequencies (natural numbers) to the stream's integers. 
We require $H$ to implement a function:  at any point in time for any $x\in [1,n]$ the dictionary will contain at most one key-value pair $(x,f)$, for some $f\in \mathbb N$. 
We denote with $H[x]$ such a frequency $f$. If the dictionary does not contain any pair of the form $(x,f)$, we define $H[x] = 0$ (in this way, $H$ implements a complete function).
We abuse notation and write $x\in H$ to indicate an integer $x\in [1,n]$ such that the dictionary contains a key-value pair of the form $(x,f)$, for some $f \in \mathbb N$.

We denote with $|H|$ the number of pairs $(x,f) \in [1,n]\times \mathbb N$ stored in $H$. 
Initially, the dictionary is empty ($|H|=0$).
We require the dictionary to use $O(|H|)$ words of space.

Increments on $H[x]$, denoted as $H[x] \gets H[x]+1$,  have the following natural meaning. If $x\notin H$, then $H[x] \gets H[x]+1$ inserts the pair $(x,1)$ in $H$. Otherwise, the operation replaces the (unique) pair $(x,f)$ in $H$ with $(x,f+1)$. Decrements have an analogous meaning. If $x\in H$, then $H[x] \gets H[x]-1$ replaces the (unique) pair $(x,f)$ in $H$ with $(x,f-1)$. If, after such a replacement, the pair is of the form $(x,0)$, then we remove it from $H$ (in particular, $|H|$ decreases by one). We will never call decrements $H[x] \gets H[x]-1$ when  $x\notin H$.

The dictionary $H$ (with the above operations) can be implemented easily with a hash table or, to avoid randomization, with a self-balancing tree.

\paragraph{Insertion of an element in the sketch}

Let $s>1$ be an integer parameter to be determined later. 
To insert a stream element $x$ in the sketch, we perform the following two operations sequentially: 

\begin{enumerate}
    \item We increment $H[x] \gets H[x]+1$.
    \item If $|H| = s$, then for every $x\in H$ we decrement $H[x] \gets H[x]-1$.
\end{enumerate}

Note that, if operation 1 increases $|H|$ by one, then (if $|H| = s$ is true) operation 2 erases the pair $(x,1)$ inserted at step 1. 
From this observation it is immediate to see that the invariant $|H| < s$ is always true, so $H$ always uses $O(s)$ words of space.

\paragraph{Estimating a frequency}

Our estimate for $y \in [1,n]$ is simply $\tilde f_y = H[y]$.

First, for any $y\in[1,n]$ clearly it holds $\tilde f_y \le f_y$. To see this, observe that $H[y]$ is incremented only when a stream element equal to $y$ is observed. 

We now prove that $\tilde f_y \ge f_y - m/s$ \emph{always holds} (deterministically!). Observe that:

\begin{remark}
    When $H[y]$ is decremented by one unit in operation 2, in total $s$ frequencies $H[x]$ are decremented. 
\end{remark}

We can imagine each decrement of $H[x] \gets H[x]-1$ as removing one occurrence of $x$ from the stream. But then, $H[y]$ cannot differ from $f_y$ by more than $m/s$: this would imply that we remove more than $(m/s)\cdot s = m$ elements from the stream (impossible, since the stream contains $m$ elements). We conclude that it must be $H[y] = \tilde f_y \ge f_y - m/s$ for any $y\in [1,n]$.

By choosing $s= \lceil \epsilon^{-1} \rceil$ for any desired error rate $\epsilon>0$, we conclude: 

\begin{theorem}
    Choose any desired error rate $\epsilon>0$.
    The Misra-Gries Sketch uses $O(\epsilon^{-1})$ words of space and, for any $y\in [1,n]$, returns an estimate $\tilde f_y$ such that $f_y - \epsilon\cdot m \le \tilde f_y \le f_y$.

    By implementing the dictionary as a self-balancing tree, frequency estimation queries are supported in $O(\log(\epsilon^{-1}))$ time. 
    When a new stream element arrives, the Misra-Gries sketch can be updated in $O(\epsilon^{-1})$ time. 
\end{theorem}

\section{Estimating $F_0$ (counting distinct elements)}\label{sec:distinct elements}

As mentioned earlier, by defining $0^0 = 0$, then $F_0 = |\{x_1,\dots, x_m\}|$ is the number of distinct elements in the stream. We first discuss naive solutions, then move to a simplified (assuming fully-uniform hashing) version of the venerable Flajolet-Martin algorithm \cite{flajolet1983probabilistic}. We then discuss a more realistic version (the \emph{Bottom-$k$} algorithm, requiring only pairwise independent hashing), and conclude with a broad discussion of cardinality-estimation algorithms used in practice.

\subsection{Naive solutions}

We start with the problem of counting
the number of \emph{distinct} integers in the stream, i.e. $F_0 = d = |\{x_1, x_2, \dots, x_m\}|$. We cannot afford to use too much memory (and $m$ and $d$ are very large --- typically in the order of billions). 

We report some illuminating examples of the practical relevance of the count-distinct problem. Some of these examples are taken from the paper \cite{estan2003bitmap}.

\begin{example}[DoS attacks]
	Denial of Service attacks can be detected by analyzing the number of distinct flows (source-destination IP pairs contained in the headers of TCP/IP packets) passing through a network hub in a specific time interval. The reason is that typical DoS software use large numbers of fake IP sources; if they were to use few IP sources, then those sources could be easily identified (and blocked) because of the large traffic they must generate in order for the DoS attack to be effective.  
\end{example}

\begin{example}[Spreading rate of a worm]
	Worms are self-replicating malware whose goal is to spread to as many computers as possible using a network (e.g. the Internet) as medium. In order to count how many computers have been infected by the worm, one needs to (1) filter packets containing the worm's code, and (2) count the number of distinct source IPs in the headers of those packets. From \url{https://www.caida.org/archive/code-red/} (an analysis of the spread of the Code-Red version 2 worm between midnight UTC July 19, 2001 and midnight UTC July 20, 2001):
	
	\begin{center}
		``On July 19, 2001 more than 359,000 computers were infected with the Code-Red (CRv2) worm in less than 14 hours. At the peak of the infection frenzy, more than 2,000 new hosts were infected each minute.''
	 \end{center}

\end{example}

\begin{example}[Distinct IPs/post views]
Suppose we wish to count how many people are visiting our web site. Then, we need to count how many \emph{distinct} IP numbers are connecting to the server that hosts the web site. 
The same problem occurs with post views; in this case, the problem is more serious since the problem must be solved for each post! For example, Reddit uses a randomized cardinality estimation algorithm (HLL) to count post views.
\end{example}

A first naive solution to the count-distinct problem is to keep a bitvector $B[1,n]$ of $n$ bits, initialized with all 0's. Then it is sufficient to set $B[x_i] = 1$ for each element $x_i$ of the stream. Finally, we count the number of 1's in the bitvector. If $n$ is very large (like in typical applications), this solution uses too much space. A second solution could be to store the stream elements in a self-balancing binary search tree or in a hash table with dynamic re-allocation. This solution uses $O(d\log n)$ bits of space, which could still be too much if the number $d$ of distinct elements is very large. 

\paragraph{Using MinHash} The MinHash sketch described in Section \ref{sec:minhash} can be used to solve the problem in $O(\log n)$ bits of space with additive error $\epsilon n$. Let $A = \{x_1, \dots, x_m\} \subseteq [n]$ be the set of which we want to estimate the cardinality $d=|A|$. Observe that $A = A \cap [n]$, so $|A| = |A \cap [n]|$. We can use MinHash to estimate the Jaccard distance $J(A,[n]) = \frac{|A \cap [n]|}{|A\cup [n]|} = \frac{|A|}{n}$ between $A$ and the universe $[n]$. Then:
$$
|A| = n\cdot J(A,[n])
$$
Recalling that MinHash gives us an additive $\epsilon$-approximation $J^+(A,[n])$ of $J(A,[n])$ with probability at least $1-\delta$, we obtain that the estimator $n\cdot J^+(A,[n])$ satisfies (with probability $\ge 1-\delta$):
$$
n\cdot J^+(A,[n]) = n\cdot (J(A,[n]) \pm \epsilon) = |A| \pm \epsilon n
$$
This estimate is good if $|A|$ is close to $n$, because then the error $\epsilon n$ is close to $|A|$. However, for small $|A|$ this estimate is very bad: the error $\epsilon n$ is much larger than $|A|$. Unfortunately, in typical scenarios $|A| \ll n$. Consider, for example, the task of counting how many distinct IPv4 addresses sent a query to a server: the answer is surely much smaller than the total number $n = 2^{32}$ of possible IPv4 addresses! In the next sections, the goal will be to achieve a \emph{multiplicative} error $\epsilon|A|$ proportional to the set's cardinality. 

An observation: recall, from Section \ref{sec:minhash}, that our estimator $\hat J_h(A,[n])$ is a Bernoullian RV which takes the value 1 if and only if $\hat h(A) = \hat h([n])$, where $\hat h(X) = \min\{h(x)\ : \ x\in X\}$ and $h$ is a uniform permutation of $[n]$. But then, $\hat h([n]) = 1$ so our estimator $\tilde d =  n\cdot \hat J_h(A,[n])$ can equivalently be defined as:
\begin{equation}
\tilde d = 
\begin{cases*}
n & if $\hat h(A) = 1$ \\
0        & otherwise
\end{cases*}
\end{equation}


\subsection{Flajolet-Martin's algorithm}\label{sec:idealized FM}

We now show algorithms for estimating set cardinalities with multiplicative (i.e. relative) error.
The following solution is an idealized version (requiring totally uniform hash functions) of the algorithm described by Flajolet and Martin in \cite{flajolet1983probabilistic}.
Let $[1,n]$ denote the range of integers $1,2,\dots, n$ and $[0,1]$ denote the range of all real values between 0 and 1, included. 
We use a uniform hash function $h: [1,n] \rightarrow [0,1]$. Note that such a function actually requires $\Theta(n)$ words of space to be stored, see Section \ref{sec:hash functions}: the algorithm is not practical, but we describe it for its simplicity and elegance.

\begin{algorithm}[t!h!]
	\caption{\texttt{FM}}
	\label{alg:FM}
	
	\SetKwInOut{Input}{input}
	\SetKwInOut{Output}{output}
	\SetSideCommentLeft
	\LinesNumbered
	
	\Input{A stream of integers $x_1, \dots, x_m$.}
	\Output{An estimate $\hat d$ of the number of distinct integers in the stream.}
	\BlankLine
	\BlankLine
	
	Initialize $y = 1$\;
	For each stream element $x$, update $y \leftarrow \min(y,h(x))$\;
	When the stream ends, return the estimate $\hat d = \frac{1}{y} - 1$\;
	
\end{algorithm}

Intuitively, why does FM work? First, note that repeated occurrences of some integer $x$ in the stream will yield the same hash value $h(x)$. Since $h$ is uniform, we end up drawing $d$ uniform real numbers $y_1 < y_2 < \dots < y_d$ in the interval $[0,1]$ \footnote{Note that we can safely assume $x \neq y \Rightarrow h(x) \neq h(y)$: since we draw uniform numbers on the real line, the probability that $h(x) = h(y)$ is zero.}. At the end, the algorithm returns $1/y_1 - 1$. The more distinct $y_i$'s we see, the more likely it is to see a smaller value. In particular, $h$ will spread the $y_i$'s uniformly in the interval $[0,1]$; think, for a moment, about the most ``uniform'' (regular) way to spread those numbers in $[0,1]$: this happens when the intervals $[0,y_1]$, $[y_i,y_{i+1}]$, $[y_d,1]$ have all approximately the same length $1/(d+1) \approx y_{i+1}-y_i \approx y_1 - 0 = y_1$. But then, our claim $1/y_1 - 1 \approx d$ follows. It turns out that this is true also \emph{on average} (not just in this idealized ``regular'' case): the \emph{average} distance between $0$ and the smallest hash $y_1$ seen in the stream is precisely $1/(d+1)$. Next, we prove this intuition. 

\begin{lemma}
    Let $y = \min\{h(x_1), \dots, h(x_m)\}$. Then, 
    $E[y] = 1/(d+1)$.
\end{lemma}
\begin{proof}
    $$
    \begin{array}{rcll}
        E[y] & = & \int_{0}^{1} P(y \geq \lambda)\,d\lambda & \tab  \mathrm{Lemma\  \ref{lem:expected CDF}} \\
        & = & \int_{0}^{1} P( \forall x_i\ :\ h(x_i) \geq \lambda )\,d\lambda  & \\
        & = & \int_{0}^{1} (1-\lambda)^d\,d\lambda & \tab h\ \mathrm{is\ uniform} \\
        & = &-\frac{(1-\lambda)^{d+1}}{d+1}\ \bigg|^{1}_{0} & \\
        & = & \frac{1}{d+1} &
    \end{array}
    $$
\end{proof}

Unfortunately, in general $E[1/y] \neq 1/E[y]$ so it is not true that $E[\hat d] = E[1/y-1]  = d$. Technically, we say that $\hat d$ is a \emph{biased} estimator for $d$, because its expected value is not equal to the value we want to compute. On the other hand, if $y$ is very close to $1/(d+1)$, then intuitively also $\hat d = 1/y -1 $ will be very close to $d$. We will prove this intuition by studying the relative error of $y$ with respect to $1/(d+1)$, and then turn this into a relative error of $\hat d$ with respect to $d$.

\begin{lemma}
    Let $y = \min\{h(x_1), \dots, h(x_m)\}$. Then, 
    $Var[y] \leq 1/(d+1)^2$.
\end{lemma}
\begin{proof}
    We use the equality $Var[y] = E[y^2] - E[y]^2$. We know that $E[y]^2 = 1/(d+1)^2$. We compute $E[y^2]$ as follows:
    $$
    \begin{array}{rcl}
        E[y^2] & = & \int_{0}^{1} P(y^2 \geq \lambda)\,d\lambda \\
        &=& \int_{0}^{1} P(y \geq \sqrt\lambda)\,d\lambda\\
        & = & \int_{0}^{1} (1-\sqrt\lambda)^{d}\,d\lambda\\
    \end{array}
    $$
    We can solve the latter integral by the substitution $u = 1-\sqrt\lambda$. We have $\lambda = (1-u)^2$ and $\frac{d\lambda}{du} = d(1-u)^2/du = -2(1-u)$, so $d\lambda = -2(1-u)\,du$. 
    Also, note that $u=0$ for $\lambda=1$ and $u=1$ for $\lambda=0$ so the integral's interval switches. 
    By applying the substitution we obtain: 
    $$
    \begin{array}{rcl}
        E[y^2] & = & \int_{0}^{1} (1-\sqrt\lambda)^{d}\,d\lambda\\
        & = & \int_{1}^{0} -2(1-u)u^d \,du\\
        & = & -2 \left( \int_{1}^0 u^d\,du - \int_{1}^0 u^{d+1}\,du \right)\\
        & = & -2 \left( \frac{u^{d+1}}{d+1} \bigg|^{0}_{1} -  \frac{u^{d+2}}{d+2} \bigg|^{0}_{1} \right)\\
        & = & -2 \left( -\frac{1}{d+1} + \frac{1}{d+2} \right)\\
        & = & \frac{2}{d+1} - \frac{2}{d+2}
    \end{array}
    $$    
    To conclude:
      $$
        \begin{array}{rcl}
            Var[y] & = & E[y^2] - E[y]^2 \\
            & =  & \frac{2}{d+1} - \frac{2}{d+2} - \frac{1}{(d+1)^2} \\
            & = & \frac{2(d+2)-2(d+1)}{(d+1)(d+2)} - \frac{1}{(d+1)^2}\\
            & = & \frac{2}{(d+1)(d+2)} - \frac{1}{(d+1)^2}\\
            & \leq & \frac{2}{(d+1)^2} - \frac{1}{(d+1)^2}\\
            & = & \frac{1}{(d+1)^2}
        \end{array}
        $$    
\end{proof}

We run $s$ independent instances of algorithm \ref{alg:FM}, for a parameter $s\ge 1$ to be determined later, and take the average $y' = \frac{1}{s}\sum_{i=1}^s y_i$ of their internal registers $y_1, \dots, y_s$. Boosted Chebyshev (Lemma \ref{lem:boosted Cheb}) gives: 
$$
P\left(\bigg|y' - \frac{1}{d+1}\bigg| >  \frac{\epsilon}{d+1} \right) \leq \frac{1}{(d+1)^2} \cdot \frac{(d+1)^2}{s\epsilon^2} = \frac{1}{s\epsilon^2}
$$
We define $\hat d' = 1/y' -1$. How much does this value differ from the true value $d$? Note that the above inequality gives us $\frac{1-\epsilon}{d+1} \leq y' \leq \frac{1+\epsilon}{d+1}$ with probability at least $1-\frac{1}{s\epsilon^2}$. 
Let us assume $0 < \epsilon < 1/2$. In this range, the following inequality holds:  $\frac{1}{1-\epsilon} \leq 1+2\epsilon$. We have:

$$
\begin{array}{rcl}
    \frac{1}{y'} - 1 & \leq  & \frac{d+1}{1-\epsilon} -1 \\
     & \leq & (1+2\epsilon)(d+1)-1 \\
     & = & d + 2\epsilon d + 2\epsilon \\
     & \leq & d + 4\epsilon d \\
     & = & d(1+4\epsilon)
\end{array}
$$

Similarly, in the interval $0 < \epsilon < 1/2$  the following inequality holds:  $\frac{1}{1+\epsilon} \geq 1-\epsilon$. We have:

$$
\begin{array}{rcl}
    \frac{1}{y'} - 1 & \geq  & \frac{d+1}{1+\epsilon} -1 \\
    & \geq & (1-\epsilon)(d+1) - 1 \\
    & \geq & d(1-2\epsilon) \\
    & \geq & d(1-4\epsilon) 
\end{array}
$$

Thus, $d'$ is a $(1\pm 4\epsilon)$ approximation of $d$ with probability at least $1-1/(s\epsilon^2)$ for any $0 < \epsilon < 1/2$. To obtain a $(1\pm \epsilon)$-approximation, we simply adjust $\epsilon$ (i.e. turn to a relative error $\epsilon' = 4\epsilon$) and obtain that $d'$ is a $(1\pm \epsilon)$-approximation of $d$ with probability at least $1-16/(s\epsilon^2)$ for any $0 < \epsilon < 1$. We force the failure probability to be 1/3 and obtain $s = \frac{48}{\epsilon^2}$:
$$
\frac{16}{s\epsilon^2} = \frac{1}{3} \Leftrightarrow s = \frac{48}{\epsilon^2}
$$

Finally, we apply the median trick (Theorem \ref{thm:median trick}). We draw $t = 72\ln(1/\delta)$ i.i.d. realizations of $d'$ and return their median $\hat d$. We obtain:

\begin{theorem}
    For any desired relative error $0< \epsilon \leq 1$ and failure probability $0<\delta < 1$, by running $O(\epsilon^{-2}\log(1/\delta))$ parallel independent instances of Algorithm \ref{alg:FM} we can 
	count the number $d$ of distinct elements in the stream with relative error at most $\epsilon$, i.e. we can compute a value $\hat d$ such that:
	$$
	P(|\hat d - d| \ge \epsilon\cdot d) \leq \delta.
	$$ 
\end{theorem}

\subsection{Bottom-k algorithm}\label{sec: bottom-k}

Motivated by the fact that a uniform $h: [1,n] \rightarrow [0,1]$ takes too much space to be stored (see Section \ref{sec:hash functions}), 
in this section we present an algorithm that only requires a pairwise-independent hash function $h: [1,n] \rightarrow [0,1]$. See Section \ref{sec:hash reals} for a discussion on how to implement such a function in practice. 

The Bottom-k algorithm is presented as Algorithm \ref{alg:bottom k}. It is a generalization of Flajolet-Martin's algorithm: we keep the smallest $k$ distinct hash values $y_1 < y_2 < \dots < y_k$ seen in the stream so far, and finally return the estimate $k/y_k$. In our analysis we will show that, by choosing $k \in O(\epsilon^{-2})$, we obtain a $(1\pm\epsilon)$-approximation with constant probability. Finally, we will boost the success probability with a classic median trick.

\begin{algorithm}[t!h!]
	\caption{\texttt{Bottom-k}}
	\label{alg:bottom k}
		
	\SetKwInOut{Input}{input}
	\SetKwInOut{Output}{output}
	\SetSideCommentLeft
	\LinesNumbered
	
	\Input{A stream of integers $x_1, \dots, x_m$ and a desired relative error $\epsilon \leq 1/2$.}
	\Output{A $(1\pm \epsilon)$-approximation $\hat d$ of the number of distinct integers in the stream, with failure probability 1/3.}
	\BlankLine
	\BlankLine
	
	Choose $k =24/\epsilon^2$\;
	Initialize $(y_1, y_2, \dots, y_k) = (1,1,\dots, 1)$\;
	For each stream element $x$, update the $k$-tuple $(y_1, y_2, \dots, y_k)$ with the new hash $y = h(x)$ so that the $k$-tuple stores (in increasing order) the $k$ smallest hashes seen so far\;
    When the stream ends, return the estimate $\hat d = k/y_k$\;
	
\end{algorithm}

\subsubsection{Analysis}

Crucially, note that the proof of the following lemma will only require pairwise-independence of $h$.

\begin{lemma}
	For any $\epsilon \leq 1/2$, Algorithm \ref{alg:bottom k} outputs an estimator $\hat d$ such that
	$$
	P(|\hat d - d| > \epsilon\cdot d) \leq 1/3
	$$
\end{lemma}
\begin{proof}
	We first compute one side of the inequality: $P(\hat d > (1+\epsilon)d)$. 
	Let $z_1, \dots, z_d$ be the $d$ distinct integers in the stream, sorted arbitrarily.
	Let $X_i$ be an indicator 0/1 variable defined as $X_i = 1$ if and only if $h(z_i) < \frac{k}{d(1+\epsilon)}$. Observe that, if $\sum_{i=1}^d X_i \geq k$, then at the end of the stream the smallest $k$ hash values must satisfy $y_1 < y_2 < \dots < y_k < \frac{k}{d(1+\epsilon)}$. But then, the returned estimate is $\hat d = k/y_k > d(1+\epsilon)$. The converse is also true: if $\hat d = k/y_k > d(1+\epsilon)$, then $y_k < \frac{k}{d(1+\epsilon)}$, thus $y_1 < y_2 < \dots < y_k < \frac{k}{d(1+\epsilon)}$ and then $\sum_{i=1}^d X_i \geq k$. To summarize:
	\begin{center}
		$\sum_{i=1}^d X_i \geq k$ if and only if $\hat d > d(1+\epsilon)$
	\end{center}
	We can therefore reduce our problem to an analysis of the random variable $\sum_{i=1}^d X_i$. Since $h(z_i)$ is uniform in $[0,1]$, $P\left(h(z_i) < \frac{k}{d(1+\epsilon)}\right) =  \frac{k}{d(1+\epsilon)} = p$. $X_i$ is a Bernoullian R.V. with success probability $p$, so $E[X_i] = p = \frac{k}{d(1+\epsilon)}$. By linearity of expectation:
	$$
	E\left[\sum_{i=1}^d X_i\right] = \frac{k}{1+\epsilon}
	$$
	The variance of this R.V. is also easy to calculate. Note that, since the $h(z_i)$'s are pairwise-independent, then the $X_i$'s are pairwise-independent (in addition to being identically distributed) and we can apply Lemma \ref{lem:variance pairwise ind} to $Var\left[\sum_{i=1}^d X_i\right]$. 
	Recall also (Corollary after Lemma \ref{lem:Ber Exp Var}) that $Var[X_i] \leq E[X_i]$. We obtain: 
	$$
	Var\left[\sum_{i=1}^d X_i\right] = \sum_{i=1}^d Var[X_i]  \leq \sum_{i=1}^d E[X_i] = E\left[\sum_{i=1}^d X_i \right] = \frac{k}{1+\epsilon} \leq k
	$$
	We can now apply Chebyshev to $\sum_{i=1}^d X_i$: 
	$$
	P\left( \left| \sum_{i=1}^d X_i - \frac{k}{1+\epsilon} \right| >  \sqrt{6k}\right) \leq \frac{Var\left[ \sum_{i=1}^d X_i \right]}{(\sqrt{6k})^2} \leq \frac{k}{6k} = 1/6
	$$
	In particular, we can remove the absolute value:
	$$
	P\left( \sum_{i=1}^d X_i - \frac{k}{1+\epsilon} >  \sqrt{6k}\right)  \leq 1/6 \Leftrightarrow P\left( \sum_{i=1}^d X_i >  \sqrt{6k} + \frac{k}{1+\epsilon}\right)  \leq 1/6
	$$
	For which $k$ does it hold that $\sqrt{6k} + \frac{k}{1+\epsilon} \leq k$? a few manipulations give 
	$$
	k \geq \frac{6(1+\epsilon)^2}{\epsilon^2}
	$$
	Moreover: $\frac{6(1+\epsilon)^2}{\epsilon^2} \leq \frac{6(1+1)^2}{\epsilon^2} = \frac{24}{\epsilon^2}$.
	Therefore, if we choose $k = 24/\epsilon^2$ then $\sqrt{6k} + \frac{k}{1+\epsilon} \leq k$ and:
	$$
	P\left( \sum_{i=1}^d X_i >  k\right)  \leq P\left( \sum_{i=1}^d X_i >  \sqrt{6k} + \frac{k}{1+\epsilon}\right)  \leq 1/6
	$$
	We finally obtain $P(\hat d > (1+\epsilon)d) \leq 1/6$.
	
	We are now going to prove the symmetric inequality $P(\hat d < (1-\epsilon)d) \leq 1/6$. The proof will proceed similarly to the previous case. 
	Let $z_1, \dots, z_d$ be the $d$ distinct integers in the stream, sorted arbitrarily.
	Let $X_i$ be an indicator 0/1 variable defined as $X_i = 1$ if and only if $h(z_i) > \frac{k}{d(1-\epsilon)}$. 
	Observe that, if $\sum_{i=1}^d X_i > d-k$, then at the end of the stream the largest $(d-k)+1$ hash values must be larger than $\frac{k}{d(1-\epsilon)}$. In particular, the $k$-th smallest hash $y_k$ is also larger than this value: $y_k > \frac{k}{d(1-\epsilon)}$.
	But then, the returned estimate is $\hat d = k/y_k < d(1-\epsilon)$. 
	The converse is also true: if $\hat d = k/y_k < d(1-\epsilon)$, then $y_k > \frac{k}{d(1-\epsilon)}$. Since $y_k$ is the $k$-th smallest hash value, all the following (larger) $d-k$ hash values must also be larger than  $\frac{k}{d(1-\epsilon)}$, i.e. $\sum_{i=1}^d X_i > d-k$. To summarize:
	\begin{center}
		$\sum_{i=1}^d X_i > d-k$ if and only if $\hat d < d(1-\epsilon)$
	\end{center}
	Note that $X_i \sim Be\left( 1 - \frac{k}{d(1-\epsilon)} \right)$, so $E[X_i] =  1 - \frac{k}{d(1-\epsilon)}$.
	The expected value of $\sum_{i=1}^d X_i$ is:
	$$
	E\left[\sum_{i=1}^d X_i\right] = dE[X_i] = d- \frac{k}{1-\epsilon}
	$$
	Recall (Corollary after Lemma \ref{lem:Ber Exp Var}) that $Var[X_i] \leq 1- E[X_i]$. Recalling that we assume $\epsilon \leq 1/2$, we have:
	$$
	Var\left[\sum_{i=1}^d X_i\right] = dVar[X_i] \leq d(1 - E[X_i]) = d\cdot \frac{k}{d(1-\epsilon)} \leq 2k
	$$
	By Chebyshev:
	$$
	P\left(\left| \sum_{i=1}^d X_i - \left( d-\frac{k}{1-\epsilon} \right)\right| > \sqrt{12k} \right) \leq \frac{2k}{12k} = 1/6
	$$
	Removing the absolute value and re-arranging terms: 
	$$
	P\left(\sum_{i=1}^d X_i > \sqrt{12k} + d - \frac{k}{1-\epsilon}  \right) \leq  1/6
	$$
	For which values of $k$ do we have $\sqrt{12k} + d - \frac{k}{1-\epsilon} \leq d-k$? after a few manipulations, we get 
	$$
	k \geq \frac{12(1-\epsilon)^2}{\epsilon^2}
	$$
	Moreover, $\frac{12(1-\epsilon)^2}{\epsilon^2} \leq 12/\epsilon^2$. Therefore, choosing $k = 24/\epsilon^2 > 12/\epsilon^2$, we have $\sqrt{12k} + d - \frac{k}{1-\epsilon} \leq d-k$. Then:
	$$
	P\left(\sum_{i=1}^d X_i > d-k  \right) \leq P\left(\sum_{i=1}^d X_i > \sqrt{12k} + d - \frac{k}{1-\epsilon}  \right) \leq  1/6
	$$
	We  conclude that $P(\hat d < (1-\epsilon)d) \leq 1/6$. Combining this with $P(\hat d > (1+\epsilon)d) \leq 1/6$ by union bound, we finally obtain the two-sided bound $P(|\hat d-d| > \epsilon\cdot d) \leq 1/3$.
\end{proof}

We finally apply the median trick (Theorem \ref{thm:median trick}): we run $t = 72\ln(1/\delta)$ parallel instances of our algorithm, and return the median result. Recall that one hash value takes $O(\log n)$ bits to be stored, and that we keep in total $kt \in O(\log(1/\delta)/\epsilon^2)$ hash values.
Theorem \ref{thm:median trick} allows us to conclude:

\begin{theorem}\label{thm:bottom+}
	For any desired relative error $0< \epsilon \leq 1/2$ and failure probability $\delta >0$,
	Algorithm  \ref{alg:bottom k} uses $O\left(\frac{\log(1/\delta)}{\epsilon^2}\log n\right)$ bits and, with probability at least $1-\delta$, counts the number $d$ of distinct elements in the stream with relative error at most $\epsilon$, i.e. it returns a value $\bar d$ such that:
	$$
	P(|\bar d - d| \geq \epsilon\cdot d) \leq \delta
	$$
\end{theorem}

\begin{example}
	We want to estimate how many distinct IPv4 addresses (32 bits each) are visiting our website. Then, $n = 2^{32}$. Say we choose a function $h' : [1,n] \rightarrow [0,M]$ that is collision-free with probability at least $1-n^{-2}$. Then (see Section \ref{sec:perfect hashing}), $M = n^4$ and each hash value requires $\log_2 M = 4\log_2 n = 128$ bits (16 bytes) to be stored. We want \texttt{Bottom-k} to return an answer that is within $10\%$ of the correct answer ($\epsilon = 0.1$, $1/\epsilon^2 = 100$) with probability at least $1-10^{-5}$ ($\delta = 10^{-5}$, $\ln(1/\delta) < 12$). Then, replacing the constants that pop up from our analysis we obtain that \texttt{Bottom-k} uses at most around 32 MiB of RAM.
\end{example}

Note that to prove our main Theorem \ref{thm:bottom+} we used rather loose upper bounds. Still, \texttt{Bottom-k}'s memory usage of $<$ 32 MiB is rather limited if compared with the naive solutions. A bitvector of length $n$ would require 4 GiB of RAM. On the other hand, C++'s \texttt{std::set} uses 32 bytes per distinct element\footnote{\url{https://lemire.me/blog/2016/09/15/the-memory-usage-of-stl-containers-can-be-surprising/}}, so it is competitive with our analysis of \texttt{Bottom-k} only for $d$ up to $\approx 10\cdot 10^5$; this is clearly not sufficient in big-data scenarios such as a search engine: with over 5 billion searches per day\footnote{\url{https://review42.com/resources/google-statistics-and-facts}}, Google would need gigabytes of RAM to solve the problem with a \texttt{std::set} (even assuming as many as 10 searches per distinct user, and even using more space-efficient data structures).
Even better, practical optimized implementations of distinct-count algorithms solve the same problem within few kilobytes of memory\footnote{\url{https://en.wikipedia.org/wiki/HyperLogLog}} (see also \cite{flajolet2007hyperloglog}). 

\subsection{The LogLog family of count-distinct algorithms}

The original algorithm by Flajolet and Martin \cite{flajolet1983probabilistic} 
is a discrete version of the algorithm presented in Section \ref{sec:idealized FM} and it  
is based on the following idea: map each element to a $N$-bits hash $h'(x_i) \in [0,2^N)$, remember the maximum number $L = lz(h'(x_i))$ of leading zeros seen in any $h'(x_i)$, and finally return the estimate $2^L$.
For example, the number of leading zeros of $00111010$ is $lz(00111010) = 2$.
To see why this works, note that in a set of hash values of cardinality $2^L$, we expect to see \emph{one} hash $h'(x_i)$ prefixed by $L$ zeroes (the pattern $0^L$ does not have anything special: it is used just because $lz()$ can be computed very efficiently on modern architectures). 
It is not hard to see that our idealized algorithm presented in Subsection \ref{sec:idealized FM} is essentially equivalent to this variant: $h'(x)/(2^N-1) \approx h'(x)/2^N = y$ is approximately a uniform number in $[0,1]$ (\emph{approximately} because $h'$ is discrete). Let $L = lz(h'(x))$. Note that $2^{N-L}$ is a 2-approximation of $h'(x)$, so $y = h'(x)/2^N \approx 2^{-L}$. Then, the analysis of the algorithm of Section \ref{sec:idealized FM} tells us that $2^L = 1/y \approx 1/y -1$ is a good estimator for $d$, the number of distinct elements.
Observe that the algorithm only needs to keep in memory the maximum value $L = lz(h'(x_i))$, which, by the discussion above, is $L\in O(\log d)$. But then, $bitsize(L) = O(\log \log d)$: the algorithm uses space proportional to $O(\log \log d)$ bits (doubly exponentially less space than the naive solution using $O(d)$ words!).

Durand and Flajolet \cite{durand2003loglog} later refined this algorithm, giving it the name LogLog (from the space used by the algorithm). Their algorithm computes a  $(1 \pm 1.3/\sqrt k)$ approximation of the result with high probability, where $k$ is the number of independent instances of the algorithm. In the same paper they proposed a more accurate variant named SuperLogLog which, by removing $30\%$ of the largest $h(x_i)$'s, improves the approximation to $(1\pm 1.05/\sqrt k)$. In 2007, Flajolet, Fusy, Gandouet and Meunier \cite{flajolet2007hyperloglog} further improved the approximation to $(1\pm 1.04/\sqrt k)$. This algorithm is named HyperLogLog and uses an harmonic mean of the estimates.  Also Google has its own version: HyperLogLog++. See Heule, Nunkesser and Hall \cite{heule2013hyperloglog}.

\section{Estimating $F_1$ (probabilistic counting)}\label{sec:counting}

Consider the basic task of counting. In order to count up to $m$ we clearly need $\log m$ bits (logarithms are base 2). What if we allow for a 2-approximation, that is, we allow our answer to be off by at most a factor of two? then, it is easy to see that $\log\log m$ bits are sufficient: instead of storing our number $x \le m$, we store $y = \lceil \log x \rceil$. 
Then, $m \le 2^y \le 2m$ (2-approximation) and, since $y$ takes integer values between $0$ and $\log m$, it uses just $\log\log m$ bits. In general, we may fix a relative error $0 < \epsilon \leq 1$ and approximate $x$ with the smallest integer power of $(1+\epsilon)$ being larger than or equal to $x$: $(1+\epsilon)^y \geq x$, with $y = \lceil \log_{1+\epsilon} x \rceil = \left\lceil \frac{\log x}{\log(1+\epsilon)} \right\rceil$. Then, $y$ requires at most $\log \log m - \log\log(1+\epsilon) + 1 \leq \log\log m + O(\log(\epsilon^{-1}))$ bits to be stored 
(logarithms are base 2 unless otherwise specified; we used the bound $\log(1+\epsilon) \geq \frac{1}{1/\epsilon+1/2}$ and assumed $0 < \epsilon < 1$) and is a $(1 + \epsilon)$-approximation of $x$. 

\begin{example}
Suppose we allow for a $10\%$ relative error (i.e. $\epsilon = 0.1$). Then, we need approximately just $\log\log m + 3$ bits. What is the largest number we can store in 8 bits? Solving $\log\log m + 3 = 8$ we obtain $\log\log m = 5$, i.e. we can store a number as large as $2^{2^5} = 2^{32}$ with 10\% relative error.
\end{example}

While the above reasoning shows how one can store approximately a large counter in a small number of bits, it does not show how to \emph{increment} such counter: in real-case applications, we may wish to start from an approximate counter initialized to 0 and increase it one unit at a time (for example, every time a certain event occurs). In the next section we see that this goal can be achieved by using \emph{randomization}, incrementing the counter with some small probability.

\subsection{Morris' algorithm}\label{sec:Morris}

In 1978 Robert Morris\footnote{\url{https://en.wikipedia.org/wiki/Robert_Morris_(cryptographer)}}, a computer scientist working at Bell labs, studied the problem of counting large numbers using very small (8 bits) registers. Using just 8 bits, the largest number that can be stored (without skipping any positive integer) is clearly 255. However, as seen above, this is true only if we wish to store \emph{exact} counts; if we allow for some error, then the register can actually hold larger numbers. 
Algorithm \ref{alg:Morris} shows the basic algorithm devised by Morris, first described in \cite{morris1978counting}. The algorithm uses just one register ($X$) and its failure probability $1/(2\epsilon^2)$ makes it useful only for $\epsilon > 1/\sqrt{2}$, since otherwise the bound on the failure probability $1/(2\epsilon^2)$ is larger than one (thus not meaningful). 
To reduce the error probability, we will run several parallel versions of the algorithm (mean + median trick).

\begin{algorithm}[t!h!]
	\caption{\texttt{Morris}}
	\label{alg:Morris}
	
	\SetKwInOut{Input}{input}
	\SetKwInOut{Output}{output}
	\SetSideCommentLeft
	\LinesNumbered
	
	\Input{A stream of $m$ events and a desired relative error $0 < \epsilon \leq 1$}
	\Output{A $(1\pm\epsilon)$-approximation $\hat m$ of the number $m$ of events in the stream, with failure probability $1/(2\epsilon^2)$}
	\BlankLine
	\BlankLine
	
	Initialize a register $X = 0$\;
	For each event to be counted: increment $X$ with probability $2^{-X}$\;
	Finally, output $2^X-1$\;
	
\end{algorithm}

\paragraph{Analysis}

This analysis of the algorithm has been adapted from \cite{gregorygundersenblog,nelsonLectures,gilbertLectures}.
We first prove that the estimator $2^X-1$ returned by the algorithm is unbiased:

\begin{lemma}
$E[2^X-1] = m$
\end{lemma}
\begin{proof}
We proceed by induction on $m$. Let us denote with $X_i$ the register's content after event $i$. For $m=0$, we have $X_0=0$ and $E[2^{X_0}-1] = 0$ so we are done. Assume inductively that the claim holds for $m$, i.e. $E[2^{X_{m}}-1] = m$ (equivalently, $E[2^{X_{m}}] = m+1$). Then, applying the law of total expectation to the partition $\{X_m = i\}$ of the event space we have:
$$
\begin{array}{ccl}
E[2^{X_{m+1}}-1] & = & E[2^{X_{m+1}}]-1  \\
& = & \left(\sum_{i=0}^\infty P(X_m=i)\cdot E[2^{X_{m+1}}\ |\ X_m=i]\right) - 1 \\
\end{array}
$$
Observe that 
$$
\begin{array}{ccl}
    E[2^{X_{m+1}}\ |\ X_m=i] & = &  P(X_m=i \mathrm{\ is\ incremented})\cdot 2^{i+1}  + \\
    && P(X_m=i \mathrm{\ is\ not\ incremented})\cdot 2^{i}\\
    & = & 2^{-i}\cdot 2^{i+1} + (1-2^{-i})\cdot 2^i \\
    &= & 2^i + 1
\end{array}
$$

so:

$$
\begin{array}{ccl}
E[2^{X_{m+1}}-1] & = & \left(\sum_{i=0}^\infty P(X_m=i)\cdot (2^i + 1 )\right)  - 1\\
& = & \left(\sum_{i=0}^\infty P(X_m=i)\cdot 2^i + \sum_{i=0}^\infty P(X_m=i)\right)  - 1\\
& = & \sum_{i=0}^\infty P(X_m=i)\cdot 2^i \\
& = & E[2^{X_m}] = m+1\\
\end{array}
$$
\end{proof}

Having established that the expected value of our estimator is exactly the count $m$ that we wish to store, we only miss to establish how much a single realization of the estimator can differ from the expected value. We first compute the estimator's variance:

\begin{lemma}
$Var[2^{X}-1] \leq m^2/2$
\end{lemma}
\begin{proof}
\begin{equation}\label{eq:var Morris}
\begin{array}{ccl}
Var[2^{X}-1] & = & E[((2^{X}-1) - m)^2] \\
 & = & E[(2^{X}- (m+1))^2] \\
 & = & E[2^{2X}] - 2(m+1)E[2^X] + (m+1)^2 \\
  & = & E[2^{2X}] - (m+1)^2 \\
\end{array}
\end{equation}
Let $X_m$ denote the value of our register after seeing $m$ events. We now compute $E[2^{2X_m}]$ and plug it into Equation \ref{eq:var Morris}. 
$$
\begin{array}{ccl}
E[2^{2X_m}] & = & \sum_{i=0}^\infty P(X_m=i)\cdot 2^{2i} \\
 & = & \sum_{i=0}^\infty 2^{2i} \cdot \left( P(X_{m-1}=i-1)\cdot 2^{-(i-1)}  + P(X_{m-1}=i)\cdot (1-2^{-i})  \right)\\
& = & \sum_{i=0}^\infty 2^{i+1}\cdot P(X_{m-1}=i-1) + \sum_{i=0}^\infty 2^{2i}\cdot P(X_{m-1}=i) - \sum_{i=0}^\infty 2^{i}\cdot P(X_{m-1}=i)\\
& = & \sum_{i=0}^\infty 4\cdot 2^{i-1}\cdot P(X_{m-1}=i-1) + \sum_{i=0}^\infty 2^{2i}\cdot P(X_{m-1}=i) - \sum_{i=0}^\infty 2^{i}\cdot P(X_{m-1}=i)\\
& = & 4\cdot E[ 2^{X_{m-1}}] + E[2^{2X_{m-1}}] - E[2^{X_{m-1}}]\\
& = & 3\cdot E[ 2^{X_{m-1}}] + E[2^{2X_{m-1}}]\\
& = & 3m + E[2^{2X_{m-1}}]\\
\end{array}
$$
The above yields a recursive definition: denoting $E_m = E[2^{2X_m}]$, we have $E_m = 3m + E_{m-1}$. Since $E_0 = E[2^{2\cdot 0}] = 1$, this series expands to $E_m = \sum_{i=1}^m 3i + 1 = \frac{3m(m+1)}{2}+1$. We can finally plug this into Equation \ref{eq:var Morris} and obtain
$$
\begin{array}{ccl}
Var[2^{X}-1] & = & E[2^{2X}] - (m+1)^2 \\
& = & \frac{3m(m+1)}{2}+1 - (m+1)^2 \\
& = & (m^2 -m)/2 \leq m^2/2 \\
\end{array}
$$
\end{proof}

Note that $Var[2^X-1] = O(E[X]^2)$. We can therefore apply the combined mean+median trick of Corollary \ref{cor:mean+median trick} and obtain (note that register $X$ in  Algorithm \ref{alg:Morris} stores a number whose expected value is $\log m$, so it requires on expectation $\log\log m$ bits):


\begin{theorem}
For any desired relative error $\epsilon >0$ and failure probability $\delta > 0$, by running $O(\epsilon^{-2} \log(1/\delta))$ parallel independent instances of Morris' algorithm we use in total 
$O\left(\epsilon^{-2} \log(1/\delta) \log\log m\right)$ bits of space on expectation and, with probability at least $1-\delta$, we  count numbers up to $m$ with relative error at most $\epsilon$, i.e.  return a value (random variable) $\hat Y$ such that:
$$
P( |\hat Y - m| \geq m\cdot \epsilon ) \leq \delta
$$
\end{theorem}

\section{Estimating $F_2$}\label{sec:F2}

We now present a very simple and elegant algorithm to estimate $F_2 = \sum_{y\in [1,n]} f_y^2$ (second-order frequency moment), then discuss some applications of the corresponding sketch: dimensionality reduction in the Euclidean space and sketches for relational algebra. 

\subsection{The tug-of-war sketch}

This technique is due to Alon, Matias, and Szegedy \cite{alon1999space} and it is also known as the \emph{AMS sketch}.
We use a $4$-wise independent hash function $f:[1,n] \rightarrow \{-1,+1\}$. As we show next, the  extremely simple (and elegant) Algorithm  \ref{alg:tow} returns an unbiased estimator for $F_2$.

\begin{algorithm}[t!h!]
	\caption{\texttt{Tug of war}}
	\label{alg:tow}
	
	\SetKwInOut{Input}{input}
	\SetKwInOut{Output}{output}
	\SetSideCommentLeft
	\LinesNumbered
	
	\Input{A stream of integers $x_1, \dots, x_m$.}
	\Output{An estimate of the second-order frequency moment $F_2$.}
	\BlankLine
	\BlankLine
	
	Initialize $z = 0$\;
	For each stream element $x$, update $z \leftarrow z + h(x)$\;
	When the stream ends, return the estimate $z^2$\;
	
\end{algorithm}

In other words, after seeing the $j$-th stream element we maintain a sketch $z = \sum_{i=1}^j h(x_i)$. 
Observe that the sketch $z$ is easy to update upon additive frequency updates (i.e. $f_y \gets f_y + \Delta$, for any $\Delta \in \mathbb N$ and $y\in[1,n]$). It is also easy to combine the sketches of two streams and obtain the sketch of  the concatenation of the two streams (solve this as an exercise). 

It is not hard to see that the estimator $z^2$ is unbiased: 
\newline
\begin{lemma}\label{lem: expected tow}
    $E[z^2] = F_2.$ 
\end{lemma}
\begin{proof}
Observe that $\sum_{i=1}^m h(x_i) = \sum_{y=1}^n f_y\cdot h(y)$. Then:
$$
\arraycolsep=1.4pt\def\arraystretch{2.2}
\begin{array}{rcl}
   E[z^2]   & = & E\left[ \left( \sum_{i=1}^m h(x_i) \right)^2 \right] \\
            & = & E\left[ \left(\sum_{y=1}^n f_y\cdot h(y)\right)^2 \right] \\
            & = & E\left[  \sum_{y=1}^n f_y^2\cdot h(y)^2 + \sum_{a\neq b\in [1,n]} f_af_b\cdot h(a) h(b) \right] \\
            & = & E\left[  \sum_{y=1}^n f_y^2\cdot h(y)^2\right] + \sum_{a\neq b\in [1,n]} f_af_b E[ h(a) h(b)]
\end{array}
$$
Note that $h(y)^2 = 1$ for all $y\in [1,n]$, so $E\left[  \sum_{y=1}^n f_y^2\cdot h(y)^2  \right] = E\left[  \sum_{y=1}^n f_y^2  \right] = \sum_{y=1}^n f_y^2 = F_2$. We are left to evaluate the second term, $\sum_{a\neq b\in [1,n]} f_af_b E[ h(a) h(b)]$. 
Since $h$ is $4$-wise independent, it is in particular pairwise independent. Then, $h(a)$ and $h(b)$ are independent random variables whenever $a\neq b$. Notice also that $E[h(a)]= \frac{1}{2}\cdot 1 + \frac{1}{2}\cdot (-1) = 0 $. Then,  
from Lemma \ref{lem:expectation product} we obtain $E[ h(a) h(b)] = E[h(a)]E[h(b)]=0$ whenever $a\neq b$. It follows that $\sum_{a\neq b\in [1,n]} f_af_b E[ h(a) h(b)] = 0$, so the thesis follows. 
\end{proof}

\begin{remark}
    Observe that, since $E[h(x)]=0$ for any $x\in [1,n]$, then by linearity of expectation it holds $E[z]=0$: the Tug-of-war algorithm seems to do almost nothing (it computes a random variable with expectation 0), but it actually estimates second order frequency moments.  
\end{remark}

In order to obtain an $(1\pm\epsilon)$-approximation with arbitrarily large probability, our next goal is to prove that the variance of $z^2$ is proportional to $O(F^2_2)$, i.e. the square of its expected value. This will allow us to apply the combined mean+median trick of Corollary \ref{cor:mean+median trick}.

\begin{lemma}\label{lem:variance tow}
    $Var[z^2] \le 3 E[z^2]^2 = 3F_2^2$.
\end{lemma}
\begin{proof}
First, observe that $Var[z^2]  = E[z^4] - E[z^2]^2 \le E[z^4]$. We evaluate $E[z^4]$ similarly to how we proceeded in Lemma \ref{lem: expected tow}, this time taking the fourth power of $z$ instead of its second power:
$$
\arraycolsep=1.4pt\def\arraystretch{2.2}
\begin{array}{rcl}
   E[z^4]   & = & E\left[ \left( \sum_{i=1}^m h(x_i) \right)^4 \right] \\
            & = & E\left[ \left(\sum_{y=1}^n f_y\cdot h(y)\right)^4 \right] \\
\end{array}
$$
We distinguish three types of terms in the expansion of the power $E\left[\left(\sum_{y=1}^n f_y\cdot h(y)\right)^4\right]$: 

\begin{enumerate}
    \item Terms of the form $E[abcd\cdot h(a)h(b)h(c)h(d)] = abcd\cdot E[h(a)h(b)h(c)h(d)]$ such that at least one of the four integers $a,b,c,d$ is distinct from the other three. Without loss of generality, let $a$ be this integer. Then, since $h$ is four-wise independent, $h(a)$ is independent from $h(b)h(c)h(d)$ and therefore $E[h(a)h(b)h(c)h(d)] = E[h(a)]E[h(b)h(c)h(d)] = 0$. We conclude that such terms do not contribute to $E\left[\left(\sum_{y=1}^n f_y\cdot h(y)\right)^4\right]$.
    
    \item Terms of the form $f_a^4\cdot h(a)^4 = f_a^4$. These terms contribute to a total amount of $\sum_{a\in [1,n]}f_a^4 = F_4$.
    
    \item Terms of the form $f_a^2f_b^2 \cdot h(a)^2\cdot h(b)^2 = f_a^2f_b^2$, for $a\neq b$. For fixed $a < b$, there are in total $\binom{4}{2}=6$ such terms. As a result, these terms contribute to a total amount of $6\sum_{a<b\in [1,n]}f_a^2f_b^2$.
\end{enumerate}

We conclude that 

$$
Var[z^4] \le  E[z^4] = E\left[\left(\sum_{y=1}^n f_y\cdot h(y)\right)^4\right] = F_4 + 6\sum_{a<b\in [1,n]}f_a^2f_b^2.
$$
Rewrite 
$$
\arraycolsep=1.4pt\def\arraystretch{2.2}
\begin{array}{rcl}
   F_2^2 & = & \left( \sum_{a\in [1,n]} f_a^2 \right)^2 \\
   & = & \sum_{a\in [1,n]} f_a^4 + 2\sum_{a<b\in [1,n]} f_a^2f_b^2 \\
   & = & F_4 + 2\sum_{a<b\in [1,n]} f_a^2f_b^2
\end{array}
$$
Re-arranging terms: $2\sum_{a<b\in [1,n]} f_a^2f_b^2 = F_2^2 - F_4$.

Putting together the above two equations, we derive that 
$$
\arraycolsep=1.4pt\def\arraystretch{2.2}
\begin{array}{rcl}
    Var[z^2] & \le & E[z^4] \\
    & = & E\left[\left(\sum_{y=1}^n f_y\cdot h(y)\right)^4\right]\\
   &= & F_4 + 6\sum_{a<b\in [1,n]}f_a^2f_b^2 \\
   & = & F_4 + 3 (F_2^2-F_4) \\
   & = & 3F_2^2 - 2F_4 \le 3F_2^2 \\
\end{array}
$$
which concludes the proof.
\end{proof}

Since we proved that the estimator $z^2$ is unbiased for $F_2$ and that $Var[z^2] \in O(E[z^2]^2)$, the combined mean+median trick of Corollary \ref{cor:mean+median trick} allows us to conclude: 

\begin{corollary}
    For any pre-defined parameters $\epsilon, \delta >0$, the tug-of-war algorithm uses $O(\epsilon^{-2}\log(1/\delta))$ words of space and estimates $F_2$ with relative error $\epsilon$ and failure probability $\delta$.
\end{corollary}

\subsection{Dimensionality reduction}\label{sec: dimensionality reduction}

We now show that the idea behind the tug-of-war sketch solves another very important problem in massive data processing: dimensionality reduction under the Euclidean distance. The main idea is to simply view the frequencies $f_y$ of the previous subsections as a size-$n$ vector. Note that the calculations of the previous subsections did not rely on the fact that the frequencies $f_y$ are integers. We will now switch to real numbers and show that this slight change of perspective allows us to perform dimensionality reduction.  
The reason why this technique is called \emph{dimensionality reduction} and not just ``distance approximation'' is that it is much powerful than simple distance approximation: our sketch will be a \emph{linear map} mapping $\mathbb R^n$ to $\mathbb R^{c}$ for some small constant $c$ (depending on the desired relative error) preserving the Euclidean distance up to some relative error $\epsilon$.

Let $a = (a_1, \dots, a_n) \in \mathbb R^n$ and $b = (b_1, \dots, b_n) \in \mathbb R^n$ be two real-valued vectors of dimension $n$. As in the previous subsections, let $h: [1,n] \rightarrow \{-1,+1\}$ be a $4$-wise independent hash function. Our final goal will be to devise a sketch allowing us to estimate 
$$
||a-b||_2 = \ell_2 = \sqrt{\sum_{i=1}^n(a_i-b_i)^2}
$$
with relative error $\epsilon$ and failure probability $\delta$.
In order to achieve this result, we first show how to use tug-of-war to approximate the square of the Euclidean distance $\ell_2^2 = \sum_{i=1}^n(a_i-b_i)^2$, and then argue that this is enough to approximate the Euclidean distance itself.

The sketch of vector $a$ (the same for $b$) is, again, the tug-of-war sketch: 

$$
z_a = \sum_{i=1}^n a_i\cdot h(i)
$$

Observe that, if we treat $z_a$ and $z_b$ as 1-dimensional vectors, then on expectation $\ell_2^2$ is preserved under this transformation. In fact, denote $z = z_a - z_b$. This is precisely the sketch of the vector $a-b$ so, as we showed in the previous subsections, $E[z^2] = \sum_{i=1}^n (a_i-b_i)^2 = \ell_2^2$. This result, however, holds just on expectation. What if we want a $(1\pm\epsilon)$-approximation with constant probability?

Observe that computing the distance between $a$ and $b$ is equivalent to computing the distance of $a-b$ from the origin, so from now on we 
work with just one vector $a$ and focus on computing $\ell_2^2 = ||a||_2^2 =  \sum_{i=1}^n a_i^2$ (i.e. the square of the Euclidean distance between $a$ and the origin).

Let's use the mean trick (Corollary \ref{cor:mean trick}). Instead of using just one hash function $h$, we use $s$ independent such hash functions $h_1, \dots, h_s$ and compute $s$ tug-of-war sketches:
$$
z_{a,j} = \sum_{i=1}^n a_i\cdot h_j(i), \ \ \ for\ j=1, \dots, s.
$$

Let $\hat \ell_2^2 = \frac{1}{s}\sum_{j=1}^s z_{a,j}^2$. 
Recall (Lemma \ref{lem: expected tow}) that $E[z_{a,j}^2] = 
\sum_{i=1}^n a_i^2 = \ell_2^2$ for all $j=1,\dots, s$.
By Corollary  \ref{cor:mean trick} (mean trick), if we choose $s =  \lceil 9\epsilon^{-2} \rceil \ge \frac{3 Var[z_{a,j}^2]}{E[z_{a,j}^2]^2\epsilon^2}$ (the inequality comes from Lemma \ref{lem:variance tow}) then:

\begin{equation}   \label{eq:distance preserving 1}
P\left( |\hat \ell_2^2 - \ell_2^2| \ge \epsilon \ell_2^2 \right) \le 1/3
\end{equation}

In other words: $\hat \ell_2^2$ is a $(1\pm \epsilon)$-approximation of $\ell_2^2$ with constant probability. 

We can express the above operations in the language of linear algebra. 
Consider the following random matrix (linear map): 

$$
H = \frac{1}{\sqrt s}\begin{bmatrix}
h_1(1) & h_1(2) & \dots & h_1(n-1) & h_1(n)\\
h_2(1) & h_2(2) & \dots & h_2(n-1) & h_2(n)\\
&&\dots\\
h_s(1) & h_s(2) & \dots & h_s(n-1) & h_s(n)\\
\end{bmatrix}
$$

Note that $H$ is simply a matrix of dimension $s\times n$ filled with uniform values from the set $\{-1/\sqrt s, 1/\sqrt s\}$.
We show that $H$ is our distance-preserving map. Note that $H$ maps vector $a$ to:

$$
H\cdot a = \frac{1}{\sqrt s} \begin{bmatrix}
z_{a,1}\\
z_{a,2}\\
\dots\\
z_{a,s}
\end{bmatrix}  
$$

But then it holds that $||H\cdot a||_2^2 = \frac{1}{s}\sum_{j=1}^s z_{a,j}^2 = \hat \ell_2^2$. So, Equation \ref{eq:distance preserving 1} can be rewritten as:

\begin{equation}   \label{eq:distance preserving 2}
P\left( \Big|\ ||H\cdot a||_2^2 - \ell_2^2\ \Big| \ge \epsilon \ell_2^2 \right) \le 1/3
\end{equation}

in other words, $H:\mathbb R^n \mapsto \mathbb R^{s}$, with $s\in \Theta(\epsilon^{-2})$, is  a linear map that preserves $\ell_2^2$ with relative error $\epsilon$ and success probability at least $2/3$. Finally note that $||H\cdot a||_2^2 = D \cdot \ell_2^2$ is equivalent to $||H\cdot a||_2 = \sqrt D \ell_2$, where $D\in [1-\epsilon, 1+\epsilon ]$ is the distortion factor. Since
$$
[\sqrt{1-\epsilon}, \sqrt{1+\epsilon} ] \subseteq 
[1-\epsilon, 1+\epsilon ]
$$

it follows that $H$ preserves the Euclidean distance as well (with relative error $\epsilon$):

\begin{equation}   \label{eq:distance preserving 3}
P\left( \Big|\ ||H\cdot a||_2 - \ell_2\ \Big| \ge \epsilon \ell_2 \right) \le 1/3
\end{equation}

Since $H$ is a linear map, $||H\cdot a - H\cdot b||_2 = || H(a-b) ||_2$ so we conclude:

\begin{theorem}[Dimensionality reduction]
  Let $a,b \in \mathbb R^n$. Let
  $s = \lceil 9\epsilon^{-2} \rceil$, for any desired error rate $\epsilon >0$. Let
  $H\in \mathbb R^{s\times n}$ be a matrix filled with uniform values from the set $\{-1/\sqrt{s},1/\sqrt{s}\}$. Then, with probability at least $2/3$ the following holds: 
  $$
  || H\cdot a - H\cdot b ||_2 \in (1\pm \epsilon)|| a-b ||_2.
  $$
  In other words: $H:\mathbb R^n \mapsto \mathbb R^{\Theta(\epsilon^{-2})}$ is  a linear map that preserves the Euclidean distance up to relative error $\epsilon$ with probability at least $2/3$.
\end{theorem}

To boost the probability of success, we can apply the median trick (Lemma \ref{thm:median trick}): we compute $t \in \Theta(\log(1/\delta))$ independent maps $H_1, \dots, H_t$, obtaining $t$ estimates for the distance between vectors $a$ and $b$. Finally, we return the median of the $t$ estimated distances. By Lemma \ref{thm:median trick}, the result is a multiplicative $(1\pm\epsilon)$-approximation of $||a-b||_2$ with probability at least $1-\delta$. 

Note that the median is not a linear operator, so the median trick does not give us a linear transformation. 
The Johnson-Lindenstrauss transform \cite{johnson1986extensions} removes the need for computing the median by filling $H$ with values following a normal distribution (rather than with uniform values from $\{-1/\sqrt s,1/\sqrt s\}$). This technique provides a linear transformation 
$H:\mathbb R^n \mapsto \mathbb R^{\Theta(\epsilon^{-2}\log(1/\delta))}$ that preserves the Euclidean distance up to relative error $\epsilon$ with probability at least $1-\delta$.

\subsection{Sketches for relational algebra (SQL)}

Relational algebra is a mathematical framework for manipulating data stored in the forms of relations (also called tables). SQL is the most influential implementation of relational algebra, being employed in most commercial implementations of relational databases.
While until now we studied sketches for conceptually-simple operations such as counting distinct elements, in the following we show that the tug-of-war sketch can be used to estimate the answer to simple SQL queries whose final result is an aggregation \cite{alon1999tracking, dobra2002processing}. This result is very powerful: imagine collapsing an entire SQL database (gigabytes) into a tiny sketch (kilobytes) allowing to compute approximate answers to SQL queries!

\paragraph{Aggregations on one relation}
Let $R$ be a relation containing tuples from  $[n]^r$. Let us call $\rho_1, \dots, \rho_r$ the $r$ attributes of $R$.
We start from a very simple case of aggregation: the sum under a \texttt{SELECT} operation on one relation. For example, let $R$ be the relation shown in Table \ref{tab:SQL1}, with $r=4$ attributes.

\begin{table}[th!]
    \centering
    \begin{tabular}{|c|c|c|c|}
        \hline
        matriculation & course\_code & module & evaluation\\\hline
        12345 & 664 & 1 & 28 \\\hline
        23315 & 664 & 2 & 18 \\\hline
        10123 & 622 & 1 & 25\\\hline
        12345 & 622 & 1 & 25\\\hline
        12345 & 622 & 2 & 27\\\hline   
        10123 & 664 & 1 & 25\\\hline    
    \end{tabular}
    \caption{A relational table with four attributes.}
    \label{tab:SQL1}
\end{table}

Then the following SQL query 

\begin{verbatim}
    SELECT sum(evaluation)
    FROM R
    WHERE matriculation = 12345
\end{verbatim}

will return $28+25+27 = 80$ (the sum of all evaluations of student 12345). 
More in general, we want to design a sketch able to estimate the answer to a query of the form:

\begin{lstlisting}
    -- query Q1
    SELECT sum($\rho_i$)
    FROM R
    WHERE $\rho_j$ = x
\end{lstlisting}

In other words: an aggregation (sum) over a one-relation query with one equality condition. The query selects all tuples where attribute $\rho_j$ is equal to $x$, and returns the sum $\sum\limits_{(x_1, \dots, x_r)\in R\ :\ x_j=x} x_i$ of their attribute $\rho_i$. We first build the following tug-of-war sketch:

$$
Q1(R) = \sum_{(x_1, \dots, x_r)\in R} h(x_j)\cdot x_i
$$

where $h : [1,n] \rightarrow \{-1,1\}$ is a four-wise independent hash function. 
Note that $Q1(R)$ collapses an entire relation (possibly, gigabytes) into \emph{one} integer!
Our estimator is:

$$
Q1(R,x) = Q1(R)\cdot h(x)
$$

We prove that $Q1(R,x)$ is an unbiased estimator for SQL query \texttt{Q1}:

\begin{theorem}\label{thm: Tow select}
    $$
    E[Q1(R,x)] = \sum_{(x_1, \dots, x_r)\in R\ :\ x_j=x} x_i
    $$ 
\end{theorem}
\begin{proof}
    $$
    \begin{array}{rcl}
        E[Q1(R,x)] & = & E[Q1(R)\cdot h(x)] \\
        & = & E\left[ \sum\limits_{(x_1, \dots, x_r)\in R} h(x_j)h(x)\cdot x_i \right]\\
        & = & \sum\limits_{(x_1, \dots, x_r)\in R} E[h(x_j)h(x)]\cdot x_i\\
        & = & \sum\limits_{(x_1, \dots, x_r)\in R\ :\ x_j=x} E[h(x)^2]\cdot x_i + \sum\limits_{(x_1, \dots, x_r)\in R\ :\ x_j\neq x} E[h(x_j)]E[h(x)]\cdot x_i \\
        & = & \sum\limits_{(x_1, \dots, x_r)\in R\ :\ x_j=x} x_i
    \end{array}
    $$
    Where we used the fact that $E[h(x_j)h(x)] =  E[h(x_j)]E[h(x)] = 0\cdot 0 = 0$ when $x_j\neq x$ (by independence of $h$) and $E[h(x)^2] = E[1] = 1$.
\end{proof}

The precision of the estimate can be improved by repetition (as done in previous sections). 

\medskip
\begin{Exercise}\noindent
    Show an estimator (not necessarily unbiased) allowing to replace \texttt{sum()} with \texttt{average()} in query \texttt{Q1}.
    \medskip
\end{Exercise}
\begin{Exercise}\noindent
    Generalize the solution to support the conjunction of multiple equality conditions in the \texttt{WHERE} clause. 
    \medskip
\end{Exercise}

\paragraph{Aggregations on two relations (equi-joins)}
Let $R, P$ be two relations containing tuples from  $[n]^{r}$ and  $[n]^{p}$, respectively. Let us call $\rho_1, \dots, \rho_r$ the $r$ attributes of $R$ and $\pi_1, \dots, \pi_p$ the $p$ attributes of $P$.
For example, let $R$ be the relation of Table \ref{tab:SQL1} and let $P$ be the relation shown in Table \ref{tab:SQL2}, with $p=2$ attributes.

\begin{table}[th!]
    \centering
    \begin{tabular}{|c|c|}
        \hline
        course\_code & active \\\hline
        664 & 0 \\\hline
        622 & 1 \\\hline
        123 & 1 \\\hline
    \end{tabular}
    \caption{A relational table with two attributes.}
    \label{tab:SQL2}
\end{table}

Then the following SQL query 

\begin{verbatim}
    SELECT sum(evaluation)
    FROM R, P
    WHERE matriculation = 12345 AND active = 1 AND R.course_code = P.course_code
\end{verbatim}

will return $25+27 = 52$: the sum of evaluations of student 12345, only for active courses.
More in general, we want to design a sketch able to estimate the answer to a query of the form:

\begin{lstlisting}
    -- query Q2
    SELECT sum($\rho_i$)
    FROM R,P
    WHERE $\rho_j$ = x AND $\pi_k$ = y AND $\rho_q$ = $\pi_t$
\end{lstlisting}

In other words: an aggregation (sum) over a join query with one equi-join condition and two other equality conditions (generalizing to any number of join and equality conditions is then immediate). 

For simplicity, we now assume that  $h: [1,n]\rightarrow \{-1,1\}$ is fully-independent (exercise: after reading this section, try to put an upper-bound to the amount of independence needed for $h$). 
Assume moreover that the domains of the attributes of $P$ and $R$ are disjoint, unless two attributes have the same name (for example, \texttt{course\_code} is an attribute of both $R$ and $P$ in our example, so the intersection of those two columns will possibly be nonempty). This is true in our example: it will never happen, for example, that a matriculation number is equal to a course code. 

We build one tug-of-war sketch for table $R$ as follows:

$$
Q2(R) = \sum_{(x_1, \dots, x_r)\in R} h(x_j)h(x_q)\cdot x_i
$$

in the above sketch, $h(x_j)$ is included for condition $\rho_j = x$, $h(x_q)$ is included for the equi-join condition $\rho_q = \pi_t$, and $x_i$ is included for the final sum() aggregate.

Similarly, the tug-of-war sketch for table $P$ is:

$$
Q2(P) = \sum_{(x_1', \dots, x_p')\in P} h(x_k')\cdot h(x_t')
$$

In the above sketch, $h(x_k')$ is included for condition $\pi_k = y$, and $h(x_t')$ is included for the equi-join condition $\rho_q = \pi_t$.

Our estimator for query \texttt{Q2}  is:

$$
Q2(R,P,x,y) = Q2(R)\cdot Q2(P) \cdot h(x) \cdot h(y)
$$

Now we show that the estimator is unbiased for query \texttt{Q2}:

\begin{theorem}
    $$
    E[Q2(R,P,x,y)] = \sum_{(x_1, \dots, x_r)\in R,\ (x'_1, \dots, x'_p) \in P :\ x_j=x\ \wedge\ x'_k=y\ \wedge\ x_q = x'_t} x_i 
    $$
\end{theorem}
\begin{proof}
    The proof works as in Theorem \ref{thm: Tow select}: the idea is that the terms that do not satisfy the \texttt{WHERE} condition will disappear on expectation. We expand the definition of $Q2(R,P,x,y)$:

    \def\arraystretch{3}
    $$
    \begin{array}{rcl}
        E[Q2(R,P,x,y)] & = & E\left[ \left(\sum\limits_{(x_1, \dots, x_r)\in R} h(x_j)h(x_q)\cdot x_i\right)\cdot \left(\sum\limits_{(x'_1, \dots, x'_p)\in P} h(x'_k)h(x'_t)\right)h(x)h(y)  \right] \\
        & = & E\left[ \sum\limits_{(x_1, \dots, x_r)\in R, (x'_1, \dots, x'_p)\in P} h(x_j)h(x_q)\cdot x_i\cdot h(x'_k)h(x'_t)h(x)h(y)  \right]
    \end{array}
    $$
    Moving the expectation inside the summation and re-grouping terms, this becomes:
    $$
    \sum\limits_{(x_1, \dots, x_r)\in R, (x'_1, \dots, x'_p)\in P}   
    E\left[\Big(h(x_j)h(x)\Big)\cdot \Big(h(x'_k)h(y)\Big)\cdot \Big(h(x_q) h(x'_t) \Big) \cdot x_i \right]
    $$
    Recall that we are assuming that the domains of attributes with different names are disjoint. Then, the expectation of the products in the equation above can be split into a product of three expectations (since the three components are independent random variables and we assume $h$ to be fully-independent):

    $$
    \sum\limits_{(x_1, \dots, x_r)\in R, (x'_1, \dots, x'_p)\in P}   
    E\left[h(x_j)h(x)\right] \cdot E\left[h(x'_k)h(y)\right]\cdot E\left[h(x_q) h(x'_t) \right] \cdot x_i
    $$
    
    In the quantity above, observe that only terms such that $x_j=x$, $x'_k=y$, and $x_q = x'_t$  survive in the summation, becoming equal to $x_i$. For the other terms, at least one of the three expectations goes to zero. This proves the claim.
\end{proof}

While we proved that our estimators for SQL queries are unbiased, we did not try to compute their variance in order to derive guarantees on the approximation ratio. Dobra et al. \cite{dobra2002processing} prove that the variance grows exponentially with $2N$, where $N$ is the number of equality conditions in the \texttt{WHERE} clause (in query \texttt{Q2}, $N=3$). Still, these estimators can be used in practice and produce accurate results with a sufficiently large number of repetitions.

\chapter{Exercises}

\setcounter{Exercise}{0}

\medskip \begin{Exercise} \noindent
You are given a string $S\in \Sigma^n$ containing $r$ equal-letter runs. For example: if $S = AAABBBBCCAAA$, then $r=4$ (the four runs are $AAA$, $BBBB$, $CC$, and $AAA$). Design a data structure using $O(r)$ words of space able to answer \emph{rank} and \emph{select} queries on $S$ in $O(\log n)$ time. Can you improve the space and query time of the data structure?
\end{Exercise}

\medskip \begin{Exercise} \noindent
Design a compressed data structure 
    implementing the  following predictive text model. The data structure should store a string $S\in \Sigma^n$ in high-order entropy-compressed space and should be able to answer quickly the following query: given a substring $T$ of $S$, what is the character that follows $T$ in $S$ more frequently? For example, if $Ta$ appears 50 times in $S$, $Tb$ appears 10 times in $S$, and $Tc$ appears 123 times in $S$, then the answer is $c$.
\end{Exercise}

\medskip \begin{Exercise} \noindent
Prove that you cannot encode (lossless) in $m$ bits a set containing $m$ integers.
\end{Exercise}

\medskip \begin{Exercise} \noindent
    We want to implement a SPAM filter for tweets formed by at most 100 words (in English). The filter should maintain a list $L$ of $10^6$ blacklisted phrases, where each phrase is formed by a minimum of one to a maximum of five words. For example, one of the phrases in $L$ could be:
    \begin{center}
    ``you have won 1000 EUR''
    \end{center}
    When concatenated together, the phrases of $L$ form an ASCII text of length $20$ MiB with second-order empirical entropy $H_2=0.2$.

    Given a tweet X, we want to put it in the SPAM folder if and only if it contains at least one sentence belonging to $L$. 

    We tolerate a $1\%$ probability that a tweet is incorrectly flagged as SPAM.
    Show how to solve the problem by using the least possible amount of space.  What if we want this probability to be $0\%$? How much space do your solutions use?
\end{Exercise}

\newpage
\medskip \begin{Exercise} \noindent
You are given two Bloom filters for sets $S_1$ and $S_2$ (the sets are unknown: you know only their filters). 
The two Bloom filters have been built with the same parameters ($M$, $k$) and the same hash functions. 
Show how to compute the Bloom filter for $S_1 \cup S_2$. Solve the same exercise with Counting Bloom filters.
\end{Exercise}

\medskip \begin{Exercise} \noindent
    The ``Algorithms and learning over massive data'' (ALMD) 12-ECTS course is formed by two modules: AMD and LMD. Let $[1,n]$ be the students enrolled in ALMD. Let moreover $A \subseteq [1,n]$ and $L \subseteq [1,n]$ be the sets of students who have passed the AMS and LMS modules, respectively. 

    Design a data sketch for $A$ and $L$ supporting the following queries:

    \begin{enumerate}
        \item Insert a new student in the set ($A$ or $L$). We assume that the student does not already belong to the set.
        \item Check if $A = L$. The answer should be correct with high probability (i.e. with probability $1-n^{-c}$ for any desired constant $c$).
        \item Estimate the fraction of students who have passed exactly one module between AMD and LMD (i.e. do not count students that have passed zero or both modules).
        For example, if there are in total $n=30$ students, 7 students have passed only AMD, and 8 students have passed only LMD, then this fraction is $(7+8)/30 = 0.5$.
        \item Focus on the set of students who passed at least one module. Among these students, estimate the fraction of students that passed both modules. For example, if $20$ students passed at least one between AMD and LMD, and $5$ students passed both AMD and LMD, this fraction is $5/20 = 0.25$.
    \end{enumerate}

    Assuming constant absolute error $0 < \epsilon < 1$ and constant failure probability $0<\delta<1$ for queries (3-4), the sketch should use $O(\log n)$ bits of space.
\end{Exercise}

\medskip \begin{Exercise} \noindent
    An embedded device with few KiB of RAM contains a sensor that measures the environment's temperature (an integer number) every 10 seconds. We want to store \emph{all} the generated data, in a streaming fashion, in a sketch using just few Kilobytes  of space (note that the sensor produces several MiB of data per year). At any point in time, the data sketch should be able to answer quickly to the following query, with relative error $\epsilon$: what was the average temperature in a specific time range? (e.g. April 15th 2021, 10:15:40 AM - September 21th 2023, 09:10:20 PM).
\end{Exercise}

\medskip \begin{Exercise} \noindent
We want to store a relational table $t(a_1, a_2, \dots, a_k)$ in a small (polylogarithmic-space in the table's size) sketch supporting the following two queries: (i) insert a new row into the table, and (ii) given $q$ attributes $a_{i_1}, \dots, a_{i_q}$ at query time, estimate with relative error $\epsilon$ and failure probablity $\delta$ the number of distinct rows in the projection $\pi_{a_{i_1}, \dots, a_{i_q}}(t)$. 
\end{Exercise}

\medskip \begin{Exercise} \noindent
We want to store a relational table $t(a_1, a_2, \dots, a_k)$, where columns $a_1, a_2, \dots, a_k$ store integer numbers, in a small (polylogarithmic-space in the table's size) sketch supporting the following two queries (we allow a small error in the answer): (i) insert a new row into the table, and (ii) given an attribute $a_i$ and two integers $\ell\leq r$, estimate the number of rows (not necessarily distinct) such that $\ell \leq a_i \leq r$. 
\end{Exercise}

\medskip \begin{Exercise} \noindent
An embedded device with few KiB of RAM contains a GPS sensor that outputs the absolute position of a train (an integer number) on a railway every second. We want to store \emph{all} the generated data, in a streaming fashion, in a small-space sketch. 
At any point in time, the data sketch should be able to answer quickly to the following query, with relative error $\epsilon$: what was the average speed of the train in a specific time range? (e.g. on April 15th 2021, from 10:15:45 AM to 12:10:22 PM).
\end{Exercise}

\medskip \begin{Exercise} \noindent
Design an algorithm to be implemented on a router with few MiB of memory, to solve the following problem. We have the original code of a huge virus, occupying 1 GiB of space. 
To elude antivirus software, after infecting a new computer the virus modifies itself by swapping some characters from a small subset of fixed positions in the code (let's say, 5 positions) that are unknown to us.
Then, the infected computer sends the modified virus' code on the network. 
Before and after sending the virus' code, the infected computer could send other unrelated data.
Assume that our router receives all packets passing through the network and, for simplicity, assume that IP packets are sorted by sender. The algorithm should be able to estimate how many computers have been infected.
\end{Exercise}

\medskip \begin{Exercise} \noindent
Design an algorithm for a web crawler that explores the web and has the task to identify clusters of web pages being similar to each other, where similarity is measured according to the words contained in the web page. 
\end{Exercise}

\medskip \begin{Exercise} \noindent
You are hired at CIA with the task of preventing terrorist attacks. 
Design an algorithm to be implemented on a router with few KiB of memory, to estimate the number of IP addresses that issued Google searches for both these questions (possibly, the searches do not belong to the same IP packet): (i) where to buy ammonium-nitrate fertilizers, and (ii) how to activate a mechanical relay from distance with a mobile phone.
\end{Exercise}

\medskip \begin{Exercise} \noindent
Design an algorithm to be implemented on a router with 3 GiB of free RAM, to solve the following problem: at any point in time, given an IPv4 address estimate how many IP packets have been sent from that address. 
You can assume that each IP address sends no more than $2^{63}$ IP packets.
The answer must not exceed \emph{relative} error
$82\%$ with probability at least $1/4$. 
\end{Exercise}

\medskip \begin{Exercise} \noindent
Design an algorithm to be implemented on a router with 100 MiB of free RAM, to solve the following problem: at any point in time, given an IPv4 address estimate how many IP packets have been sent from that address. 
You can assume that the number of IP packets sent in total by all addresses does not exceed $2^{32}$.
The answer must not exceed \emph{absolute} error 1000.
\end{Exercise}

\medskip \begin{Exercise} \noindent
Let $c_1, \dots, c_n$ be a list of computers that send IP packets to two servers $A$ and $B$; we see the packets arriving in a streaming fashion as a list of pairs $(c_i,X)$, where $i\in [n]$ and $X \in \{A,B\}$. Assuming that every computer $c_i$ sends at least one packet to either $A$ or $B$, estimate with small absolute error how many computers sent a packet to both $A$ and $B$ (i.e. computers $c_i$ such that both pairs $(c_i,A)$ and $(c_i,B)$ appear on the stream).
\end{Exercise}

\medskip \begin{Exercise} \noindent
Let $c_1, \dots, c_n$ be a list of computers that send IP packets to two servers $A$ and $B$; we see the packets arriving in a streaming fashion as a list of pairs $(c_i,X)$, where $i\in [n]$ and $X \in \{A,B\}$. Estimate with small absolute error how many computers sent a packet to both $A$ and $B$ (i.e. computers $c_i$ such that both pairs $(c_i,A)$ and $(c_i,B)$ appear on the stream). Differently from the previous exercise, there could exist computers that do not send packets to $A$ nor $B$.
\end{Exercise}

\medskip \begin{Exercise} \noindent
A router filters all IP packets --- of the form $(S,M)$ --- from our company's intranet, where $S$ is the sender's IP address, and $M$ is the message. 
We know that two worms $w_1$ and $w_2$ (self-replicating viruses) circulate on the network. The code of each of the two worms fits in $n$ bytes, and we can assume that if an IP packet $(S,M)$ contains a worm's code, the whole worm's code is contained inside the message $M$.
Design a streaming algorithm that uses $O(\log n)$ words of space and solves the following problem. 
Let $W_i^S$, for $i\in \{1,2\}$, be the event ``$S$ has been infected by $w_i$'', where $S$ is a uniform IP address in our intranet. Estimate the correlation
$$
P(W_1^S | W_2^S)
$$
i.e. the probability of being infected by $w_1$ given that we've been infected by $w_2$ where, again, $S$ is a uniform IP address in our intranet.  
\end{Exercise}

\bibliographystyle{plain}
\bibliography{notes}

@book{leskovec2020mining,
  title={Mining of massive data sets},
  author={Leskovec, Jure and Rajaraman, Anand and Ullman, Jeffrey David},
  year={2020},
  publisher={Cambridge university press}
}

@article{misra1982finding,
  title={Finding repeated elements},
  author={Misra, Jayadev and Gries, David},
  journal={Science of computer programming},
  volume={2},
  number={2},
  pages={143--152},
  year={1982},
  publisher={Elsevier}
}

@article{johnson1986extensions,
  title={Extensions of Lipschitz maps into Banach spaces},
  author={Johnson, William B and Lindenstrauss, Joram and Schechtman, Gideon},
  journal={Israel Journal of Mathematics},
  volume={54},
  number={2},
  pages={129--138},
  year={1986},
  publisher={Springer}
}

@inproceedings{park2004reservoir,
  title={Reservoir-based random sampling with replacement from data stream},
  author={Park, Byung-Hoon and Ostrouchov, George and Samatova, Nagiza F and Geist, Al},
  booktitle={Proceedings of the 2004 SIAM International Conference on Data Mining},
  pages={492--496},
  year={2004},
  organization={SIAM}
}

@misc{ChakrabartiNotes,
  title = {{Data Stream Algorithms - Lecture Notes}},
  howpublished = {\url{https://www.cs.dartmouth.edu/~ac/Teach/data-streams-lecnotes.pdf}},
  author={Chakrabarti, Amit},
  note = {Accessed: 2023-01-03}
}

@misc{BoweRRR,
  title = {{RRR: A Succinct Rank/Select Index for Bit Vectors}},
  howpublished = {\url{https://www.alexbowe.com/rrr/}},
  author={Bowe, Alex},
  note = {Accessed: 2024-02-20}
}

@article{mullin1983second,
  title={A second look at Bloom filters},
  author={Mullin, James K},
  journal={Communications of the ACM},
  volume={26},
  number={8},
  pages={570--571},
  year={1983},
  publisher={ACM New York, NY, USA}
}

@misc{pizzachiliText,
  title = {{PizzaChili Corpus}},
  howpublished = {\url{https://pizzachili.dcc.uchile.cl/texts.html}},
  author={Navarro, Gonzalo and Ferragina, Paolo},
  note = {Accessed: 2024-02-19}
}

@inproceedings{mitzenmacher2008simple,
  title={Why simple hash functions work: exploiting the entropy in a data stream.},
  author={Mitzenmacher, Michael and Vadhan, Salil P},
  booktitle={SODA},
  volume={8},
  pages={746--755},
  year={2008},
  organization={Citeseer}
}

@article{GAGIE2006246,
title = {Large alphabets and incompressibility},
journal = {Information Processing Letters},
volume = {99},
number = {6},
pages = {246-251},
year = {2006},
issn = {0020-0190},
doi = {https://doi.org/10.1016/j.ipl.2006.04.008},
url = {https://www.sciencedirect.com/science/article/pii/S0020019006001414},
author = {Travis Gagie}
}

@misc{gregorygundersenblog,
  title = {{Approximate Counting with Morris's Algorithm}},
  howpublished = {\url{http://gregorygundersen.com/blog/2019/11/11/morris-algorithm/}},
  author={Gundersen, Gregory},
  note = {Accessed: 2023-01-03}
}

@inproceedings{flajolet1983probabilistic,
  title={Probabilistic counting},
  author={Flajolet, Philippe and Martin, G Nigel},
  booktitle={24th Annual Symposium on Foundations of Computer Science (sfcs 1983)},
  pages={76--82},
  year={1983},
  organization={IEEE}
}

@inproceedings{durand2003loglog,
  title={Loglog counting of large cardinalities},
  author={Durand, Marianne and Flajolet, Philippe},
  booktitle={European Symposium on Algorithms},
  pages={605--617},
  year={2003},
  organization={Springer}
}

@inproceedings{ferragina2000opportunistic,
  title={Opportunistic data structures with applications},
  author={Ferragina, Paolo and Manzini, Giovanni},
  booktitle={41st Annual Symposium on Foundations of Computer Science, 2000.},
  pages={390--398},
  year={2000},
  organization={IEEE}
}

@article{grandi2018analysis,
  title={On the analysis of bloom filters},
  author={Grandi, Fabio},
  journal={Information Processing Letters},
  volume={129},
  pages={35--39},
  year={2018},
  publisher={Elsevier}
}

@article{bloom1970space,
  title={Space/time trade-offs in hash coding with allowable errors},
  author={Bloom, Burton H},
  journal={Communications of the ACM},
  volume={13},
  number={7},
  pages={422--426},
  year={1970},
  publisher={ACM New York, NY, USA}
}

@inproceedings{CSA03,
author = {Grossi, Roberto and Gupta, Ankur and Vitter, Jeffrey Scott},
title = {High-order entropy-compressed text indexes},
year = {2003},
isbn = {0898715385},
publisher = {Society for Industrial and Applied Mathematics},
address = {USA},
booktitle = {Proceedings of the Fourteenth Annual ACM-SIAM Symposium on Discrete Algorithms},
pages = {841–850},
numpages = {10},
location = {Baltimore, Maryland},
series = {SODA '03}
}

@misc{burrows1994block,
  title={A block-sorting lossless data compression algorithm},
  author={Burrows, Michael and Wheeler, David J},
  year={1994},
  publisher={Citeseer}
}

@inproceedings{heule2013hyperloglog,
  title={Hyperloglog in practice: Algorithmic engineering of a state of the art cardinality estimation algorithm},
  author={Heule, Stefan and Nunkesser, Marc and Hall, Alexander},
  booktitle={Proceedings of the 16th International Conference on Extending Database Technology},
  pages={683--692},
  year={2013}
}

@inproceedings{flajolet2007hyperloglog,
  title={Hyperloglog: the analysis of a near-optimal cardinality estimation algorithm},
  author={Flajolet, Philippe and Fusy, {\'E}ric and Gandouet, Olivier and Meunier, Fr{\'e}d{\'e}ric},
  booktitle={Discrete Mathematics and Theoretical Computer Science},
  pages={137--156},
  year={2007},
  organization={Discrete Mathematics and Theoretical Computer Science}
}

@inproceedings{estan2003bitmap,
  title={Bitmap algorithms for counting active flows on high speed links},
  author={Estan, Cristian and Varghese, George and Fisk, Mike},
  booktitle={Proceedings of the 3rd ACM SIGCOMM conference on Internet measurement},
  pages={153--166},
  year={2003}
}

@misc{nelsonLectures,
  title = {{CS229r: Algorithms for Big Data}},
  howpublished = {\url{http://people.seas.harvard.edu/~minilek/cs229r/fall15/lec.html}},
  author={Nelson, Jelani},
  note = {Accessed: 2023-01-03}
}

@article{datar2002maintaining,
  title={Maintaining stream statistics over sliding windows},
  author={Datar, Mayur and Gionis, Aristides and Indyk, Piotr and Motwani, Rajeev},
  journal={SIAM journal on computing},
  volume={31},
  number={6},
  pages={1794--1813},
  year={2002},
  publisher={SIAM}
}

@misc{gilbertLectures,
  title = {{CS5234 - Algorithms at Scale}},
  howpublished = {\url{https://www.comp.nus.edu.sg/~gilbert/CS5234/}},
  author={Gilbert, Seth},
  note = {Accessed: 2023-01-03}
}

@article{cormode2017data,
  title={Data sketching. \url{https://cacm.acm.org/magazines/2017/9/220427-data-sketching/fulltext}},
  author={Cormode, Graham},
  journal={Communications of the ACM},
  volume={60},
  number={9},
  pages={48--55},
  year={2017},
  publisher={ACM New York, NY, USA}
}

@article{morris1978counting,
  title={Counting large numbers of events in small registers},
  author={Morris, Robert},
  journal={Communications of the ACM},
  volume={21},
  number={10},
  pages={840--842},
  year={1978},
  publisher={ACM New York, NY, USA}
}

@inproceedings{clifford2019streaming,
  title={The streaming k-mismatch problem},
  author={Clifford, Rapha{\"e}l and Kociumaka, Tomasz and Porat, Ely},
  booktitle={Proceedings of the Thirtieth Annual ACM-SIAM Symposium on Discrete Algorithms},
  pages={1106--1125},
  year={2019},
  organization={SIAM}
}

@inproceedings{clifford2016k,
  title={The k-mismatch problem revisited},
  author={Clifford, Rapha{\"e}l and Fontaine, Allyx and Porat, Ely and Sach, Benjamin and Starikovskaya, Tatiana},
  booktitle={Proceedings of the twenty-seventh annual ACM-SIAM symposium on Discrete algorithms},
  pages={2039--2052},
  year={2016},
  organization={SIAM}
}

@inproceedings{porat2009exact,
  title={Exact and approximate pattern matching in the streaming model},
  author={Porat, Benny and Porat, Ely},
  booktitle={2009 50th Annual IEEE Symposium on Foundations of Computer Science},
  pages={315--323},
  year={2009},
  organization={IEEE}
}

@article{breslauer2014real,
  title={Real-time streaming string-matching},
  author={Breslauer, Dany and Galil, Zvi},
  journal={ACM Transactions on Algorithms (TALG)},
  volume={10},
  number={4},
  pages={1--12},
  year={2014},
  publisher={ACM New York, NY, USA}
}

@inproceedings{gupta2014pattern,
  title={Pattern matching algorithms for intrusion detection and prevention system: A comparative analysis},
  author={Gupta, Vibha and Singh, Maninder and Bhalla, Vinod K},
  booktitle={2014 International Conference on Advances in Computing, Communications and Informatics (ICACCI)},
  pages={50--54},
  year={2014},
  organization={IEEE}
}

@inproceedings{alon1999tracking,
  title={Tracking join and self-join sizes in limited storage},
  author={Alon, Noga and Gibbons, Phillip B and Matias, Yossi and Szegedy, Mario},
  booktitle={Proceedings of the eighteenth ACM SIGMOD-SIGACT-SIGART symposium on Principles of database systems},
  pages={10--20},
  year={1999}
}

@inproceedings{dobra2002processing,
  title={Processing complex aggregate queries over data streams},
  author={Dobra, Alin and Garofalakis, Minos and Gehrke, Johannes and Rastogi, Rajeev},
  booktitle={Proceedings of the 2002 ACM SIGMOD international conference on Management of data},
  pages={61--72},
  year={2002}
}

@article{alon1999space,
  title={The Space Complexity of Approximating the Frequency Moments},
  author={Alon, Noga and Matias, Yossi and Szegedy, Mario},
  journal={Journal of Computer and System Sciences},
  volume={58},
  pages={137--147},
  year={1999}
}

@book{navarro2016compact,
  title={Compact data structures: A practical approach},
  author={Navarro, Gonzalo},
  year={2016},
  publisher={Cambridge University Press}
}

@book{stojmenovic2007handbook,
  title={Handbook of applied algorithms: Solving scientific, engineering, and practical problems. Chapter 8: Algorithms for Data Streams. \url{http://www.dei.unipd.it/~geppo/PrAvAlg/DOCS/DFchapter08.pdf}},
  author={Stojmenovic, Ivan and Nayak, Amiya},
  year={2007},
  publisher={John Wiley \& Sons}
}

@article{alon1999linear,
  title={Linear hash functions},
  author={Alon, Noga and Dietzfelbinger, Martin and Miltersen, Peter Bro and Petrank, Erez and Tardos, G{\'a}bor},
  journal={Journal of the ACM (JACM)},
  volume={46},
  number={5},
  pages={667--683},
  year={1999},
  publisher={ACM New York, NY, USA}
}

@inproceedings{knudsen2016linear,
  title={Linear hashing is awesome},
  author={Knudsen, Mathias B{\ae}k Tejs},
  booktitle={2016 IEEE 57th Annual Symposium on Foundations of Computer Science (FOCS)},
  pages={345--352},
  year={2016},
  organization={IEEE}
}

@inproceedings{bender2011don,
  title={Don't thrash: how to cache your hash on flash},
  author={Bender, Michael A and Farach-Colton, Martin and Johnson, Rob and Kuszmaul, Bradley C and Medjedovic, Dzejla and Montes, Pablo and Shetty, Pradeep and Spillane, Richard P and Zadok, Erez},
  booktitle={3rd Workshop on Hot Topics in Storage and File Systems (HotStorage 11)},
  year={2011}
}

@article{rabin1981fingerprinting,
  title={Fingerprinting by random polynomials},
  author={Rabin, Michael O},
  journal={Technical report},
  year={1981},
  publisher={Center for Research in Computing Technology, Harvard University}
}

@inproceedings{broder2000min,
  title={Min-wise independent permutations: Theory and practice},
  author={Broder, Andrei Z},
  booktitle={International Colloquium on Automata, Languages, and Programming},
  pages={808--808},
  year={2000},
  organization={Springer}
}

@inproceedings{rajeev2002succinct,
  title={Succinct indexable dictionaries with applications to encoding k-ary trees and multisets},
  author={Raman, Rajeev and Raman, Venkatesh and Rao, S Srinivasa},
  booktitle={Proceedings of the thirteenth annual ACM-SIAM symposium on Discrete algorithms},
  pages={233--242},
  year={2002}
}

@InProceedings{10.1007/3-540-45061-0_30,
author="Raman, Rajeev
and Rao, Satti Srinivasa",
editor="Baeten, Jos C. M.
and Lenstra, Jan Karel
and Parrow, Joachim
and Woeginger, Gerhard J.",
title="Succinct Dynamic Dictionaries and Trees",
booktitle="Automata, Languages and Programming",
year="2003",
publisher="Springer Berlin Heidelberg",
address="Berlin, Heidelberg",
pages="357--368",
isbn="978-3-540-45061-0"
}

@article{indyk2001small,
  title={A small approximately min-wise independent family of hash functions},
  author={Indyk, Piotr},
  journal={Journal of Algorithms},
  volume={38},
  number={1},
  pages={84--90},
  year={2001},
  publisher={Elsevier}
}

@article{patracscu2012power,
  title={The power of simple tabulation hashing},
  author={Pǎtra{\c{s}}cu, Mihai and Thorup, Mikkel},
  journal={Journal of the ACM (JACM)},
  volume={59},
  number={3},
  pages={1--50},
  year={2012},
  publisher={ACM New York, NY, USA}
}

@inproceedings{dahlgaard2017fast,
  title={Fast similarity sketching},
  author={Dahlgaard, S{\o}ren and Knudsen, Mathias B{\ae}k Tejs and Thorup, Mikkel},
  booktitle={2017 IEEE 58th Annual Symposium on Foundations of Computer Science (FOCS)},
  pages={663--671},
  year={2017},
  organization={IEEE}
}

\end{document}